\documentclass[10pt,superscriptaddress,compress]{elsarticle}
\usepackage{geometry}
\usepackage{subcaption}
\usepackage[english]{babel}
\usepackage{hyperref}
\usepackage{array}
\usepackage{graphicx}
\usepackage{tcolorbox}
\usepackage{quantikz}
\usepackage{tabularx}
\usepackage{tikz}
\usetikzlibrary{shapes, arrows, calc, positioning, bending}
\usepackage{pgfplots}
\pgfplotsset{compat=1.18}
\usepackage{amsthm,amsmath,braket,amssymb,bm}
\usepackage[normalem]{ulem}
\newtheorem{definition}{Definition}
\newtheorem{lemma}{Lemma}
\newtheorem{proposition}{Proposition}
\newtheorem{theorem}{Theorem}

\newtheorem{remark}{Remark}
\usepackage[utf8]{inputenc}
\usepackage{booktabs} 
\newcommand{\LC}{\mathcal{L}}
\newcommand{\QC}{\mathcal{Q}}
\newcommand{\NSC}{\mathcal{NS}}

\newcommand{\blk}{\color{black}}

\newcounter{procedurecount}
\newcommand{\protocoltitle}{}
\newenvironment{procedurelist}[1]{%
	\renewcommand{\protocoltitle}{#1}%
	\begin{center}%
		\setlength{\parskip}{0pt}%
		\hrule\vspace{2pt}\hrule\vspace{4pt}%
		\emph{\protocoltitle} \\[1ex] 
		\hrule
		\begin{list}{\arabic{procedurecount}.}{%
				\usecounter{procedurecount}%
				\setlength{\leftmargin}{1em}%
				\setlength{\itemindent}{0pt}%
				\setlength{\labelwidth}{0pt}%
				\setlength{\labelsep}{0.5em}%
				\setlength{\itemsep}{0pt}%
				\setlength{\parsep}{0pt}%
			}%
		}{%
		\end{list}%
		\vspace{4pt}\hrule\vspace{2pt}\hrule%
	\end{center}%
}
\AtBeginEnvironment{thebibliography}{\scriptsize} 
\journal{Physics Reports}

\begin{document}

\begin{frontmatter}
\title{The future of secure communications: device independence in quantum key distribution}

\author[bratislava,brno]{Seyed Arash Ghoreishi}
\ead{arash.ghoreishi@savba.sk}
\author[bari1,bari2,gdansk]{Giovanni Scala}
\ead{giovanni.scala@poliba.it}
\author[zurich]{Renato Renner}
\ead{renner@ethz.ch}
\author[chile1,chile2]{Letícia Lira Tacca}
\ead{letacca@udec.cl}
\author[brno]{Jan Bouda}
\ead{bouda@fi.muni.cz}
\author[chile1,chile2]{Stephen Patrick Walborn}
\ead{swalborn@udec.cl}
\author[gdansk]{Marcin Paw\l owski}
\ead{marcin.pawlowski@ug.edu.pl}
\address[bratislava]{
RCQI, Institute of Physics, Slovak Academy of Sciences, D\'ubravsk\'a cesta 9, 84511 Bratislava, Slovakia}
\address[brno]{Faculty of Informatics, Masaryk University, Botanická 68a, 602 00 Brno, Czech Republic}
\address[bari1]{Dipartimento Interateneo di Fisica, Politecnico di Bari, 70126 Bari, Italy}
\address[bari2]{INFN, Sezione di Bari, 70126 Bari, Italy}
\address[gdansk]{International Centre for Theory of Quantum Technologies, University
of Gda\'nsk, Jana Ba\.zy\'nskiego 1A, Gda\'nsk, 80-309, Poland}
\address[zurich]{Institute for Theoretical Physics, ETH Zürich, 8093 Zürich, Switzerland}
\address[chile1]{Departamento de Fisica, Universidad de Concepci\'{o}n, Concepci\'{o}n, B\'{\i}o-B\'{\i}o, Chile}
\address[chile2]{Millennium Institute for Research in Optics, Universidad de Concepci\'{o}n, Concepci\'{o}n, B\'{\i}o-B\'{\i}o, Chile}

\begin{abstract}
In the ever-evolving landscape of quantum cryptography, Device-independent Quantum Key Distribution (DI-QKD) stands out for its unique approach to ensuring security based not on the trustworthiness of the devices but on nonlocal correlations. Beginning with a contextual understanding of modern cryptographic security and the limitations of standard quantum key distribution methods, this review explores the pivotal role of nonclassicality and the challenges posed by various experimental loopholes for DI-QKD. Various protocols, security against individual, collective and coherent attacks, and the concept of self-testing are also examined, as well as the entropy accumulation theorem, and additional mathematical methods in formulating advanced security proofs. In addition, the burgeoning field of semi-device-independent models (measurement DI--QKD, Receiver DI--QKD, and One--sided DI--QKD) is also analyzed. 
The practical aspects are discussed 
through a detailed overview of experimental progress and the open challenges towards the commercial deployment 
in the future of secure communications. 
\end{abstract}

\end{frontmatter}

{\footnotesize \tableofcontents}
\section{Introduction}\label{chap1}
\subsection{Overview of modern cryptography}
\textit{Facta lex inventa fraus} --- the principle that every established law is followed by the invention of a way to circumvent it --- does not hold, \textit{in theory}, for modern cryptography. With the advent of Quantum Cryptography \cite{Gisin2002, Pirandola2020, Zapatero2023}, the security of communication protocols has shifted from complex, yet vulnerable algorithms, to fundamental quantum principles (uncertainty, entanglement, complementarity, no-cloning, non-locality, etc), providing a mechanism for inherently secure communication channels. 
The pioneering application of these revolutionary techniques has led to Quantum Key Distribution (QKD), (and beyond~\cite{Bozzio2024}), representing a significant leap forward in security compared to traditional public-key standards such as Diffie-Hellman \cite{Diffie1976} and RSA (Rivest--Shamir--Adleman)\cite{Rivest1978}. 
\textit{Facta lex inventa fraus} --- the principle that every established law is followed by the invention of a way to circumvent it --- does not hold, \textit{in theory}, for modern cryptography.  
While traditional public-key cryptography employs the concept of one-way functions to encrypt/decrypt information, QKD detects the potential intrusion of an eavesdropper, Eve, due to the principle of no-cloning or, equivalently, through the uncertainty principle.  Fig. \ref{fig:traditionalvsQKD} visually compares RSA \footnote{In Fig. \ref{fig:AES}, one might think spectroscopy could reveal each purple's components, but the process resembles a one-way cryptographic function: easy to mix, hard to reverse. Like password hashing, the color mixture conceals the original inputs, preventing unauthorized decryption.} with the \textit{prepare-and-measure} scenario of BB84
 \cite{Bennett2014} (emblematic of QKD).
The players' preparation (in Alice's lab $\mathcal{A}$) and measurements (in Bob's lab $\mathcal{B}$) can be the red or blue buttons. While in the classical case, Eve has no button (she can make perfect copies of the encoded messages sent and manipulate the copies however she chooses), in QKD Eve must choose and perform a measurement to obtain information.  If Alice selects red and Eve selects blue, then the effect of Eve's disturbance appears in Bob's measurements as a purple ball (a mix of red and blue). 
\begin{figure}
    \centering
    \begin{subfigure}{.49\textwidth}
        \centering
        \includegraphics[width=\linewidth,trim=0in 6.5in 0.in 1.in,clip]{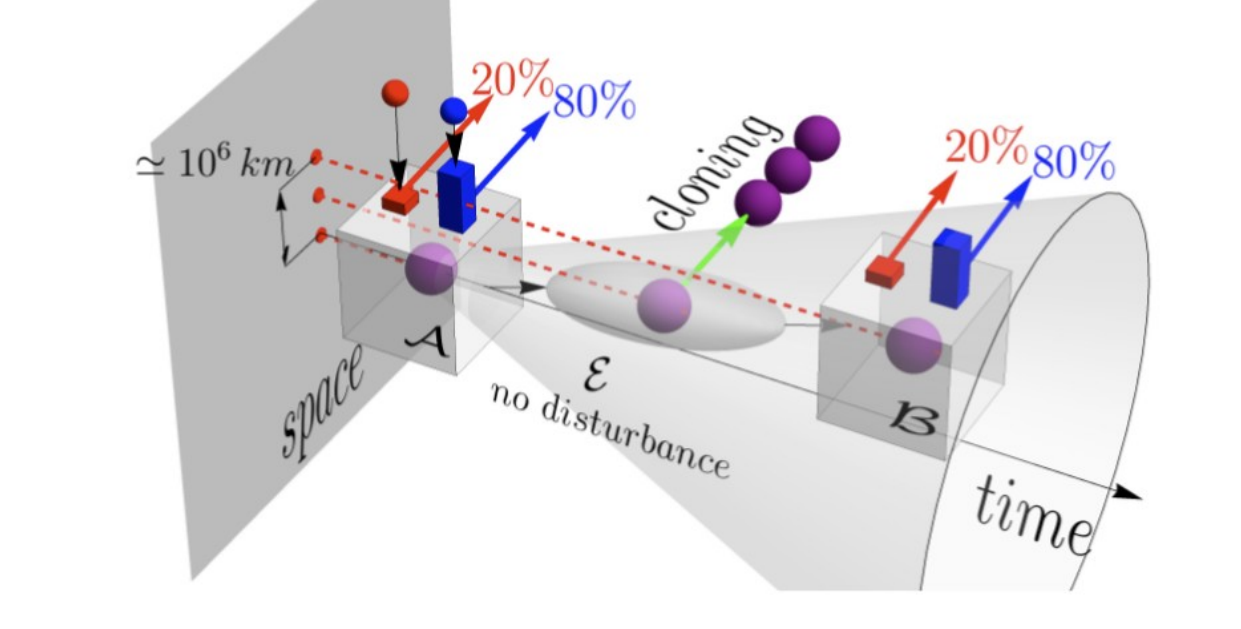}
        \caption{\textit{One--way function in classical public-key encryption}}
        \label{fig:AES}
    \end{subfigure}
    \hfill
    \begin{subfigure}{.49\textwidth}
        \centering
        \includegraphics[width=\linewidth,trim=0in 6.5in 0.in 1.in,clip]{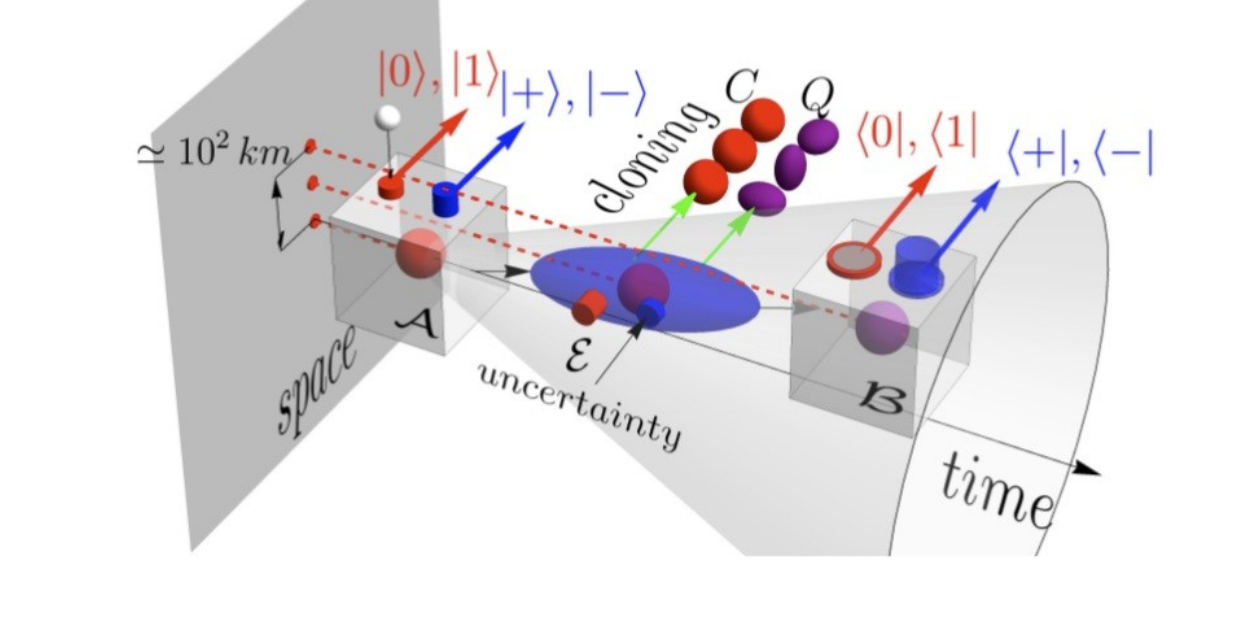}
        \caption{\textit{Quantum key distribution (BB84).}}
        \label{fig:QKDBB84}
    \end{subfigure}
    \caption{\textit{Illustration on the fundamental physical principles behind the need of quantum cryptography} -- 
    In Fig. (\ref{fig:AES}) a color-mixing analogy represents the encoding in public-key cryptography 
    as a purple sphere, symbolizing an encrypted message open to all. Yet, only holders of the private key can accurately decrypt it. Alice creates the purple sphere with a specific combination of colors (20\% red, 80\% blue) mimicking a one way function (Eve cannot perfectly decompose the purple shade into component colors).  Bob, having some information about the precise mix (the private key), can decrypt it.  In Fig. (\ref{fig:QKDBB84}) quantum cryptography. Colors represent the basis (red $\{\ket{0}, \ket{1}\}$, blue $\{\ket{+}, \ket{-}\}$). Due to no-cloning, Eve's interference changes the color and shape of the ball. If Alice uses the red button and Eve guesses the blue button, the result in Bob's box is purple. Contrary to classical cryptography, in the quantum case, Eve's intrusion affects the outcome at Bob's station. Bob detecting purple with a red button, signals Eve's presence.
    Traditional cryptography and QKD protocols are realized in the same causal cone at today's distances, which are on the order of  $\sim10^6\mathrm{km}$ for classical techniques, while QKD can reach $\sim10^2\mathrm{km}$ for fiber-based schemes without quantum repeaters.  For satellite QKD, much longer distances on the order of $\sim10^3-10^4\mathrm{km}$ have been achieved \cite{Liao2017Sat, Li2025Sat}    
}
    \label{fig:traditionalvsQKD}
\end{figure}
Not only has the security of many public-key ciphers never been formally proven, but it is also well known that many of those currently in use are vulnerable to quantum computers~\cite{Gidney2021}. 
QKD, on the other hand, creates robust encryption methods based on Kerckhoffs's principle \cite{Petitcolas2011}, which posits that a cryptosystem's security should be maintained even if everything about it is public knowledge, except the secret key. However, theoretical security and practical security are different issues. Experiments have demonstrated vulnerabilities in QKD systems \cite{Qi2005,Lamas-Linares:07,zhao08,lydersen10,wiechers11,qin18}, raising questions 
about whether the theory or the experiments of QKD need refinement. Claiming that the principles of quantum theory are fundamentally flawed would be an exaggeration. The real issues predominantly reside in the practical implementation of QKD. While the theoretical basis of QKD is robust, its real-world application involves new technology, such as single-photon detectors, and can be compromised by unavoidable imperfections in the devices.
\begin{quote}
    ``Theory and experiment are the same, in theory, but experimentally different.''
    
        \tiny{\textit{(The Yale Literary Magazine, Feb, 1882, B. Brewster)}}
    
\end{quote}
 
Therefore, to truly ensure security at the paranoid level under QKD necessitates addressing an additional layer of scrutiny: the independence from underlying devices \cite{Primaatmaja2023}. This prerequisite gave rise to the concept of Device-Independent Quantum Key Distribution (DI-QKD). 
DI-QKD, and its slightly more lenient version, Semi-Device-Independent (SDI) QKD, ensure security based solely on the principles of quantum mechanics, without dependence on the specifics of the hardware used. Thus, potential vulnerabilities or backdoors due to device malfunctions or imperfections are eliminated, providing a robust mechanism for secure communication.   

\subsection{From classical to quantum cryptography}
Classical cryptography is broadly categorized into two main types: secret (or symmetric) key cryptography and public-key (asymmetric) cryptography. In secret key cryptography, a single key is employed for both encrypting and decrypting messages, exemplified by the one-time pad (OTP) \cite{Singh2011,tutorial} or the Advanced Encryption Standard (AES) \cite{daemen1999aes}. The OTP can achieve perfect information-theoretic security against adversaries with unlimited computational power, as discussed in Ref. \cite{Gisin2002}. Moreover, the threat of quantum attacks on AES requires only doubling the size of the key to achieve equivalent levels of security \cite{Bernstein2017}\footnote{Known quantum attacks on AES use Grover's search algorithm, which provides quadratic speedup \cite{Grover1996}. Thus, to achieve ``quantum-safe" security equivalent to AES256 (256 key bits) under classical attacks requires upgrading to AES512.}. The primary challenge with symmetric cryptography lies in the secure distribution of the secret key prior to communication.   
\par
Public-key cryptography, such as the Diffie-Hellman \cite{Diffie1976} and RSA \cite{Rivest1978} protocols, circumvents this issue by employing a pair of keys for each participant: a public key, which can be shared openly, and a private key, which remains confidential. This enables Alice to encrypt a message using Bob's public key, ensuring that only Bob can decrypt it with his private key. This eliminates the need to exchange secret keys in advance.  Importantly, public-key cryptography also provides a mean for authentication: Bob can sign a message with his private key, and Alice (or anyone) can use the public key to confirm that it was indeed signed by Bob. From a practical standpoint, public-key systems  are slower in that they require larger keys and more communication between users, compared to symmetric encryption.  Thus, in current communication protocols such as Transport Layer Security (TLS), for example, a public-key method is used for authentication and key exchange in an initial handshaking session, while subsequent data encryption employs symmetric encryption \cite{TLS12,TLS13}. 
\par
Nowadays, the security of many public key ciphers is built on the computational difficulty of mathematical problems like \textit{integer factorization} or the \textit{discrete logarithm problem}, making it potentially vulnerable to advances in quantum computing. Notably, algorithms capable of solving these problems in polynomial time on a quantum computer have already been proposed~\cite{Shor1997,Bruss2007}.
It is thus through quantum mechanics that  \textit{Facta lex inventa fraus} is realized, through the emergence of quantum computing as a significant threat, (the \textit{fraus}) to the security of current public key cryptography. Still, quantum physics itself offers a new and robust set of laws (\textit{lex}) through QKD, capable of providing unconditionally secure key distribution in theory. 
\subsection{Standard Quantum Key Distribution}
\subsubsection{Theoretical security of QKD}\label{subsec:QKD}
To introduce standard prepare and measure QKD, we specifically elaborate the BB84 protocol \cite{tutorial} or conjugate coding \cite{Wiesner1983}, sketched in Fig. \ref{fig:QKDBB84}. A general protocol for prepare-and-measure (PM) QKD can be found in Box 1. 

\textit{Step 1} --(Data generation) 
Alice prepares eigenstates of 
$\sigma_z$ or $\sigma_x$ bases (red or blue of Fig. \ref{fig:traditionalvsQKD}) and attached them with a classical register. Then the classical-quantum preparation is
\begin{equation}    \psi_{C_{A_i}Q_i}=\frac14\sum_{a_i,x_i\in\{0,1\}}\ket{x_ia_i}\bra{x_ia_i}\otimes H^{x_i} \ket{a_i}\bra{a_i}H^{x_i},
\end{equation}
where the first system corresponds to her classical register, storing values of classical bits $x_i,a_i$ ( $x_i,a_i\in\{0,1\})$.  The second system is the quantum state $\psi_{Q_i}$, which she sends to Bob. Here $H$ is the Hadamard matrix, such that $H^{0}=\bm 1$ is identity.  
\par
Bob is unaware of Alice's input $x_i$, so he randomly selects a measurement basis $y_i$ and obtains result $b_i$ (Here $y_i,b_i\in\{0,1\})$. To each result he attaches a random classical bit $T_i$, so that with probability $p(T_i=1)=\gamma$, $b_i$ will be used for security check ($T_i=1$), else it will be used to generate final key ($T_i=0$). The classical-quantum state describing his measurement result is
\begin{equation}    
N_{b_i|y_i}^{T_i}=\ket{T_i}\bra{T_i}\otimes\frac12 H^{y_i}\ket{b_i}\bra{b_i}H^{y_i}.
\end{equation}
\textit{Step 2} -- (Public discussion and raw key) Alice and Bob must partially compare preparation and measurement results stored in classical registers $C_{A_i}=(x_i,a_i)\in \mathcal{C}_{A_i}$ and $C_{B_i}=(y_i,b_i,t_i)\in \mathcal{C}_{B_i}$, respectively. To do so, Bob publicly announces $(y_i,t_i)$  (but not $b_i$) so that Alice can inform Bob in which rounds $x_i = y_i$, so that Alice can define a raw key bit $\kappa_{A_i}=a_i$ and Bob $\kappa_{B_i}=b_i$.
When $x_i \neq y_i$, both parties discard their results, defining null bits $\kappa_{A_i}=\kappa_{B_i}=\mathrm{Null}$.    
Provided no errors occurred or no one manipulated the qubits sent, Bob has a string of bits identical to Alice's: $\kappa_B=\kappa_A=\{\kappa_{A_i}\neq \mathrm{Null}\}_i$.
\par
\textit{Step 3} -- (Error correction and Security Check) Both noise and/or intrusion by Eve will produce errors in Bob's bit string $\kappa_B$.  To correct them, Alice and Bob publicly communicate $k^\mathrm{EC}$ for error correction (cascade, LDPC, parity check) on their key bits ($T_i=0$). Let us say that is Bob to perform a security check on results with $T_i=1$. 
For each result, he defines errors using (see Box 1)
\begin{equation}
    c_i = 
    \begin{cases}
        \mathrm{Null} & \text{if } x_i \neq y_i \lor T_i = 0, \text{ no useful check}, \\
        1 & \text{if } x_i = y_i \land T_i = 1, \text{ check passed}, \\
        0 & \text{else, check failed}.
    \end{cases}
\end{equation}
Then, Alice and Bob can estimate $Q_i$'s, ($i\in\{Z,X\}$), the Quantum Bit Error Rate (QBER) for each basis,

\begin{align}
\label{eq:QBER}
Q_Z =
  \frac{\bigl|\{\,i \mid x_i = y_i = 0,\; c_i = 0\}\bigr|}
       {\bigl|\{\,i \mid x_i = y_i = 0,\; c_i \neq \text{Null}\}\bigr|},\quad
Q_X =
  \frac{\bigl|\{\,i \mid x_i = y_i = 1,\; c_i = 0\}\bigr|}
       {\bigl|\{\,i \mid x_i = y_i = 1,\; c_i \neq \text{Null}\}\bigr|}.
\end{align}
Whenever a protocol involves only a single key-generation basis, we denote its QBER by $Q$, omitting the subscripts $Z$ and $X$.
A QBER below a threshold indicates minimal interference or eavesdropping, so Alice and Bob can agree under a certain level of confidence that they final keys $\{k_{A_i}=k_{B_i}\}_{i|T_i=0}$ are correctly distributed and the technique for the step 3.4 discussed in the tutorial can be applied \cite{tutorial}. This refined key, now highly secure, is suitable for encrypting messages (see Sec. \ref{sec:chap4} to bound Eve's knowledge about the key). 

\begin{tcolorbox}[colframe=black, colback=gray!10, title= Box 1: General QKD prepare-and-measure protocol]
\noindent The most general QKD prepare-and-measure protocol can be defined as \cite{Metger2024}: \\
		\textit{-- Data generation:} for $i=1,\dots ,n$, where $n$ is the number of rounds Alice prepares $\psi_{C_A^nQ^n}=\psi_{C_AQ}^{\otimes n}$ and stores in a classical register $(x_i,a_i)\in\mathcal{C}_{A_i}$ her $i-$th preparation label by $a_i$ and setting $x_i$. She sequentially sends  $\psi_{Q_i}$ via a public channel to Bob; Bob chooses $y_i$ and measures $N_{y_i}=\{N_{b_i|y_i}\}_{b_i=1}^{d_B}$ storing  $(y_i,b_i) \in \mathcal{C}_{B_i}$ at each round, where $b_i$ labels one of the $d_B$ possible outcomes.\\
    \textit{-- Public discussion for the raw key generation:} Alice and Bob publicly exchange information, i.e. $\mathrm{PD}:\mathcal{C}_{A_i}\times \mathcal{C}_{B_i}\mapsto T_i$ with $\iota_i=\mathrm{PD}((x_i,a_i),(y_i,b_i))$ such that Alice can compute the raw key $\kappa_A=\{\kappa_{A_i}\}_i$, with $\kappa_{A_i}=\mathrm{RK}((x_i,a_i),\iota_i)\in\mathcal{S}_i$ via $\mathrm{RK}:\mathcal{C}_{A_i}\times T_{i}\mapsto \mathcal{S}_i$.\\
    \textit{-- Post-processing:} The players exchange a string $\kappa^\mathrm{EC}\in\{0,1\}^{\lambda_{EC}}$ to define the final key $k_A=k_B$ via
    \begin{enumerate}
        \item \textit{Error correction:} the players exchange $\kappa^\mathrm{EC}$ from $\mathcal{C}_{A_i},\mathcal{C}_{B_i}$ and $T_i$ so that Bob compute $\kappa_B=\{\kappa_{B_i}\}_i\in \mathcal{S}$ where $ \kappa_{B_i}(\kappa^\mathrm{EC}_i,(y_i,b_i),\iota_i)\in \mathcal{S}_i$.
        \item \textit{Raw key validation:} for $\varepsilon_{\mathrm{KV}}>0$ Alice chooses an universal hash function $\mathrm{HASH}:\mathcal{S}\mapsto\{0,1\}^{\lceil -\log\varepsilon_\mathrm{KV}\rceil}$ and publishes a description of it and the value $\mathrm{HASH}(\kappa_A)$. Bob computes $\mathrm{HASH}(\kappa_B)$ and if $\mathrm{HASH}(\kappa_B)\neq \mathrm{HASH}(\kappa_A)$ the protocol aborts.
        \item \textit{Statistical security check:} Bob sets $\mathrm{EV}:\mathcal{C}_{B_i}\times T_i\times\mathcal{S}_i\mapsto \mathcal{C}\ni q_{B_i}$. Bob then computes $q_B=\mathrm{CA}(\mathrm{freq}(q_B))$, where $\mathrm{CA}$ is an affine function corresponding to collective attack bound $q_{\mathrm{CA}}$. If the required amount of single-round entropy generation is $q_B<q_{\mathrm{CA}}$, he aborts the protocol.
        \item \textit{Privacy amplification:} Alice and Bob respectively have $\kappa_A,\kappa_B\in\{0,1\}^m$. Alice chooses a seed $\mu\in\{0,1\}^m$ uniformly at random and publishes her choice. Alice and Bob independently compute $\ell-$bit string $k_A=\mathrm{EXT}(\kappa_A,\mu)$ and $k_B=\mathrm{EXT}(\kappa_B,\mu)$ where $\mathrm{EXT}:\{0,1\}^m\times \{0,1\}^m \mapsto \{0,1\}^\ell$ is a quantum-proof strong extractor. 
    \end{enumerate}
    \end{tcolorbox}
Notably, in BB84 the no-cloning theorem~\cite{Wootters1982} prohibits the duplication of quantum states (we represent no-cloning in Fig. \ref{fig:traditionalvsQKD} as a not perfect copy process giving deformated spheres), ensuring that any attempt by Eve to intercept and replicate the qubits would introduce detectable errors. Additionally, if Eve measures a qubit without knowing the correct basis, (only Alice knows $x$), the original information encoded in the other basis is irreversibly lost due to the uncertainty principle. Consequently, any eavesdropping attempt increases the QBER, which can be measured by Alice and Bob. Specifically, after comparing $m=|\{c_{i}|c_i=1\lor c_i=0\}_{i}|$ bits, the probability that Eve can eavesdrop without being detected drops to $(3/4)^m$. Here the 3/4 probability of a single correct guess follows from the 1/4 error probability in an  intercept/resend attack (see below): there is a 50\% chance that Eve guesses the correct basis (same Alice's color in Fig. \ref{fig:traditionalvsQKD}), and thus a 50\% chance that Eve guesses wrong, and within those wrong cases, there is another 50\% chance that Bob's measurement will yield an incorrect result (Alice's color differs from Eve and Bob's color in Fig. \ref{fig:traditionalvsQKD}).
The theoretical security proofs depend on Eve's ability to perform \textit{(i) individual attacks}, measuring states separately; \textit{(ii) collective attacks}, measuring individually with joint classical post-processing; \textit{(iii) coherent attacks}, using joint quantum measurements on all states stored in a memory.
\subsubsection{Implementation issues and Quantum Hacking}\label{subsec:implementationissues}
As is traditionally advertised in regards to QKD, any attempt by Eve to uncover information of the key results in an increase in the QBER.  A simple and straightforward example are
{Intercept-and-Resend Attacks}, where Eve intercepts $\psi_{Q_i}$ sent by Alice, measures it in a chosen basis, prepares a new photon state based on her measurement result and sends it to Bob \cite{Barnett1993,Nguyen2004,EZZAHRAOUY2009}.  Since she chooses the wrong basis some of the time, her disturbance increases the QBER and is thus detectable. Evaluation of the QBER can give an upper bound for the amount of information Eve has about the key. Thus,  
QKD can be \textit{theoretically} secure.  However, even in a noise-free scenario, the difference between theory and practice can result in vulnerabilities.  Indeed, if $\psi_{Q_i}\in\mathcal{H}$ with uncontrolled $\dim \mathcal{H}$ and no Bell Inequality (BI) violation is measured, QKD is insecure, because the same BB84 correlations ( \( p(a_i=b_i|x_i=y_i)=1 \) and \( p(a_i \neq b_i|x_i \neq y_i)=1/2 \)) produced by $\ket{\psi}=(\ket{00}+\ket{11})/\sqrt{2}\in \mathbb{C}^2\otimes\mathbb{C}^2$ are also reproduced by a four-qubit separable state \cite{Acin2006b}, ( \cite{magniez2006self} in app. A), 
\[
\rho=\frac{1}{4}\left( \ket{00}\bra{00} + \ket{11}\bra{11}\right)\otimes\left( \ket{++}\bra{++} + \ket{--}\bra{--}\right),
\]
when Alice measures the first (third) qubit in the \(\sigma_z\) (\(\sigma_x\)) basis, and Bob measures the second (fourth) qubit in the  \(\sigma_z\) (\(\sigma_x\)) basis. 
As there is no quantum correlation, a secure key cannot be established.
\par
While this type of quantum state manipulation might seem to give too much power to Eve, it is indeed true that operational imperfections present considerable opportunities for hacking \cite{SCARANI2014,Okula2024}. For example, Eve can exploit the fact that weak coherent pulses (WPCs),  used in some QKD systems, can contain more than one photon to implement the {Photon Number Splitting} (PNS) attack.   By separating and storing one of the photons from a WCP, Eve can measure it later, once the measurement basis has been publicly announced. In this way, she can obtain full information without disturbing the state of the photons sent to Bob \cite{Huttner1995,Luetkenhaus2002,Wang2005}.  Other examples include  side-channel~ \cite{standaert2010introduction,braunstein2012side,Baliuka2023,Zhang2022b}, trojan horse~\cite{Gisin2006,Jain2014}, and device calibration~\cite{Fung2007,zhao08,Xu2010,Jouguet2013,Mao2020,Emde2024} attacks (for a full list see Ref. \cite{tutorial}). To effectively counter these vulnerabilities, the best approach is to use security proofs based on minimal principles and strategies that reduce or eliminate reliance on trusted components. Among these, DI-QKD stands out as the ultimate solution.

\subsection{Overview of Device-independent QKD}
The internal workings and security of the quantum devices involved in QKD protocols, as we analyzed, are often faulty and vulnerable to quantum hacking.
 DI-QKD represents a significant advance in that it aims to ensure the utmost security of QKD, irrespective of the reliability or trustworthiness of the devices used. 
This security is achieved through nonlocal correlations verified by BI violation, as depicted in Fig. \ref{fig:diqkd}. In general, the correlations, or \textit{behaviors} are indicated as points $\bm{p}=\{p(ab|xy)\}_{a,b,x,y}$ in the convex correlation space characterized by the regions $\mathcal{L}\subset\mathcal{Q}\subset\mathcal{NS}$ respectively for local and realistic, quantum, and no-signaling behaviors, respectively. A BI violation ($\bm{p}\not\in \mathcal{L}$), classified as ``strongly nonclassical'' \cite{Spekkens2016}, implies one of two possibilities, or both: (\textit{1}) $a$ and $b$ are determined only when observed; (\textit{2}) a nonlocal influence ensures that the key is established solely through interactions between the trusted parties. In either case,  Eve cannot access the information without being detected because any interference would deviate from the expected nonlocal correlations.
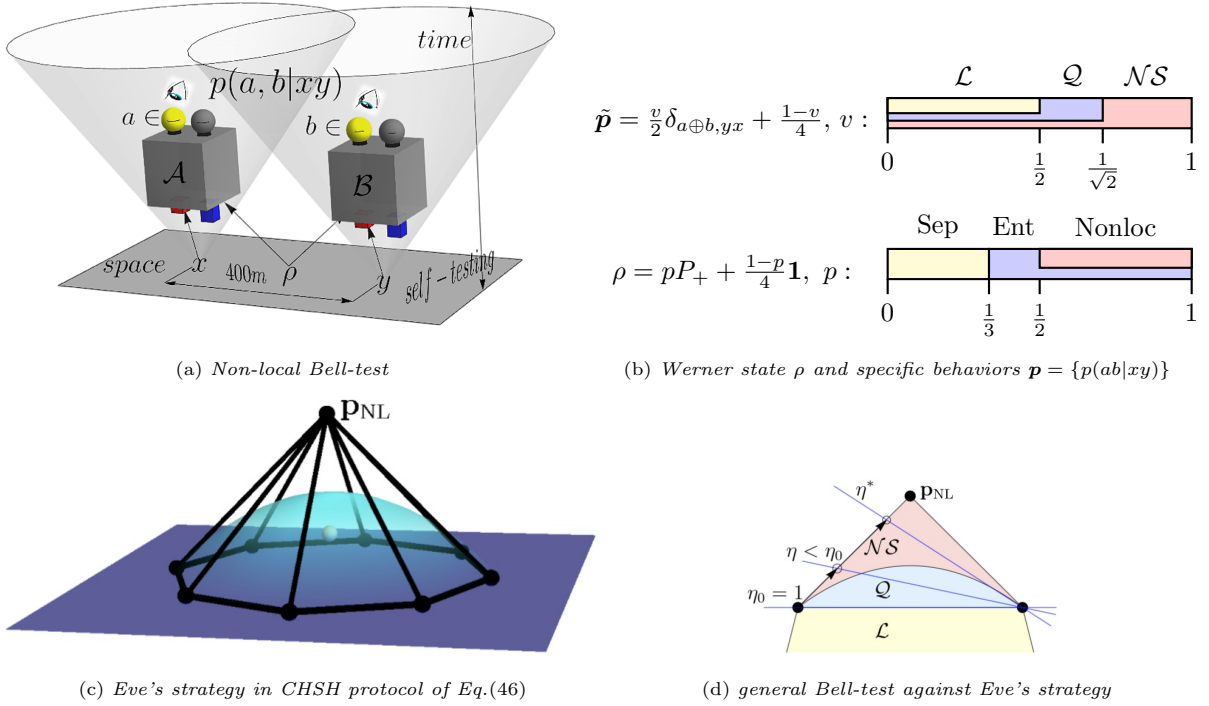
\begin{figure}[h]
    \centering
    \begin{subfigure}{.47\textwidth}
        \centering
        \includegraphics[width=\linewidth,trim=.2in 6.95in .8in 2.1in,clip]{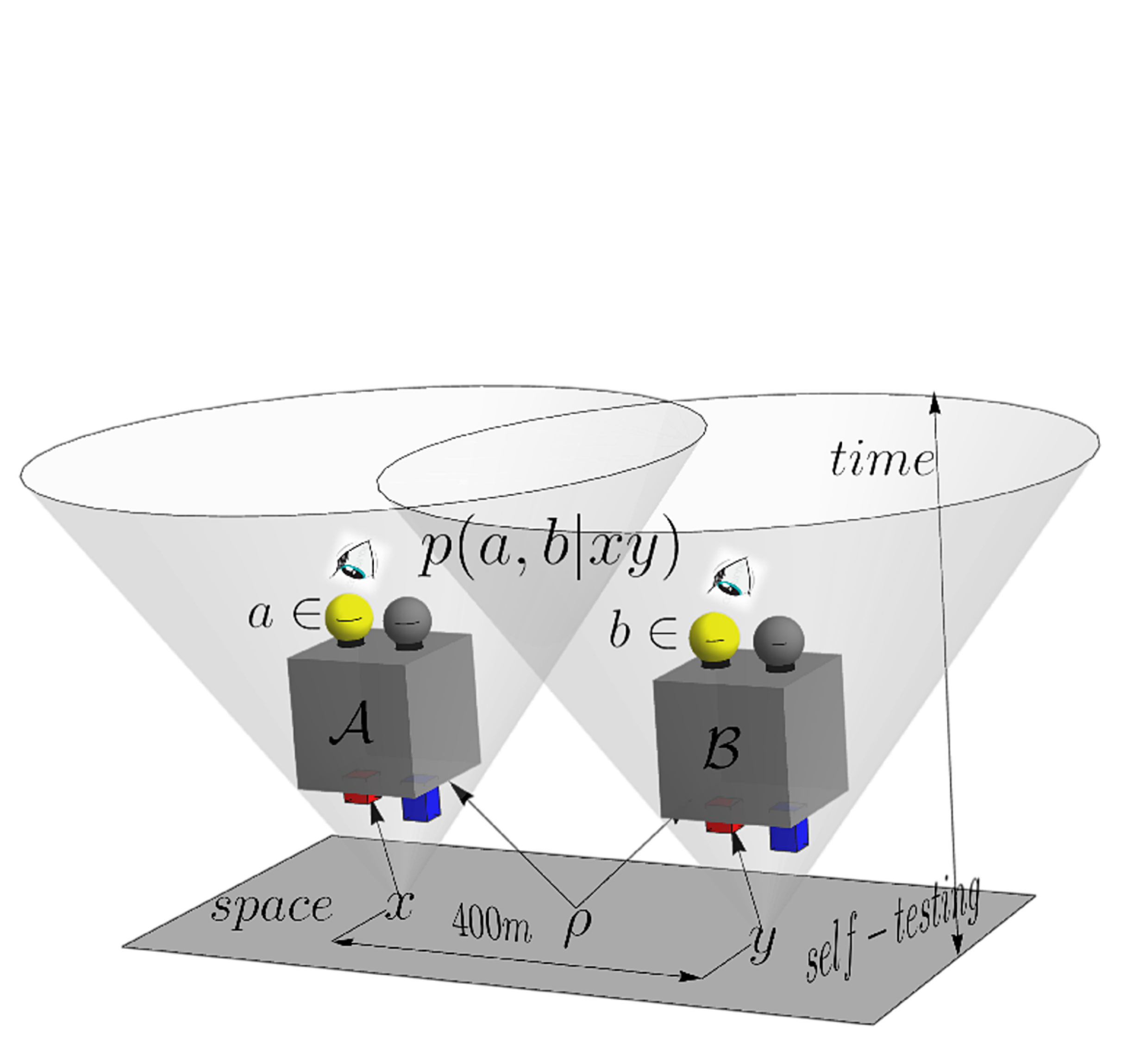}
        \caption{\textit{Non-local Bell-test}}
        \label{fig:nonlocaltestbell}
    \end{subfigure}
    \hfill
    \begin{subfigure}{.52\textwidth}
        \centering   
\begin{tikzpicture}

\def\fullwidth{4}

\def\pone{4/3}
\def\ptwo{2}
\def\pthree{2.83}
\def\inizio{.85}
\def\yshift{2}

\draw[thick] (\inizio,.1) -- (\inizio,-0.1) node[below] {0};
\draw[thick] (\inizio+\pone,.1) -- (\inizio+\pone,-0.1) node[below] {\(\frac13\)};
\draw[thick] (\inizio+\ptwo,.1) -- (\inizio+\ptwo,-0.1) node[below] {\(\frac12\)};
\draw[thick] (\inizio+\fullwidth,.1) -- (\inizio+\fullwidth,-0.1) node[below] {1};

\node at (-1.2, 0.2) {$\rho_W=\nu P_++\frac{1-\nu}{4}\bm{1},\,\,p:$};
\node at (1.51, 0.8) {Sep};
\node at (2.51, 0.8) {Ent};
\node at (3.85, 0.8) {Nonloc};

\draw[thick, fill=yellow!20] (\inizio,0.1) rectangle (\inizio+\pone,0.5);
\draw[thick, fill=blue!20] (\inizio+\pone,0.1) rectangle (\inizio+\fullwidth,0.5);
\draw[thick, fill=red!20] (\inizio+\ptwo,0.25) rectangle (\inizio+\fullwidth,0.5);

\draw[thick] (\inizio,\yshift+.1) -- (\inizio,\yshift-0.1) node[below] {0};
\draw[thick] (\inizio+\ptwo,\yshift+.1) -- (\inizio+\ptwo,\yshift-0.1) node[below] {\(\frac12\)};
\draw[thick] (\inizio+\pthree,\yshift+.1) -- (\inizio+\pthree,\yshift-0.1) node[below] {\(\frac{1}{\sqrt{2}}\)};
\draw[thick] (\inizio+\fullwidth,\yshift+.1) -- (\inizio+\fullwidth,\yshift-0.1) node[below] {1};

\node at (-1.2, \yshift+0.2) {$\tilde{\bm{p}}=\frac{v}{2}\delta_{a\oplus b,yx}+\frac{1-v}{4},\,v:$};
\node at (1.85, \yshift+0.8) {$\mathcal{L}$};
\node at (3.27, \yshift+0.8) {$\mathcal{Q}$};
\node at (4.26, \yshift+0.8) {$\mathcal{NS}$};

\draw[thick, fill=red!20] (\inizio,\yshift+0.1) rectangle (\inizio+\fullwidth,\yshift+0.5);
\draw[thick, fill=blue!20] (\inizio,\yshift+0.2) rectangle (\inizio+\pthree,\yshift+0.5);
\draw[thick, fill=yellow!20] (\inizio,\yshift+0.3) rectangle (\inizio+\ptwo,\yshift+0.5);

\end{tikzpicture}
        \caption{\textit{Werner state $\rho_W$ and specific behaviors $\bm{p}=\{p(ab|xy)\}$}}
        \label{fig:WernerCorr}
    \end{subfigure}
    \begin{subfigure}{.49\textwidth}
        \centering
        \includegraphics[width=\textwidth]{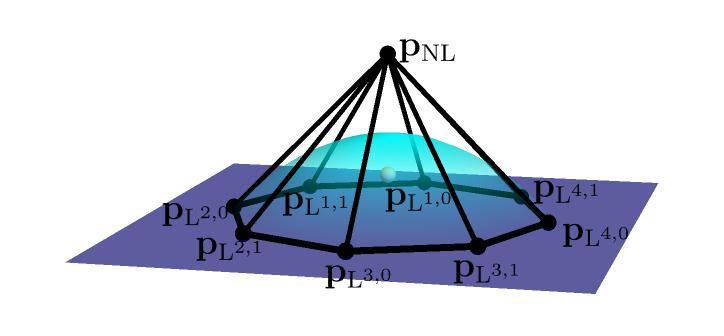}
        \caption{\textit{Eve's strategy in CHSH protocol of Eq.\eqref{eq:Eve_ombrella}}}
        \label{fig:Eve_ombrella}
    \end{subfigure}
    \begin{subfigure}{.49\textwidth}
        \centering

        \includegraphics[width=.65\linewidth]{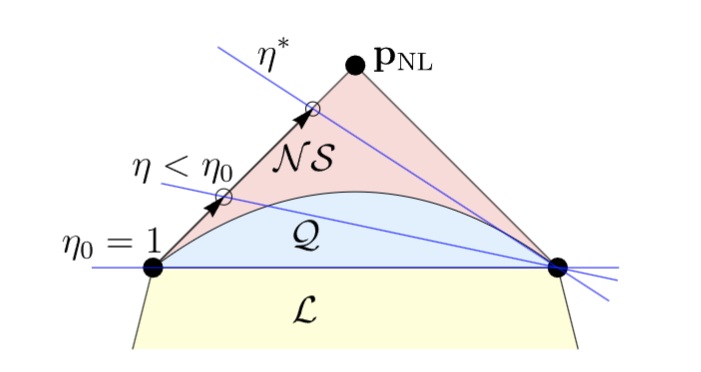}
        \caption{\textit{general Bell-test against Eve's strategy}}
        \label{fig:corr_hierarchy1}
    \end{subfigure}
   \caption{\textit{DI-QKD and Bell's Theorem} — only by the observed correlations $\bm p$ from two causal cones in Fig. \ref{fig:nonlocaltestbell}, the security of DI-QKD is tested by BI determining if $\bm{p} \not\in\mathcal{L}$ (see Sec. \ref{sec:chap2}). Self-testing may also be possible, a time retrodictive process that infers the inputs $x, y, \rho$ from $\bm{p}$ \cite{Mayers2004}. Figure \ref{fig:WernerCorr} shows the tolerance level $\nu$ in a Werner state $\rho_W$ required to witness nonlocality, along with the visibility in a specific $\tilde{\bm{p}}$ across the different regions $\mathcal{L}$, $\mathcal{Q}$, and $\mathcal{NS}$ in the space of correlations of Fig. \ref{fig:Eve_ombrella}-\ref{fig:corr_hierarchy1} (see \ref{sec:CHSHprotocol}). $\eta^*$ in Fig. \ref{fig:corr_hierarchy1} is the critical detection efficiency, if $\eta<\eta^*$ then $\not\exists$ BI to assert $\bm p \in \QC\setminus \LC$. In Sec. \ref{sec:chap2} we will introduce the behavior $\bm p_{\mathrm{NL}}$, a.k.a. PR box.}
    \label{fig:diqkd}
\end{figure}
Fig. \ref{fig:timeline} 
shows the evolution of DI-QKD, from BB84  and E91 protocol \cite{E91} up to the formalization of theoretical techniques and the first implementations in 2022 \cite{Nadlinger2022, Liu2022, Zhang2022a} (details in Fig. \ref{fig:qidkdplot}).
Remarkably, the first successful implementation of DI-QKD was reached after overcoming all the Bell test loopholes, highlighting the challenges in realizing DI-QKD and its connecting with BI experiments. Fig. \ref{fig:qidkdplot} compares the current experimental reach of DI-QKD - specifically, the distances of $2$ m, approximately $200$ m, and $400$ m — with those achieved using Measurement Device Independent-QKD (MDI-QKD), a related but distinct approach that is easier to implement.
\begin{figure}
    \centering
    \includegraphics[width=\linewidth]{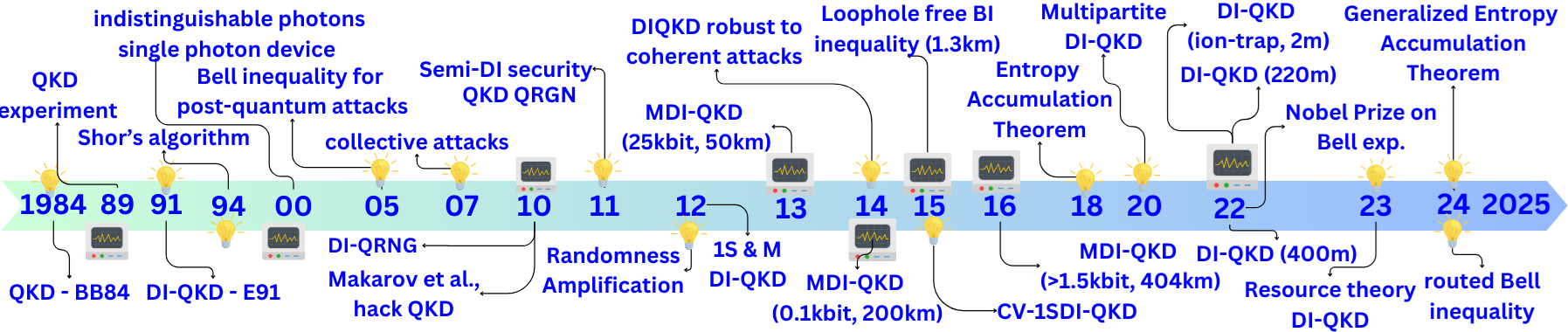}
    \caption{Timeline highlighting key events using a lamp and oscilloscope, distinguishes theoretical and experimental contributions (MDI - measurement device independent; 1S – One-sided; QRGN – Quantum random generator number; CV- continuous variable).}
    \label{fig:timeline}
\end{figure}
\begin{figure}
    \centering
    \begin{subfigure}{.5\textwidth}
        \centering 
        \includegraphics[width=.8\textwidth]{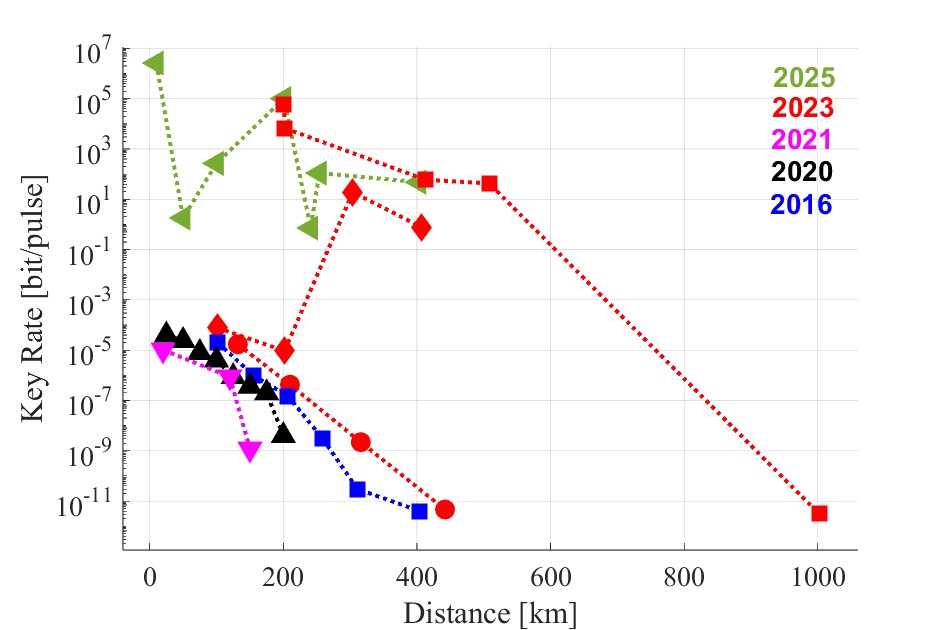}
        \caption{\textit{MDI--QKD}}
        \label{fig:mdiqkd}
    \end{subfigure}
    \hfill
    \begin{subfigure}{.425\textwidth}
    	\centering
        \includegraphics[width=\textwidth]{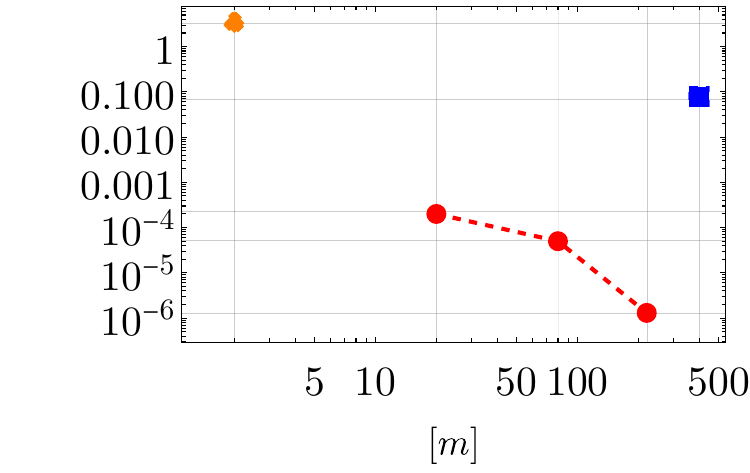}
        \caption{\textit{DI--QKD}}
        \label{fig:diqkd2}
    \end{subfigure}
    \caption{\textit{Comparative Analysis of DI-QKD and MDI-QKD Experiments} -- Fig. \ref{fig:mdiqkd} encapsulates the progress in quantum communication distances achieved through MDI-QKD implementations~\cite{yan2025measurement,shao2025high,li2023twin,zhou2023experimental,liu20231002,zhan2025experimental,zhang2025experimental,Hajomer_2025,liu2025hybrid,pittaluga2025long} (see. \ref{sec:DI-QKD}). In contrast fully-DI-QKD in Fig. \ref{fig:diqkd2} at distances: 2 m (yellow)\cite{Nadlinger2022}, 20 m, 100 m, and 200 m (red)\cite{Liu2022}, and 400 m (blue)\cite{Zhang2022a}. 
    } 
    \label{fig:qidkdplot} 
\end{figure}
Despite the challenges that lie ahead in terms of practical implementation and scalability, as the technology readiness level \cite{Heder2017} currently stands at 2-3 (concept formulated/experimental proof of concept), with expectations to advance to level 4-5 in the coming years (technology validated in lab/relevant environment), the pursuit of DI-QKD continues to push the boundaries of what is possible for secure communications negating the possibility of \textit{inventa fraus}.

\subsection{Focus of this Review}
 Here, we provide an in-depth presentation of DI-QKD, while also introducing Semi-device independent methods, including standard SDI-QKD, as well as MDI-QKD (see Fig.~\ref{fig:mdiqkd}), receiver DI-QKD (RDI-QKD), and one--sided DI-QKD (1SDI-QKD). As shown in Fig.~\ref{fig:timeline}, this review integrates DI-QKD theoretical proofs~\cite{
E91,
Acin2006b,
Acin2006,
acin2007device,
Barrett2005,
scarani2006secrecy,
pabgs09andSangouard.4,
GBHA10,
pironio2010random,
hanggi2010device,
masanes2011secure,
barrett2013memory,
lim2013device,
Slater2014,
vazirani2014fully,
vazirani2019fully,
Kaur2020,
tan2020advantage,
Brown2021,
Sekatski2021,
Woodhead2021,
Schwonnek2021,
Xu2022,
farkas2021bell,
Zhen2023a,
chaturvedi2024extending,
Lobo2024,
Tan2024,
LeRoyDeloison2025,
Sekatski2025,
Chaturvedi2025,
Farkas2024}, 
and experimental challenges~\cite{
Giustina2013,
Hensen2015a,
Lucamarini2018,
Yin2020,
Holz2020,
Valcarce2023} (with BI loopholes), resulting in fully DI-QKD experiments \cite{Nadlinger2022,Liu2022,Zhang2022a} (see Fig.~\ref{fig:diqkd2}).
We present simulations that bridge the gap with experiments in \cite{
Tan2022,
Melnikov2020,
Kolodynski2020}, 
as well as advanced mathematical methods, such as entropy accumulation and bounds in \citep{dfr16,
adfrv17,
Dupuis2020,
Grasselli2021,
Brown2021,
Metger2022a,Metger2024}.
DI randomness generation (theory\cite{colbeckrenner2012free,Acin2016,Li2023d}, experiments\cite{CZYM16,MVV17,Avesani2018}), and broader DI-QKD versions (theory~\cite{braunstein2012side,PB11,branciard2012one,Lo2012,Ioannou2022}, experiments ~\cite{liu13,tangpan14,panayi2014memory,pirandola2015high,gehring2015implementation,yin16}, are also discussed. 
A number of very nice review papers have covered theoretical, experimental and implementation aspects of QKD and DI-QKD ~\cite{Gisin2002,Pirandola2020,Zapatero2023,Primaatmaja2023,SCARANI2014,Alleaume2014,Scarani2009,Xu2020,Mehic2020Review,Cao2022Review}. As the DI framework relies on Bell nonlocality, we also refer the reader to reviews on this subject \cite{Brunner2014,Scarani2019,Tan2021a}.  In the present review, we have attempted to build on this previous work by including the most recent results, and providing alternative approaches when possible.  For example, the topic of Bell nonlocality in section 2 is presented using the modern approach of causal structures.  While touching upon mathematical and technological advancements, our review, starting with a pedagogical focus, remains concise, without claiming to cover all developments exhaustively, but providing references to relevant details in references, as well as a repository of simulations and tutorials, which we have made available online \cite{tutorial}.
\section{Nonclassicality in quantum cryptography}\label{sec:chap2}
\noindent Not all entangled states violate a BI (see Fig. \ref{fig:WernerCorr}, \cite{werner1989quantum}).  
Indeed, different types of non-classical behavior lead to distinct communication tasks. In this section, we introduce the ones related to DI cryptography.
\subsection{Bell nonlocality}\label{sec:BI}
In 1862, Boole laid out conditions for probabilities and logical constraints that any consistent probability theory should follow \cite{boole1862xii} a.k.a. \textit{Boole's conditions of ``possible experience''}, or \textit{causal instruments} \cite{Pearl2009}
. Boole’s work was essentially about the constraints on observable correlations, an early classical analogue to BI's constraints on local and realistic correlations \cite{bell1964einstein}. A Bell test, in its simplest form, involves two random measurement settings $X$\footnote{We refer with the capital letter to the random variable and its lower case the values that it can assume.} and $Y$ assuming values $x,y \in \{0,1\}$, with dichotomous outcomes $A$ and $B$ with values \(a,b \in \{0,1\}\) for Alice (\(\mathcal{A}\)) and Bob (\(\mathcal{B}\)), who are space-like separated, as illustrated in Fig. \ref{fig:diqkd}. Generally, the measurement process is denoted as $M_{A|X}$, a map depending on the specific $X=x$ and $A=a$. A Bell test serves as a causal instrument, represented by an inequality that must be satisfied to ensure the compatibility of certain causal structures, e.g. in Fig. \ref{fig:causal} with the statistics $\bm{p}$ in the 8 dimensional affine subspace of correlations \cite{tutorial}(e.g. replace $\bm{p}$ from Fig. \ref{fig:WernerCorr} in Eq. \eqref{eq:p_entries} gives Eq. \eqref{eq:pNL}).
\begin{equation}\label{eq:p_entries}
    \bm{p} = \{p(ab|xy)\} =
    \begin{array}{|c|c|c|c|c|}
    \hline
    ab \backslash xy & 00 & 01 & 10 & 11 \\ \hline
    00 &    p_{00|00} & p_{00|01} & p_{00|10} & p_{00|11} \\ \hline
    01 &    p_{01|00} & p_{01|01} & p_{01|10} & p_{01|11} \\ \hline
    10 &    p_{10|00} & p_{10|01} & p_{10|10} & p_{10|11} \\ \hline
    11 &    p_{11|00} & p_{11|01} & p_{11|10} & p_{11|11}\\  \hline
    \end{array}
    \in [0,1]^{16}.
\end{equation}
\begin{figure}[h!]
\centering
    \begin{subfigure}{0.24\textwidth}
    \resizebox{\textwidth}{!}{
    \begin{tikzpicture}
    \tikzset{circle node/.style={circle,draw=black,fill=blue!20,minimum size=1cm,inner sep=0pt}}
      \node[circle node] (S) at (0,1.7) {$\Lambda$};
      \node[circle node] (X) at (-1.2,1.7) {$X$};
      \node[circle node] (Y) at (1.2,1.7) {$Y$};
      \node[circle node] (A) at (-1.5,3.4) {$A$};
      \node[circle node] (B) at (1.5,3.4) {$B$};
      \draw[-latex] (S) -- (A);
      \draw[-latex] (S) -- (B);
      \draw[-latex] (X) -- (A);
      \draw[-latex] (Y) -- (B);
    \end{tikzpicture}
    }
    \caption{\textit{in $\mathcal{L}$}}
    \label{fig:L}
    \end{subfigure}
    \hfill
    \begin{subfigure}{0.24\textwidth}
    \resizebox{\textwidth}{!}{
    \begin{tikzpicture}
    \tikzset{circle node/.style={circle,draw=black,fill=blue!20,minimum size=1cm,inner sep=0pt}}
      \node[circle node] (S) at (0,1.7) {$\mathcal{\rho}$};
      \node[circle node] (X) at (-1.2,1.7) {$X$};
      \node[circle node] (Y) at (1.2,1.7) {$Y$};
      \node[circle node] (A) at (-1.5,3.4) {$A$};
      \node[circle node] (B) at (1.5,3.4) {$B$};
      \draw[-latex] (S) -- (A);
      \draw[-latex] (S) -- (B);
      \draw[-latex] (X) -- (A);
      \draw[-latex] (Y) -- (B);
    \end{tikzpicture}
    }
    \caption{\textit{in $\mathcal{L}\subset\mathcal{Q}$}}
    \label{fig:Q}
    \end{subfigure}
    \hfill
    \begin{subfigure}{0.24\textwidth}
    \resizebox{\textwidth}{!}{
    \begin{tikzpicture}
    \tikzset{circle node/.style={circle,draw=black,fill=blue!20,minimum size=1cm,inner sep=0pt}}
      \node[circle node] (S) at (0,1.7) {$\Lambda$};
      \node[circle node] (X) at (-1.2,1.7) {$X$};
      \node[circle node] (Y) at (1.2,1.7) {$Y$};
      \node[circle node] (A) at (-1.5,3.4) {$A$};
      \node[circle node] (B) at (1.5,3.4) {$B$};
      \draw[-latex] (S) -- (A);
      \draw[-latex] (S) -- (B);
      \draw[-latex] (X) -- (A);
      \draw[-latex] (Y) -- (B);
      \draw[-,decorate,decoration={snake,amplitude=0.5mm,segment length=2mm}] (A) -- (B);
    \end{tikzpicture}
    }
    \caption{\textit{in $\mathcal{L}\subset\mathcal{Q}\subset\mathcal{NS}$}}
    \label{fig:NL-NS}
    \end{subfigure}
    \hfill
    \begin{subfigure}{0.24\textwidth}
    \resizebox{\textwidth}{!}{
    \begin{tikzpicture}
    \tikzset{circle node/.style={circle,draw=black,fill=blue!20,minimum size=1cm,inner sep=0pt}}
      \node[circle node] (S) at (0,1.7) {$\Lambda$};
      \node[circle node] (X) at (-1.2,1.7) {$X$};
      \node[circle node] (Y) at (1.2,1.7) {$Y$};
      \node[circle node] (A) at (-1.5,3.4) {$A$};
      \node[circle node] (B) at (1.5,3.4) {$B$};
      \draw[-latex] (S) -- (A);
      \draw[-latex] (S) -- (B);
      \draw[-latex] (X) -- (A);
      \draw[-latex] (Y) -- (B);
      \draw[<->,decorate,decoration={snake,amplitude=0.5mm,segment length=2mm}] (A) -- (B);
      \draw[-latex] (X) -- (B);
    \end{tikzpicture}
    }
    \caption{\textit{out $\mathcal{NS}$}}
    \label{fig:out-NS}
    \end{subfigure}
    \caption{\textit{Bell-test causal structure} -- directed acyclic graphs (DAGs) with nodes for random variables and arrows for direct causal influence\cite{Pearl2009,Spirtes1993}
    . From Fig. \ref{fig:WernerCorr} the correlations with $0\le v\le 1/2$ are compatible with \ref{fig:L}; for $ v\le 1/\sqrt{2}$ with \ref{fig:Q} where nonlocal correlations arise from the entangled state; for $v\le 1$ the nonlocal correlations in \ref{fig:NL-NS} come from a post-quantum common cause (correlations stronger than quantum are represented as a wavy connection between $A$ and $B$, but satisfying no-signaling); for $v>1$ faster-than-light signals are allowed, e.g. $X$ directly influences $B$ or between $A$ and $B$ (the wavy connection can signalize).}
    \label{fig:causal}
\end{figure}
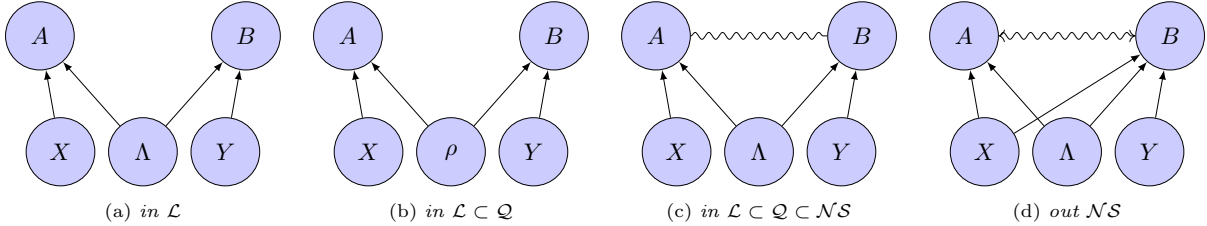
To explain how the causal structures in Fig. \ref{fig:causal} work, let us consider \ref{fig:L} a.k.a \textit{Local Hidden Variable (LHV) model}. The node $A$ ($B$) represents the output random variable and is influenced only by 
classical random variables $X$ and $\Lambda$ ($Y$ and $\Lambda$). Therefore \(p_{A|X\Lambda}\) and \(p_{B|Y\Lambda}\) are the probability distributions associated with variables \(A\) and \(B\), influenced respectively by \(X,\Lambda\) and \(Y,\Lambda\). The distributions \(P_X\), \(P_Y\), and \(P_\Lambda\) represent the probability distributions of \(X\), \(Y\), and \(\Lambda\), respectively. BI can be obtained from the causal structure in \ref{fig:L}. 
\textit{Locality} means that no arrow occurs between the two cones of Fig. \ref{fig:diqkd}, i.e. between $\{A,X\}$ and $\{B,Y\}$. Then:
\begin{equation}\label{eq:locality}
    p_{AB|XY\Lambda}=p_{A|BXY\Lambda}\, p_{B|XY\Lambda}
    \stackrel{\text{Loc}}{=}
p_{A|X\Lambda}\, p_{B|Y\Lambda}.
\end{equation}
Note that this condition also implies \textit{no-signaling} 
\begin{equation}\label{eq:no-signalling}
    p_{A|XY}\stackrel{\text{NS}}{=}p_{A|X},\qquad
    p_{B|XY}\stackrel{\text{NS}}{=}p_{B|Y}.
\end{equation}
Therefore the entries of the \textit{local} correlation $\bm{p}$ are the marginal of $p_{AB\Lambda|XY}=p_{AB|XY\Lambda}p_\Lambda$. From \eqref{eq:locality}
\begin{equation}\label{eq:pabxy}
  \bm{p}\equiv  p_{AB|XY}(ab|xy) = \sum_\lambda 
p_{A|X \Lambda}(a|x,\lambda) p_{B|Y \Lambda}(b|y,\lambda)\,p_{\Lambda}(\lambda).
\end{equation}
The common source is described by a joint probability distribution $P_{\Lambda}(\lambda)=P_{\Lambda_A\Lambda_B}(\lambda_A,\lambda_B)$.
Note that, local correlations $p_{AB|XY}$  can be reproduced by parties equipped only with shared randomness in $p_\Lambda$ so that Alice (Bob) samples from the distribution $P_{A|X\Lambda}$ ($P_{B|X\Lambda}$). \(\Lambda\) can be any system with arbitrary dimension but for \(A, B, X, Y \in \{0,1\}\), the cardinality of \(\Lambda\) is 16 (see \cite{pirsa_PIRSA:23030073} min. 48:30). 
\begin{definition}[Realism]\label{def:realism}
    The outcome of $A$ represents an element of reality, namely it satisfies the realism condition if it is pre-determined by a function \( f(X, \Lambda_A)\stackrel{\text{real}}{=} A\). This can be rewritten as $\lambda_A(X)=A$ once redefining \(f = \lambda_A\) as a pre-existing outcome of \(\Lambda_A\).  
\end{definition}
The role of the measurement process, once $X$ is chosen, is to select a specific function \(M_{A|X}=\lambda_A(X) = A\).
Given that \(X\) is dichotomous,  \(\lambda_A \in \{r_0, r_1, \mathrm{fp}, \mathrm{id}\}\) represents all possible deterministic functions (discard and replace $a\mapsto r_k(a)=k$, flip $a\mapsto\bar{a}$ and identity $a\mapsto a$) unveiling a pre-defined element of reality. Similarly, for \(N_{B|Y}=\lambda_B(Y)=B\), thus with the assumption of \textit{realism}, Eq. \eqref{eq:pabxy} becomes:
\begin{equation}\label{eq:realism}
p_{AB|XY}(ab|xy)\stackrel{\text{Loc,real}}{=}\sum_{\lambda_A,\lambda_B} 
\delta_{A,\lambda_A(X)}\delta_{B,\lambda_B (Y)}P_{\Lambda_A \Lambda_B}(\lambda_A,\lambda_B).
\end{equation} 
With \(\lambda_A,\lambda_B \in \{r_0, r_1, \mathrm{fp}, \mathrm{id}\}\), and denoting \(P_Z(z) = p_z\), we have:
\begin{align}
p_{00|00}&=q_{r_0,r_0}+q_{\mathrm{id},r_0}+q_{r_0,\mathrm{id}}+q_{\mathrm{id},\mathrm{id}}\nonumber\\
p_{00|01}&=q_{r_0,r_1}+q_{\mathrm{id},r_1}+q_{r_0,\mathrm{fp}}+q_{\mathrm{id},\mathrm{fp}}\nonumber\\
\vdots\nonumber\\
p_{11|11}&=q_{r_1,r_1}+q_{\mathrm{id},r_1}+q_{r_1,\mathrm{id}}+q_{\mathrm{id},\mathrm{id}}.
\end{align}
These 16 equations, along with the 16 constraints of probabilities of applying specific \(\lambda_A, \lambda_B\), \(0 \le q_{\lambda_A \lambda_B} \le 1\)  represent a \textit{linear quantifier elimination} on the 16  probabilities \(q\). It satisfies the \textit{no-signaling} relations of Eq. \eqref{eq:no-signalling}
that denies superluminal causal arrows (from node $X$ to $B$ in Fig. \ref{fig:NL-NS}), and characterize the polytope $\mathcal{L}$ by the following type of Clauser-Horne-Shimony-Holt (CHSH) inequalities \footnote{Specifically, there are eight CHSH inequalities correspond to the 8 frustrated 4-node networks \cite{pirsa_PIRSA:23030073}. By taking the absolute value, only 4 CHSH inequalities are relevant and they can be represented in a tetrahedron \cite{Brunner2014,Gigena2022}.}\cite{pirsa_PIRSA:23030073,Lee2017}:
\begin{equation}\label{eq:CHSH1}
    \hat\beta=\sum_{xy=0}^1(-1)^{xy}M_{A|x}N_{B|y}, \qquad\beta=\beta_\uparrow - \beta_\downarrow 
    =\sum_{x,y=0}^1 p(a\oplus b=xy|xy)-p(a\oplus b\neq xy|xy)\le 2
\end{equation}
where $\beta=\langle \hat\beta\rangle$ and $\beta_\uparrow=4-\beta_\downarrow$.
\begin{definition}[Optimal CHSH quantum strategy] \label{optimalchshstrategy}
Let Alice and Bob share the singlet state
\(
\ket{\psi^-}_{AB}=\tfrac{1}{\sqrt2}\!\left(\ket{01}-\ket{10}\right).
\)
Choose dichotomic observables $A_0=\sigma_z,\;A_1=\sigma_x$ for Alice and
$B_0=(\sigma_z+\sigma_x)/\sqrt2,\; B_1=(\sigma_z-\sigma_x)/\sqrt2$ for Bob, which results in the outcome projectors
$M_{a|x}=\frac{1}{2}\bigl(\bm{1}+(-1)^aA_x\bigr)$ and
$N_{b|y}=\frac{1}{2}\bigl(\bm{1}+(-1)^bB_y\bigr)$, therefore, 
$p(ab|xy)=\langle\psi^-|M_{a|x}\otimes N_{b|y}|\psi^-\rangle$ achieves the maximal quantum value
$\beta(p)=2\sqrt2$ for the CHSH expression in Eq.\eqref{eq:CHSH1}.\\
\end{definition}
Other Bell states also reach this maximum value, provided the measurement settings are adjusted accordingly. CHSH can be rewritten as
\begin{equation}\label{eq:CHSHSS}
     \beta_\uparrow=p_{a=b|00}+
     p_{a=b|01}+
     p_{a=b|10}+
     p_{a\neq b|11}
     \le 3\quad \lor \quad \beta_\downarrow\ge 1,
\end{equation}
Note that each CHSH inequality corresponds to a facet of the local polytope $\LC$. This motivates the following definition
\begin{definition}
\label{definition:tightBI}
    A Bell inequality is \emph{tight} (or \emph{facet-defining}) if the hyperplane it defines corresponds to a facet of the local polytope $\LC$. 
\end{definition}
The CHSH inequality can also be understood as a \textit{non-local game}: once Alice and Bob receives $x,y\in\{0,1\}$ they cannot communicate. They respond with bits $a,b\in\{\ 0,1\}$ and win if $a\oplus b=x\cdot y$.
The correlator expression $\langle M_{A|x} N_{B|y} \rangle = p_{a=b|xy} - p_{a \neq b|xy}$, where $p_{a=b|xy} = \sum_{a=b} p_{ab|xy}$, elucidates the relationship between $\beta$, $\beta_\uparrow$, and $\beta_\downarrow$. Classically, $P_\mathrm{win}=\sum_{x,y}p(x,y)\sum_{a,b}p(a\otimes b=xy|xy)=3/4$. Notably, from the behavior $\tilde{\bm{p}}$ in Fig. \ref{fig:WernerCorr}, we have $\beta = 4v$. For $v \leq 1/2$, $\tilde{\bm{p}} \in \mathcal{L}$, consistent with the causal structure in Fig. \ref{fig:L}. However, for $v > 1/2$, $\tilde{\bm{p}} \notin \mathcal{L}$ and cannot be derived from Eq. \eqref{eq:realism}, which holds under the assumptions of locality and realism.

In quantum theory, neither locality nor realism is assumed. Instead of Eq. \eqref{eq:realism}, the Born rule determines the entries of the conditional probability distribution $\bm{p}$. Given a quantum state $\rho \in \mathcal{D}(\mathcal{H}_A \otimes \mathcal{H}_B)$, where $\mathcal{H}_A \cong \mathbb{C}^{d_A}$ and $\mathcal{H}_B \cong \mathbb{C}^{d_B}$, and local measurements described by POVMs $M_{A|x} = \{ M_{a|x} \}_a \in \mathcal{B}(\mathcal{H}_A)$ and $N_{B|y} = \{ N_{b|y} \}_b \in \mathcal{B}(\mathcal{H}_B)$, the Born rule yields

\begin{equation}\label{eq:bornrule}
    \bm{p} \equiv \{p_{AB|XY}(ab|xy)\}_{abxy} = \{\mathrm{Tr}\left(M_{a|x} \otimes N_{b|y} \rho \right)\}_{abxy}\in \mathcal{Q}.
\end{equation}
These behaviors are consistent with the causal structure in Fig. \ref{fig:Q} and with $\tilde{\bm{p}}$ in Fig. \ref{fig:WernerCorr} for $v \leq 1/\sqrt{2}$. If we use Eq.\eqref{eq:bornrule} there exists $\rho$, $M_{A|X}$ and $N_{B|Y}$ such that $p_\mathrm{win}\approx0.85$.

The Hilbert space structure and the non-commutativity of observables imply $\mathcal{L}\subsetneq \mathcal{Q} \subsetneq \mathcal{NS}$ \footnote{In relativistic quantum field theory the set of quantum correlation is $\tilde{\mathcal{Q}}\supseteq \mathcal{Q}$. The question $\tilde{\mathcal{Q}}\equiv \mathcal{Q}?$ is named \textit{Tsirelson problem} and the answer is no, unless finite dimensional Hilbert spaces \cite{AliAhmad2022,Zhang2020,Khrennikov2007}.}. In fact, certain $\tilde{\bm{p}}$ for $1/\sqrt{2} < v \leq 1$ can satisfy the no-signaling constraints ($\bm{p} \in \mathcal{NS}$) while still not belonging to $\mathcal{Q}$ \footnote{These correlations can, in principle, violate the monogamy of entanglement, which asserts that if two parties ($\mathcal{A}$ and $\mathcal{B}$) are maximally entangled, neither can be maximally entangled with a third party ($\mathcal{C}$). Nonetheless, in quantum theory, non-local correlations must still respect monogamy of entanglement.}.
There is ongoing research into fundamental physical principles that could explain why $\mathcal{Q}\subsetneq\mathcal{NS}$. One example is the \textit{information causality principle}\cite{Pawlowski2009,Miklin2021}. It states that the amount of information that one party ($\mathcal{B}$) can gain about another party's ($\mathcal{A}$’s) data, even using shared correlations, cannot exceed the amount of classical communication exchanged between them. This principle is respected only for $\bm{p}\in\mathcal{Q}$, as it imposes limits compatible with Tsirelson's bound $v=\/\sqrt{2}$ (for details see \cite{tutorial} and the review on $\QC$ \cite{Le2023}).
In conclusion, behaviors outside the no-signaling polytope contradict special relativity as shown in Fig. \ref{fig:out-NS}.

In general, $\dim \LC=\dim \QC =\dim \NSC $ and the extremal points of $\LC$ is a finite subset of the set of infinite extremal points of $\QC$.
BI violation remains a necessary condition to detect Eve and ensures secure communication \cite{farkas2021bell}. It turns out that the shared state $\rho$ must necessarily be entangled. 
Unlike entanglement witnesses, which rely on assumptions about Hilbert space structure, BI violation is a stronger test for witnessing entanglement of $\rho$ by $\bm p\in \QC\setminus\LC$. It depends \textit{solely} on the observed statistical behaviors $\bm{p}$, making protocols device-independent. If the measurement outcomes of entangled particles violate BI, it guarantees that the correlations are genuinely quantum.
Eve cannot reproduce these correlations without being detected, ensuring the security of the key. Indeed, let $\beta_\mathcal{EA}$ be the CHSH value between Eve and Alice, and $\beta_{\mathcal{AB}}$ between Alice and Bob, then quantum theory predicts that $\beta_\mathcal{AB}^2+\beta_\mathcal{EA}^2\le 8$ \cite{Cieslinski2024}.
Therefore if $\beta_\mathcal{AB}>2\sqrt{2} \Longrightarrow \beta_\mathcal{EA}<2$. 

Next, we analyze the numerical and experimental tool to assert that $\bm{p}\in \QC\setminus \LC$.

\subsection{The Navascués-Pironio-Acin hierarchy}\label{subsec:NPA}
\noindent The Navascués-Pironio-Acín (NPA) hierarchy is a systematic approach to check if $\bm{p}=p(ab|xy)=\mathrm{Tr}(\rho M_{a|x}N_{b|y})\in \QC\setminus\LC$ \cite{Navascues2007,Navascues2008}. It provides a sequence of increasingly tighter outer approximations to the set of quantum correlations $\QC_1 \supseteq \cdots \QC_k \supseteq \cdots \supseteq  \QC$, where each $k$--th level in the hierarchy defines the following semidefinite programming (SDP) relaxation:
\begin{equation}
             \operatorname{maximize} \{\varphi|\,\;
             \operatorname{Tr}\Gamma^kJ^i=0,\;
              \operatorname{Tr}\Gamma^k\Phi^i=p_i,\;
              \Gamma^k-\varphi \bm{1}\ge0,\;\Gamma\succeq0\}
\end{equation}
where $J^i$ and $\Phi^i$ are linked to the moment matrix $\Gamma^k$, which encodes the constraints derived from quantum mechanics:
\textit{(i) Definition of the moment matrix $\Gamma^k$} --
for a given level \( k \) in the hierarchy,
\[
\Gamma_{i, j}^k = \operatorname{Tr}\rho (\tau^k_i)^\dagger \tau^k_j, \quad \tau^1=\{\bm{1},M_{a|x},N_{b|y}\}_{abxy},\quad \tau^{k+1}=\{\tau^k,\tau^k_i \tau^1_j\}_{ij}
\]
where \( \tau^k = \{ \tau_i^k \} \) is the set of monomials of measurement operators up to degree \( k \), e.g., consisting of products like \( M_{a|x} N_{b|y} \) or \( M_{a'|x} M_{a|x} \). The size of \(\Gamma\) grows with \(k\), encompassing higher-order correlations between measurement operators.

\noindent
\textit{(ii) Constraints} -- The condition $\operatorname{Tr}\Gamma^k J^i=0$ and $\operatorname{Tr}\Gamma^k\Phi^i=p_i$ with opportune $J$ and $\Phi$ suitably rewrite the constraints that $\bm{p}\in\QC_k$ only if $\exists \Gamma \succeq 0$ with $\bm{p}=\{p_i\}_i$ and (similarly for $N_{b|y}$)
\[
M_{a|x} M_{a'|x} = \delta_{a,a'}\bm{1}, \qquad
\sum_a M_{a|x} = \bm 1,\qquad
[M_{a|x},N_{b|y}]=0.
\]
These constraints are incorporated into the structure of \(\Gamma\), imposing relations between the matrix elements and reducing the affine subspace of the possible correlations.

\noindent
\textit{(iii) Feasibility} -- If a feasible \(\Gamma\) that solves the ST problem exists at level \(k\), then $\bm{p}\in\QC_k$ is ``k-quantum'', meaning it can be approximated by a quantum behavior up to $k$--th hierarchy level.
If $\bm{p}\in \QC_{k+1}\Longrightarrow \bm{p}\in \QC_{k}$ since $\Gamma^k$ is a sub-matrix of $\Gamma^{k+1}$. In the limit $\lim_{k\to\infty}\QC_k=\QC$, therefore if $\bm{p}\notin \QC\Longrightarrow \exists k$ s.t. the problem is unfeasible.

\noindent
\textit{(iv) Generalization} -- The NPA hierarchy can be seen as a specific instance of a more general framework for optimizing polynomial expressions over noncommuting variables. This generalized formulation was developed in \cite{Pironio2010}, where the moment matrix $\Gamma^k$ is understood as a noncommutative moment matrix, and the SDP relaxation is extended to arbitrary Hermitian polynomial constraints over operator variables. In this setting, one defines a sequence $y = \{y_w\}$ indexed by words $w$ in the measurement operators, where $y_w = \langle \phi, w(X) \phi \rangle$ for some Hilbert space $\mathcal{H}$, state $\phi \in \mathcal{H}$, and operator tuple $X = (X_1, \dots, X_n)$. The feasibility conditions $\Gamma^k \succeq 0$ and the constraints encoded by moment and localizing matrices ensure that the hierarchy of SDPs converges to the global optimum under an Archimedean condition on the constraint set.

This formulation subsumes the NPA hierarchy by treating quantum constraints (e.g., projection, completeness, commutation) as polynomial equalities or inequalities. Notably, it includes methods to detect when a given hierarchy level already yields the global optimum and to extract an explicit quantum realization $(\mathcal{H}, X, \phi)$ from the SDP solution, as shown in \cite[Theorem~2]{Pironio2010}.

In practice, the NPA hierarchy offers a tractable approximation of the quantum set via a sequence of SDPs, each solvable by efficient algorithms, though the computational cost increases with the level \(k\). For many applications, low levels (e.g., \(k = 2\) or \(3\)) already yield tight enough bounds. Intermediate levels, such as \(\tau^{1+AB} = \tau^1 \cup \{M_a^x N_b^y\}_{abxy}\), are also commonly used.
Replacing the objective function \(\varphi\) with any linear function of the elements in \(\tau^k\), the NPA hierarchy becomes a powerful computational tool—e.g., for estimating min-entropy in security proofs (see Sec.~\ref{sec:chap4}). Additional methods for DI applications are discussed in \cite{Navascues2014, Vertesi2014, Navascues2015} and in the SDP review \cite{Mironowicz2024}.

\subsection{Self-testing}\label{sec:self-testing}
\noindent
In particular cases, DI--protocols not only identify $\bm{p}\in \QC\setminus\LC$, but from the behaviors $\bm{p}$ can also infer the \textit{input state and measurements realization}  $R=(\ket{\tilde{\psi}}_{AB}, \tilde{M}_{A|X},\tilde{N}_{B|Y})$ adopted in the experiment up to some local invariance $\Phi$ (Fig. \ref{fig:nonlocaltestbell}). When this is possible, we say that the behaviors self-test the realization,  $\bm{p}\stackrel{\text{self-test}}{\mapsto} \Phi(R)$ \cite{Mayers98,Mayers2004}. 
\begin{definition}
Let identically and independent distributed (iid) $\bm{p}=\{p_{AB|XY}(ab|xy)\}\equiv p_{ab|xy}$
with locality and measurement-dependence loophole closed (sec. \ref{subsec:loopholes}), then $\forall \dim\mathcal{B}(\mathcal{H}_A),\dim\mathcal{B}(\mathcal{H}_A)$

\begin{equation}
    \Sigma:\bm{p}\stackrel{\text{self-test}}{\mapsto}(\tilde{\psi}_{AB},\tilde{M}_{A|X},\tilde{N}_{B|Y}) \Longrightarrow
    \exists | \Sigma^{-1}(\tilde{\psi}_{AB},\tilde{M}_{A|X},\tilde{N}_{B|Y})=
    \bra{\tilde{\psi}}\tilde{M}_{A|X}\otimes\tilde{N}_{B|Y}\ket{\tilde{\psi}}=p_{AB|XY}
\end{equation}
up to some \textit{gauge of freedom} characterized by the following local invariance $\Phi=\Phi_A \otimes \Phi_B$:
\begin{itemize}
    \item [\textit{(i)}] $\tilde{M}_{A|x}\mapsto U \tilde{M}_{A|x} U^\dagger$, $\tilde{N}_{B|y}\mapsto V \tilde{N}_{B|y} V^\dagger$, $\ket{\tilde{\psi}}\mapsto U\otimes V \ket{\tilde{\psi}}$
    \item [\textit{(ii)}] Given $\ket{\psi}_{ABE}\in \mathcal{H}_{ABE}$, 
    $\{M_{A|x}\}_x\in \mathcal{B}(\mathcal{H}_A)$, $\{N_{B|y}\}_y\in \mathcal{B}(\mathcal{H}_B)$, exists
    \begin{equation}\label{eq:self-state}
        \Phi:\mathcal{H}_{AB}\mapsto \mathcal{H}_{\tilde{A}\bar{A}\tilde{B}\bar{B}E} \mbox{ s.t. }
        \ket{\psi}_{AB}\mapsto \ket{\psi}_{\tilde{A}\tilde{B}}\ket{\mathrm{junk}}_{\bar{A}\bar{B}E}
    \end{equation}
    and
    \begin{align}\label{eq:self-msmt}
        \Phi_A:\mathcal{B}(\mathcal{H}_A)&\mapsto \mathcal{B}(\mathcal{H}_{\tilde{A}}\otimes\mathcal{H}_{\bar{A}})
        &
        \Phi_B:\mathcal{B}(\mathcal{H}_B)&\mapsto \mathcal{B}(\mathcal{H}_{\tilde{B}}\otimes\mathcal{H}_{\bar{B}}) \nonumber\\
        M_{a|x}&\stackrel{\Phi_A}{\mapsto}\tilde{M}_{a|x}\otimes \bm{1}_{\bar{A}}
        &
        N_{b|y}&\stackrel{\Phi_B}{\mapsto}\tilde{N}_{b|y}\otimes \bm{1}_{\bar{B}}
    \end{align}
    such that
    \begin{equation}\label{selftesting}
       (\Phi_A\otimes\Phi_B\otimes\mathrm{id}_E) (M_a^x\otimes N_b^y\otimes\bm{1}_E \ket{\psi}_{ABE})=
(\tilde{M}_{a|x}\otimes\tilde{N}_{b|y}\ket{\tilde{\psi}}_{\tilde{A}\tilde{B}})\otimes \ket{\text{junk}}_{\bar{A}\bar{B}E}.
    \end{equation}
\end{itemize}
\end{definition}
A simple case of self-testing is given by the maximal violation of CHSH of Eq. \eqref{eq:CHSH1}. It is easy to observe that any realization such that $\beta =2\sqrt{2}$ consists of anticommuting operators on the support of the state 
$\ket{\psi}$, 
$\{M_{A|0},M_{A|1}\}\ket{\psi}=\{N_{B|0},N_{B|1}\}\ket{\psi}$. Indeed, let 
$M_{A|\pm}=\frac{M_{A|0}\pm M_{A|1}}{\sqrt{2}}$, the sum-of-square (SOS) decomposition of the shifted CHSH operator assures that:
\begin{equation}
    2\sqrt{2}\bm{1}-\hat\beta=\frac{(M_{A|+}-N_{B|0})^2+(M_{A|-}-N_{B|1})^2}{\sqrt{2}}\succeq 0.
\end{equation}
Then the anticommutation comes from $\beta=2\sqrt{2} \Longrightarrow M_{A|+}=N_{B|0}\ket{\psi}$ and $ M_{A|-}\ket{\psi}=N_{B|1}\ket{\psi}$. The explicit isometry of Eq. \eqref{eq:self-state}  is given in the circuit  \ref{fig:swap-circuit}. Analogously for the isometries on the measurements of Eq. \eqref{eq:self-msmt}.
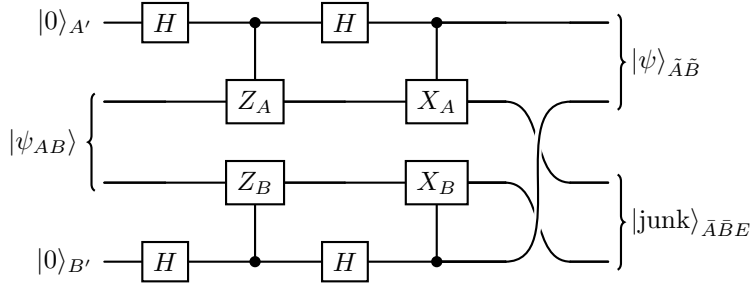
\begin{figure}
    \centering
    \begin{subfigure}{0.6\textwidth}
    \resizebox{\textwidth}{!}{
    \begin{quantikz}
    \lstick{\(\ket{0}\)$_{A'}$}    & \gate{H} & \ctrl{1}   & \gate{H} & \ctrl{1}   &\permute{1,3,4,2}&\qw\rstick[2]{$\ket{\psi}_{\tilde{A}\tilde{B}}$}\\
    \lstick[2]{$\ket{\psi_{AB}}$}  & \qw      & \gate{Z_A} & \qw      & \gate{X_A} & \qw & \qw \\
                                   & \qw      & \gate{Z_B} & \qw      & \gate{X_B} & \qw &\qw\rstick[2]{$\ket{\text{junk}}_{\bar{A}\bar{B}E}$} \\
    \lstick{\(\ket{0}\)$_{B'}$}    & \gate{H} & \ctrl{-1}  & \gate{H} & \ctrl{-1} & \qw &\qw\\
    \end{quantikz}
    }
    \caption{Isometry of Eq. \eqref{eq:self-state}}
    \label{fig:swap-circuit}
     \end{subfigure}
     \begin{subfigure}{0.32\textwidth}
    \resizebox{\textwidth}{!}{
    \begin{tikzpicture}[every node/.style={font=\small}, 
                    box/.style={draw, minimum width=1.cm, minimum height=1.cm, fill=black},
                    >=Latex]

    \node[box] (Alice) at (-1.3,0) {\textcolor{white}{$\mathcal{A}$}};
    \node[box] (Bob) at (1.3,0) {\textcolor{white}{$\mathcal{B}$}};
    \node[circle, draw, inner sep=1pt] (S) at (0,0) {$\psi$};
    
    \draw[->, thick] (S) -- (Alice);
    \draw[->, thick] (S) -- (Bob);
    
    \draw[->] (Bob.north) ++(0,0.3) node[right] {$y$} -- (Bob.north);
    
    \draw[->] (Bob.south) -- ++(0,-0.35) node[right] {$b$};
    
    \draw[->] (Alice.north) ++(0,0.3) node[right] {$x$} -- (Alice.north);
    
    \draw[->] (Alice.south) -- ++(0,-0.35) node[right] {$a$};
    
    \draw[thick, dashed] (-2,1) rectangle (2,-1);
    
    \draw[->, thick] (-1.4,-1) -- (-1.4,-1.5) node[below] {$\{p(ab|xy)\}_{a,b,x,y}$};
    \draw[-, thick] (-1.8,-2.2) -- (-1.8,-2.8); 
    \draw[->, thick] (-1.8,-2.8) -- (-.2,-2.8) node[ right] {$(\tilde\psi, \tilde M_{A|X}, \tilde N_{B|Y})$}; 
    \node (see) at (-1.,-2.55) {selftesting};

    \end{tikzpicture}
    }
        \caption{Selftesting as retrodictive task}
        \label{fig:selftesting}
    \end{subfigure}
    \caption{In \ref{fig:swap-circuit} an explicit implementation of the isometry (details in Ref. \cite{Supic2020}) that characterizes the realization $(\tilde\psi, \tilde M_{A|X}, \tilde N_{B|Y})$ inferred only from $\bm p=\{p_{AB|XY}(a,b,|x,y)\}_{a,b,x,y}$ treating the labs as black boxes \ref{fig:selftesting}-\ref{fig:nonlocaltestbell}.}
        \label{fig:selftesting-pic}
\end{figure}
Similar calculations hold when only one detector is inefficient \cite{AMP2012randomness} (see sec. \ref{subsec:detection-loophole} putting $\alpha_B=0$ in \eqref{eq:cq_alphabeta}) and the tilted Bell operator can be obtained, i.e. $\hat\beta+\alpha_A M_{A|0}$(see Eq. \eqref{eq:cq_alphabeta}) such that
\begin{equation}
    \sqrt{2+\alpha_A^2}\bm{1}-(\hat\beta+\alpha_A M_{A|0})=\sum_i P_i^\dagger P_i
\end{equation}
in terms of polynomials $P_i\in\{\bm{1}, M_{A|x}, N_{B|y}, M_{A|x}\otimes N_{B|y}\}$. SOS decomposition allows to prove that if maximal violation $\beta_\QC=\sqrt{2+\alpha_A^2}$ is obtained then the optimal realization $(\ket{\psi}_{AB},M_{A|\pm}, N_{B|y})$ is \textit{self-tested} 
\cite{tiltedasym, Supic2020}, with
\begin{equation}\label{eq:realization2}
    \ket{\psi}_{AB} = \cos{\theta}\ket{00} + \sin{\theta}\ket{11},\quad
    N_{B|0} = \sigma_z,\quad
    N_{B|1} = \sigma_x,\quad  
    M_{A|\pm} = \cos{\mu}\,\sigma_z \pm \sin{\mu}\,\sigma_x 
\end{equation}
where $\alpha_A = 2/\sqrt{1+2\tan^2{2\theta}}$, $\tan(\mu) = \sin(2\theta)$. If the polynomials $P_i$ are written in terms of the operators of \textit{any} optimal realization, then $\forall i$ $P_i\ket{\psi}=0$. These conditions implies the existence of operators $\{Z_A, X_A, Z_B, X_B \}$ satisfying
\begin{equation}\label{isometry_conditions}
    Z_B\ket{\psi}_{AB} = Z_B\ket{\psi}_{AB}, \qquad
    \sin\theta X_A(\bm{1}+Z_B)\ket{\psi}_{AB} = \cos\theta X_A(\bm{1}-Z_A)\ket{\psi}_{AB}. 
\end{equation}
In  turn, Eq.~\eqref{isometry_conditions} ensures the existence of local isometries $\Phi_A$ and $\Phi_B$ such that
\begin{equation}
    \begin{split}
        \Phi_A\otimes\Phi_B \ket{\psi}_{AB} &= \ket{\psi}_{\tilde{A}\tilde{B}}\otimes\ket{\mathrm{junk}}_{\bar{A}\bar{B}E} \\
        \Phi_A\otimes\Phi_B (M_{A|x}\otimes N_{B|y} \ket{\psi}_{AB}) &= M_{A|x}'\otimes N_{B|y}' \ket{\psi}_{\tilde{A}\tilde{B}}\otimes \ket{\mathrm{junk}}_{\bar{A}\bar{B}E}. \label{isometries}
    \end{split}
\end{equation}
Self-testing can be made \textit{robust} in the sense that in a neighborhood of the maximal quantum value $\mathcal{I}_{\beta_\QC} $, there exists a physical realization $R$ that is close—up to a local isometry—to the ideal realization $R_{\QC}$ (see numerical SWAP technique in \cite{bancal2015nswap, yang2014nswap}).
The most general case involving two inefficient detectors, the SOS decomposition is analyzed with NPA hierarchy (see Sec. \ref{subsec:NPA}) without finding a simple expression for the polynomial $P_i$, unless the inefficiency of the detectors is the same \cite{Gigena2024}. The solution in this case is obtained with Jordan's lemma \cite{Scarani2019} and Groebner basis. 
\begin{lemma} \label{Jordanlemma}
(Jordan's lemma)
In CHSH, $\{M_{a|x}\}_{a,x=0,1}$ and $\{N_{b|y}\}_{b,y=0,1}$ can be projective w.l.o.g., then there must exist a local unitary transformations that simultaneously block-diagonalize the observables $M_{A|x}, N_{B|y}$, with blocks of size $1$ or $2$. But, to compute $\langle M_{A|x}\rangle_\psi, \langle N_{B|y}\rangle_\psi$ we can always complete a one-dimensional block by adding to it a projector over a state in the null space of the corresponding reduced state $\rho_{A(B)}=\mathrm{Tr}_{B(A)}{\ket{\psi}\bra{\psi}}$. We can thus assume all blocks to be two-dimensional and write Alice's measurement operators as
\begin{equation} 
    M_{A|0}= \bigoplus_i M_{A|0}^{(i)}= \bigoplus_i {\sigma}_{Z} \label{qubit_param_A1}, \qquad
    M_{A|1}= \bigoplus_i M_{A|1}^{(i)}= \bigoplus_i(\cos{\theta_i^A}\,{\sigma} _{Z}+\sin{\theta_i^A}\,{\sigma}_{X}),
\end{equation}
where index $i$ iterates over the Jordan blocks. Similarly, for Bob's observables. 
\end{lemma}
\noindent Using Jordan's lemma, one can decompose the Bell operator as \( \hat{\beta} = \bigoplus_i \hat{\beta}_i \), where each \( \hat{\beta}_i \) acts on a two-dimensional subspace. This decomposition implies that \textit{self-testing is independent of the local Hilbert space dimensions} \( \dim \mathcal{B}(\mathcal{H}_A) \) and \( \dim \mathcal{B}(\mathcal{H}_B) \), being invariant under local isometries \( \Phi_A \otimes \Phi_B \) that preserve physical predictions.
This leads to the so-called \textit{qubit reduction argument}, a key security feature of device-independent QKD: although the actual devices may operate in high-dimensional spaces, only two-dimensional subspaces contribute to the Bell inequality violation.
Suppose the state decomposes as \( \rho = \bigoplus_i p_i \rho_i \), then the observed value is
\(
    \beta = \mathrm{Tr}(\rho \hat{\beta}) = \sum_i p_i \langle \hat{\beta}_i \rangle.
\)
Each \( \hat{\beta}_i \) is bounded above by Tsirelson's bound \( 2\sqrt{2} \). If \( \rho \) is entirely supported on the block achieving this bound, the full violation is maximal. Otherwise, contributions from blocks with lower eigenvalues reduce the total violation. \textit{This dilution effect makes the security proof robust against dimension attacks}, where an adversary might try to hide extra information in higher-dimensional components (see Sec.~\ref{subsec:implementationissues}).
Consequently, the derivation of bounds on the adversary’s information—such as bounding the guessing probability from the observed Bell violation—becomes \textit{independent of the internal structure of the devices}, which can be treated as black boxes. The existence of a map \( \Sigma : \bm{p} \mapsto R \) (via Jordan’s lemma) is sufficient to certify private randomness extraction. Specifically, condition \eqref{selftesting} implies:
\[
    \sigma_{AE} = \sum_a \ket{a}\bra{a} \otimes \mathrm{Tr}_{AB} \left[(\tilde{M}_{a|x} \otimes \mathbb{I}_B \otimes \mathbb{I}_E) \ket{\psi}\bra{\psi} \right]
    = \left[ \sum_a p_A(a|x) \ket{a}_{\tilde{A}}\bra{a} \right] \otimes \sigma_E,
\]
where \( \sigma_E = \mathrm{Tr}_{\tilde{A}\tilde{B}} \ket{\text{junk}}\bra{\text{junk}} \in \mathcal{H}_{\bar{A} \bar{B} E} \), and \( p_A(a|x) = \sum_b p(a,b|x,y) \) is Alice’s marginal.
Thus, \textit{Alice’s outcomes are completely random from Eve’s perspective}~\cite{Farkas2024}, and for any correlation satisfying condition~\eqref{selftesting}, one may optimize over Bob’s measurements accordingly. A comprehensive review of self-testing is given in~\cite{Supic2020}, and a geometric characterization of self-testing via nonlocal extremal points \( \bm{p} \) is presented in~\cite{Goh2018,Le2023}.
In the next section, we examine how to account for \textit{experimental imperfections} in Bell tests. 

\subsection{Experimental Loopholes}\label{subsec:loopholes}
\noindent Experimental validation that $\bm{p}\in \mathcal{Q}\setminus\mathcal{L}$ requires careful treatment of imperfections in Bell tests. Such imperfections—typically due to transmission losses, detector inefficiencies, or other technical limitations—can open \textit{loopholes} that allow an LHV model to reproduce data that would otherwise appear nonlocal.
These loopholes undermine the assumption that the observed statistics $\bm p$ truly violate the causal constraints illustrated in Fig.\ref{fig:L}~\cite{Kaiser_2022,Larsson2014}. Although this may seem like a conspiratorial behavior of nature, in practice, an adversary could exploit such loopholes to forge fake BI violations using only classical resources~\cite{Gerhardt2011}, potentially compromising device-independent cryptographic protocols.
Below, we briefly summarize the main loopholes (see Refs.~\cite{Kaiser_2022,Larsson2014} for a detailed discussion).

\subsubsection{Detection Efficiency Loophole} \label{subsec:detection-loophole}

Consider the ideal scenario depicted in Fig.\ref{fig:nonlocaltestbell}, where the behavior $\bm{p}=\{p(ab|xy)\}$ satisfies $\beta(\bm{p})=2\sqrt{2}$ as in Eq. \eqref{eq:CHSH1}. Detection is illustrated as ``eyes'' observing which lamp turns on, but in reality, it involves two detectors that click with respective probabilities $\eta_A,\eta_B<1$. For simplicity, let us consider only Bob's detector, then he measures only a set $\mathcal{D}$ of detected particle from a set $\mathcal{E}$, of emitted particle, where  
$\eta=|\mathcal{D}|/|\mathcal{E}|$ \footnote{Here $\eta$ is the probability that a photon emitted by the source is indeed detected. A discussion on the experimental parameters that contribute to $\eta$ is provided in \ref{sec:losses}.}.
The most general way of accounting for no-click events is to consider an additional outcome, which enlarges the Bell scenario\footnote{This approach can be used to describe null events for analyzers with a number of outcomes that span the space associated to the degree of freedom of interest, such as two-outcome polarizing beam splitters, for example.} \cite{Eberhard1995,PhysRevA.83.032123}. As a results, characterizing the sets $\mathcal{L}$ and $\mathcal{Q}$ becomes considerably more complex~\cite{Jones2005,Barrett2005b,Barrett2005a}. To remain in the same Bell scenario, Bob fixes an outcome $b$ to assign at each no-click event with probability $q_B(b|y)$. Similarly, Alice assigns $a$ with probability $q_A(a|x)$
\footnote{This approach is best suited for experiments using ``pass/fail" measurement devices, such as a polarization filter, where one cannot distinguish a null event from a projection onto the state that does not pass through the filter.}. 
An affine map $\hat{\bm{p}}=\Omega_{\eta_A\eta_B}(\bm{p})=\{\hat{p}(ab|xy)\}$ describes the events of both detectors: both, only one, or none of them click with related probabilities $\eta_A\eta_B$, $\eta_A(1-\eta_B)$ or $(1-\eta_A)\eta_B$, and $(1-\eta_A)(1-\eta_B)$, such that
\begin{align} \label{eq:phat}
\hat{p}(ab|xy)=\eta_A\eta_B p(ab|xy) +&\eta_A(1-\eta_B)p_A(a|x)q_B(b|y)+(1-\eta_A)\eta_B q_A(a|x)p_B(b|y)\nonumber\\
+&(1-\eta_A)(1-\eta_B)q_A(a|x)q_B(b|y).
\end{align}
Taking into account also the ``no click'' events, the inequality \eqref{eq:CHSH1} turns out to be $\beta(\hat{\bm{p}})<2$. It is well known that optimal local assignment gives $\beta(p_A(a|x)q_B(b|y))=2\langle M_{A|0}\rangle$, $\beta(q_A(a|x)p_B(b|y))=2\langle N_{B|0}\rangle$, and $\beta(q_A(a|x)q_B(b|y))=2$. This yields
\begin{equation}\label{eq:Sphat}
\beta(\hat{\bm{p}})=\eta_A\eta_B \beta(\bm{p})+2\eta_A(1-\eta_B)\langle M_{A|0}\rangle + 2(1-\eta_A)\eta_B\langle N_{B|0}\rangle +2(1-\eta_A)(1-\eta_B)\le 2.
\end{equation}
This can be rewritten as ($\alpha_A=2(1-\eta_B)/\eta_B$ and $\alpha_B=2(1-\eta_A)/\eta_A$)
\begin{equation}\label{eq:cq_alphabeta}
\beta_{\mathcal{Q}}\equiv \beta(\bm{p})+\alpha_A\langle A_0\rangle +\alpha_B\langle B_0\rangle \le 2+\alpha_A +\alpha_B\equiv \beta_{\mathcal{L}}\le 4.
\end{equation}
The last inequality comes from $\mathcal{L}\subset \mathcal{NS}$ ($\beta(\bm{p}_\mathrm{NL})=4$). Therefore $\alpha_A+\alpha_B\le 2$ (or equivalently, $\eta_A^{-1}+\eta_B^{-1}>3)\Longrightarrow \QC\setminus\LC=\emptyset$.
There is no room for quantum violation as shown in Fig. \ref{fig:corr_hierarchy1} as the local vertex is moving up towards what is known as the \textit{critical detection efficiency} (CDE) (observe in Fig.\ref{fig:cQetas} that $\beta_\QC$ approaches the local bound).
Graphically, one can imagine that the plane of Fig. \ref{fig:Eve_ombrella} and  \ref{fig:corr_hierarchy1} with the local vertices approaches the no-signaling vertex $\bm{p}_{\mathrm{NL}}$ (more geometrical details are in Refs.~\cite{wilms08,branciard2011,czechlewski2018,sauer2020}).
\begin{figure}
\centering
\begin{subfigure}{.49\textwidth}
	\centering
	\includegraphics[width=\textwidth]{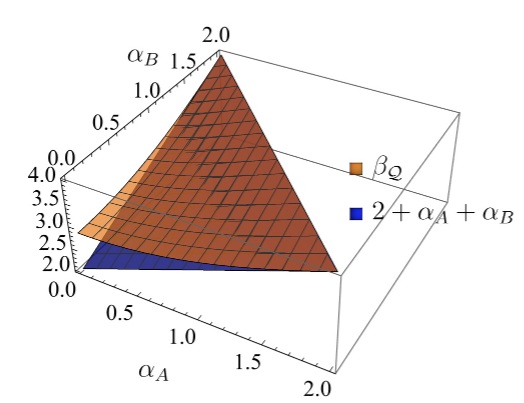}
	\caption{\textit{maximum quantum value}}
	\label{fig:cQ}
\end{subfigure}
\hfill
\begin{subfigure}{.49\textwidth}
	\centering
	\includegraphics[width=\textwidth]{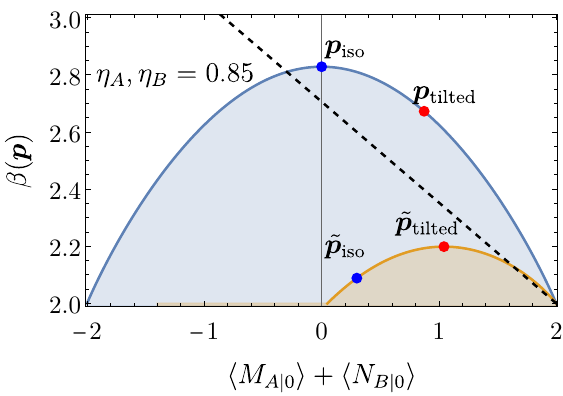}
	\caption{\textit{correlations with inefficient detectors
	}}
	\label{fig:cQetas}
\end{subfigure}
\caption{\textit{CHSH with inefficient detectors} -- Fig. \ref{fig:cQ} encapsulates the optimal value of $\beta_{\mathcal{Q}}$ at given value of $\alpha_A+\alpha_B\le 2$. The \ref{fig:cQetas} illustrates the impact of detector inefficiencies on nonlocal quantum correlations within the simplest Bell scenario. The blue region represents the set of quantum correlations $\bm{p}\in\QC$ in ideal conditions. With the detection efficiencies $\eta_A=\eta_B=0.85$, and the local assignment strategy $q_A(a|x)=\delta_{a,0},q_B(b|y)=\delta_{b,0}$, the effective quantum correlations $\tilde{\bm{p}}=\Omega_{\eta_A\eta_B}(\bm{p})$ are constrained to the smaller orange subset (see Fig. \ref{fig:corr_hierarchy1}). The blue dot on the blue curve corresponds to the isotropic behavior $\bm{p}_{iso}$ that maximally violates the CHSH inequality, $\beta(\bm{p}_{iso})=2\sqrt{2}$, in ideal conditions, while the corresponding effective behavior (blue dot on the orange curve) $\tilde{\bm{p}}_{iso}=\Omega_{\eta_A\eta_B}(\bm{p}_{iso})$ no longer attains the maximum violation of the CHSH inequality,  $\beta(\tilde{\bm{p}}_{iso})\approx 2.08854$. Instead, the red dot on the blue curve corresponds to the quantum behavior $\bm{p}_{tilted}$ which maximally violates the doubly-tilted CHSH inequality (dashed black line)\cite{Gigena2024,tutorial}.
}
\label{fig:detectionloophole}
\end{figure}
\begin{definition}
	The open detection loophole refers to the implication $\beta(\bm{p})>2 \Longrightarrow \bm{p}\in \mathcal{Q}\setminus\mathcal{L}$ mistakenly (assuming $\bm p = \hat{\bm{p}}$) ignoring the ``no click'' events.
\end{definition}
\noindent In many cases, local models can be constructed that are compatible with the experimental data \cite{Pearle1970}. It has been shown that manipulation of measurement devices cannot only lead to fake violations of BI \cite{Gerhardt2011} but also to violations of Tsirelson's bound $2 \sqrt{2}$ \cite{tasca09b}.
In DI-QKD, low values of $\eta_A$ and $\eta_B$ allow Eve to intercept and hide herself more effectively because many ``no-click'' events would already occur naturally due to losses from attenuation and imperfect detectors. On the contrary, high detection efficiency $\eta\lesssim 1$ helps to distinguish Eve's attacks from natural losses by maintaining a high value of $\beta_\QC-\beta_\mathcal{L}$, which translates to a reliable measure of nonlocality. 
\begin{definition}
	The detection loophole is closed on the test $\beta_Q=\beta(\hat{\bm{p}})>\beta_{\mathcal{L}}$ as genuinely implies (if all other loopholes are closed) $\hat{\bm{p}}\in \mathcal{Q}\setminus\mathcal{L}$. 
\end{definition}

 The problem of low detection efficiencies in Bell tests had been noted earlier~\cite{Pearle1970}, but it was Eberhard in 1993~\cite{Eberhard1993} who first derived a Bell-type inequality specifically optimized for overcoming this loophole. His approach used only the probabilities of conclusive outcomes and allowed for non-maximally entangled states to minimize the critical detection efficiency. Originally it was written as
\begin{equation}
\label{eq:Eberhard}
S_E = p(11|00) + p(11|01) + p(11|10)
      - p(11|11) - p(1|10) - p(1|01) \;\le 0,
\end{equation}

It is a limiting case of Eq.~\eqref{eq:cq_alphabeta}, obtained by setting \(\alpha_A = \alpha_B = 1\) and expressing the expectation values through raw detection probabilities.

\begin{proposition}
	A necessary condition to close the detection loophole in Bell experiments implies $\eta>\eta^*$, where $\eta^*$ is the CDE.
	(see a representation in Fig. \ref{fig:corr_hierarchy1}).
\end{proposition}
\noindent Indeed, $\eta^*$ is a characteristic of an ideal nonlocal correlation, below which $\QC\setminus\LC=\emptyset$, and limits the distance across which nonlocality can be operationally (quantum) certified. 
In the simplest Bell scenario of Fig. \ref{fig:diqkd}, the quantum strategy maximally violating Eq. \eqref{eq:CHSH1} (in ideal conditions) ceases to yield $\bm{p}\in\QC\setminus\LC$ for CDE below $\eta^*=2\sqrt{2}-2\simeq 0.82$ \cite{Garg1987}.
This comes from Eq. \eqref{eq:Sphat} for $\eta_A=\eta_B$, independent on the measurement nor on each other. Then $\langle M_{A|0} \rangle =0 =\langle N_{B|0}\rangle$ because the results will be uncorrelated (the detected particle is in a maximally mixed state). 
It follows a list of less recent achievements:
\textit{(i)}
The CDE is lowered to $2/3\simeq 0.66$ \cite{Eberhard1993}, which comes at the cost of very low robustness to background noise as the state is almost separable (see fig. \ref{fig:cQetas}). 
\textit{(ii)}
For more general scenarios, involving more measurements, the extra-outcome approach presents lower CDEs \cite{wilms08}.

\textit{(iii)}
Overall, if $\rho\in\mathcal{B}(\bigotimes_{i=1}^n\mathbb{C}^d)$ is used, then higher $d$ and/or $n$ implies lower (exponentially) CDE \cite{m02}, but at the costs of more experimental complexity. For example, an improvement for CHSH is only for $d\gtrsim 1600$. In asymmetric (symmetric) Bell tests $\eta^*\sim\frac{1}{d}$ ($\eta^*\sim61.8\%$) \cite{VPB10}.
\textit{(iv)}
For the BI \( I_{3322} \), with one ideal detection efficiency (\( \eta_A = 1 \)), a CDE is \( \eta_B = 43\% \) (or \( \eta_B = 66.7\% \)) for non-maximally (or maximally) entangled states~\cite{CG04};  
\textit{(v)} an LHV model cannot describe $\bm{p}\in\QC$ when the number of measurement settings at each site $m_A$ and $m_B$ satisfy~\cite{massar2003violation,massar2002bell}  
\begin{equation}
\eta\geq\frac{m_A+m_B-2}{m_Am_B-1}.
\end{equation}
Below more recent achievements: 
\textit{(i)}
 a family of $n$-party BI with binary outcomes and $m>2$ measurement settings per party can obtain BI violation with lower CDE \cite{Pal2015}; 
\textit{(ii)}
 BI using multiple copies of the two-qubit maximally entangled state and Pauli measurements, defining a Bell setup with $m = 2^n$ settings and $2^n$ outcomes reduces the CDE below $0.8214$ for $n \geq 2$ \cite{Marton2023a}; 
\textit{(iii)} 
An exponential reduction of CDE was demonstrated in \cite{miklin2022exponentially} by violating $N$ BI in parallel using $N$ entangled states shared by a single particle pair. 
\textit{(iv)} 
 the BI $I_{4422}^4$ is experimentally violated using four-dimensional entangled photons closing the detection loophole with  $\eta\sim71.7\%$ \cite{hu2022high}\footnote{In an atom-photon system for example, the atomic system can have $\eta_A$, $\eta_B$ near the unity (see Sec. \ref{sec:experiments}).
 }. 

\textit{(v)} 
to experimentally increase detection efficiency (qubits constructed in trapped ions, atoms, or nitrogen-vacancy (NV) centers in diamond) is also used an ``event-ready" setup, in which the presence of particles at the measurement stations is heralded by an additional event-ready protocol ~\cite{bell80,zukowski93,simon03} (more detail in Sec. \ref{sec:eventready}). In point-to-point photonic experiments, both link losses and detector losses are more difficult to overcome.  Superconducting single-photon detectors (SNSPDs), achieving efficiencies of over 90\%, have been instrumental in recent loophole-free experiments \cite{Marsili2013,US15,Vienna15and32_1}.

\paragraph{Fair sampling loophole} -- The losses that naturally appear (e.g., in optical fibers) and affect the particles independently of the measurement settings, are solely responsible for $|\mathcal{D}|\ll|\mathcal
{E}|$. The \textit{fair sampling assumption} (FSA) is often invoked to justify ignoring the detection loophole.
\begin{definition}
	A Bell tests in which $|\mathcal{D}|\ll|\mathcal{E}$ invoked FSA when it imposes that $p(\mathcal{D})\simeq p(\mathcal{E})$.
\end{definition} 
\noindent Eve can exploit the fair sampling loophole by applying a biased coarse-graining \(\mu_{DL}\) to the distribution \(p(\mathcal{E})\), resulting in \(p(\mathcal{D}) = \mu_{DL}[p(\mathcal{E})] \neq p(\mathcal{E})\). This manipulation skews the observed data \(\bm{p}\), making it falsely appear that \(\bm{p} \in \mathcal{Q} \setminus \mathcal{L}\), as if BI were violated. She can achieve this by influencing detection efficiency, introducing selective transmission losses, or tampering with data processing. For instance, Eve may ensure that only particles with hidden variables producing strong correlations are detected, while others are discarded.
\begin{proposition}
   Although high-efficiency detectors, with \(|\mathcal{D}|/|\mathcal{E}| \simeq 1\), limit Eve’s manipulation, they do not guarantee \(p(\mathcal{D}) \simeq p(\mathcal{E})\). 
\end{proposition}
\noindent Even in experiments with $\eta\sim 1$, hidden mechanisms can bias which particles are detected based on certain hidden variables (e.g., emission angle or polarization). These variables could correlate with measurement settings in a way that skews the detected sample to favor results violating BI. Thus, while nearly all particles are detected (addressing the detection loophole), the sample may still not represent the full emitted set (leaving the fair sampling loophole open).  To avoid such bias, careful calibration is crucial like using space-like separation and random detector calibration \cite{weihs1998violation}.
However, $\eta\sim 1$~\cite{rowe2000} makes it easier to verify that the detected pairs are a fair representation of the emitted set, helping to close the fair sampling loophole, e.g. in Ref.~\cite{rowe2000}\footnote{Despite the detection loophole in Ref.~\cite{rowe2000} is closed, but the separation distance was not sufficient to close the locality loophole.}.

\subsubsection{Locality and Measurement-dependence Loophole}
\label{sec:locloop}
\begin{definition}[space-like separation] 
    \label{def:space-like}
	For two events in Alice and Bob's lab respectively with coordinates $(t_A,x_A)$ and $(t_B,x_B)$ in Minkowski spacetime are \textit{causally space-like separated} iff the invariant spacetime interval $\Delta s^2=c^2(t_A-t_B)^2-|x_A-x_B|^2<0$, or equivalently Eq. \eqref{eq:locality} holds, which implies that the spatial distance between the events is greater than the distance light could travel in the time interval separating them.
\end{definition}
Because no causal influence (which is limited by the speed of light) can bridge a space-like interval, there is no possible way for one event to affect the other. Fig.\ref{fig:nonlocaltestbell} shows two black boxes $\mathcal{A}$ and $\mathcal{B}$ representing the Alice and Bob's laboratories \textit{causally space-like separated} in causal cones to prevent any influence the other detector's measurement from the other lab.
\begin{definition}[Locality loophole] 
	The locality loophole is open if Eq. \eqref{eq:locality} is not certified.
\end{definition}
\noindent To close the locality loophole the Bell experiment must be realized such that the entire measurement process, consisting of the random choice of basis, the adjusting the analyzer, and the detection of the particle satisfied the space-like separation condition~\cite{aspect1982experimental}. Locality loophole was certified in the late ``90s using \textit{(i)} entangled photons from SPDC sources, \textit{(ii)} increasing the space-like separation between the analyzers to tens of km~\cite{tittel98,marcikic04}, \textit{(iii)} employing fast, unpredictable and random switching of measurement settings to further eliminate the possibility of communication between the detectors \textit{(iv)} using fast electronics and quantum random number generators (QRNG) to choose the settings of the analyzers~\cite{weihs1998violation}.   

The first Bell test to close both the detection and locality loopholes was reported in 2015~\cite{Hensen2015a}. It used electron spins that were entangled using an event-ready protocol~\cite{bell80,zukowski93,simon03}. The experiment demonstrated the first statistically significant BI violation without relying on additional assumptions such as fair sampling.
\begin{definition}[Measurement-dependence loophole]
The measurement-dependence loophole, also known as the \textit{freedom-of-choice} or \textit{the free-will loophole}, questions whether the choices of measurement settings could be influenced by hidden variables, i.e. $P_{X}= P_{X|\Lambda}$?, $P_{Y}= P_{Y|\Lambda}$?
\end{definition} 
\noindent This arises from the observation that the local and realistic causal structure in Fig~\eqref{fig:L} implicitly assumes $P_{X}= P_{X|\Lambda}$ and $P_{Y}= P_{Y|\Lambda}$ that there is no common cause between the local settings $X$ and $Y$ and the source $\Lambda$.
A small amount of correlation is required to produce a false BI violation, therefore, a Bell test must use QRNGs to randomly determine the measurement settings in real-time, ensuring that no prior knowledge could influence the results, hence closing the measurement dependence loophole~\cite{Gallicchio2014}. In 2017, a groundbreaking experiment, known as the ``Cosmic Bell Test'' the light from distant stars was used to choose measurement settings, arguing that the light had traveled for hundreds of years and thus could not be influenced by hidden variables~\cite{handsteiner2017}.

\paragraph{Other loopholes}
\textit{Coincidence-Time Loophole} -- Coincidence windows can create spurious correlations if the time window for considering detection events as part of the same pair is too wide. Then,  nanosecond-level timing precision are used for tight synchronization and narrow coincidence time windows~\cite{tittel1998}.
Future quantum networks employing repeater stations and tight coincidence timing windows will further ensure the proper pairing of entangled photons~\cite{Sangouard2011}.

\textit{Memory Loophole} -- The memory loophole arises if detectors have some form of memory from previous trials, which could influence future results. Experiments must randomized trials and reset the system after each trial to avoid memory effects~\cite{rowe2000} (see Sec.\ref{sec:memoryloopholeDIQKD}).
 
Finally, \textit{super-determinism} is a theoretical loophole that challenges the assumption of \textit{free will} in choosing measurement settings \cite{Hance2022}. Although superdeterminism is not directly testable in the traditional sense, the scientific consensus generally assumes that randomness and independence in quantum processes are valid. However, for DI-QKD its significance is largely philosophical—if all variables were predetermined, any “key” would already be known—. In contrast, for DI quantum random number generation, whose aim is to certify fresh randomness, the loophole poses a more substantive concern. A concise review of this philosophical loophole can be found in Ref.,\cite{brans1988}.

\subsubsection{Experimental Breakthroughs}

\noindent The timeline in Fig.\ref{fig:timeline} refers to the first definitive closure of the detection and locality loopholes, simultaneously referred to as ``loophole-free'' Bell tests.  In 2015, the first experiment was reported, using entangled photons and electrons in NV centers located at a distance of about 1.3 kilometers \cite{Hensen2015a}. All photonic experiments were also reported that same year:  Ref. \cite{US15} used high-efficiency photon detectors and random measurement settings, and similarly, the experiment in Ref. \cite{Vienna15and32_1} used highly efficient detectors and a large spatial separation between detectors.
Compared to previous Bell tests using entangled photons, the critical component here was high-efficiency superconducting photo-detectors, which permitted the realization of experiments above the CDE. 

A loophole-free Bell test using an event-ready setup with entangled neutral atoms in ~\cite{Munich17}, where atom-photon entanglement and entanglement swapping to prepare entangled spin states of two atoms separated by 398 m;
In Ref. \cite{li18}, all three major loopholes were addressed using randomness from photons emitted by cosmic sources to determine the measurement settings. This approach effectively closes the locality loophole by ensuring that the measurement settings are not influenced by any local hidden variable by using events that occurred 11 years prior (see sec. \ref{sec:experiments} for experimental details).

\subsection{Other notions of nonclassicality that can power Cryptography}
\paragraph{Quantum Steering}
Similarly to the LHV framework introduced, we can define a Local Hidden State (LHS) model. The operational diagrams are \cite{pcavalcanti2024}
\begin{equation}
\label{operationaldiagram}
p(ab|xy) = \vcenter{\hbox{
	\begin{tikzpicture}
		\node  (0) at (0, -1.2) {};
		\node  (1) at (-0.75, -0.45) {};
		\node  (2) at (0.75, -0.45) {};
		\node  (3) at (0, -0.7) {};
		\node  (4) at (0, -0.75) {$\Lambda$};
		\node  (5) at (-1.25, 0.3) {};
		\node  (6) at (-1.25, -0.2) {};
		\node  (7) at (-0.25, -0.2) {};
		\node  (8) at (-0.25, 0.3) {};
		\node  (9) at (-1, -0.45) {};
		\node  (10) at (-0.75, 0.3) {};
		\node  (11) at (-0.5, -0.2) {};
		\node  (12) at (-1, -0.2) {};
		\node  (13) at (-0.75, 0.05) {$M_{A|x}$};
		\node  (14) at (-0.25, -0.45) {};
		\node  (15) at (-0.75, 0.5) {};
		\node  (16) at (-1.25, 0.5) {};
		\node  (17) at (-0.25, 0.5) {};
		\node  (18) at (-0.75, 1) {};
		\node  (19) at (-0.75, 0.75) {};
		\node  (20) at (-0.75, 0.75) {$a$};
		\node  (21) at (-1, -0.75) {$x$};
		\node  (22) at (1.25, 0.3) {};
		\node  (23) at (1.25, -0.2) {};
		\node  (24) at (0.25, -0.2) {};
		\node  (25) at (0.25, 0.3) {};
		\node  (26) at (1, -0.45) {};
		\node  (27) at (0.75, 0.3) {};
		\node  (28) at (0.5, -0.2) {};
		\node  (29) at (1, -0.2) {};
		\node  (30) at (0.75, 0.05) {$N_{B|y}$};
		\node  (31) at (0.25, -0.45) {};
		\node  (32) at (0.75, 0.5) {};
		\node  (33) at (1.25, 0.5) {};
		\node  (34) at (0.25, 0.5) {};
		\node  (35) at (0.75, 1) {};
		\node  (36) at (0.75, 0.75) {};
		\node  (37) at (0.75, 0.75) {$b$};
		\node  (38) at (1, -0.75) {$y$};
		\draw (0.center) to (1.center);
		\draw (0.center) to (2.center);
		\draw (2.center) to (1.center);
		\draw (5.center) to (6.center);
		\draw (5.center) to (8.center);
		\draw (8.center) to (7.center);
		\draw (7.center) to (6.center);
		\draw [bend right=45, looseness=0.50] (11.center) to (14.center);
		\draw (15.center) to (10.center);
		\draw (18.center) to (17.center);
		\draw (17.center) to (16.center);
		\draw (16.center) to (18.center);
		\draw (12.center) to (9.center);
		\draw (22.center) to (23.center);
		\draw (22.center) to (25.center);
		\draw (25.center) to (24.center);
		\draw (24.center) to (23.center);
		\draw [bend left=45, looseness=0.50] (28.center) to (31.center);
		\draw (32.center) to (27.center);
		\draw (35.center) to (34.center);
		\draw (34.center) to (33.center);
		\draw (33.center) to (35.center);
		\draw (29.center) to (26.center);
	\end{tikzpicture}
	}},\qquad
	\sigma_{a|x} = \vcenter{\hbox{
	\begin{tikzpicture}
		\node  (0) at (0, -1.2) {};
		\node  (1) at (-0.75, -0.45) {};
		\node  (2) at (0.75, -0.45) {};
		\node  (3) at (0, -0.7) {};
		\node  (4) at (0, -0.75) {$\Lambda$};
		\node  (5) at (-1.25, 0.3) {};
		\node  (6) at (-1.25, -0.2) {};
		\node  (7) at (-0.25, -0.2) {};
		\node  (8) at (-0.25, 0.3) {};
		\node  (9) at (-1, -0.45) {};
		\node  (10) at (-0.75, 0.3) {};
		\node  (11) at (-0.5, -0.2) {};
		\node  (12) at (-1, -0.2) {};
		\node  (13) at (-0.75, 0.05) {$M_{A|x}$};
		\node  (14) at (-0.25, -0.45) {};
		\node  (15) at (-0.75, 0.5) {};
		\node  (16) at (-1.25, 0.5) {};
		\node  (17) at (-0.25, 0.5) {};
		\node  (18) at (-0.75, 1) {};
		\node  (19) at (-0.75, 0.75) {};
		\node  (20) at (-0.75, 0.75) {$a$};
		\node  (22) at (-1, -0.75) {$x$};
		\node  (23) at (1, 0.3) {};
		\node  (24) at (1, -0.2) {};
		\node  (25) at (0.25, -0.2) {};
		\node  (26) at (0.25, 0.3) {};
		\node  (28) at (0.75, 0.3) {};
		\node  (29) at (0.5, -0.2) {};
		\node  (31) at (0.55, 0.05) {$\rho_{\Lambda}$};
		\node  (32) at (0.25, -0.45) {};
		\node  (33) at (0.75, 0.75) {};
		\draw (0.center) to (1.center);
		\draw (0.center) to (2.center);
		\draw (2.center) to (1.center);
		\draw (5.center) to (6.center);
		\draw (5.center) to (8.center);
		\draw (8.center) to (7.center);
		\draw (7.center) to (6.center);
		\draw [bend right=45, looseness=0.50] (11.center) to (14.center);
		\draw (15.center) to (10.center);
		\draw (18.center) to (17.center);
		\draw (17.center) to (16.center);
		\draw (16.center) to (18.center);
		\draw (23.center) to (24.center);
		\draw (23.center) to (26.center);
		\draw (26.center) to (25.center);
		\draw (25.center) to (24.center);
		\draw [bend left=45, looseness=0.50] (29.center) to (32.center);
		\draw (33.center) to (28.center);
		\draw (12.center) to (9.center);
	\end{tikzpicture}
}}
\end{equation}
The left-hand side of Eq.~\eqref{operationaldiagram} represents an LHV model for the Bell test. We note that such graphical equations follow the framework introduced in~\cite{Chiribella2011}.

\begin{definition}[Local Hidden State (LHS) model]
	Let us consider Alice's measurements with the POVM $M_{A|x}=\{M_{a|x}\}_a$ on $\rho_{AB}\in\mathcal{B}(\mathbb{C}^{\mathrm{d}_A}\otimes\mathbb{C}^{\mathrm{d}_B})$, such that the update conditional state on Bob's side is given by
	\begin{equation}
		\rho_{a|x}=\frac{\sigma_{a|x}}{p_{A|X}(a|x)},\qquad
		\sigma_{a|x}=\mathrm{Tr}_A[(M_{a|x}\otimes \bm{1}_B)\rho_{AB}], \qquad
		p_{A|X}(a|x)=\mathrm{Tr}\sigma_{a|x}>0.
	\end{equation} 
The collection $\{\sigma_{a|x}\}_{a,x}$, a.k.a. assemblages, is said to admit a LHS model if there exists: (i) a classical random variable \(\lambda\) with probability distribution \(p(\lambda)\),
	(ii) a set of conditional probability distributions \(p(a|x,\lambda)\), (iii) a collection of normalized quantum states \(\{\sigma_\lambda\}_\lambda\in\mathcal{B}(\mathcal{H}_B)\),
such that the following decomposition holds (discrete case):
\begin{equation}\label{eq:steering}
\sigma_{a|x} = \sum_\lambda p_\Lambda(\lambda)\, p_{A|X\Lambda}(a|x,\lambda) \, \sigma_\lambda, \quad
\land \quad
\rho_B = \operatorname{Tr}_A[\rho_{AB}] = \sum_a \sigma_{a|x} \quad \forall x,a.
\end{equation}
\end{definition}
\noindent Bob performs full tomography of the quantum state $\rho_{a|x}$ that is effectively prepared in his lab after Alice's action. Then the LHS correlations are:
\begin{equation}
	p(a,b|x,y) = p_{B|Y}(b|y,\sigma_{a|x})=p_{A|X}(a|x) p_{B|Y}(b|y,\rho_{a|x})
	= \sum_\lambda p_{A|X\Lambda}(a|,\lambda) p_{B|Y}(b|y\rho_{a|x}\lambda)  p_{\Lambda}(\lambda),
\end{equation}
where Bob’s conditional probability is $p_{B|Y,\rho}(b|y,\rho)=\text{tr}[N_{b|y} \rho]$, for $\rho\in \{\sigma_{a|x}\}_{a,x}$. 
\begin{definition}[Quantum Steering]
A bipartite quantum state \(\rho_{AB}\) is said to be unsteerable (from Alice to Bob) if Eq. \eqref{eq:steering} holds, otherwise is said to be steerable (from Alice to Bob), $\rho_{AB}\in\mathcal{S}_{A\to B}$.
\end{definition}
In other words, quantum steering is exhibited when the correlations between Alice’s outcomes and Bob’s conditional states cannot be explained by a classical mixture of preexisting states on Bob’s side (in the diagram rather than $\Lambda$ is required $\rho_{AB}$).
Notice that steering is directional ($\rho_{AB}\in \mathcal{S}_{A\to B} \land \rho_{AB}\notin \mathcal{S}_{B\to A}$).  Whether in DI-QKD the nonlocal correlation $\bm{p}\in\QC\setminus\LC \Longrightarrow$ Alice and Bob are untrusted (their measurement devices are ``black boxes'' -- unknown to the experimenter), a \textit{certification} of a steering state (Steering Inequality SI violation) allows \textit{one-sided DI-QKD (1SDI-QKD)}: only Alice can be trusted (her measurement devices are well-characterized), while Bob’s devices remain untrusted \cite{Cavalcanti2022} (see section \ref{sec:1SDIQKD}). Specifically, $\bm{p}(\rho_{AB})\in\QC\setminus\LC\Longrightarrow \rho_{AB}\in\mathcal{S}_{A\to B}\Longrightarrow \rho_{AB}$ entangled, but the \textit{only if} does not hold.
Follows a series of interesting facts: 
\textit{(i)}
 SI violation requires a lower CDE than its analogous BI violation. Indeed for loophole-free steering with qubits and $N$ measurement settings $\eta^* \propto 1/N$ \cite{branciard2012one,Cavalcanti15,mattar17}. 
\textit{(ii)}
 SI are easier to test than BI using continuous-variable (CV) systems, due to the fact that  high-efficiency Gaussian measurements (such as homodyne) are not sufficient to violate BIs in Gaussian states \cite{revzen05gauss}.  This is not an issue for demonstration of quantum steering, and indeed the first demonstration was realized in 1992 using homodyne measurements on entangled optical fields \cite{ou1992realization}. Since then, a series of experiments have demonstrated one-way Gaussian steering \cite{handchen12oneway} and non-Gaussian steering \cite{walborn11a}. 
\textit{(iii)}  
Loophole-free SI violations in discrete-variable (DV) systems were demonstrated in 2012 using polarization-entangled photon pairs and superconducting detectors for high detection efficiency~\cite{Smith12,Wittmann2012,Bennet12}. Ref.~\cite{Smith12} closed the detection loophole with superconducting detectors; Ref.~\cite{Bennet12} achieved steering over 1\,km of optical fiber despite lower efficiencies. Ref.~\cite{Wittmann2012} simultaneously closed detection, locality, and measurement-dependence loopholes using measurements 48\,m apart. See also~\cite{Cavalcanti_2017,Uola20_RMP} and Sec.~\ref{sec:experiments}.

\paragraph{Contextuality}
\noindent Rooted in the Kochen-Specker (KS) paradox, contextuality, another nonclassicality notion, reveals the impossibility of \textit{realism} of \ref{def:realism}, i.e. assigning pre-existing values to quantum observables independently of measurement context. A contextuality-based DI-QKD scheme, exemplified by the Peres-Mermin square \cite{Peres1990,Mermin1990}, uses a bipartite system satisfying KS paradox conditions locally while exhibiting perfect distant correlations \cite{Horodecki2010}. This ensures secure key extraction, as any eavesdropping attempt by Eve introduces detectable errors. 
Unlike Bell-based methods, contextuality relies on the trade-off between information gain and disturbance, tied to quantum uncertainty \cite{Pusey2014,Catani2022} and wave-particle duality \cite{Catani2023,Wagner2024}. Variants like \textit{generalized contextuality} \cite{Spekkens2016}, hyperbits \cite{Scala2024}, Kirkwood-Dirac distribution \cite{ArvidssonShukur2024}, witwords \cite{Cavalcanti2022}, and overall Generalized Probabilistic Theories \cite{Janotta2014,Mazurek2021} highlight quantum advantages for DI cryptography.
\section{Fully Device Independent Quantum Key Distribution (DI-QKD)}\label{sec:DI-QKD}
\noindent We discussed that device-dependent cryptography permits \textit{inventa fraus} \cite{lydersen10}; \textit{facta lexia}, DI-QKD eliminates them via BI violation or other nonclassical notions. In this section, we are going to introduce the DI protocols that are powered by \textit{quantum theory} to enhance security.
\subsection{Equivalence principle for simulation-based security}
DI-QKD protocols are analyzed in \textit{simulation-based security}, where the real execution of a protocol is compared to an ideal one: if no adversary can distinguish the two, then the real protocol is considered secure.
\begin{definition}[Indistinguishable protocols]
Let \(\Pi\) A protocol which take inputs and produce outputs. Let $\Pi_1=\mathcal{Z}_1(\Pi) $ and $\Pi_2=\mathcal{Z}_2(\Pi')$  the protocols affected by the presence of an external environments. We define an equivalence relation \(\Pi_1\sim_{\epsilon(n)} \Pi_1\) if for all probabilistic polynomial-time (PPT) environment \(\mathcal{Z}_1,\mathcal{Z}_2\), there exists a negligible function $\epsilon(n)$ (in the security parameter \(n\)) such that 
\begin{equation}
   \operatorname{Adv}(\Pi_1, \Pi_2) \leq \epsilon(n),
   \qquad
   \operatorname{Adv}(\Pi, \Pi') \coloneqq \Bigl|\Pr\big[\Pi_1=1\big] - \Pr\big[\Pi_2=1\big]\Bigr|.
\end{equation}
where $\operatorname{Adv}$ is the distinguishing advantage function.
\end{definition}
\noindent In other words, no efficient environment can tell $\Pi$ apart from $\Pi'$ with more than the negligible advantage $\epsilon(n)$.
\begin{definition}[Simulation Security]
We denote a real protocol executed among parties which may be corrupted by a real adversary $\mathcal{E}_R$ as a function \(\mathcal{E}_R(\Pi^R)\) and an ideal protocol \(\Pi^I\) that receives inputs from the parties and returns outputs that are guaranteed to satisfy the security properties in the presence of a simulated adversary (simulator) $\mathcal{E}_S$ also as a function \(\mathcal{E}_S(\Pi^I)\). The protocol \(\mathcal{E}_R(\Pi^R)\) is said to securely realize the ideal functionality \(\Pi^I\) if for every PPT adversary \(\mathcal{E}_R\) there exists a PPT simulator \(\mathcal{E}_S\) such that $\mathcal{E}_S(\Pi^I)\sim_{\epsilon(n)}\mathcal{E}_R(\Pi^R)$.
\end{definition}
\noindent This definition captures the intuition that any attack on \(\Pi^R\) in the real world can be simulated in the ideal world, so that no environment can distinguish the two executions except with negligible probability.
In practice, protocols are rarely executed in isolation. The following UC framework requires that security be preserved even when the protocol is composed with an arbitrary set of other protocols.
\begin{definition}[Universal Composable (UC) Security \cite{Canetti2000,Canetti2001,BenOr2004,BenOr2005,Renner2005a,renner2008security}]
A protocol \(\Pi^R\) UC-realizes an ideal functionality \(\Pi^I\) if for every PPT adversary \(\mathcal{E}_R\) even in the presence of arbitrary concurrent protocol executions $\{\Pi_i\}_i$) there exists a PPT simulator \(\mathcal{E}_S\) such that for every PPT environment \(\mathcal{Z}\), $\mathcal{Z\circ\mathcal{E}_R}(\Pi^R,\{\Pi\}_i)\sim_{\epsilon(n)}\mathcal{Z}\circ\mathcal{E}_S(\Pi^I,\{\Pi\}_i)$, that is
\begin{equation}\label{sim_real}
  \Bigl|\Pr\big[\mathcal{Z}\circ\mathcal{E}_R\left(\Pi^R,\{\Pi\}_i\right)=1\big] - \Pr\big[\mathcal{Z}\circ \mathcal{E}_S\left(\Pi^I,\{\Pi\}_i\right)=1\big]\Bigr| \leq \epsilon(n).  
\end{equation}
\end{definition}
Here, the environment \(\mathcal{Z}\) is allowed arbitrary interactions with all components (including \(\Pi\) as a subroutine, and any other concurrently running protocols $\{\Pi_i\}_i$), and the security guarantee must hold regardless of the surrounding context and any efficient environment.
In the context of DI-QKD, this indistinguishability principle is operationalized via the following question: \textit{Can Eve distinguish whether she is interacting with the real protocol executed with untrusted, potentially malicious quantum devices, or with an ideal protocol that outputs a uniformly random secret key independent of her side information?} If no efficient test exists that can tell the difference, then the protocol is deemed secure. The security of the protocol is ensured by the fact that Eve’s intervention breaks the expected equivalence, which is detected in the security check step of Box 1 (Sec.~\ref{subsec:QKD}) through the \textit{lack} of BI violation.

\subsection{Bell inequalities bound eavesdropper's knowledge}
\noindent 
The \textit{statistical security check step} notifies Eve's presence. 
Indeed, her action with the protocol is modeled by a non-signaling distribution $p=\prod_{i=1}^np_{A^1B^1E^1|X^1Y^1Z^1}$$(a_i,b_i,e_i|x_i,y_i,z_i)$, or among $n$ rounds with $p=p_{A^nB^nE^nT|X^nY^nZ^n}(\{a_i,b_i,e_i\}_i,t|\{x_i,y_i,z_i\}_i)$ including $T=\{(a_i,b_i,x_i,y_i),\text{hash}\}_{i}$ and $\mathrm{hash}$ is the collection of post-processing functions for error correction and privacy amplification.
The simulator's job (security proof) is then to construct a hypothetical ideal non-signaling distribution $p^I$ that matches Eve’s observations, but where the secret key is perfectly secure. It must hold for any possible initial non-signaling distribution $p$ where only the marginal $\bm{p}=p_{AB|XY}$ is used by Alice and Bob for the \textit{security check} that estimate Eve's knowledge by a function of the Bell value (e.g. \eqref{eq:CHSHSS}) $f[\beta_\QC(\bm p)]$. Given that, a DI-QKD protocol can be realized as a transformation of the form
\begin{equation}\label{p2p}
    p\mapsto \Pi_{K_A^n,K_B^n,f,T,E^n|X^nY^nZ^n}
\end{equation} 
and it is sub-routine that satisfy UC-security~\cite{Masanes2014} \eqref{sim_real} when warrantees that the composed scheme that uses QKD is secure as if an ideal secret key were used instead
\begin{equation}
    \sum_{K_A,K_B,t} \max_z \sum_e|p^{R}_{K_A,K_B,T,E|Z}-p^{I}_{K_A,K_B,T,E|Z}|=\mathcal{O}(1/N), \qquad p^{I}_{K_AK_BTE|Z}=\frac{1}{2^{N_S}}\delta_{k_A,k_B}p^{R}_{TE|Z}.
\end{equation}
Th transformation \eqref{p2p} captures the operational steps of the protocol in Box 1 of Sec.\ref{subsec:QKD}.
Importantly, this modular perspective separates the description of the protocol from the security proof: the protocol defines how the classical and quantum data are processed, while the security proof (the simulator) analyzes the transformation's outcome under the appropriate assumptions. This modularity highlights that the protocol structure can be independently specified without requiring the proof to be embedded in its definition. 
However, each protocol is secure by simulating an adversary constrained to perform specific attacks \cite{Metger2024}.  
\begin{definition}[Simulation attacks]
\label{def:attacks}
Given a global probability distribution $\bm p$ during the data generation step, an attack is simulated by computing  Eve's knowledge about the outcome, let us say, $b$  if Bob measures $Y=y_1$, by the optimal guessing probability
\begin{equation}
p_{\mathrm{g}}(b|E)=\max_z\sum_e \bm p,\qquad
\bm p=
\begin{cases}
        p_{B^1E^1|Y^1Z^1}(b,e|y=y_1,z), &\text{individual}\\
    \prod_{i=1}^n p_{A^1B^1E^1T|X^1Y^1Z^1}(a_i,b_i,e_i,t|x_i,y_i,z_i),& \text{collective}\\
    p_{A^nB^nE^nT|X^nY^nZ^n}(\{a_i,b_i,e_i\}_i,t|\{x_i,y_i,z_i\}_i),& \text{coherent}
\end{cases}
\end{equation}

- \textit{Individual attack} -- before the public discussion (PD), Eve tries to guess at each round performing an identical operation on each qubit on the quantum channel and keeps only classical information $Z$. Shannon entropy measures $p_g$ ;

- \textit{Collective attack} -- after the PD the protocol $\Pi^R$  has already generated the keys $k_A,k_B$ and Eve's can read all the public discussion in $T$. Memorylessness via space-like separation (Def.~~\ref{def:space-like}) of devices \textit{at each run} implies the factorization. Eve keeps quantum information but still performs the same operation on every round. Von Neumann entropy measures $p_g$
 
- \textit{Coherent attack} -- Eve has a (classical or quantum) memory to correlate inter-round communication. To measure $p_g$ see Sec. \ref{sec:chap4}.

- \textit{Non-signaling attack} (a.k.a known as \textit{post quantum attack}) -- Eve can prepare $p^\lambda$, where $\lambda$ is a post-quantum state allowing correlations stronger than quantum, as shown in Fig. \ref{fig:NL-NS}.
\end{definition}
Though memorylessness is impractical, in reality, Alice and Bob reuse single devices with a refresh routine to cancel memory effects. However, confined Eve to quantum collective attacks is already far weaker than (non-DI) QKD assumptions, because all these attacks are bounded by a functional $f$ of the Bell value $
P_{\mathrm{g}}(b|E) < f[\beta_\QC(\bm{p})]$ regardless of how $\bm{p}=p_{AB|XY}$ is produced. 
\begin{proposition}
    Given a specific power to Eve that fix the simulation, a real and an ideal protocol $\Pi^R$ $\Pi^I$ satisfy a given simulation security iff $p_{\mathrm{g}}(b|E) < f[\beta_\QC(\bm{p})]$ with $\bm p\in\NSC\setminus \LC$.
\end{proposition}
Apart from protocol in Ref.~\cite{Barrett2005} that shows a security proof in a post-quantum theory, let us relate $p_g$ with theoretical information quantities within quantum theory.
\begin{definition}\label{def:c-min-entropy}
Let $\rho_{BE}=\sum_b p_b \ket{b}\bra{b} \otimes \rho_{b|E} \in \mathcal{B}(\mathcal{H}_B \otimes \mathcal{H}_E)$ a cq-state the quantum,
\begin{itemize}
    \item conditional min-entropy of $B$ given $E$ is (last equality for binary outcomes):
\begin{equation}\label{H=-logP}
    H_{\min}(B|E)_\rho:=\sup\{\lambda>0:2^{-\lambda}\bm 1_B\otimes \sigma_E-\rho_{BE}\succeq0,\,\sigma_E\in\mathcal{B}(\mathcal{H}_E)\} = -\log_2 p_{\mathrm{g}}(b|E)   
\end{equation}
\item Conditional $\epsilon$-smoothed min-entropy \cite{Dupuis2020}, $\varepsilon\in[0,1]$ is

    \begin{align}\label{eq:minHep}
	H_{\min}^\varepsilon(B|E)_\rho :=& \sup_{\tilde{\rho}_{BE} \in \mathcal{B}^\varepsilon(\rho_{BE})} H_{\min}(B|E)_{\tilde{\rho}}= -\log \inf_{\tilde{\rho}\in \mathcal{B}_\varepsilon(\rho_{BE})} \inf_{\sigma_E\in \mathcal{B}(\mathcal{H}_E)} \parallel \tilde{\rho}_{BE}^{\frac{1}{2}}\sigma_{E}^{-\frac{1}{2}}\parallel_{\infty}^2, \\
    \mathcal{B}_\varepsilon(\rho_{E})=&\{\tilde{\rho}|\tilde{\rho}\succ 0,\mathrm{Tr}\tilde{\rho}<1,\,\sqrt{1-||\sqrt{\rho_{BE}}\sqrt{\tilde{\rho}}||_1^2}\le \varepsilon\}.
\end{align}
\end{itemize}
\end{definition}

$H_{\min}(B|E)$ is a worst-case entropy measure used for UC-security proofs. It quantifies Eve’s maximum ability to guess Bob’s value \( b \) given access to the quantum system \( E \), by bounding \( \rho_{BE} \) above by a scaled product state \( \mathbb{I}_B \otimes \sigma_E \). The smoothed version \( H_{\min}^{\varepsilon}(B|E)_\rho \) accounts for finite-size effects by optimizing over subnormalized states \( \tilde{\rho} \) within purified distance \( \varepsilon \) of \( \rho_{BE} \).
Operationally, \( H_{\min}^{\varepsilon}(K|E)_\rho \) determines the maximum number of secret bits extractable from raw data \( K \), secure against an adversary holding \( E \), up to error \( \varepsilon \). Unlike von Neumann entropy, which captures average uncertainty, smoothed min-entropy provides tight bounds reflecting worst-case adversarial knowledge.

\begin{proposition}
In the iid asymptotic limit, this entropy converges (after smoothing) to the conditional von Neumann entropy:
  $$
  \lim_{\varepsilon \to 0} \lim_{n \to \infty} \frac{1}{n} H_{\min}^\varepsilon(B^n|E^n) = H(B|E),
  \qquad
  H(B|E)=H(\rho_{BE})-H(\rho_E),
  \quad
  H(\rho)=-\mathrm{Tr}\rho\log\rho,
  $$
  where the von Neumann entropy $H(\rho)$ generalized the Shannon entropy $H(X)=-\sum_ip_X(x_i)\log_2p_X(x_i)$ of a classical variable $X$ with probability distribution $p_X$.
\end{proposition} 
\( H_{\min}^{\varepsilon}(K|E)_\rho \) is related to the achievable key length, i.e. the number of uniform secret bits that can be extracted from the raw key \( K \), given the adversary’s quantum side information \( E \), up to \( \varepsilon \). In the finite-key regime, the secret key rate $r$ and the extractable key length \( \ell \) are bounded as
\begin{equation}\label{eq:betabound}
    r=\eta\frac{\ell}{n},\qquad\ell \leq H_{\min}^{\varepsilon}(K|E)_\rho - \mathrm{leak}_{\mathrm{EC}} - \mathcal{O}(\log \tfrac{1}{\varepsilon})=
    \begin{cases}
        n(1- h(\mathrm{QBER}))- \mathcal{O}(\log \tfrac{1}{\varepsilon}) & \text{QKD}\\
        n(1- h(\mathrm{QBER})-f(\beta_\QC))- \mathcal{O}(\log \tfrac{1}{\varepsilon}) & \text{DI-QKD}
    \end{cases}
\end{equation}
where $h$ is the binary entropy, \( \mathrm{leak}_{\mathrm{EC}} \) accounts for information revealed during error correction, $\eta=\eta_{\mathrm{link}}+\eta_d$ modifies the effective number of rounds contributing to key generation where \( \eta_l = 10^{-a d} \) represents line losses (transmittivity) over a distance \( d \) with attenuation coefficient \( a \) and \( \eta_d \) is the efficiency of the detectors (see Sec.\ref{sec:experiments}). The last equality holds for one-way error correction and ideal privacy amplification\footnote{In one-way postprocessing ``error correction and privacy amplification'' $T$ flows only from Bob to Alice or vice-versa, efficient in terms of secret key rate;
in two-way postprocessing ``advantage distillation'' $T$ flows from both Alice and Bob. This is inefficient for small errors but tolerating larger errors (see \ref{sec:CHSHprotocol}).}
. $f(\beta_\QC)$ upper bounds Eve’s information as a function of $\beta_\QC(\bm p)$ such that $\bm p \in \QC\setminus\LC$. The QBER may still appear due to experimental imperfections, or alternatively the bound may depend only by $\beta_\QC$.
The \emph{secret key rate} \( r \) is defined as the number of secret bits extracted per use of the quantum channel. 
In the asymptotic limit, it reduces to the Devetak–Winter bound 
\begin{equation}\label{devatekwinterformula}
    r \geq \sup_{T}  I(A:B)-I(A:E)=\sup_{T} H(K|E) - H(K|B),
\end{equation}
showing consistency between one-shot and large-scale QKD analyses.
$I$ is the mutual information and $T$ represents \emph{preprocessing}  \footnote{In QKD, preprocessing denotes any local operation applied to the raw key \emph{before} information reconciliation and privacy amplification, with the aim of improving security or noise tolerance. For the specific noisy–preprocessing scheme used in this work, see Sec.~\ref{subsec:CHSHT}.}

In the next, we discuss the developments and the ideas in the variants of protocol in Box 1 listed in Tab.~\ref{tab:my_label}. We are interested on the \textit{data generation} and the \textit{statistical security check}  to focus only on the term $f(\beta_\QC)$ in Eq.~\eqref{eq:betabound} since the \textit{raw key validation}, that involves perfect correlated outcomes, is the same of all QKD protocols. Recall that in DI-QKD Eve controls the source and the fabrication of Alice and Bob's devices, therefore for the users the preparations and measurements become merely input and output classical labels.

\begin{table}[ht]
    \centering
    \begin{tabular}{|p{2.1cm}|p{13cm}|}
    \hline
    \textbf{Protocol} & \textbf{Description}\\
    \hline
        E91 \ref{sec:E91}~\cite{E91}& Within QT, E91 uses CHSH for secrecy without quantifying $p_g$ \\  \hline
        BHK ``05 
        \ref{sec:nosignalQKD}~\cite{Barrett2005}&$p_g$ via post-quantum individual attacks is bounded by chain inequality, but no rate is computed.\\ \hline
        CHSH ``06 \ref{sec:CHSHprotocol}~\cite{Acin2006b,scarani2006secrecy}&First key rate $r$ is computed bounding $p_g$ via post-quantum individual attacks using CHSH and CGLMP in 1w and 2w PD.
        \\ \hline
        Chain ``06 \ref{sec:CHAIN-protocol} \cite{Acin2006}
        &Improvement of $r$ bounding $p_g$ via post-quantum individual attacks enlarging Bell scenario in 1w PD via chain inequality (UC-security~\cite{Masanes2014}, CV version~\cite{Marshall2014} in ``14).\\ \hline
        CHSH\textsubscript{$\chi$} ``07 \ref{DIQKD-collective} \cite{acin2007device} & Improvement of $r$ bounding $p_g$ obtained with post-quantum collective attacks via Holevo quantity $\chi$ in 1-w PD (rigorous proof in ``09~\cite{pabgs09andSangouard.4}, UC-secure in ``09~\cite{Masanes2009}).\\ \hline
        CHSH$_{p_g}$ ``10 \ref{CHSHp}~\cite{masanes2011secure,hanggi2010device} & $r$ improved bounding $p_g$ in 1-w PD from coherent attacks (only independent devices assumed -- memory attacks ``13~\cite{barrett2013memory}) via $H_{\min}$ achieving $Q=7.1\,\%$ ($\beta_\QC=2.423$). 2-w PD secure in ``20~\cite{tan2020advantage}, $\eta^*=0.891$.\\ \hline
        CHSH$_T$ ``20 \ref{subsec:CHSHT}~\cite{Ho2020}& Implement Entropy Accumulation theorem (EAT) (see sec.~\ref{sec:chap4}) for the bound on $r$ robust under coherent attacks reducing $\eta^*=0.832$. \\ \hline
        CHSH\textsubscript{$\ell$} ``13 \ref{subsec:CHSHl}~\cite{lim2013device}& Detection loophole is closed by performing a CHSH test entirely within Alice's lab using causally independent devices.\\ \hline
        CH-SH ``21 \ref{ch-sh}~\cite{Sekatski2021,Woodhead2021}& $p_g$ bounded from refined CHSH data using either asymmetric expressions or full ``$(X=CH,Y=SH)$'' statistics; both improve $r$ and match $\eta^*=0.826$. \\ \hline
        
         CHSH$_{\text{IMD}}$``21 \ref{CHSHIMD}~\cite{Brown2021} & Rényi‑like quantities—the iterated‑mean divergences—to with tighter entropy bounds than previous numerical techniques. \\ \hline
        CHSH\textsubscript{$2e$} ``21 \ref{subsubsec:randomkeybais}~\cite{Schwonnek2021}& $r$ improved by doubling the events for the raw key in a single round under quantum collective attacks in 1-w PD achieving $Q=8.2\,\%$ ($\beta_\QC=2.362$). \\ \hline 
        \hline
        CHSH$_{p_{\mathcal{V}_p}}$ ``22 \ref{subsub:randompostselection}~\cite{Xu2022}& $r$ improved under quantum collective attacks using random postselection of raw key bit in 1-w PD, reducing the effect of no-detection events. $\eta^* = 0.685$ \ref{tab:diqkd_threshold}.\\ \hline
       MPG ``23 \ref{msg-di-qkd}~\cite{Zhen2023a}& A protocol based on the Mermin-Peres magic square game, surpassing CHSH-based approaches in regimes of high state visibility and detection efficiency.\\ \hline
       rDI-QKD ``24 \ref{sec:routedBI} \cite{chaturvedi2024extending,Lobo2024,Tan2024,LeRoyDeloison2025,Sekatski2025,Chaturvedi2025}& Introduces routing-based Bell tests to enable long-range DI-QKD; security relies on witnessing long-range nonlocality. Routed Bell inequalities certified even under loss or weak detection, achieving critical efficiency $\eta_{A_1}>1/2$, with exponential robustness in parallel repetition. \\
       \hline
    \end{tabular}
    \caption{Summary of DI-QKD protocols. QT: quantum theory; $p_g$: Eve's guessing probability of $k_i$, the sifted bit; 1-w: one-way public discussion (PD); UC: universal composite security; r: secure key rate; CV: continuous variables; $\eta^*=$ critical efficiency. We discuss the advanced security proof started from 2014 in Sec.~\ref{sec:chap4}.}
    \label{tab:my_label}
\end{table}

\subsection{Ekert (1991) -- Bell inequality for secrecy without bound}\label{sec:E91}

E91 is the first quantum key distribution (QKD) protocol leveraging Bell inequality violation for security~\cite{E91}, but with no bound on Eve's knowledge.

\textbf{Protocol} — \textit{Data generation}. Alice and Bob select $x,y\in\{0,1,2\}$ obtaining $a,b\in\{0,1\}$. In QT $M_{A|x}=\{\Pi_{0|x},\Pi_{1|x}\}$ and $N_{B|y}=\{\Pi_{0|y},\Pi_{1|y}\}$: 
\begin{equation}
    \Pi_{a|x}=U_{\theta_x}\ket{a}\bra{a}U_{\theta_x}^\dagger,\quad
    U_{\theta}=\mathrm{e}^{-\mathrm{i}\frac{\theta}{2}\sigma_y},\quad
    \theta_{x=(0,1,2)}=\left(\frac{\pi}{2},\frac{\pi}{4},0\right),\quad
    \theta_{y=(0,1,2)}=\left(\frac{\pi}{4},0,-\frac{\pi}{4}\right).
\end{equation}

\textit{PD and raw key validation}.
They publicly announce $x,y$. For $(x,y) = (0,2)$ they evaluate $\beta_\QC(\bm{p})$ and abort if $\beta_\QC(\bm{p}) \le 2$. But when settings match: $(x,y) = (1,0) \lor (2,1)$, they use the outcomes to extract the key. If they share $\ket{\phi^+} = \frac{1}{\sqrt{2}}(\ket{00} + \ket{11})$ they exploit perfect outcome correlations.

\textbf{Security} — 
in the absence of eavesdropping, E91 predicts $\beta_\QC(\bm{p}) = 2\sqrt{2}$. If Eve attempts to eavesdrop, she must interact with the system, inevitably introducing disturbances that degrade the observed nonlocal correlations. Any such unnoticed tampering constrains the statistics $\bm{p}$ to admit a local hidden variable (LHV) model, for which $\beta_\QC(\bm{p}) \le 2$.
More precisely, all LHV strategies can be written as convex mixtures over deterministic outcomes defined by hidden variables $n_a, n_b$:
$$
\langle A_x B_y \rangle = \int \rho(n_a, n_b) (a_x \cdot n_a)(b_y \cdot n_b) \, \mathrm{d}n_a \mathrm{d}n_b.
$$
This class of models forms the local polytope $\LC$, and defines the boundary of what is attackable without to be noticed.
Hence, a violation $\beta_\QC(\bm{p}) > 2 \Longrightarrow\bm{p} \notin \LC$ guarantees no eavesdropper.

Although the E91 protocol does not address the detection loophole caused by no-click events (see Sec.~\ref{subsec:detection-loophole}), it pioneered the principle that a Bell inequality violation enables secrecy independently of device trust, by certifying entanglement without assuming a specific Hilbert space structure. This insight paved the way to later protocols that guarantee security even beyond quantum correlations, i.e., for behaviors $\bm{p} \in \NSC \setminus \QC$.

More recently, proofs based on the (Generalized) Entropy Accumulation Theorem have established a rigorous, finite-size, composable security proof of E91 against coherent attacks. For details of this analysis and explicit bounds, we refer the reader to Sec.~\ref{subsec:E91GEAT}.

\subsection{DI-QKD Secure Against No-signaling Adversaries}
We start with presenting the earliest DI-QKD protocols that achieve security against a general non-signaling adversary—that is, an eavesdropper constrained only by the no-signaling principle, rather than by the laws of quantum mechanics (see definition \ref{def:attacks}). These adversaries may exploit post-quantum correlations, including those stronger than any quantum strategy, provided they do not enable faster-than-light communication. Security proofs in this regime rely on monogamy properties of non-local correlations, typically quantified via Bell inequality violations. The protocols discussed here preceded the development of entropy-based techniques and thus typically achieve limited key rates, but they remain important foundational milestones in the theory of device-independent cryptography.
\subsubsection{Barrett-Hardy-Kent (2005) -- bound on Eve's knowledge without positive rate}\label{sec:nosignalQKD}
Once the correlations $\bm p\notin \LC$ are information-theoretic resource~\cite{Barrett2005} a protocol the first protocol is derived with a security proof (against collective no-signaling) bounding Eve's knowledge with $f[\beta_\QC(\bm p)]$ but without computing the secrete key rate~\cite{Barrett2005}\footnote{Ref.~\cite{Barrett2005}. Tested in QT, it uses anticorrelation from $\ket{\psi}=\frac{1}{\sqrt{2}}(\ket{01}+\ket{10})$, here we keep $\ket{\psi}=\frac{1}{\sqrt{2}}(\ket{00}+\ket{11})$ as in E91.}.

\textbf{Protocol}. \textit{Data generation} --
Alice and Bob select $x,y\in\{0,\dots,m-1\}$ obtaining $a,b\in\{0,1\}$. In QT $M_{A|x}=\{\Pi_{0|x},\Pi_{1|x}\}$ and $N_{B|y}=\{\Pi_{0|y},\Pi_{1|y}\}$:
\[
\Pi_{a|x} = U_x \ket{a}\bra{a} U_x^\dagger, \qquad U_x = e^{-i \frac{\pi x}{2m} \sigma_y}.
\]
 \textit{PD and raw key validation} --
given $n=Mm^2$ with $M\in\mathbb{N}$, the protocol continues if 
\begin{itemize}
    \item [(i)] \begin{equation}
     2 M m \le \sum_{x=0}^{m-1} \sum_{c = -1}^1 \mathcal{M}_{x,c}, \qquad \mathcal{M}_{x,c} := \left| \left\{ j : X_j = x, Y_j = x + c \right\} \right|.
\end{equation}
    \item [(ii)] $\forall j = 1, \dots, s-1, s+1, \dots, \mathcal{M}_{x,c} $ rounded that reveals the outcomes $(X_j, Y_j) = (x, x+c)$
    \item [(iii)] Let $X_j\in\{0,\dots,m-1\}$ Alice's input at round $j$ and Bob's input $Y_j\in\{x-1,x,x+1\}$ for a given $X_j=x$. For all $j$ the following BI (a variant of the chained Bell inequality~\cite{Braunstein1990,Pearle1970}) must be violated:
    \begin{equation}\label{chainBI0}
    \beta_j(\bm{p}^\lambda) = \frac{1}{3m} \sum_{c = -1}^1 \sum_{x = 0}^{m-1} p^\lambda(a_j = b_j \mid X_j = x, Y_j = x + c)
        \le 1 - \frac{2}{3m}
    \end{equation}
\end{itemize}

The outcomes are kept private only for a designated setting pair $(X_s, Y_s) = (x, x + c)$, with $c \in \{-1, 0, 1\}$. The secret bit is then derived from the outcomes of this unrevealed round, defined as $a_s = b_s$.

\textbf{Security Proof}.
For her collective attack Eve prepares a post-quantum state \( \lambda \) so that she keeps a subsystem for her and 2n subsystems where at each run Alice and Bob believe to share a state as in E91. Thus, \( \lambda \) defines joint probability distributions \( p^\lambda_{A^nB^nEX^nY^nZ} \), where \( X^n = (X_1, \dots, X_n) \), \( Y = (Y_1, \dots, Y_n) \) are the players' measurement choices, and \( E = \{E_1\} \) is Eve’s \emph{time-independent} measurement. This means that Eve's choice and statistics do not depend on the players’ measurement rounds, preventing adaptive attacks. Otherwise, the theory would be pathological, even if it respects the no-signaling principle~\cite{Kent2005}.

Then, for any partition \( X = X^1 \cup X^2 \), \( Y = Y^1 \cup Y^2 \), and \( E = E^1 \cup E^2 \), the no-signaling condition (see Eq.~\eqref{eq:no-signalling}) imposes:
\[
p^{\lambda}_{X^1 Y^1 E^1} = p^{\lambda}_{X^1 Y^1 E^1 | X^2 Y^2 E^2}.
\]
This expresses that no local measurement on one subsystem can transmit information to a distinct subsystem.

Let \( (X_j, Y_j) = (x, y) \) be a random measurement pair chosen with uniform probability \( 1/m^2 \), yielding outcomes \( (a_j, b_j) \in \{0,1\}^2 \). From \eqref{chainBI0}:
\begin{equation}\label{chainBI}
        \beta_j(\bm{p}^\lambda)\le 1 - \frac{2}{3m}\equiv\beta_\LC(m) \Longrightarrow \bm{p}^\lambda \in \mathcal{L}, \qquad
        \beta_j(\bm{p}^\lambda)= 1 - \mathcal{O}(1/m^2)\equiv \beta_{\NSC\setminus\LC}(m) \Longrightarrow \bm{p} \notin \mathcal{L}.
\end{equation}
For \( m \gg 1 \), the gap \( \beta_{\NSC\setminus\LC}(m) > \beta_\LC(m) \) is detectable, computed either:\textit{(i)} using the Born rule in QT:
    \(
    p^\lambda(a, b \mid x, y) = \mathrm{Tr}\left[ \rho_\lambda \left( \Pi_{a|x} \otimes \Pi_{b|y} \right) \right],
    \)
     or in any no-signaling probability model (Hahn-Banach theorem\footnote{Exists always, a hyperplane that separate a point from a convex set such as $\LC$ and $\QC$ (both contained in $\NSC$).}):
    \(
    p^\lambda(a, b \mid x, y) \in \NSC\setminus\QC.
    \)

This gap is used to infer limits on Eve's knowledge. Specifically, she cannot simultaneously satisfy:
(a) Passing the Bell test (i.e., the protocol does not abort),
(b) Possessing high confidence in the value of the secret bit $a_s = b_s$.

\begin{lemma}\label{lemma02}
    If the protocol passes with probability $p^\lambda(\mathrm{pass}) > \epsilon$, then the correlation of the secret bit is bounded:
\begin{equation}\label{p-pass}
1 - \frac{1}{2 M m \epsilon} < p^\lambda(a_s = b_s \mid \mathrm{pass}).
\end{equation}
From \eqref{p-pass}, the no-signaling constraint and chain rule:
\begin{equation}\label{p-pass2}
 1 - \frac{1}{2 M m \epsilon}< \beta_s(\bm{p}^\lambda).
\end{equation}
\end{lemma}
The maximum probability that Eve guesses correctly that $a_s = b_s$ (the secret bit), conditioned on her outcome $e$, for a fixed input setting pair $(x, x+c)$ $p_{\mathrm{g}}^\lambda$, defined as:

$$
p_{\mathrm{g}}^\lambda(b|E):= \max_e p^\lambda(a_s = b_s \mid X_s = x, Y_s = x + c, e),
$$

Assume that Eve has partial information about the secret bit. That is, there exists an outcome $e_0$ occurring with probability $\delta > 0$ (i.e. how likely it is for Eve to land on a "lucky" outcome — one that gives her some knowledge), such that her guessing probability is strictly better than random:
$$
 p_{\mathrm{g}}^\lambda(b|E=e_0)= p^\lambda(a_s = b_s \mid X_s = x, Y_s = x + c, e_0) > \frac{1 + \delta'}{2}, \qquad \delta' > 0.
$$
 $\delta'$ is  the bias in Eve's knowledge, i.e., how much better than random she is at guessing the secret bit when she obtains outcome $e_0$\footnote{By considering both $\delta$ and $\delta'$, the bound becomes robust and general. Even if Eve’s knowledge is good but rare (large $\delta'$, small $\delta$), or frequent but weak (small $\delta'$, large $\delta$), the product $\delta \delta'$ controls her effective information}.
This implies that the conditional distribution $p^\lambda(a_s, b_s \mid X_s = x, Y_s = x + c, e_0)$ is biased towards equal outcomes, allowing Eve to guess $a_s = b_s$ with success probability greater than $1/2$.
\begin{proof} 
From the derivation in~\cite{Barrett2005}, one can use the triangle inequality and the no-signaling property to propagate this bias across the chained Bell expression. The result is a \textit{reduction} in the observed Bell value:

$$
\beta_s^\lambda(\bm{p}) \le 1 - \frac{\delta \delta'}{3m},
$$

where the factor $\delta \delta'/3m$ quantifies the average amount of correlation that is lost due to Eve’s biased knowledge.
On the other hand, if the protocol passes (i.e., the observed Bell value is high), we have a previously established lower bound \eqref{p-pass2} from Lemma \ref{lemma02}, which assumes that the test was passed with probability at least $\epsilon$.
We now reach a contradiction if both bounds are applied at once. Specifically, if:

$$
\frac{1}{2 M m \epsilon} < \frac{\delta \delta'}{3m} \quad \Longrightarrow \quad \text{Contradiction:  }\quad p_{\mathrm{g}}^\lambda > \frac{1 + \delta'}{2} \quad \text{is incompatible with} \quad \beta_s^\lambda(\bm{p}) > 1 - \frac{1}{2 M m \epsilon}.
$$

In short, if the observed value of the chained Bell parameter is high enough, Eve’s guessing probability must remain close to $1/2$, meaning she gains vanishing information about the secret bit. Thus, we can write the final conclusion as:

$$
p_{\mathrm{g}}^\lambda \le f(\beta_s^\lambda) := \frac{1}{2} + \mathcal{O}\left(1 - \beta_s^\lambda\right),
$$

where $f$ is a function that decreases as the BI violation increases, and vanishes as $m \to \infty$~\cite{Barrett2005a,Barrett2006}.
\end{proof}
However, this protocol has zero key rate (defined in Sec. \ref{sec:CHSHprotocol}) ascribing correlations with noiseless states. In the following protocol, we describe a more practical scheme, but without tackling the most powerful adversary. It was proposed in \cite{Acin2006b} (extended version \cite{scarani2006secrecy}) and improved with higher noise tolerance and key rate in \cite{Acin2006}.
\subsubsection{CHSH Protocol (2006) -- from polytopes to positive rates}\label{sec:CHSHprotocol} 
The idea is to generate a protocol from the geometry of $\LC$, $\QC$, $\NSC$ and the specific Eve's attacks. In fact, suppose that Eve prepare $\lambda = \sum_i c_i \lambda_i$ such that $\bm{p}^\lambda \in \NSC$~\cite{Barrett2005,Acin2006b}. 
If Alice and Bob observe a Bell violation (e.g., $\beta > \beta_{\LC}$), then:
$
\bm{p}^\lambda = \sum_i c_i \bm{p}^{\lambda_i} \notin \LC.
$
So at least one $\lambda_i$ must generate a distribution $\bm{p}^{\lambda_i} \notin \LC$. But this alone does not guarantee security: such $\bm{p}^{\lambda_i}$ could still be \textit{deterministic} or even \textit{signaling}.
This leads to three possibilities:
\begin{itemize}
    \item $\bm{p}^{\lambda_i} \in \LC \Longrightarrow$  contradiction as no BI violation is observed;
    \item $\bm{p}^{\lambda_i} \in \NSC \setminus \LC$. Eve’s hidden variable $\lambda_i$ is no-signaling but nonlocal. Such correlations are monogamous: Nonlocality cannot be shared \textit{freely} with a third party without violating no-signaling~\cite{Barrett2005,Barrett2006}. Therefore, Eve cannot simultaneously reproduce the nonlocal statistics between Alice and Bob and retain side-information allowing her to predict outcomes. Otherwise, her model would allow signaling. Let us explain what \textit{freely} means:
    \begin{enumerate}
        \item Given $\ket{\mathrm{EPR}}_\mathrm{AB}$,$\ket{\mathrm{EPR}}_\mathrm{AC}$,$\ket{\mathrm{EPR}}_\mathrm{BC}$, the state of the 3 subsystems we cannot really talk about tripartite entanglement because the correlations are always between any two of the three parties. When Alice and Bob share $\rho_{\mathrm{AB}}=\ket{\mathrm{EPR}}\bra{\mathrm{EPR}}_\mathrm{AB}$, being a pure maximally entangled state, it has eigenvalues $\lambda_1=1$ and $\lambda_k=0$ ($k>1$). For Uhlmann theorem any $\rho_\mathrm{ABC}=\rho_{\mathrm{AB}}\otimes\rho_\mathrm{C}$ still has eigenvalues $\lambda_1=1$ and $\lambda_k=0$ ($k>1$), but the Schmidt rank across the partition $\mathrm{AB}:\mathrm{C}$ is 1, thus $\rho_{\mathrm{AB|C}}$ is not entangled, hence cannot share any correlation with $\ket{\mathrm{EPR}}_{\mathrm{AB}}$, not only that, the subsystem $\mathrm{C}$ does not know nothing about $\mathrm{A}$ since $\rho_\mathrm{AC}=\frac{\bm 1}{2}\otimes\rho_{\mathrm{C}}$. 
        \item Given $\ket{\mathrm{GHZ}}_\mathrm{ABC}=\frac12(\ket{000}+\ket{111})$ this has genuine entanglement in all the partition $\mathrm{AB:C}$, $\mathrm{AC:B}$, $\mathrm{BC:A}$ but the entanglement is not maximal if one qubit is dropped, e.g. $\rho_{\mathrm{AB}}$ is separable. 
    \end{enumerate}
    Consider a \textit{3-players CHSH game}. The referee selects 2 players to perform the \textit{CHSH game} of Sec. \ref{sec:BI}. If Alice is selected, she does not know who will be her partner. Therefore, she does not know what to choose $\ket{\mathrm{EPR}}_\mathrm{AB}$ or $\ket{\mathrm{EPR}}_\mathrm{AC}$. Or maybe $\ket{\mathrm{GHZ}}_\mathrm{ABC}$? But its marginals $\rho_{\mathrm{AC}}$, $\rho_{\mathrm{AB}}$ are separable. Therefore the classical strategy is also the quantum one.
    \item $\bm{p}^{\lambda_i} \notin \NSC$. If Eve uses such a $\lambda_i$, the associated strategy enables faster-than-light signaling. This is physically excluded: such distributions violate relativistic causality~\cite{Valentini2002} (see Fig.~\ref{fig:Eve_ombrella} and~\ref{fig:out-NS}).
\end{itemize}
Let us study the simplest scenario with $x,y\in\{0,1\}$, $a,b\in\{0,1\}$, and individual attacks
\begin{equation}\label{eq:individual_attack}
     p_{ABE|XYZ}= p_{E|XYZ} p_{AB|XYZE}= p_{E|Z} p_{AB|XYZE}
\end{equation}
such that 
\begin{equation}\label{eq:AB-behavior}
p(ab|xy) = \sum_e p(abe|xyz) \sum_e p(ab|xyze) p(e|z).    
\end{equation}
The no-signaling condition in Eq. \eqref{eq:no-signalling} is applied in \( p_{E|XYZ} = p_{E|Z} \).
Being \( \mathcal{NS} \) a convex polytope, $p_{AB|XYZE}$ can be expressed as a convex combination of extremal points, but in Eq. \eqref{eq:individual_attack}, $p_{AB|XYZE}$  is taken extremal because individual attacks satisfy two key properties: 
\begin{itemize}
    \item [\textit{(i)}] the interconvertibility property \cite{Barrett2005b, Jones2005} ensures that extremal points \( p(ab|xyze) \) describe the most general individual attack; 
    \item[\textit{(ii)}] local operations and public communication between Alice and Bob do not enhance their security. 
\end{itemize}
For binary $a,b$ the extreme points $\{p_{AB|XYZE}\}=\bm{p}_{\mathrm{ext}}\in\mathcal{NS}$ are fully characterized (binary input \cite{Barrett2005a,Barrett2005b}, arbitrary input \cite{Jones2005}). For binary input and output Eve's strategy is sketched in Fig. \ref{fig:Eve_ombrella}. Let us analyze it in detail.
\begin{proposition}
For $\bm{p}\in\mathcal{NS}$ with binary input and output the no-signaling conditions requires that $\sum_a p(ab|xy)=p(b|y)$ and $\sum_b p(ab|xy)=p(a|x)$ so that
\begin{itemize}
\item [\textit{(i)}] $\exists\, \bm{p}_{\mathrm{NL}}=\frac12 \delta_{a\oplus b,yx}$ isotropic correlation (as in Fig. \eqref{fig:WernerCorr} with $v=1$) know as Popescu-Rohrlich-Tsirelson box \cite{Popescu1994,Barrett2005a,tsirelson1993some}, or nonlocal machine \cite{Cerf2005}, or unit of nonlocality \cite{Jones2005,Barrett2005b} and it is the vertex on the top in Fig. \eqref{fig:Eve_ombrella}
\begin{equation}\label{eq:pNL}
  \bm{p}_{\mathrm{NL}}= 
    \begin{array}{|c|c|c|c|c|}
    \hline
    ab \backslash xy & 00 & 01 & 10 & 11 \\ \hline
    00& 1/2 & 1/2 & 1/2 & 0  \\ \hline
    01& 0 & 0 & 0 & 1/2  \\ \hline
    10& 0 & 0 & 0 & 1/2  \\ \hline
    11& 1/2 & 1/2 & 1/2 & 0  \\ \hline
    \end{array}
\end{equation}  
\item [\textit{(ii)}] As in Fig. \eqref{fig:Eve_ombrella}, $\exists$ 8 extreme points $\bm{p}_{\mathrm{L}}^{j,r}$ with entries $p_{\mathrm{ext}}(ab|xy)=\delta_{a,\lambda_A(x)}\delta_{b,\lambda_B(y)}$ (satisfying \ref{def:realism}) s.t. $\beta_\uparrow=3$ defined as
\begin{equation}\label{eq:l-strategies}
    \bm{p}_{\mathrm{L}}^{j,r}: (x,y)\in\{0,1\}^2\mapsto\{0,1\}^2\ni \bm{p}_{\mathrm{L}}^{j,r}(x,y)=(a(x),b(y))
\end{equation}
where at each $j\in\{1,\dots,4\}$, and $r=0,1$ the output $(a(x),b(y))$ are\begin{equation}
    \begin{array}{|c|c|c|c|c|c|c|c|}
    \hline
       \bm{p}_{\mathrm{L}}^{1,0}  & \bm{p}_{\mathrm{L}}^{1,1} & \bm{p}_{\mathrm{L}}^{2,0} & \bm{p}_{\mathrm{L}}^{2,1} & \bm{p}_{\mathrm{L}}^{3,0} & \bm{p}_{\mathrm{L}}^{3,1} & \bm{p}_{\mathrm{L}}^{4,0} & \bm{p}_{\mathrm{L}}^{4,1}  \\
       (0,0)  & (1,1)  & (x,0)  & (x+1,1)  & (0,y)  & (1,y+1)  & (x,y+1)  & (x+1,y)  \\
       \hline
    \end{array}
\end{equation}
for example
\begin{equation}
\bm{p}_{\mathrm{L}}^{1,0}=
\begin{array}{|c|c|c|c|c|}
    \hline
    ab \backslash xy & 00 & 01 & 10 & 11 \\ \hline
    00& 1 & 1 & 1 & 1  \\ \hline
    01& 0 & 0 & 0 & 0  \\ \hline
    10& 0 & 0 & 0 & 0  \\ \hline
    11& 0 & 0 & 0 & 0  \\ \hline
    \end{array},
     \,\quad \bm{p}_{\mathrm{L}}^{4,1}=
\begin{array}{|c|c|c|c|c|}
    \hline
    ab \backslash xy & 00 & 01 & 10 & 11 \\ \hline
    00& 0 & 0 & 1 & 0  \\ \hline
    01& 0 & 0 & 0 & 1  \\ \hline
    10& 1 & 0 & 0 & 0  \\ \hline
    11& 0 & 1 & 0 & 0  \\ \hline
    \end{array},\mbox{ etc.}
\end{equation}
    \item [\textit{(iii)}] If $\bm{p}\in \mathcal{L} \Longrightarrow$ Eve knows $a_0,a_1,b_0,b_1$
    \item [\textit{(iv)}] If Alice and Bob would observe $\bm{p}=\bm{p}_\mathrm{NL}$ $\Longrightarrow$ Eve cannot be correlated (perfect monogamy \cite{Barrett2005a})
    \end{itemize}
\end{proposition}
To mimic the marginal $\bm{p}=p_{AB|XY}$ of Eq. \eqref{eq:AB-behavior} observed by Alice and Bob, Eve's optimal attack $\bm{p}_\mathcal{E}$ with entries $\{p(abe|xyz)\}$ from Eq. \eqref{eq:individual_attack}  then consists of the combination of the 9 extreme points with the minimal $p_{\mathrm{NL}}=1-p_L=2v-1$ (see Fig. \ref{fig:Eve_ombrella}):
\begin{equation}\label{eq:Eve_ombrella}
    \bm{p}_\mathcal{E}=\sum_{j=1}^4\sum_{r=0}^1 p_j^r \bm{p}_{\mathrm{L}}^{j,r} +p_{NL}\bm{p}_{NL}, 
    \mbox{ with } 
    \sum_{j=1}^4\sum_{r=0}^1 p_j^r=p_L.
\end{equation}
We label the Eve input $z\in\{1,\dots,9\}\mapsto \{\bm{v}_i\}$ with $\{\bm{v}_i\}\in\{\bm{p}_{\mathrm{L}}^{j,r}\}\cup\{\bm{p}_{\mathrm{NL}}\}$ which provides the following knowledge $e\in\{(a,b),(a,?),(?,?)\}$ at given $x$ and $y$.
Then, resulting marginal probability distribution $p(ab|xy)=\sum_e p(abe|xyz)$ reads as
\begin{equation}\label{eq:tab_pabxyCHSH}
\begin{array}{|c|c|c|c|c|}
\hline
ab \backslash xy & 00 & 01 & 10 & 11 \\ \hline
00 & \frac{p_{\mathrm{NL}}}{2} + \sum_{j\neq 4}p_j^0 & \frac{p_{\mathrm{NL}}}{2}  + \sum_{j\neq 3}p_j^0  & \frac{p_{\mathrm{NL}}}{2}  + p_1^0 + p_3^0 + p_4^1 & p_1^0 \\ \hline
01 & p_4^0 & p_3^0 & p_2^1 & \frac{p_{\mathrm{NL}}}{2}  + p_2^1 + p_3^0 + p_4^1 \\ \hline
10 & p_4^1 & p_3^1 & p_2^0 & \frac{p_{\mathrm{NL}}}{2}  + p_2^0 + p_3^1 + p_4^0 \\ \hline
11 & \frac{p_{\mathrm{NL}}}{2}  + \sum_{j\neq 4}p_j^1  & \frac{p_{\mathrm{NL}}}{2}  + \sum_{j\neq 3}p_j^1  & \frac{p_{\mathrm{NL}}}{2}  + p_1^1 + p_3^1 + p_4^0 & p_1^1 \\ \hline
\end{array}
\end{equation}
From the bipartite distribution $p(ab|xy)$ in \eqref{eq:tab_pabxyCHSH} the best procedure to extract the secret key is unknown.  The CHSH protocol \cite{Acin2006b,scarani2006secrecy} is a good candidate because it provides high correlations between Alice and Bob and reduces Eve's information. From \eqref{eq:tab_pabxyCHSH} we see that Alice and Bob are highly anticorrelated only for $x=y=1$. It is therefore natural to devise the following procedure that transforms these anticorrelations into correlations (see tutorial \cite{tutorial}), this is done in the \textit{Pseudosifting} step, which is the only variant from Box 1 to discuss.

\textbf{CHSH protocol (\textit{Pseudosifting})} -- Alice reveals $x=0$ or $x=1$ and Bob without announcing the value of $y$, if $(x,y)=(1,1)\Longrightarrow b\mapsto \bar{b}$. The anticorrelation becomes correlations and the distribution $p(ab|xy)$ in Eq. \eqref{eq:tab_pabxyCHSH} is updated to $p(ab|x=0,y=k)$ and $p(ab|x=1,y=k)$ conditioned to the knowledge of $x$ and outcome probability $p(y=k)=\xi_k$. 

\begin{equation}\label{eq:Pabecond}
\begin{array}{|c|c|}
\hline
ab &   x=0,\,p(y=k)=\xi_k \\ \hline
00 & \frac{p_{\mathrm{NL}}}{2} +p_1^0+p_2^0+\xi_0p_3^0+\xi_1p_4^0 \\ \hline
01 & \xi_1p_3^0+\xi_0p_4^0 \\ \hline
10 & \xi_1p_3^1+\xi_0p_4^1 \\ \hline
11 & \frac{p_{\mathrm{NL}}}{2}  + p_1^1 + p_2^1 +\xi_0p_3^1+\xi_1p_4^1 \\ \hline
\end{array}\qquad
\begin{array}{|c|c|}
\hline
ab &   x=1,\,p(y=k)=\xi_k \\ \hline
00 & \frac{p_{\mathrm{NL}}}{2} +\xi_0p_1^0+\xi_1p_2^1+p_3^0+p_4^1 \\ \hline
01 & \xi_1p_1^0+\xi_0p_2^1 \\ \hline
10 & \xi_1p_1^1+\xi_0p_2^0 \\ \hline
11 & \frac{p_{\mathrm{NL}}}{2}  + \xi_0p_1^1 +\xi_1 p_2^0 +p_3^1+p_4^0 \\ \hline
\end{array}
\end{equation}
To maximize Eve's uncertainty $\xi_k=1/2$. An interesting property for the pseudosifting about all the eight local points is that Alice's outcome $a$ from Eq. \eqref{eq:l-strategies} is always known to Eve because $x$ is publicly announced, and if one $\bm{p}_{\mathrm{L}}$ provides the knowledge of $b$ to Eve for $x=0$, the same point leaves Eve ignorant for $x=1$ and vice-versa. For example, given Eve's strategy in \eqref{eq:Eve_ombrella} she knows $(a,b)$ if she sent out $\bm{p}_{\mathrm{L}}^{1,0}$ (with probability $p_1^0$) and Alice announces $x=0$,  then $\forall y\in\{0,1\} \Longrightarrow b=0$. In this case, Eve's uncertainty on Bob's symbol is null, we write it as $H(b|E=\bm{p}_{\mathrm{L}}^{1,0},x=0)=0$. But if Eve sent out $\bm{p}_{\mathrm{L}}^{1,0}$ and Alice announces $x=1$,  then $y=0 \Longrightarrow b=0$ and $y=1 \Longrightarrow \bar{b}=1$. Since Eve does not know $y$ the uncertainty is maximum $H(b|E=\bm{p}_{\mathrm{L}}^{1,0},x=1)=1$. Because of pseudosifting, she does not know if Bob's outcome is $b=0$ or $\bar{b}=1$, and at the same time, the outcomes are correlated $a=b$ when $x=y=1$.
Given the above distributions $p(ab|x=0,y=k)$ and $p(ab|x=1,y=k)$ in Eq.\eqref{eq:Pabecond}, then $p(a\neq b|0)=\frac12(p_3^0+p_3^1+p_4^0+p_4^1)\equiv e_{\mathrm{AB|0}}$ and $p(a\neq b|1)=\frac12(p_1^0+p_1^1+p_2^0+p_2^1)\equiv e_{\mathrm{AB|1}}$ (with $\xi_k=1/2$). Then Eve's uncertainty on $b$ is the conditional Shannon entropy: 
\begin{equation}\label{eq:nosignaling_uncertainty}
    H(B|E,X)=\sum_{e} P(E=e,X=x)H(b|E=e,X=x)
    =1-2e_{\mathrm{AB|x+1}},
\end{equation}
with fixed $x\in\{0,1\}$ and $e$ determined by the values of $z$ that chooses the strategies $\{\bm{v_i}\}=\{\bm{p}_{\mathrm{L}}^{j,r},\bm{p}_{\mathrm{NL}}\}$. This is the first evidence of an analogue of quantum mechanical uncertainty relations in a generic no-signaling theory. 
The pseudosifting process is optimized to extract correlations from the nonlocal strategy $\bm{p}_{\mathrm{NL}}$, but it also affects deterministic strategies differently. In particular, for $\bm{p}_{\mathrm{L}}^{1,r}$ and $\bm{p}_{\mathrm{L}}^{2,r}$, pseudosifting yields no errors when $x = 0$, and Eve fully learns Bob’s outcome $b$, since $b(y)|_{y=0} = b(y)|_{y=1} = a$. However, when $x = 1$, errors occur in half of the cases, and Eve gains no information about $b$, as $b(y)|_{y=0} \ne b(y)|_{y=1}$. The reverse holds for $\bm{p}_{\mathrm{L}}^{3,r}$ and $\bm{p}_{\mathrm{L}}^{4,r}$~\cite{scarani2006secrecy}.

\textbf{Quantitative security}. To compute the key rate, the so-called \textit{depolarization procedure} (individual attacks) transform w.l.o.g. Eve's strategy of Eq.  \eqref{eq:Eve_ombrella} into the isotropic distribution (cf. Appendix C of \cite{scarani2006secrecy})
$$
\bm{p}_{\mathcal{E}} = p_L \bm{1} + p_{\mathrm{NL}} \bm{p}_{\mathrm{NL}}, \quad \text{with } p_j^r = \frac{p_L}{8}.
$$
This leads to the probability distributions in Eq. \eqref{eq:Pabecond} $$p(ab|x=0,y=k)=p(ab|x=1,y=k)=
\sum_e p(abe|X=x,Y=k,z=\bm{p}_\mathcal{E}),$$ and the tripartite distribution:
\begin{equation}\label{eq:Pabecond-iso}
p(abe|0,k,\bm{p}_\mathcal{E})=\,\begin{array}{|c|c|c|c|}
\hline
ab\backslash e &  (?,?) &  (a,?) &  (a,b) \\ \hline
00 & \frac{p_{\mathrm{NL}}}{2} &\frac{p_{\mathrm{L}}}{8} & \frac{p_{\mathrm{L}}}{4} \\ \hline
01 & 0 & \frac{p_{\mathrm{L}}}{8} & 0\\ \hline
10 & 0 & \frac{p_{\mathrm{L}}}{8} & 0 \\ \hline
11 & \frac{p_{\mathrm{NL}}}{2} &\frac{p_{\mathrm{L}}}{8} & \frac{p_{\mathrm{L}}}{4} \\ \hline
\end{array}
\end{equation}
\paragraph{One-way classical postprocessing} --
The tapescript $T(q)$ flow goes from Alice to Bob and from the distribution in \eqref{eq:Pabecond-iso} Bob's error probability is $\epsilon_\mathrm{AB}=p_L/4$, after preprocessing, the quantity to be corrected in error correction is $e_{\mathrm{AB}}^\prime=(1-q)e_{\mathrm{AB}}+q(1-e_{\mathrm{B}})$ while Eve's information is $I(A:E)=p_L/2(1-h(q))$. Therefore Eq. \eqref{devatekwinterformula} yields~\cite{tutorial}

\begin{equation}\label{eq:rCHSH}
    r(D)=\max_{q\in[0,1/2]}\left(1-h(e_{\mathrm{AB}}^\prime)-\frac{p_L}{2}(1-h(q))\right), \mbox{   with }p_\mathrm{NL}=\sqrt{2}(1-2D)-1.
\end{equation} 
The critical disturbance $D$ characterizes the properties of the channel linking Alice and Bob.
$r(p_\mathrm{NL})>0$ is obtained with optimal preprocessing at $p^\prime_\mathrm{NL}\gtrsim 0.236$ ($D\lesssim 6.3\%$) and without preprocessing at $p^{\prime\prime}_\mathrm{NL}\gtrsim 0.318$.

Since $\bm{p}_\mathcal{E}(p_\mathrm{NL})\in \QC\Longleftrightarrow  p_\mathrm{NL}\leq \sqrt{2}-1\simeq0.414$, both $p^\prime_\mathrm{NL},p^{\prime\prime}_\mathrm{NL}\in \QC$ (see Fig.~\ref{fig:rCHSH}).
\begin{definition}\label{def:Bell_limit}[Bell limit]
A family of distributions $ \bm{p}=p(ab|xy)$ reach the Bell limit if leads to a secret key $r>0$ for any amount of nonlocality. 
\end{definition}
\begin{remark}
    In the case of $p_1^0=p_2^0$, $p_1^1=p_2^1$, and $p_3^r=p_4^r=0$, even neglecting preprocessing
    $
        p_\mathrm{NL}>0 \Longrightarrow r=1-\frac{h(p_L/2)}{2}-\frac{p_L}{2}>0.
    $
    Notice that, $\exists \bm{p}$ reaching the Bell limit, despite the fact $\bm{p}\notin \mathcal{Q}$ hence it cannot be broadcasted using quantum preparations (solid orange line in Fig.\ref{fig:rCHSH}). Indeed $\mathcal{L}\subsetneq\mathcal{Q}\subsetneq\mathcal{NS}$ and $\sum_{j,r}p_j^r\bm{p}_\mathrm{L}^{j,r}=0 \Longleftrightarrow \forall p_j^r=0$ (see Ref. \cite{scarani2006secrecy} sec.III.E.3).
    
\end{remark}
    Only in 2024, Refs.~\cite{Farkas2024} show a protocol that for $\bm{p}\in \QC$ a quantum strategy reaches the Bell limit extending the BI scenario beyond CHSH. Exploiting a family of $d$-outcome Bell inequalities and the additivity of the von Neumann entropy, they obtained the analytic relation
$H(A|E)=H(A)$, which gives the key-rate bound
$r \;\ge\; \log d$,
so that $\log d$ bits of secret key can be extracted for every integer $d\ge 3$.  
In other words, \textit{unbounded DI key rates can be certified from correlations exhibiting arbitrarily small non-locality}.  
Wooltorton et al. \cite{Wooltorton2024} obtained the same result by examining a family of three-parameter Bell inequalities \cite{Le2023}, whose maximal quantum violations self-test a unique state and measurement settings. Within this family, there exist nonlocal correlations—across a range of CHSH violations, including those arbitrarily close to the classical bound—that exhibit the remarkable property of producing outputs arbitrarily close to a perfect key for a specific pair of inputs. Additionally, they extended this result by demonstrating the existence of quantum correlations capable of generating both perfect key and perfect randomness simultaneously, while exhibiting arbitrarily small CHSH violations. 

\paragraph{Two-way classical postprocessing: repetition code advantage distillation}\label{par:2ways-c} --
In two-way postprocessing, no optimal procedure or tight bound analogous to Eq. \eqref{devatekwinterformula} is known. The most common method, \textit{Advantage Distillation} (AD) \cite{chau2002,gottesman2003,renner2008security}, especially in the form of repetition-code AD, say that
\begin{equation}\label{eq:AD}
    \exists \tilde{B},\tilde{E} \mbox{ s.t. } I(A:B) < I(B:E) \stackrel{AD}{\Longrightarrow} I(\tilde{A}:\tilde{B}) > I(\tilde{B}:\tilde{E})
\end{equation}
enabling one-way postprocessing on $\tilde{B},\tilde{E}$.
While repetition-code AD is often introduced as a classical two-way step, the so-called B-step from entanglement purification is typically described in terms of bilateral CNOT operations. However, as explained in Sec.~6.1 of~\cite{Ma2008}, these quantum operations can be commuted with the final measurements and reduce to classical parity (XOR) operations on the raw data ~\cite{gottesman2003}. Thus, repetition-code AD and the B-step are operationally equivalent at the level of classical data processing, even though their derivations originate from different viewpoints. Below, we follow the AD terminology for consistency.

In AD, Alice publicly reveals \(N\) instances where her bits are equal, i.e., \(a_{i_1} = a_{i_2} = \dots = a_{i_N} = a\). Bob checks his corresponding bits and announces whether all his bits are also equal. If Bob’s bits are all equal, Alice and Bob keep one of these instances, $(a_{i_k},b_{i_k})$; otherwise, they discard the \(N\) instances. The error rate between Alice and Bob after this process, denoted as \(\tilde{e}_{AB}\), becomes exponentially smaller:
\[
\tilde{e}_{AB} = \frac{e_{AB}^N}{(1 - e_{AB})^N + e_{AB}^N},
\]
where \(e_{AB}\) is the initial error rate between Alice and Bob.
As \(N \to \infty\), \(\tilde{e}_{AB} \to 0\), meaning that Alice and Bob almost always share identical bits after a sufficiently large number of instances.

The probability that Eve makes an error on Bob’s symbols after AD is approximated by
\[
\tilde{e}_E \gtrsim C [f(e_{AB})]^N,
\]
where \(f(\cdot)\) is a function that depends on the probability distribution and \(C\) is a constant. As long as the condition \(f(e_{AB}) > e_{AB}/(1 - e_{AB})\)
is satisfied, Eve’s error increases exponentially with \(N\). There is always a finite value of \(N\) such that Eve’s error becomes greater than Bob’s, ensuring that a secret key can be extracted. 

The bound on the tolerable error after AD is derived by solving this inequality and provides the necessary condition for secrecy extraction. Without preprocessing $p_{\mathrm{NL}}\gtrsim 0.2$; with preprocessing (allowing Alice and Bob to flip some bits before AD) $p_{\mathrm{NL}}\gtrsim 0.09$ \cite{renner2008security}. It remains an open question if in two-way postprocessing, a Bell limit can be reached. Although one might consider that a two-way communication would increase interceptions, overall AD is more noise tolerant (lower \(p_{NL}\)) than one-way post-processing, by iteratively improving correlations and discarding mismatched rounds to reduce error rates. Preprocessing enhances this by scrambling Eve’s knowledge before AD.

\paragraph{Intrinsic information}
Given a tripartite probability distribution \(\tilde{\bm p}=\{p_{ABE|XYZ}(abe|xyz)\}_{a,b,e}\), the intrinsic information \(I_\downarrow=I(A:B \downarrow E)= \min_{E \to \bar{E}} I(A:B|\bar{E})\ge r\) \cite{maurer1999unconditionally}.
This upper bound represents the mutual information between Alice and Bob conditioned on Eve's knowledge. $I_\downarrow <0 \Longrightarrow r=0$ witness the impossibility of secret correlations in $\tilde{\bm p}$. The vice-versa is unknown \cite{Tyagi2013,Mista2015,Khesin2023}.
Furthermore, $\exists \, \tilde{\bm p}$ with a $I_\downarrow > 0$ but $r=0$, indicating the presence of bound information (similarly to bound entanglement). While bound information has been proven in multipartite settings, its existence in bipartite scenarios remains unknown.

For the CHSH protocol, \(I(A:B|E)=p_{NL}\) when Alice and Bob are perfectly correlated (see Fig.~\ref{fig:iCHSH}). In other cases, they are uncorrelated. A conjectured optimal map for minimizing this conditioned mutual information is introduced, supported by numerical evidence giving
\begin{equation}\label{eq:ichsh}
    I_\downarrow=(1-p_L/2)[1-h(\,p_L/(4-2p_L)\,)].
\end{equation}
This conjectured intrinsic information remains positive for \(p_{NL} > 0\), which leads to two possibilities: (i) $r>0\,\forall \tilde{\bm p}\notin \LC$ (ii) In the Bell limit (\(p_{NL}\simeq 0\)) $\tilde{\bm p}$ might represent bipartite bound information.

The CHSH protocol in quantum theory (Alice and Bob share a quantum state of two qubits, agreeing on specific measurements, while Eve holds a purification) is shown to be equivalent to the BB84 protocol with added classical preprocessing, having the same robustness to noise. However, BB84 achieves a higher secret key rate at low error rates but cannot be used for device-independent proofs, as its security becomes compromised if the Hilbert space dimensionality is unknown.

Then how to improve the rate in  DI-QKD? For example, CHSH protocol admits larger-dimensional outcomes generalization using CGLMP inequalities with more extractable secrecy as the outcomes increase (see Ref. \cite{scarani2006secrecy} sec.IV) or via chain BI (Sec. \ref{sec:CHAIN-protocol}).
The former is a tight family of BI, while chain BI, though not tight, can be efficiently implemented in the next protocol (see definition \ref{definition:tightBI} for tight BI). 

\begin{figure}
    \centering
    \begin{subfigure}{.49\textwidth}
        \centering
        \includegraphics[width=\textwidth]{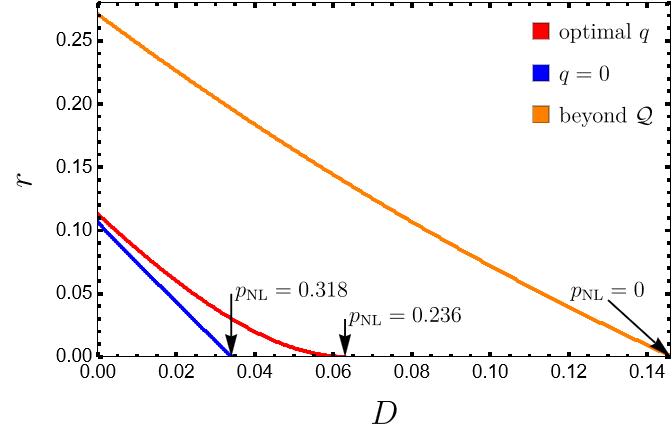}
        \caption{\textit{CHSH protocol}}
        \label{fig:rCHSH}
    \end{subfigure}
    \hfill
    \begin{subfigure}{.47\textwidth}
        \centering
        \includegraphics[width=\textwidth]{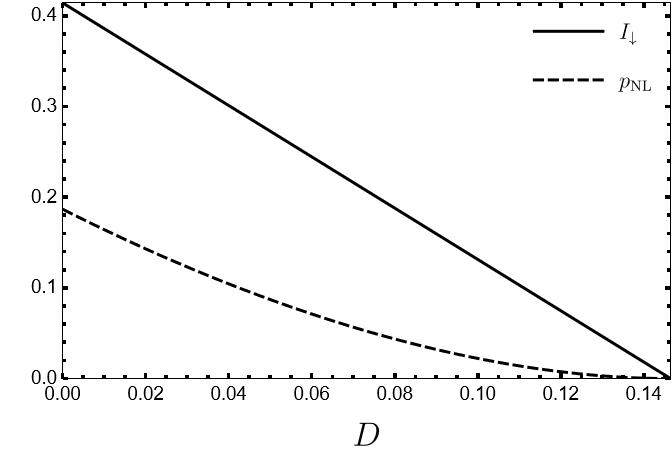}
        \caption{\textit{Intrinsic information in CHSH protocol}}
        \label{fig:iCHSH}
        \end{subfigure}
        \caption{\textit{Key rate versus noise in individual attack} -- In \ref{fig:rCHSH}, key rate for one way PD of Eq.\eqref{eq:rCHSH}, with preprocessing (red); without preprocessing (blue); without preprocessing, but using a distributions compatible with CHSH protocol, but that cannot be
 broadcasted using quantum states. In \ref{fig:iCHSH} Eq.~\eqref{eq:ichsh} is positive $\forall p_{\mathrm{NL}}$.}
    \label{fig:keyrate}
\end{figure}

\subsubsection{Chain Protocol (2006) --  enlarged Bell scenario}\label{sec:CHAIN-protocol}
The previous protocol considered the simplest CHSH scenario with inputs and outputs \(x, y, a, b \in \{0,1\}\), providing a preliminary step towards key-rate estimation. In this section, we begin by revisiting a slightly extended CHSH-based analysis, where Alice’s input \(x\) can take three values, \(x \in \{0,1,2\}\). We then generalize this setting to the more powerful chained Bell inequalities (CHAIN), introduced in~\cite{Acin2006}, which offer improved key rates and greater noise tolerance while remaining secure against individual no-signaling adversaries.

\textbf{(3,2)--CHSH protocol}. \textit{Data generation}. Alice selects $x\in\{0,1,2\}$ with probability $q, (1-q)/2, (1-q)/2$, Bob selects $y\in\{0,1\}$ with probability $q^\prime, 1-q^\prime$. They obtain $a,b\in\{0,1\}$. In QT $M_{A|x}=\{\Pi_{0|x},\Pi_{1|x}\}$ and $N_{B|y}=\{\Pi_{0|y},\Pi_{1|y}\}$:
\begin{equation}
    \Pi_{a|x}=U_{\theta_x}\ket{a}\bra{a}U_{\theta_x}^\dagger,\quad
    U_{\theta}=\mathrm{e}^{-\mathrm{i}\frac{\theta}{2}\sigma_z},\quad
    \theta_{x=(0,1,2)}=\left(\frac{\pi}{4},0,\frac{\pi}{2}\right),\quad
    \theta_{y=(0,1)}=\left(\frac{\pi}{4},-\frac{\pi}{4}\right).
\end{equation}
acting on the Werner state $\rho_W$ of fig. \ref{fig:WernerCorr}.
 with $P_+=\ket{\phi^+}\bra{\phi^+}$.

\textit{PD and raw key validation}.
They publicly announce $x,y$. For $x\in\{1,2\}$ they evaluate $\beta_{\downarrow \QC}(\bm{p})=2-\sqrt{2}p$ and abort if $\beta_{\downarrow \QC}(\bm{p})\ge 1$. When the settings unmatch: $(x,y) = (0,1)$, the outcomes are uncorrelated and they reject the round, instead for $(x,y)=(0,0)$ they quantify the correlation $ \langle C \rangle_\rho\equiv P(a_0 = b_0)-P(a_0 \neq b_0)=\nu$ and extract the key. If they share $\ket{\phi^+} = \frac{1}{\sqrt{2}}(\ket{00} + \ket{11})$ they exploit perfect outcome correlations.

\textbf{Security} In Eve's strategy \cite{Acin2006}, for each Alice's measurement, there might be output: predetermined (D)   $p(a|x)=\delta_{a,\lambda_A(x)}$ or uniformly random (R) $p(a|x)=1/2$. Similarly for Bob.
If $y=0$ is (D), then all the measurements are (D). Eve's strategies can be classified as in Tab. \ref{tab:Evestrategies} into three sets, according to whether $(x,y)=(0,0)$ yields predetermined (D) or uniformly random outcomes (R). For each strategy a bound on the \textit{security check} $\beta_{\downarrow \QC}(\bm{p})$ and \textit{raw key validation} $\langle C \rangle $ is computed:
\begin{equation}\label{eq:rchain}
    r\ge\sqrt{2}\nu-1-h\left(\frac{1+\nu}{2}\right)
    ,\qquad
    r
    \le
    I_{\downarrow}\le   \langle C\rangle_\rho-\beta_{\downarrow \QC}(\bm{p})=(1+\sqrt{2})\nu-2.
\end{equation}
\begin{table}[h]
    \centering
\begin{tabular}{lcccccc}
\hline
&Strategies  &$S_{\bar{*}}$& $\langle C\rangle_\rho$ & $H(A|E)$ & $H(B|E)$& $I(A:B|E)$ \\ 
\hline
$p_1$&(D,D)  &$\ge 1$& $\le 1$ & 0 & 0 & 0\\
$p_2$&(D,R)  &$\ge 0$& 0 & 0 & 1 & 0\\
$p_3$&(R,R)  &$\ge 0$& $\le 1$ & 1 & 1 & 1\\
\hline
\end{tabular}
        \caption{\textit{Eve's extremal strategies for $(x,y)=(0,0)$ with probability $p_i$} (details in \cite{Acin2006}).}
    \label{tab:Evestrategies}
\end{table}
\begin{figure}
    \centering
    \begin{subfigure}{.47\textwidth}
    \includegraphics[width=\linewidth]{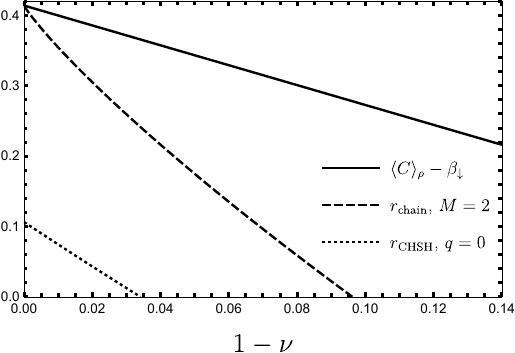}
    \caption{$(3,2)$--CHSH vs.  CHSH protocol rate }
    \label{fig:rchain}
    \end{subfigure}
    \begin{subfigure}{.5\textwidth}
    \includegraphics[width=\linewidth]{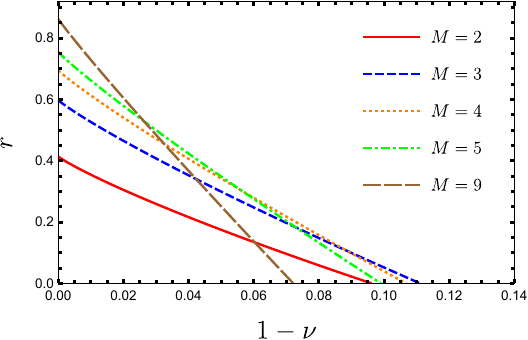}
    \caption{$(M,M+1)$--chain protocol rate}
    \label{fig:Mrchain}
    \end{subfigure}
    \caption{In \ref{fig:rCHSH} the rate of $(3,2)$--CHSH protocol in Eq.\eqref{eq:rchain} vs. the rate of CHSH protocol without preprocessing where $1-\nu$ in $\rho_W$ in \eqref{fig:WernerCorr} in one-way communication protocols  bounded via $\beta_{\downarrow \QC}(\bm{p})$ of Eq. \eqref{eq:CHSHSS}. If Eve saturates the inequalities of Tab. \ref{tab:Evestrategies} the rate is saturated. Without noise $r\lvert_{\nu=1}\ge \sqrt{2}-1\simeq 0.414$ and $r=0$ for $\nu=0.9038$; upper--bounded by intrinsic information $I_{\downarrow}$ using two-way key distillation protocols  \cite{maurer1999unconditionally}
    For $\nu=2/(1+\sqrt{2})\simeq0.8284\Longrightarrow r\le I_\downarrow=0$, however CHSH is violated for $p\ge 0.7071$ always with $I_\downarrow> 0$ \cite{Masanes2006}. In \ref{fig:Mrchain} Eq. \eqref{eq:KN} generalized $M=2$.}
\label{fig:chain}
\end{figure}
\begin{proof}
 1) From Tab. \ref{tab:Evestrategies}, $I(B:E)\le I(A:E)$ then the privacy amplification goes from Bob to Alice, thus $r=I(A:B)-I(B:E)$. The mutual information is $I(A:B)=1-h(\frac{1+\nu}{2})$, where $h$ is the binary entropy and 
\begin{equation}
    I(B:E)=H(B)-\sum_{k=1}^3p_kH_k(B|E)=p_1\le \beta_{\downarrow\QC}(\bm p), \mbox{ where } H(B)=1.
\end{equation} 
This follows from $p_k\ge 0$, $\sum_kp_k=1$ and measured values of $\beta_{\downarrow\QC}(\bm p)$ and $\langle C\rangle$. 

2) The intrinsic information $I_\downarrow$ is upper--bounded, since $I(A:B\downarrow E)\le I(A:B|E)$, thus $I_\downarrow\le\sum_{k=1}^3p_kI_k(A:B|E)=p_3$. Additionally $p_1+p_3\ge \langle C\rangle_\rho$ concludes the proof.
\end{proof}
Whether a key can be extracted from such data, and if it can, what is the best protocol for achieving it, remains an open challenge. 
\paragraph{(M+1,M)-- chain protocol} As the performance depends on the BI, a generalization extends from $(3,2)\to(M+1,M)$ settings
\begin{equation}
    \theta_{x=0}=\frac{\pi}{2M},\quad\theta_{x>0}=\frac{x\pi}{2M},\qquad \theta_{y}=-\frac{\pi(y+1/2)}{M}, \mbox{ for }x,y=1,\dots,M.
\end{equation}
Similarly, for $(x,y)=(0,0)$ the measurement gives highly correlated bit used for the secret key, the other choices are used to violate the chain BI (used in \eqref{chainBI} with a quantum state $\rho_{AB}$ from a post-quantum tripartite $\lambda$) from~\cite{Pearle1970,Braunstein1990,Jogenfors2017} based on Franson interferometer \cite{Franson1989}
where it is known \cite{Aerts1999,Jogenfors2014} that the CHSH is insufficient as a security test. Full security can be reestablished with \cite{Jogenfors2015,Tomasin2017}

\begin{equation}\label{eq:BC2}
    \beta_{\downarrow\QC}^M(\bm{p}) =\sum_{i=1}^M \left(p(a_i\neq b_{i-1})+p(a_i\neq b_{i})\right)=
    \begin{cases}
        2(M-1) & \bm{p}\in \mathcal{L}\\
        M\left(1-\nu\cos\left(\frac{\pi}{2M}\right)\right)& \bm{p}\notin \LC
    \end{cases}
\end{equation}

where $b_M\equiv b_0=1\,\mod 2$.
One-way privacy amplification lower bounds as (see Fig. \ref{fig:Mrchain})
\begin{equation}\label{eq:KN}
    r_M\ge 1-h\left(\frac{1+\nu}{2}\right)-M\left(1-\nu\cos\left(\frac{\pi}{2M}\right)\right).
\end{equation}
A variant with pre-processing from Ref. \cite{Kraus2005} improves the noise-resistance.
For $M$ increase, $r_M\gtrsim 1-\pi^2/8M$ and $\bm{p}=\sum_ic_i\bm p_i$ become maximally non-local with any local component: Eve must always use non-local strategies for which has zero knowledge about Bob's outcome $I(B:E)=0$ (see end of Sec. \ref{sec:nosignalQKD}). 
Each BI provides a different estimation of Eve's knowledge so that for $M$ large, the corresponding protocols are very sensitive to noise, but in the absence of noise, Alice and Bob extract one secret bit per e-bit asymptotically. 
For a post-quantum and quantum individual attacks, the maximal value of the resistance to noise are respectively $\nu=0.86$ and $\nu=0.75$. However, the chain protocol originally did not take into account the detection loophole that would lower such values to certify security (see Sec. \ref{sec:locloop}) nor collective attacks, but nowadays we know the following.

\begin{theorem}
    $\forall p_{A,B,E|X,Y,Z}$ there always exists $\mathrm{hash}$ such that the transformed  distribution $P_{K_A,K_B,\text{BI},T,E|Z}$ via the $(M+1,M)$--chain protocol is universal composable secure \cite{Masanes2014}.
\end{theorem}

\subsection{CHSH\textsubscript{$\chi$} protocol secured by Holevo quantity} \label{DIQKD-collective}
Instead of extending the Bell scenario, to improve the key rate, a tighter bound that relates $\beta_\QC$ with Holevo quantity $\chi$, itself related to $p_g(B|E)$ can be derived against \textit{quantum} collective attacks~\cite{acin2007device,pabgs09andSangouard.4}. Specifically, given $n$ the number of instances, the state prepared by Eve is the same at each instance $\ket{\Psi}=\ket{\psi}^{\otimes n}\in \mathcal{H}_{ABE}$ and Alice and Bob receive $\rho_{\mathrm{AB}}=\mathrm{Tr}_E(\ket{\Psi}\bra{\Psi})=\sum_cp_c\rho_{\mathrm{AB}}^c$ without unknowing $\dim \mathcal{H}_{AB}$ for Alice and Bob similarly to Eqs. \eqref{eq:individual_attack}-\eqref{eq:AB-behavior} for individual attacks. Because of \textit{Jordan's lemma}~\ref{Jordanlemma}  (any pair of binary measurements can be decomposed as the direct sum of pairs of measurements acting on two-dimensional spaces) $\rho_{\mathrm{AB}}^c\in \mathcal{B}(\mathbb{C}^2\otimes\mathbb{C}^2)$ without restrictions. At each $\rho_{\mathrm{AB}}^c$, the classical ancilla $c$ known to Eve determines which measurements Alice and Bob choose. 

\textbf{CHSH\textsubscript{$\chi$} protocol}. \textit{Data generation}. 
The same as CHSH protocol (Sec. \ref{sec:CHSHprotocol}), but a variant allows Bob's choices $y\in\{0,1\}$ (or $y\in\{0,1,2\}$ as in E91 \ref{sec:E91} and \cite{Slater2014})
so that $M_{A|0} = N_{B|1} =\sigma_z$, $ N_{B|2} = \sigma_x$, $M_{A|x} = (\sigma_z \pm \sigma_x)/\sqrt{2}$ for $x=1,2$; the raw key is extracted from the pair $(x,y)=(0,1)$, and  $\mathrm{QBER}=\mathrm{Q}=p(a\neq b|01)$ (from Eq. \eqref{eq:QBER}).

\textbf{Security 1 way PD}. As detailed in \cite{acin2007device,pabgs09andSangouard.4} (see Sec.\ref{sec:analyticalBounds}) Eq.~\eqref{devatekwinterformula} valid for in the asymptotic $n\to \infty$ one-way public discussion (PD) gives,
\begin{equation}\label{eq:2007device}
    r_{\mathrm{CHSH}}=1-h(Q)-h\left[\frac12 +\frac12\sqrt{\left(\frac{\beta_\QC }{2}\right)^2-1}\right]
    \le r_{\mathrm{BB84}}=1-h(Q)-h\left(\frac{\beta_\QC }{2 \sqrt{2}}+Q\right),\quad \beta_\QC=2\sqrt{2}(1-2Q).
\end{equation}
\begin{figure}
    \centering
    \begin{subfigure}{.49\textwidth}
        \centering
        \includegraphics[width=\textwidth]{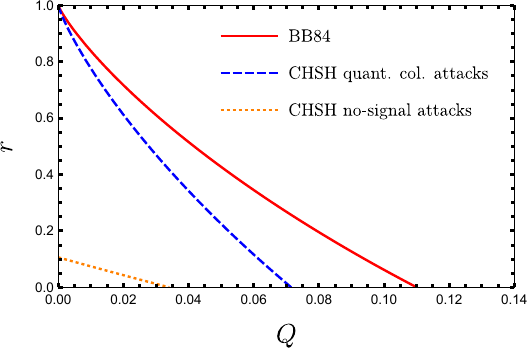}
        \caption{\textit{CHSH protocol improvement}}
        \label{fig:2007devicea}
    \end{subfigure}
    \hfill
    \begin{subfigure}{.49\textwidth}
        \centering
        \includegraphics[width=\textwidth]{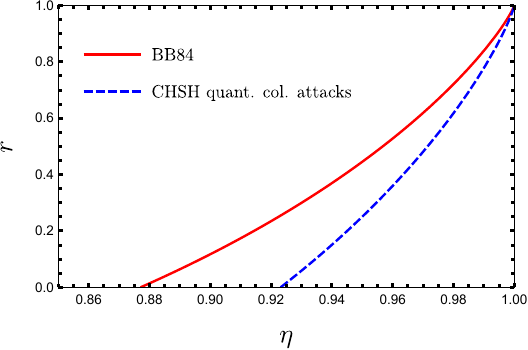}
        \caption{\textit{CHSH protocol vs detection efficiency}}
        \label{fig:2007deviceb}
        \end{subfigure}
        \caption{In \ref{fig:2007devicea} \textit{Key rate versus noise}-- The noise in orange dotted line is $Q=D$ as in Eq.\eqref{eq:rCHSH} obtained for the CHSH protocol under no-signaling individual attacks. Within quantum collective attacks, $r_\mathrm{CHSH}$  is enhanced in Eq.\eqref{eq:2007device} (dashed blue) and compared with standard QKD BB84 (solid red). In \ref{fig:2007deviceb} $Q=\eta(1-\eta)$ and $\beta_\QC=2\sqrt{2}\eta^2+(1-\eta)^2$ following the discussion below Eq.~\eqref{eq:Sphat}.}
    \label{fig:Acin2007}
\end{figure}
Eq.~\eqref{eq:2007device} is obtained as $r=I(a:b|x=0,y=1)-\max\{I(a:E|x=0),I(b:E|y=1)\}$ where $I(a:b|x=0,y=1)=1-h(Q)$ is computed assuming uniform marginals and $\max\{I(a:E|x=0),I(b:E|y=1)\}=\chi(b:E|y=1)$ (PD from Alice to Bob)
can be tightly upper bounded by the binary entropy $h$
\begin{equation}
    \chi(b:E|y=1)\equiv H(\rho_E)-\frac12\sum_{b=0}^1H(\rho_{E|b})\le h\left[\frac12+\frac12\sqrt{\left(\frac{\beta_\QC}{2}\right)^2-1}\right]\equiv \chi_0(\beta_\QC).
\end{equation}
The Holevo quantity upper bounds the amount of classical information (about the variable $b$ (when $y=1$) that can be extracted by any measurement on the quantum state $\rho_{E|b}$. In other words, if you encode classical information $b$ in quantum states $\rho_{E|b}$, then the accessible information (mutual information between $b$ and the outcome of any measurement $E$ is at most $\chi(b:E|y=1)$. The expression is valid for $\beta_\QC>2$\footnote{Note that in Sec.~\ref{sec:CHAIN-protocol} $\beta_{\downarrow\QC}=2-\sqrt2p$, then $\beta_QC=2\sqrt2p$, then $Q=(1-p)/2$.}. In short, $r_\mathrm{CHSH}$ of Eq.~\eqref{eq:2007device} reads as
\begin{equation}\label{eq:2007device2}
r=1-h(\mathrm{Q})-\chi_0(\beta_\QC).
\end{equation}

\textit{Eve's guessing measurement $p_G^M$}. If Eve has some probability $p_g^M$ of making a correct guess on the choice of measurement settings (a flag $f=1$) so that she fixes a priori the players' outcomes engineering
\begin{equation}
    \beta_\QC^\prime(\bm p)=4p_g^M+(1-p_g^M)\beta_\QC(\hat{\bm p}),
     \text{ to mimic  }\beta_\QC^\prime(\bm p)\in[2,2\sqrt2]:\,p_g^M\in[0,\sqrt{2}-1],
\end{equation}
(otherwise her guess is uncorrelated, $f=0$), then Eq. \eqref{eq:2007device2} becomes~\cite{Kofler2006},
\begin{equation}\label{DevetakWinter}
		r=1-h(Q)-p_g^M-(1-p_g^M)\,\chi_0(\beta_\QC^\prime)
\end{equation}
To take into account the detection loophole, the key rate in Eqs \eqref{DevetakWinter} must be computed on $\hat{\bm{p}}$ rather than the ideal $\bm{p}$ (with $\eta_A=\eta_B=1$) as discussed in Eq. \eqref{eq:phat} of Sec. \ref{subsec:detection-loophole}, therefore for $\eta_A=\eta_B=\eta$, $Q=\eta(1-\eta)$ and $\beta(\hat{\bm{p}})=2\sqrt{2}\eta^2+2(1-\eta)^2$ ( $\langle M_{A|0} \rangle=\langle N_{B|0} \rangle=0$ in Eq. \eqref{eq:Sphat})
. 

Fig~\ref{fig:rate_vs_eta-q} shows $r\to0 \Longrightarrow \eta>0.924$ at $p_g^M=0$. 
\begin{figure}
    \centering
    \begin{subfigure}{.49\textwidth}
        \centering
        \includegraphics[width=\textwidth]{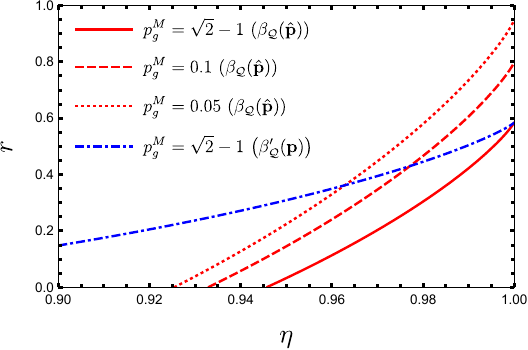}
        \caption{\textit{detection efficiency vs. fake BI violation}}
        \label{fig:rate_vs_eta-q}
    \end{subfigure}
    \hfill
    \begin{subfigure}{.49\textwidth}
        \centering
        \includegraphics[width=\textwidth]{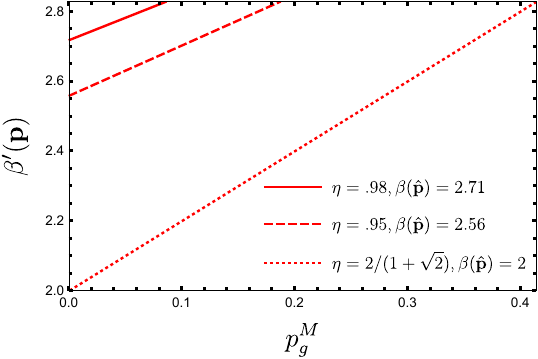}
        \caption{\textit{fake BI vs. Eve's measurement guessing probability}}
        \label{fig:rate_vs_eta-q2}
        \end{subfigure}
        \caption{In \ref{fig:rate_vs_eta-q} -- Eve's guessing probability $p_g$ open the detection loophole (dash-dotted blue line) where $\beta_\QC(\hat{\bm p})=2$, but $\beta_\QC^\prime(\bm p)\in(2,2\sqrt2]$. To close this loophole for large values of $p_g$ a higher critical detection efficiency $\eta^*$ is required. In \ref{fig:rate_vs_eta-q2} higher is $\beta_\QC(\hat{\bm p})$ less manoeuvre is left to Eve to mimic BI violation by $\beta_\QC^\prime(\bm p)$.}
    \label{fig:Acin2007a}
\end{figure}
A variant of CHSH$_\chi$ protocol that
instead of $M_{A|0}=\{M_{a|0}\}_{a=0}^1$ assume $M_{A|0}=\{M_{a|0}\}_{a=0}^3$ reduces $\eta^*=0.909$ \cite{Ma2012a}. 
Ref. \cite{McKague2009} shows an analogue of Eq.\eqref{eq:2007device} for coherent attacks with memoryless measurement devices.
\paragraph{CHSH$_\chi$ protocol vs. BB84}
For the entanglement-based BB84, $r_{\mathrm{BB84}}$ in Eq.~\eqref{eq:2007device} has $\chi^\prime(b:E|y=1) \leq h(\beta_\QC/(2\sqrt{2})+\mathrm{Q})$ with $\beta_\QC$ computed on $\rho_{\mathrm{AB}}=\sum_{j=\pm } (1+jC)/2\ket{\Phi_j}\bra{\Phi_j}$, 
where \( C = \sqrt{(\beta_\QC/2)^2 - 1} \) is the concurrence that maximize $\beta_\QC$ at a given amount of entanglement. The measurements defined by Eve are \(\sigma_z,\sigma_x \), fr $y=1,2$, and \(\frac{1}{1+C^2}\sigma_z \pm \frac{C}{1+C^2}\sigma_x \) for $x=1,2$. Under collective attack, with this realization (state and settings), CHSH$_\chi$ protocol saturates Holevo bound and BB84 turns out \textit{unsafe}. Indeed, not only \( \rho_{AB} \), but also guessing Alice’s measurements depend explicitly on the observed values \( \beta_\QC \) and \( Q \).
\paragraph{CV - CHSH$_\chi$ protocol} 
The encoding of CHSH$_\chi$ protocol (see Sec. \ref{DIQKD-collective} and \cite{acin2007device}) is possible with continuous variable (CV) \cite{Marshall2014}. Typically, CV adopts Gaussian states that alone cannot violate a BI \cite{Paternostro2009}. As a result, non-Gaussian states or measurements are necessary, though they are harder to produce experimentally. This poses challenges for developing a CV-based version of DI-QKD, as most current CV-QKD protocols rely on Gaussian states. 
This challenge, however, can be addressed by utilizing, for instance, a single mode of the electromagnetic field as the harmonic oscillator
by \textit{GKP encoding} (Gottesman, Kitaev, Preskill): embedding a two-level Hilbert space into the full infinite-dimensional space \cite{Gottesman2001}. A qubit is encoded in the infinite-dimensional space of an oscillator, allowing protection against arbitrary but small shifts in the canonical variables such as position and momentum. 

However, this result is possible in 2-input and 2-output scenarios thanks to the \textit{qubit reduction argument} (Jordan's lemma).
Do similar results exist for more complex scenarios for $m$ measurements of $d$
 outcomes \cite{Navascues2008}? Along this line, some progress for $n=2, m\to \infty$ in \cite{Vertesi2009,Briet2009} conjectured in \cite{Brunner2008a}.

We analyzed CHSH protocol under the no-signaling attack and recalled it CHSH$_\chi$ under quantum collective attacks by related the key rate to Holevo quantity $\chi$. In the next section, instead, the key rate is obtained via $H_\mathrm{min}=-\log_2p_g$ as we discussed above, and we recall the protocol CHSH$_{p_g}$.
\subsection{CHSH$_{p_g}$ protocol - bounding the attacks with $H_{\min}$ assuming independent measurement devices}\label{CHSHp}
It is remarkable to highlight the issues due to memory effects in the data generation~\cite{barrett2013memory}.
\textbf{CHSH$_{p_g}$ protocol}. \textit{Data generation}.  Alice and Bob select $x,y\in\{1,\dots,M\}$
to generate a raw key $K_A=K_B\in\{0,1\}^n$ with $n=M$ via separate pair of commuting measurements \cite{masanes2011secure,hanggi2010device} 
	\begin{equation}\label{commutingmeasurements}
		[M_{a|x},N_{b|y}]=0,\qquad [M_{a|x},M_{a'|x'}]=[N_{b|y},N_{b'|y'}]=0,\quad x,y=1,\dots,M.
	\end{equation}
The commutation relations can be satisfied by either using the devices in parallel where the $n$ bits of the raw key are generated by $n$ separate and non-interacting devices or using in a sequential way in which the raw key is generated by repeatedly performing measurements provided that the functioning of the devices do not depend on any internal memory storing the quantum states and measurement results obtained in each round so that $\bm{p}$ has only one entry
 \begin{equation}
p(a_1,\dots,a_N,b_1,\dots,b_N|x_1,\dots,x_N,y_1,\dots,y_N)=\textrm{tr}\left(\rho_{AB}\prod_{i=1}^{N}M_{a_i|x_i}N_{b_i|y_i}\right), 
 \end{equation} 
where \(\rho_{AB} = \textrm{tr}_E(\rho_{ABE})\).  After Alice's \(n\) systems have been measured, $\rho_{AE} = \sum_{\boldsymbol{a}} p(\boldsymbol{a}|\boldsymbol{x}_{raw}) \ket{\boldsymbol{a}} \bra{\boldsymbol{a}} \otimes \rho_{E|a}$ where \(\rho_{E|a}\) is the reduced state of Eve conditioned on Alice measuring outcome \(\boldsymbol{a}\)\footnote{In a parallel execution the devices receive all inputs $\bm{x}, \bm{y}$ simultaneously, and return the outputs $\bm a , \bm b$ at once. The security proof is equivalent to the sequential execution \cite{Jain2020}, but experimentally more challenging as higher-dimensional entanglement is required.
}
In the protocol Alice always uses a particular measurement $x_{\text{raw}}$ (call it $a_0$) to produce the sifted key bit.

\textbf{Security 1 way PD}~\cite{masanes2011secure,hanggi2010device}. After PD Eve possesses some (possibly huge-dimensional) quantum system $E$.
Whatever coherent processing she performs, she must end with a single classical guess $\hat a\in\{0,1\}$ for Alice’s bit $a_0$.
In each round Eve wants to maximize her success probability $p_{\text{succ}}=p(\hat a=a_0\,|\,x_{\text{raw}})$. The best she can do is always output the most likely value of $a_0$. Therefore the optimal success probability is
$
    p_g=\max_{a\in\{0,1\}} P(a\,|\,x_{\text{raw}}),
$  (For a binary outcome variable this coincides with the operationally meaningful smooth min-entropy $H_{\min}(A|E)=-\log_2p_g$ as expressed in Eq.~\eqref{H=-logP}). Because the two outcomes are encoded as $\pm1$, $p_g$ can be rewritten as $p_g=\bigl(1+\lvert\langle a_0\rangle\rvert\bigr)/2$. Again, via Jordan's lemma it is enough to consider $\ket{\psi0}=\cos\theta\ket{00}+\sin\theta\ket{11}$ where $\theta\in[0,\frac{\pi}{4}]$ determines the amount of the entanglement. Then, $\lvert\langle a_0\rangle\rvert=\cos2\theta\equiv\alpha\Longrightarrow p_g=(1+|\cos2\theta|)/2$ with $M_{A|0}=\sigma_z$ and $M_{A|1}=\sigma_x$. It is known (and easy to verify by direct optimization) that Bob should measure $N_{B|y}=\cos\mu\,\sigma_z+(-1)^y\sin\mu \,\sigma_x$ (inverting the role of the players in Eq.~\eqref{eq:realization2}), with $\mu$ chosen to maximize the CHSH value.
A straightforward calculation gives~\cite{Horodecki1995}

\begin{equation}
    \beta_\QC(\alpha,\mu)
  = 2\bigl(\cos 2\theta\cos\mu + \sin 2\theta\sin\mu\bigr)
  = 2\bigl(\alpha\cos\mu + \sqrt{1-\alpha^{2}}\,\sin\mu\bigr).
\end{equation}
Maximizing the right-hand side over $\mu$ yields
\begin{equation}
    \beta_\QC(\alpha)
  = 2\sqrt{2-\alpha^{2}} \Longrightarrow  p_g
    \;=\;
    \frac12\bigl(1+|\alpha|\bigr)
    \;\le\;
    \frac12\Bigl(1+\sqrt{\,2-\tfrac{\beta_\QC^{2}}4}\Bigr)
    \;=\;
    f(\beta_\QC).
\end{equation}
For $n$ independent rounds (the “no memory” or “causal-independence” assumption), $p_g$ translates into a smooth min-entropy $H_{\min}(A|E)=-\log_2 p_g^n=-n\log_2 f(\beta_\QC)$.
The amount of public communication needed to reconcile Bob with Alice is
$n\,H(A|B)=n\,h(Q)$. Privacy amplification therefore yields the asymptotic key rate
\begin{equation}
    r\;\ge\;\frac{1}{n}\Bigl[H_{\min}(A|E)-n\,h(Q)\Bigr]\;=\;-\log_2 f(\beta_\QC)-h(Q)=1-h(Q)\;-\log_2\!\Bigl[1+\sqrt{2-\tfrac{\beta_\QC^{2}}4}\Bigr].
\end{equation}
Last step uses $h\left(\frac12(1+x)\right)=1-\log_2(1+x)$. For the usual depolarized noise ($\beta_\QC =2\sqrt2\,\nu,\;Q=(1-\nu)/2$),
\begin{equation}
    r(\nu)\;\ge\;1-h\!\Bigl(\tfrac{1-\nu}{2}\Bigr)
          -\log_{2}\!\Bigl[1+\sqrt{2(1-\nu^{2})}\Bigr]\;,
\end{equation}
this rate vanishes at
$Q\simeq5\%$ (visibility $\nu \simeq0.95$).
\begin{table}[h!]
\centering
\renewcommand{\arraystretch}{1.3}
\begin{tabularx}{\textwidth}{@{}lX X@{}}
\toprule
\textbf{Aspect} &
\textbf{CHSH$_{p_g}$}~\cite{masanes2011secure,hanggi2010device} &
\textbf{CHSH$_\chi$}~\cite{pabgs09andSangouard.4} \\
\midrule
\textbf{Attack} &
Arbitrary (coherent) attacks on \emph{any} dimension, only causally independent uses assumed &
Collective IID attacks only \\
\textbf{Security metric} &
Smooth min-entropy $\Rightarrow$ \emph{explicit} upper bound on Eve’s success probability, implies universally composable secrecy &
Holevo information $\chi(b\!:\!E|y=1)$ $\Rightarrow$ bound on \emph{average} quantum information; extra steps needed for composability \\
\textbf{Technique} &
Direct operator inequality gives $p_g \le f(\beta_\QC)$ that is concave, tight, and easily generalized to other Bell inequalities &
Two-qubit reduction, optimization over eigenvalues and entropic inequalities \\
\textbf{Generality} &
Same machinery works for any Bell inequality once $f(\beta_\QC)$ is computed (via SDP) &
Tailored to CHSH; extending to other inequalities is non-trivial \\
\textbf{Noise threshold} &
More conservative (5\% QBER) because the proof holds for stronger adversaries &
Higher tolerance (7.1\% QBER) but under weaker assumptions \\
\textbf{Conceptual clarity} &
Security expressed in \emph{operational} terms: “Eve cannot guess the raw bit better than $f(\beta_\QC)$” &
Security expressed via $\chi$, which is less intuitive and only bounds Eve’s \emph{knowledge in principle} \\
\bottomrule
\end{tabularx}
\caption{Why the guessing-probability proof improves on the Holevo-$\chi$ approach}
\end{table}

\textbf{Security 2-w PD}~\cite{tan2020advantage}. \emph{advantage distillation} (AD, see Eq.~\eqref{eq:AD}), improves the tolerable imperfections of quantum devices. The security proof is fully composable and phrased in the \emph{min-entropy} framework: for every set of observed correlators, a level-two NPA semidefinite programme is used to upper-bound the trace distance \( d(\rho_{E|00},\rho_{E|11}) = \frac{1}{2}\lVert \rho_{E|00} - \rho_{E|11} \rVert_1 \) between Eve’s conditional states when Alice and Bob both obtain outcome \(0\) or \(1\). A sufficient condition for security is then provided by Corollary~1 in \cite{tan2020advantage}, which states that \( 1 - d > \varepsilon/(1 - \varepsilon) \), where \( \varepsilon \) is the bit-error rate after advantage distillation. This replaces earlier security analyses based on the Holevo quantity \( \chi \), and yields direct bounds on the smooth min-entropy \( H_{\min}(A|E) \), which governs universally composable key extraction.
The authors analyze two noise models: \textit{(i)} depolarizing noise, parameterized by \( q \in [0, 1/2] \), where the observed distribution is \( \bm{p} = (1 - 2q)\bm{p}_T + q/2 \) for an ideal target distribution \( \bm{p}_T \); and \textit{(ii)} detection inefficiency, where each outcome is independently flipped by a \( \sigma_z \) channel with probability \( 1 - \eta \), for efficiency \( \eta \in [0,1] \). In both cases, advantage distillation raises the noise tolerance and lowers the detection threshold (see Tab.\ref{tab:tan2020}).
\begin{table}[h!]
\centering
\renewcommand{\arraystretch}{1.4}
\small
\begin{tabular}{>{\raggedright}p{4.7cm} >{\raggedright}p{7.3cm} >{\centering\arraybackslash}p{1.5cm} >{\centering\arraybackslash}p{1.5cm}}
\toprule
\textbf{Description of $\bm{p}_T$} & \textbf{State and measurements for $\bm{p}_T$} & $q_t$ & $\eta_t$ \\
\midrule
(i) Achieves maximal CHSH value with the measurements $A_0$, $A_1$, $B_1$, $B_2$. & 
$A_0 = B_0 = \sigma_z$, $A_1 = \sigma_x$, $B_1 = (\sigma_x + \sigma_z)/\sqrt{2}$, $B_2 = (\sigma_x - \sigma_z)/\sqrt{2}$. & 6.0\% & 93.7\% \\
\addlinespace
(ii) Modification of a distribution exhibiting the Hardy paradox for improved robustness against limited detection efficiency. & 
$\lvert \psi \rangle = \sqrt{\kappa}(|01\rangle + |10\rangle) + \sqrt{1 - 2\kappa}|11\rangle$, with $\kappa = (3 - \sqrt{5})/2$; the 0 outcomes correspond to projectors onto:
\[
\begin{aligned}
|a_0\rangle &= |b_0\rangle \propto (\sqrt{1 + 2\kappa} - \sqrt{1 - 2\kappa})|0\rangle + 2\sqrt{\kappa}|1\rangle, \\
|a_1\rangle &= |b_1\rangle \approx 0.37972|0\rangle + 0.92510|1\rangle, \\
|a_2\rangle &= |b_2\rangle \approx 0.90821|0\rangle + 0.41851|1\rangle.
\end{aligned}
\]
& 3.2\% & 92.0\% \\
\addlinespace
(iii) Includes the Mayers–Yao self-test and the CHSH measurements. & 
$A_0 = B_0 = \sigma_z$, $A_1 = B_1 = (\sigma_x + \sigma_z)/\sqrt{2}$, $A_2 = B_2 = \sigma_x$, $A_3 = B_3 = (\sigma_x - \sigma_z)/\sqrt{2}$. & 6.8\% & 92.7\% \\
\addlinespace
(iv) Achieves maximal CHSH value with the measurements $A_0$, $A_1$, $B_0$, $B_1$. & 
$A_0 = \sigma_z$, $A_1 = \sigma_x$, $B_0 = (\sigma_x + \sigma_z)/\sqrt{2}$, $B_1 = (\sigma_x - \sigma_z)/\sqrt{2}$. & 7.7\% & 91.7\% \\
\addlinespace
(v) Similar to (iv), but with measurements optimized for robustness against depolarizing noise. & 
Measurements are in the $x$-$z$ plane at angles:
\[
\theta_{A_0} = 0.42, \, \theta_{A_1} = 1.79, \,
\theta_{B_0} = 0.86, \, \theta_{B_1} = 2.63.
\]
& 9.1\% & 90.0\% \\
\addlinespace
(vi) Similar to (iv), but states and measurements maximizing CHSH for each value of detection efficiency $\eta$. & 
$\lvert \psi \rangle = \cos\Omega |00\rangle + \sin\Omega |11\rangle$, with $\Omega = 0.6224$;  
0 outcomes correspond to projectors onto states of the form:
\(
\cos(\theta/2)|0\rangle + \sin(\theta/2)|1\rangle,
\)
with
$
\theta_{A_0} = -\theta_{B_0} = -0.36, \,
\theta_{A_1} = -\theta_{B_1} = 1.15.
$
& 7.3\% & 89.1\% \\
\bottomrule
\end{tabular}
\caption{Noise thresholds for advantage distillation in various DI-QKD scenarios. $\bm{p}_T$ is the ideal target distribution, $q_t$ is the tolerated depolarizing noise, and $\eta_t$ is the minimum detection efficiency under which a positive key rate is provable~\cite{tan2020advantage}.}
\label{tab:tan2020}
\end{table}
Because the SDP formulation depends only on observed Bell correlators, the method is general and adaptable to other inequalities and more advanced Bell tests.

\subsection{CHSH$_T$ -- Noisy preprocessing}\label{subsec:CHSHT}
\noindent Noisy preprocessing denotes the deliberate random flipping of raw key bits prior to information reconciliation and privacy amplification. When appropriately tuned, this strategy reshapes the trade-off between Eve’s information and the observed error rate, thereby increasing the tolerable noise level for a given Bell violation and improving finite-size key rates. In the following, we focus on the CHSH$_T$-based realization, which provides a concrete and widely studied implementation of this idea.

\textbf{Protocol}. \textit{Data generation}. If Bob applies the preprocessing $T$ of Eq.\eqref{devatekwinterformula} (see \cite{Kraus2005,Acin2006})
generating a new raw key $T(K_B)$ by flipping each bit independently with probability $q$ before the post-processing~\cite{Ho2020}. 

\textbf{Security}. The secret key rate is derived using a min-entropy uncertainty relation that bounds Eve’s knowledge about Bob’s (post-processed) raw key $\widehat{B}_1$ in terms of the observed CHSH violation $\beta_\QC$. The rate takes the form: $r \geq H_{\min}(\widehat{B}_1 | E) - H(\widehat{B}_1 | A_0)$,
where the first term is bounded via a refined uncertainty relation, and the second is the conditional Shannon entropy (the error correction cost).
The bound on $H_{\min}(\widehat{B}_1 | E)$ is obtained using the \textit{entropy accumulation theorem (EAT)}, and incorporates the effect of Bob’s deliberate bit-flip preprocessing (probability $q$) through the function $I_p(S)$. This function is not a Holevo quantity, but rather a derived bound on min-entropy that depends concavely on the CHSH violation $\beta_\QC$ and reads~\cite{Ho2020}
\begin{equation}
    I_p(\beta_\QC) =\chi_0(\beta_\QC) - h\left(\frac{1+\sqrt{1-q(1-q)(8-\beta_\QC^2)}}{2}\right),\quad (p_g^M=0).
\end{equation}
so that the asymptotic rate becomes 
\begin{equation}
    r\ge 1-h(Q) -I_p(\beta_\QC),\qquad
     Q=\frac{1}{2}\left(1-\frac{\beta_\QC}{\sqrt{2}}\right),
\end{equation}
where we model $H(\widehat{B}_1|A_0)=h(Q)$ with a binary key $\widehat{B}_1$ and four-valued outcome $A_0 \in \{0,1,2,3\}$ (representing 0-click, 1-click-left, 1-click-right, 2-click), this entropy is:
 $$
H(\widehat{B}_1 | A_0) = \sum_{a=0}^3 p_{A_0}(a) \cdot h\left( p(\widehat{B}_1 = 1 | A_0 = a) \right)=h(Q),
$$
where $p_{A_0}(a)$ is the empirical probability distribution of Alice’s outcome $A_0$, $p(\widehat{B}_1 = 1 | A_0 = a)$ the empirical conditional probability of Bob's bit being 1 given $A_0 = a$ and assuming binary errors with some effective QBER $Q(\beta_\QC)$.
Again, including the artificial pre-processing noise $q$ to $K_B$ damages both the correlation between Alice and Bob and the correlation to Eve. However, since the possibility of generating a key depends on the difference between the strengths of these correlations, the net effect is positive.
optimization over $q$ therefore lowers the security threshold without changing the physical devices or the Bell test itself.  The security proof applies to arbitrary, memory-free devices and is fully composable.
A realistic performance analysis for a pulsed SPDC source, non–photon-number-resolving detectors and finite-mode multiplexing shows a dramatic improvement: the global detection-efficiency threshold drops from $92.7\%$ (original CHSH protocol) to $90.9\%$ when four-outcome error correction is used, and further to $83.2\%$ once the noisy pre-processing is included—an effective $78\%$ relative gain that approaches the fundamental Bell limit of $82.8\%$.  At unit efficiency the protocol still produces nearly one secret bit per entangled pair despite multi-pair emissions, and for $\eta>0.85$ the predicted key rate is within a factor $\lesssim2$ of the BB84 Shannon limit because the only extra requirement is software-level bit-flipping advances a near-term DI-QKD photonic demonstration.
 This represents an improvement compared to the efficient post-processing method described in \cite{Ma2012a}.
\subsection{CHSH\texorpdfstring{$_\ell$}{l} protocol - DI-QKD with local test and entanglement swapping}\label{subsec:CHSHl}
A \emph{local Bell-test} approach is introduced to close the detection loophole arising from the lossy channel between Alice and Bob~\cite{lim2013device}.
\begin{figure}[h!]
\centering    
\begin{tikzpicture}[
    box/.style={draw, rounded corners, minimum width=1.cm, minimum height=.8cm, thick, fill=white},
    circ/.style={draw, circle, minimum size=.8cm, thick, fill=white},
    font=\small
  ]

\begin{scope}
  \fill[red!20] (-2.,1.3) rectangle (2.,-1.3);
\end{scope}

\node[box] (mtest) at (1.2, 0) {$\mathrm{M_{test}}$};
\node[circ] (s) at (0, 0) {$S$};
\node[box] (mkey) at (-1.2, 0) {$\mathrm{M_{key}}$};
\node (alice) at (s.north) [yshift=0.5cm] {Alice};
\draw[->, thick] (s) -- (mtest);
\draw[->, thick] (s) -- (mkey);

\node (zxlabel) at (mkey.north) [yshift=0.6cm] {$\{\sigma_z, \sigma_x\}$};
\node (a01label) at (mkey.south) [yshift=-0.6cm] {$\{0,1\}$};
\draw[->, thick] (zxlabel) -- (mkey.north);
\draw[->, thick] (mkey.south) -- (a01label);

\node (UVPlabel) at (mtest.north) [yshift=0.6cm] {$\{U, V, P\}$};
\node (a01label2) at (mtest.south) [yshift=-0.6cm] {$\{0,1\}$};
\draw[->, thick] (UVPlabel) -- (mtest.north);
\draw[->, thick] (mtest.south) -- (a01label2);
\begin{scope}
  \fill[blue!20] (4.4,1.3) rectangle (7,-1.3);
\end{scope}
\node[box] (mkeyB) at (6.3, 0) {$\mathrm{M_{key}^\prime}$};
\node[circ] (sB) at (5, 0) {$S^\prime$};
\node (bob) at (sB.north) [yshift=0.5cm] {Bob};

\draw[->, thick] (sB) -- (mkeyB);

\node (zxlabelB) at (mkeyB.north) [yshift=0.6cm] {$\{\sigma_z, \sigma_x\}$};
\node (a01labelB) at (mkeyB.south) [yshift=-0.6cm] {$\{0,1\}$};

\draw[->, thick] (zxlabelB) -- (mkeyB.north);
\draw[->, thick] (mkeyB.south) -- (a01labelB);
%
%

\node[box] (sC) at (3.2, 0) {BSM};
\node (charlie) at (sC.north) [yshift=0.5cm] {Charlie};
\draw[->, thick] (mtest.east) -- (sC);
\draw[->, thick] (sB) -- (sC);

\node (c01label) at (sC.south) [yshift=-0.6cm] {$\{0,1\}^2$};
\draw[->, thick] (sC) -- (c01label);

\end{tikzpicture}
    \caption{The CHSH$_\ell$ protocol~\cite{lim2013device}, inspired by the time-reversed BB84, introduces an untrusted party, Charlie, who performs a Bell-state measurement on $\rho$ sent by Alice and Bob. He outputs a pass or fail indicating success, and if successful, provides two bits for Alice to apply bit-flip corrections.}
\end{figure}

\textbf{Protocol}. \textit{Data generation}. The CHSH test is performed entirely inside Alice’s lab using a setup composed of three separate devices: a source device $\mathcal{S}$ that generates entangled systems $\rho$, and two measurement devices $\mathcal{M}_{\text{key}}$ and $\mathcal{M}_{\text{test}}$. The device $\mathcal{M}_{\text{key}}$ implements two settings $\{\sigma_x, \sigma_z\}$ to generate the raw key, while $\mathcal{M}_{\text{test}}$ has three settings—two used to evaluate $\beta_\QC$, and one to send $\rho$ into the quantum channel towards Charlie. Bob operates a similar setup with measurement $\mathcal{M}_{\text{key}}'$ and source $\mathcal{S}'$. A Bell-state measurement (BSM) is then performed by Charlie on qubits $A'$ and $B'$; when a successful detection occurs in orthogonal spatial modes, bosonic symmetry implies a collapse into the antisymmetric Bell state $\ket{\psi^-}_{A'B'}$, which by entanglement swapping also projects the remaining qubits $AB$ into $\ket{\psi^-}_{AB}$. This clever scheme ensures the overall detection efficiency $\eta = \eta_\ell + \eta_d$ has $\eta_\ell = 0$ due to the post-selected BSM (see Sec.~\ref{subsec:detection-loophole}), effectively closing the detection loophole that plagues long-distance DI-QKD.

\textbf{Security}. The security proof is proven via an entropic-uncertainty framework that links the observed local CHSH value $\beta_\QC$ to the incompatibility of Alice’s measurement bases, yielding a lower bound on the smooth min–entropy of the raw key. This distinguishes the protocol from standard DI-QKD schemes, which rely on the monogamy of nonlocal correlations. The asymptotic secret key rate is
\[
r=\frac{\ell}{m_x}=1-2 h(Q)-\log_2\left(1+\frac{\beta_\QC}{4\eta}\sqrt{8-\beta_\QC^2}\right).
\]
where $\ell$ is the secret key length, $m_x$ is the number of sifted key bits, and $\beta_\QC, \eta, Q$ are the tolerated CHSH value, Bell-state measurement efficiency, and quantum channel error rate, respectively. This rate is independent of transmission losses and matches that of BB84 in the ideal limit $\beta_\QC \to 2\sqrt{2}$. Ref.~\cite{tan2020advantage} shows how the key rate degrades with decreasing $\eta$ for visibility $\nu = 0.99$ and $0.999$, yet still allows key generation at the channel transmission $t \simeq 45\%$ (about 17 km of optical fiber), demonstrating remarkable robustness under realistic visibilities.
\subsection{CH-SH protocol -- Splitting CHSH parameter estimation}\label{ch-sh}
CH-SH protocol \cite{Sekatski2021,Woodhead2021}, so-called because Eve's knowledge is bounded by the quantity $X$ and $Y$ obtained by splitting the parameter estimation of $\beta_\QC$ more finely than the single CHSH score:
\begin{equation*}
    X = \langle A_0 \otimes (B_0 + B_1) \rangle, \qquad
    Y = \langle A_1 \otimes (B_0 - B_1) \rangle. 
\end{equation*}
Apart from that, the protocol is the same as CHSH$_\chi$ protocol.
Keeping the two CHSH sub-sums separate and derive a tight, closed-form lower bound $H(A_0|E)\ge 1-h_q(z)$ with $z=\frac12\!\left(\frac{Y}{\sqrt{4-X^2}}+1\right)$ and the noisy-pre-processing parameter $q$.  Because the full pair $(X,Y)$ is used, the bound is always as good as— and often strictly better than— the CHSH-only bound. Indeed, the set of points \((X, Y)\) where Eve's conditional entropy is bounded below by a constant is convex. Consequently, combining two quantum models into a new one results in Eve's conditional entropy being bounded by the weighted sum of the individual models' entropy bounds. Using this fact and considering all possible quantum models \((\psi_{ABE}, A_x, B_y)\) where \( X_{\textrm{model}} \geq X \) and \( Y_{\textrm{model}} \geq Y \), it is equivalent to bound Eve's conditional entropy with the following linear constraint:
\begin{equation}
    \frac{\cos(\Omega)}{2} X_{\textrm{model}} + \frac{\sin(\Omega)}{2} Y_{\textrm{model}} \geq \beta, 
    \label{generalizedchsh}
\end{equation}
where \(\beta = \frac{1}{2} (\cos(\Omega) X + \sin(\Omega) Y)\) is deduced from the quantities \( X \) and \( Y \). An improved bound on Eve's knowledge can be obtained in two identified regimes: \(\Omega \leq \frac{\pi}{4}\) and \(\Omega > \frac{\pi}{4}\). In the first regime, \(\Omega \leq \frac{\pi}{4}\), the optimal value is given by \(\cot(\Omega) = \frac{XY}{4 - X^2}\), and it is verified that the bound in this regime is better than the CHSH formula if \(\frac{4 - X^2}{XY} < 1\). The regime \(\Omega > \frac{\pi}{4}\) is more complicated, but numerical evaluation shows the advantage of the generalized CHSH inequality \eqref{generalizedchsh}. When applying this result to photonic implementations of DI-QKD using an SPDC source, the key rate improved by up to $37\%$ with $\eta=1$ while the efficiency threshold ($\approx 82.8 \%$ for photon-loss-limited optics) is unchanged  with $T=\mathrm{fp}$ \cite{Ho2020}.  The analysis remains analytic thanks to the \textit{qubit reduction argument} and shows that the advantage is largest when the correlations lie far from the line $X(X\!+\!Y)=4$ (where usual CHSH attacks are optimal).

Ref.~\cite{Woodhead2021} tackles the complementary mismatch that only Alice’s measurement $A_1$ is used for key generation although the CHSH expression is symmetric.  They introduce an \textit{asymmetric} CHSH family
$\beta_\alpha=\alpha\langle A_1B_1\rangle+\alpha\langle A_1B_2\rangle+\langle A_2B_1\rangle-\langle A_2B_2\rangle$
and, combining a qubit decomposition with BB84-type techniques, derive the exact conditional-entropy bound
$H(A_1|E)\ge\bar g_{q,\alpha}(S_\alpha)$ valid for every $\alpha$ and any noise $q$.  Optimising $\alpha$ and $q$ raises the tolerable depolarising error from 7.15 \% to 8.34 \%, and pushes the loss-only detection-efficiency threshold for maximally-entangled qubits from 90.8 \% down to 90.0 \%.  In the high-loss regime still larger gains arise when Bob re-optimizes his bases; conversely, in realistic noisy channels the two refinements (asymmetric tests + noise pre-processing) are most beneficial at moderate to high visibilities where they deliver the best device-independent key rates reported so far without changing the implementation.

\subsection{CHSH\textsubscript{\text{IMD}} protocol--conditional entropies via iterated-mean divergences} \label{CHSHIMD}
The CHSH$_{\text{IMD}}$~\cite{Brown2021} exploits a new family of Rényi-like quantities---the iterated‑mean divergences-to certify secret key directly from CHSH statistics with tighter entropy bounds than previous numerical techniques.

\textbf{Protocol}. \textit{Data generation}. This stage is identical to the CHSH setting, but the rate analysis departs from min‑entropy methods.  Alice and Bob use their raw correlation sample to solve an SDP that outputs the IMD‑based conditional entropy under the observed CHSH $\beta$.  This bound is evaluated at a low hierarchy level (degree 3), making the computation fast enough to be iterated in real time while the experiment runs.  Optimizing measurement angles within this loop raises the certified asymptotic key rate by $15–35\%$ in the high‑violation, low‑noise regime and reaches the analytic maximum of two bits of global randomness already with qubit systems.

\textbf{Security}
Security against collective quantum and general non‑signaling attacks follows because every IMD conditional entropy lower bounds the true von Neumann entropy. By deriving affine trade-off functions from the SDP dual, this approach enables a plug-and-play use of EAT. As a result, finite-key analysis inherits the same modular proof structure as modern BB84 protocols.  Importantly, the new technique lowers the minimal global detection efficiency for two‑qubit CHSH DI‑QKD without noisy preprocessing to $\eta_{\text{crit}}=0.843$.  See the mathematical technique of the protocol in Sec. \ref{subsubsec:numeriacallowerbounds}.
\subsection{CHSH\textsubscript{$2e$} protocol -- random key basis}\label{subsubsec:randomkeybais}
We name the CHSH$_{2e}$ protocol in Ref. \cite{Schwonnek2021} a variant of CHSH$_{p_g}$ protocol because it provides an extra Bob's choice that doubles the event to generate the raw key for $x=y=0,1$.

\textbf{Protocol}. \textit{Data generation}. Alice selects $x\in\{0,1\}$ with probability weight $p,1-p$. The same of CHSH protocol, but Bob selects $y\in\{0,1,2,3\}$, with probability respectively $q p, q(1-p), (1-q)/2, (1-q)/2$. The extra Bob's choice doubles the event to generate the raw key for $x=y=0,1$.
This enhances the security besides BI violation as decreases $p_g^M$, i.e. Eve does not even know the pair $(x,y)$ used to generate the key anymore.

\textbf{Security 1-w PD}. In CHSH$_{2e}$ protocol Eq. \eqref{DevetakWinter} yields
\begin{equation}\label{eq:r-DW2}
    r=\max_{\lambda} p_s\left[\lambda H(A_0|E)+(1-\lambda) H(A_1|E)-\lambda h(Q_0) - (1-\lambda)h(Q_1)\right],
\end{equation} 
with  $p_s=q(p^2+(1-p)^2)=p(x=y)$
is the matching basis probability; $Q_0$ and $Q_1$ is the $\mathrm{QBER}$ to generate the key when $x=y=0$ and $x=y=1$ respectively (in the depolarizing noise, \( Q_0 = Q_1 = \frac{1}{2} \left( 1 - \beta_\QC/\sqrt{8} \right) \)), $\lambda=p^2/p_s$, and $E$ is Eve's variable gathered before the error correction step with the quantum side-information. 
To estimate the LHS of Eq. \eqref{eq:r-DW2} with a function $C^\star(\beta_\QC)$depending only on $\beta_\QC$ the first step is to consider in the asymptotic limit regime $q\to 1$ (CHSH$_{2e}\to$ CHSH) and reformulate the tripartite problem among Alice, Bob, and Eve into a bipartite one so that the conditional entropy $H(A_i|E) \mapsto H(T_X(\rho_{AB}))-H(\rho_{AB})=D(\rho_{AB}||T_X(\rho_{AB}))$ is mapped into the \textit{quantum relative entropy}, as well as the entropy production of the quantum channel $T_X$ on $\rho_{AB}$, e.g. the \textit{pinching channel}
$T_X(\rho)=\sum_{a=0}^1(M_{a|x}\otimes \bm{1})\rho(M_{a|x}\otimes \bm{1})$. Eve's action mix $\rho_{AB}\in \mathcal{B}(\mathbb{C}^2\otimes\mathbb{C}^2)$, as $q\to 1$, and CHSH$_{2e}\to$ CHSH the same argument (see sec. \ref{sec:CHSHprotocol}, Jordan's lemma \ref{Jordanlemma})
applies here. It implies that
\begin{equation}
    C^\star\ge \inf_\mu \int_2^{2\sqrt{2}}C^\star_{M_\mathbb{C}(4)}(\beta_\QC)\,\mu(\mathrm{d}\beta_\QC),\quad
    \mbox{ s.th. : }
    \mu([2,2\sqrt{2}])\le 1,\,\mu\ge0,\,
    \int_{\beta_\QC^\prime=2}^{2\sqrt{2}}\beta_\QC^\prime\,\mu(\mathrm{d}\beta_\QC^\prime)=\beta_\QC,
\end{equation}
yielding
\begin{equation}
    C^\star(\beta_\QC)\;:=\;1-h\left(\frac{1+\tfrac{\beta_\QC}{2\sqrt2}}{2}\right),\qquad
    \xi(S)\;:=\;
  \begin{cases}
    \dfrac12 & \text{for }S\le 2.5,\\[6pt]
    1        & \text{for }S> 2.5,
  \end{cases}
\Longrightarrow\qquad
  r=\;\xi(S)\Bigl[C^\star(\beta_\QC)-h\left(Q(\beta_\QC)\right)].
\end{equation}

A direct computation of the $4\times 4$ complex operator $\hat\beta_\QC$ is still an open problem as no known proof techniques can be applied. To that end, a refined version of Pinsker's inequality formulates the SPD problem in terms of trace norm, i.e. $D(\rho||T(\rho))\ge \log 2 - h(1/2-||\rho-T(\rho)||_1/2)$, obtaining a closed, efficiently computable expression for the one-way secret fraction in the asymptotic IID regime. The optimization involves similar methods discussed in Sec. \ref{subsec:NPA} and \ref{sec:self-testing} (see details in Ref.~\cite{Schwonnek2021} and Sec.~\ref{sec:chap4}) to estimate the lower bound of uncertainty relations.
Most importantly, their bound is \emph{near-optimal}: for CHSH violations $S\lesssim2.5$ the choice $\lambda=\tfrac12$ (fully random basis) maximizes Eve’s conditional entropy $H(Z|E\Theta)$ essentially up to the linear algebraic limit, while for $S\gtrsim2.5$ the protocol smoothly reverts to the original single-basis variant (\(\lambda=1\)).  Thus the family interpolates between minimal and maximal entropy production without sacrificing simplicity or implementability.

Assuming isotropic depolarising noise, the new protocol tolerates QBER up to $Q=8.2\,\%$, corresponding to a critical CHSH value of $\beta_\QC=2.362$, whereas the CHSH$_{p_g}$~\cite{masanes2011secure} ceased to produce key beyond $Q=7.1\,\%$ $\beta_\QC=2.423$).  This $\sim\!15\,\%$ noise-budget increase brings today’s loophole-free, event-ready Bell experiments firmly into the positive-rate region: the 2015 NV-center test by Hensen \textit{et al.}~\cite{Hensen2015a} and its 2017 cold-atom successor by Rosenfeld \textit{et al.}~\cite{ Munich17} both admit asymptotic key rates of order $10^{-2}$–$10^{-1}$ bits per channel use.  A finite-key analysis based on EAT shows that $10^{8}$–$10^{10}$ rounds suffice to beat finite-size overheads against \emph{coherent} attacks, squarely within reach of next-generation entanglement-swapping platforms.  In contrast, all-photonic loophole-free experiments still fall short, primarily due to detector inefficiencies, but the authors’ Monte-Carlo study indicates that threshold efficiencies could drop below the often-quoted $94\,\%$ mark once modest multiplexing or precertification techniques mature.  Overall, the random-basis variant narrows the gap between theory and experiments.

\subsection{CHSH$_{p_{\mathcal{V}_p}}$ protocol - random postselection pre-processing} \label{subsub:randompostselection}
Ref.~\cite{Xu2022} introduces a simple yet powerful modification of the standard two–basis\,/\,three–setting DI-QKD protocol.

\textbf{Protocol}. \textit{Data generation}. After the usual CHSH  rounds, Alice and Bob keep only those \emph{key-generation} outcomes for which at least one detector clicks, and they do so \emph{independently at random}.  
Formally, every raw bit pair $(a,b)$ obtained with inputs $(\bar x,\bar y)=(1,3)$ is accepted with weight  
$\omega_{ab}\in\{1,p,p,p^{2}\}$, where $p$ is the local post-selection probability and $V_{p}=\{(a,b)|(a,b)=(0,0),(0,1),(1,0),(1,1)\}$ the set of potentially retained events.  
Because the secrecy analysis still uses the \emph{complete} detection statistics—including all double-clicks and no-clicks—the postselection does \textit{not} re-open the detection loophole~(Sec.\ref{subsec:detection-loophole}), yet it strips away the high-error, low-correlation data that would otherwise dominate the error-correction cost.  

\textbf{Security}. The authors cast Eve’s optimal attack as a semidefinite programme (see \ref{subsec:numericaltechniques}) over the \emph{unnormalized} quantum behaviors $p^{0}_{e}(a,b|x,y)$ compatible with the observed distribution $p(a,b|x,y)$, and bound her information through the conditional min-entropy $H_{\min}(A|E,V_{p})$ of the postselected string.

Then, in the asymptotic IID regime and against collective attacks, the extractable secret key rate per \emph{channel use} obeys the Devetak–Winter lower bound  and yields
\begin{equation}
    r
\;\ge\;
p_{V_{p}}\Bigl[
      H_{\min}\!\bigl(A\mid E,V_{p}\bigr)
      \;-\;
      H\!\bigl(A\mid B,V_{p}\bigr)
\Bigr]
.
\end{equation}
Here $H_{\min}(A|E,V_{p})=-\log_{2}p_g(A|E,V_{p})$ is the guessing probability of Alice’s bit $A$ by an eavesdropper Eve, \emph{conditioned on} the fact that the round was retained (i.e., passed the postselection $V_p$)\footnote{Formally, it is defined as:
$
p_g(A|E, V_p) = \max_{\{E_e\}} \sum_{a \in \{0,1\}} p(A = a, \hat{A}_E = a \mid V_p)
$
where:
$\{E_e\}$ is the POVM Eve applies to her system; $\hat{A}_E$ is Eve’s guess based on her measurement outcome; the probability is conditioned on the postselected events $V_p = \{00, 01, 10, 11\}$; the overall state is a classical–quantum (cq) state of the form
  $\rho_{AE} = \sum_a p_a |a\rangle\langle a| \otimes \rho_E^a$
  after the postselection.}, while $H(A|B,V_{p})$ is the (classically estimated) one-way error-correction cost for the retained bits.  All three quantities—$p_{V_{p}}$, $H_{\min}$ and $H(A|B)$—are directly computable from the experimentally observed correlations and the chosen postselection parameter $p$.
Specifically, the usual secrecy term is multiplied by the \emph{sifting factor} $p_{V_{p}}$, 
\begin{equation}
    p_{V_{p}}\;=\;\sum_{(a,b)\in V_{p}}\omega_{ab}\,p(a,b\!\mid\!\bar x,\bar y),
\qquad
\omega_{00}=1,\;
\omega_{01}=\omega_{10}=p,\;
\omega_{11}=p^{2},
\end{equation}
denote the probability that a key-generation round survives the random postselection. 
Numerical optimization over the quantum realization (quantum state, measurement settings and the single parameter $p$) shows that a positive key rate is already achieved for a global detection efficiency of $\eta\simeq68.5\%$—a dramatic improvement over the $92.4\%$ required by CHSH$_{p_g}$~\cite{masanes2011secure} and the $82.6\%$ achieved by noisy preprocessing.  
At $\eta=80\%$ the protocol cuts the Shannon error-correction term from $0.65$ to $0.033$ bits per retained round and \emph{simultaneously} raises Eve’s uncertainty, yielding $r\approx7.9\times10^{-5}$ bits per \emph{attempted} channel use—two orders of magnitude better than without postselection (see Fig. \ref{tab:diqkd_threshold}. 

\paragraph{Remark} As we emphasized, the security proof above is against collective attacks. More recently, it was shown in~\cite{Sandfuchs2023} that CHSH-based DI-QKD protocols with random post-selection are vulnerable to coherent attacks. This vulnerability, however, is specific to the random post-selection variant and does not extend to standard DI-QKD protocols such as E91, which do not employ post-selection.

\subsection{Mermin-Peres Magic Square Game-based DI-QKD} \label{msg-di-qkd}
Ref.~\cite{Zhen2023a} proposes a DI-QKD protocol based on the Mermin-Peres Magic Square Game (MPG) \cite{Mermin1990,Peres1990}. In this game, there are two players, Alice and Bob, who are forbidden to communicate with each other. In each round of the game, a referee generates two random 'trits' \(x, y \in \{0, 1, 2\}\) and sends the index \(x\) to Alice and the index \(y\) to Bob. Alice and Bob then reply to the referee with \([a_{0}^x, a_{1}^x, a_{2}^x]\) and \([b_{0}^y, b_{1}^y, b_{2}^y]\) under the conditions that \(a_{2}^x = a_{0}^x + a_{1}^x\) and \(b_{2}^y = b_{0}^y + b_{1}^y \oplus 1\). The winning condition is when \(a_{i=y}^x = b_{j=x}^y\). After the game, the referee decides whether Alice and Bob win or not according to the average winning probability 
\begin{equation}
    \omega = \frac{1}{9}\sum_{x,y} p(a_{y}^x = b_{x}^y | x, y)
\end{equation}
where \(p(a_{y}^x = b_{x}^y | x, y)\) is the winning probability of Alice and Bob with respect to \((x, y)\). The MPG DI-QKD in \cite{Zhen2023a} is based on the fact that for all classical strategies \(\omega \leq \frac{8}{9}\) and since all classical strategies are equal to local hidden variables, then \(\omega \leq \frac{8}{9}\) is actually a Bell inequality that some quantum strategies can violate \cite{Cabello2001, GISIN2007}. \\

\textbf{Protocol}. Alice and Bob initially generate data by playing the MPG. They announce their inputs and record the overlapped bits. To estimate parameters, Alice communicates to Bob which part of the bits serves as raw keys, with the remaining part of the bits announced to play the MPG. If the average winning probability estimated from the announced data is less than an expected value, they abort the protocol; otherwise, they perform data reconciliation on raw keys to obtain the final keys.

\textbf{Security}. The security analysis is against collective attacks in the asymptotic scenario and showed that if \(\omega > 0.9575\), it generates higher secret keys than the CHSH-based protocol for $\rho_p$ of Fig. \ref{fig:WernerCorr} with $\nu>0.978$ ($\eta=1$) or $\eta\ge 0.982$ ($\nu=1$). Later in \cite{CerveroMartin2025}, this protocol is generalized to the finite-round regime under general attacks.

\subsection{rDI-QKD: DI-QKD based on routed Bell tests}\label{sec:routedBI}
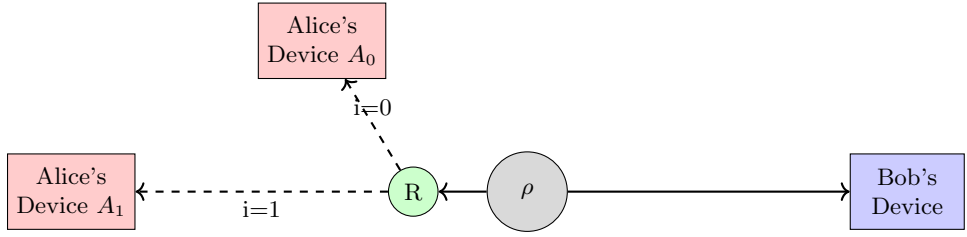
\begin{figure}[h]
\begin{tikzpicture}[scale=1, every node/.style={font=\small}]
\node[draw, circle, fill=gray!30, minimum size=30pt] (S) at (1,0) {$\rho$};

\node[draw, rectangle, fill=blue!20, minimum width=1.5cm, minimum height=1cm, align=center] 
(B) at (6,0) {Bob's \\ Device };

\node[draw, rectangle, fill=red!20, minimum width=1.5cm, minimum height=1cm, align=center] 
(A0) at (-1.7,2) {Alice's \\ Device $A_0$} ;

\node[draw, rectangle, fill=red!20, minimum width=1.5cm, minimum height=1cm, align=center] 
(A1) at (-5,0) {Alice's \\ Device $A_1$} ;

\node[draw, circle, fill=green!20, minimum size=10pt] (R) at (-0.5,0) {R};

\draw[thick, ->] (S) -- (B) node[midway, above] {};
\draw[thick, ->] (S) -- (R) node[midway, above] {};

\draw[thick, ->, dashed] (R) -- (A0) node[midway, above] {i=0};
\draw[thick, ->, dashed] (R) -- (A1) node[midway, below] {i=1};
\end{tikzpicture}
\caption{A schematic of routed Bell experiment introduced in \cite{chaturvedi2024extending}. In each round of experiment, Alice transmits her choice of $i\in\{0,1\}$ to the relay switch,$R$, which then transmits the quantum system from the source to either the measurement device,$A_0$, placed close to the source or $A_1$, placed further away, based on $i$.} 
\label{fig:routedBI}
\end{figure}
The current state-of-the-art combination of detection efficiency and visibility $v$ \eqref{fig:WernerCorr} (see Sec.\ref{sec:experiments}) is achieved only for distances less than 400 m. For DI-QKD to become a widely adopted near-term technology, operationally certifiable robust nonlocal correlations must be sustained over distances that are orders of magnitude greater ($\gg$100 km). Due to the high sensitivity inherent in traditional approaches for establishing nonlocal quantum correlations, these methods prove ineffective for long-distance DI-QKD. Consequently, developing alternative methods to establish nonlocality over large distances is essential.  \\
An approach to extending nonlocal correlations over large distances involves generalizing standard Bell experiments to the \textit{routed Bell experiments} introduced in \cite{chaturvedi2024extending}. In a routed Bell experiment, as illustrated in Figure \ref{fig:routedBI}, the measuring parties randomly select the location of their measurement in each round (the relay $R$ in Fig.\ref{fig:CVP2024}).\\
For a general quantum strategy, the correlations in a routed Bell experiment can be expressed as follows \cite{Lobo2024} 
\begin{equation}
    p(a,b|x,y,i)=\textrm{Tr}[C_i\otimes \mathcal{I}(\rho_{AB})M_{a|xi}\otimes M_{b|y}]=\mathrm{Tr}[\rho_{AB}\tilde{M}_{a|xi}\otimes M_{b|y}],
\end{equation}
where $C_i$ is the CPTP map describing the transmission of Alice's system on the short-path ($i=0$) or the long-path ($i=1$), and $\tilde{M}_{a|xi}=C_{i}^{\dagger}(M_{a|xi})$ are the elements of a valid POVM. Thus general correlation in a routed Bell experiment coincides with those of regular bipartite Bell experiment where Alice has $m_0+m_l$ inputs ($m_0$ and $m_1$ are the number of measurement settings at short-distance $i=0$ and long-distance $i=1$, respectively). Various subsets of general quantum correlations $\QC$ then can be defined as follows \cite{Lobo2024}:
\begin{definition}
\label{routedBIcorrelations}
\begin{itemize}
    \item \textbf{Short-range quantum correlations:} denoted as $\mathcal{Q}_{SR}$, short-range quantum (SRQ) correlations refer to the correlations achieved without the distribution of any entanglement to $A_1$. For these correlations, $C_1$ is an entanglement-breaking channel $C_1(\rho)=\sum_{\lambda}\mathrm{Tr}[N_{\lambda}\rho]\rho_{\lambda}$, where $N_{\lambda}$'s are the elements of a POVM. Therefore, the POVM elements of $M_{a|y_1}$ are maps to 
    \begin{equation}
        \tilde{M}_{a|x_1}=\sum_{\lambda}p(a|x,\lambda)N_{\lambda}, \quad \text{where}  \enspace\enspace p(a|x,\lambda)=\mathrm{Tr}[\rho_{\lambda}M_{a|x_1}], 
    \end{equation}
    which is equivalent to the statement that the measurement $M_{a|x1}$  are jointly-measurable. Therefore, the SRQ correlations can be expressed as 
    \begin{equation}
    \label{SRQcorrelation}
        p(a,b|x,y,i)=
\begin{cases} 
\textrm{Tr}[\rho_{AB}\tilde{M}_{a|x0}\otimes M_{b|y}] & \text{if } i=0, \\
\sum_{\lambda}p(a|x,\lambda)\textrm{Tr}[\rho_{AB}N_{\lambda}\otimes M_{b|y}] & \text{if } i=1.
\end{cases}
    \end{equation}
This operationally means that if the relay selects the short path $i=0$, the correlations are obtained by measuring a shared entangled state $\rho_{AB}$ as in a regular Bell
experiment. If it selects the long path $i=1$, a fixed measurement $N_{\lambda}$ is performed
on Alice’s system, yielding a classical outcome $\lambda$ with the probability distribution $p(a|x,\lambda)$, and transmit it to $A_1$.\\
    \item \textbf{Fully quantum marginal correlations:} denoted as $\mathcal{M}_{qq}$, the fully quantum marginal correlations are where the source prepares on Alice's side a pair of systems $A=(A_0, A_1)$ and if $i=0$ ($i=1$) the relay routes the first subsystem to the device $A_0$ ($A_1$). The resulting correlations can be encompassed to the bipartite marginals of qqq-correlations 
    \begin{equation}
        p(a_0,a_1,b|x_0,x_1,y)=\mathrm{Tr}[\rho_{A_0A_1B}M_{a|x,0}\otimes M_{a|x,1}\otimes M_{b|y}].
        \label{eq:Mqq}
    \end{equation}
    \item \textbf{Quantum-classical marginal correlations:} denoted as $\mathcal{M}_{qc}$, the quantum-classical marginal correlations can be obtained by further restricting the state $\rho_{A_0A_1B}$ to be a qqc-state as 
    \begin{equation}
        \rho_{A_0A_1B}=\sum_{\lambda}p(\lambda)\rho_{A_0B}\otimes\ket{\lambda}\bra{\lambda}_{A_1}.
    \end{equation}
\end{itemize}
\end{definition}
All the above correlations can be written as 
\begin{equation}
    p(a,b,i)=\mathrm{Tr}[\rho M_{a|xi}M_{b|y}],
\end{equation}
where $M_{a|xi}$ and $M_{b|y}$ are projectors that satisfy the following commutation relations for each subset 
\begin{align}
    [M_{a|xi}, M_{b|y}] &= 0, \quad &&\text{if } p \in \mathcal{Q}, \\
    [M_{a|xi}, M_{b|y}] &= 0, \quad [M_{a|x1}, M_{a'|x'1}] = 0, \quad &&\text{if } p \in \mathcal{Q}_{SR}, \\ 
    [M_{a|xi}, M_{b|y}] &= 0, \quad [M_{a|x0}, M_{a'|x'1}] = 0, \quad &&\text{if } p \in \mathcal{M}_{qq}, \\
    [M_{a|xi}, M_{b|y}] &= 0, \quad [M_{a|x0}, M_{a'|x'1}] = 0, \quad [M_{a|x1}, M_{a'|x'1}] = 0, \quad &&\text{if } p \in \mathcal{M}_{qc}.
    \label{routedBIcommutation}
\end{align}
which follows from the fact that each of the tensor products between subsystems in definition \ref{routedBIcorrelations} can be replaced by commutation relations.\\
As a result, the above representation fits in the framework of non-commutative polynomial optimization, which means that the sets defined in \ref{routedBIcorrelations} can be outer-approximated through SDP hierarchies. 

\paragraph{Trade-off relations between $\beta_0$ and $\beta_1$}
Consider a realistic CHSH scenario in which the source and measurement devices are imperfect. Considering the situation described in fig. \ref{fig:routedBI}, when $i=0$ Alice places her measurement device $A_0$ close to the source achieving effective detection efficiencies, $\eta_{A_0}$, while when $i=1$, she places her device $A_1$, further away from the source, therefore $\eta_{A_1}\leq \eta_{A_0}$. Similarly, for effective visibilities $v_{A_1}\leq v_{A_0}$, which result in the following quantum state shared between $(A_{i},B)$ 
\begin{equation}
    \rho(v_i)=v_{A_i}v_{B}\ket{\psi}\bra{\psi}+v_{A_i}(1-v_B)\left(\rho_B\otimes\frac{\bm{I}}{2}\right)+(1-v_{A_i})v_B\left(\frac{\bm{I}}{2}\otimes\rho_A\right)+(1-v_{A_i})(1-v_B)\frac{\bm{I}_4}{4},
\end{equation}
where $\rho_{A(B)}=\mathrm{Tr}_{B(A)}(\ket{\psi}\bra{\psi})$.
Treating $i\in\{0,1\}$ as an additional  Alice's input, denoting the location of her measurement device, the CHSH value is also dependent on $i$, which we denote by $\beta_i$. The following theorem captures the tradeoffs between $\beta_0$ and $\beta_1$ \cite{chaturvedi2024extending}

    \begin{theorem}
    \label{routedBItradeoff}
 If a loophole-free nonlocal correlation between $(A_0,B)$  is witnessed i.e. $\beta_0>2$, then loophole-free nonlocal correlations between $(A_1,B)$ can be certified whenever the following inequality is violated
\begin{equation}
\label{analyticaltradeoff}
    \beta_1\leq\sqrt{8-\beta_{0}^2}. 
\end{equation}
\end{theorem}
\begin{proof}
    The short-range correlations that were considered in \cite{chaturvedi2024extending} are represented as
\begin{equation}
    p(a,b|x,y,i)=
\begin{cases} 
\sum_{\lambda}p(\lambda)\textrm{Tr}[\rho_{A_{0}B}^{\lambda} M_{a|x0}\otimes M_{b|y}] & \text{if } i=0, \\
\sum_{\lambda}p(\lambda)p(a|x,\lambda)p(b|y,\lambda) & \text{if } i=1.
\end{cases}
\label{CVPshortrange}
\end{equation}
The intuition for the proof is as follows: for simplicity, consider the case where  $\beta_0=2\sqrt{2}$. Then, by standard
self-testing result \ref{sec:self-testing}, it can be inferred that the measurement $\{B_{b|x}\}$ corresponds
to a Pauli measurement on a two-dimensional subspace of $B$ that is maximally entangled
with $A_0$.  In particular, the
Bob's measurement outcomes must be fully random and uncorrelated with the classical instructions
$\lambda$ shared with $B$ which result in $p(b|y,\lambda)=p(b|y)=\frac{1}{2}$ for all $\lambda$'s which gives $\beta_1=0$
Although the assumption that $\beta_0=2\sqrt{2}$ is too strong, for any value $\beta_0>2$ there is a bound on how much the outcomes of  POVM $\{B_{b|y}\}$ can be correlated to other systems besides. Particularly $p(b|x,\lambda)\leq\frac{1}{2}(1+\frac{\sqrt{8-\beta_0^2}}{2})$ for all $\lambda$ \cite{pironio2010random}, which result in the bound \eqref{analyticaltradeoff}.
\end{proof}
\noindent Since $\sqrt{8-\beta_{0}^2}$ is a monotonically decreasing function of $\beta_0$, the inequality $\eqref{analyticaltradeoff}$ implies that any amount of loophole-free violation of the CHSH inequality between $(A_0,B)$ ($\beta_0>2$), results in reducing the threshold value of the CHSH expression $\beta_1$ \\
The significance of this result depends on the assumptions underlying the derivation of the bound \eqref{analyticaltradeoff} (see proof in \cite{chaturvedi2024extending}). Specifically, it relies on the interpretation of \textit{nonlocality} and \textit{ruling out classical models} in routed Bell experiments, since the local variable model can already be excluded by performing a simple CHSH test between $(A_0,B)$. The main purpose is that by observing the nonlocality in $(A_0,B)$ one can conclude about the classicality or nonclassicality of the observed outcomes of the remote device $A_1$. Therefore, theorem \ref{routedBItradeoff} provides a convenient tool to estimate the critical parameters, $(\eta_{A_1}^{*},v_{A_1}^{*})$, required for the loophole-free certification of nonlocal correlation between $(A_1,B)$. Given a tuple of experimental parameters $(\eta_B,\eta_{A_0},v_B,v_{A_0})$, the critical parameters $(\eta_{A_1}^{*},v_{A_1}^{*})$ are those that saturate inequality \eqref{analyticaltradeoff}.\\
Considering both asymmetric and symmetric cases, if the parties share a maximally nonlocal isotropic strategy, the following results hold:
\begin{enumerate}
    \item \textbf{Asymmetric Case}:  
    Suppose $B$ is placed extremely close to the perfect source, such that $\eta_B=1$, and all devices exhibit perfect visibilities. In this scenario:
    \[
    \beta_0 = 2\sqrt{2}\eta_{A_0}, \quad \beta_1 = 2\sqrt{2}\eta_{A_1}.
    \]
    The violation of \eqref{analyticaltradeoff} yields:
    \[
    \eta_{A_1} > \sqrt{1-\eta_{A_0}^2} = \eta_{A_1}^{*}.
    \]

This result highlights a central insight: the closer $A_0$ is placed to the source, the higher its effective detection efficiency, $\eta_{A_0}$, becomes. Consequently, the critical detection efficiency of $\eta_{A_1}$ decreases, allowing $A_1$ to be placed further from the source while still retaining loophole-free non-local correlations with $B$.

\item \textbf{Symmetric case}:  $(A_0,B)$ are equidistant from the source, such that $\eta_B=\eta_{A_0}=\eta$ and in the presence of perfect visibility for all devices, then 
\begin{eqnarray}
    \beta_0&=&2\sqrt{2}\eta^2+2(1-\eta)^2, \nonumber \\
    \beta_1&=&2\sqrt{2}\eta_{A_1}\eta +2(1-\eta_{A_1})(1-\eta). 
\end{eqnarray}
where a loophole-free violation can be observed for $\eta\in (\frac{2}{1+\sqrt{2}},1]$.  Fig.\ref{fig:CVP2024} shows the detection efficiency values $\eta_{A_1}^{*}$ versus $\eta$. The critical detection efficiency starts to decline after $\eta$ exceeds the threshold value ,$\eta=\frac{2}{1+\sqrt{2}}\approx 0.828$.  
\end{enumerate}
    \begin{figure*} 
    \centering
    \includegraphics[width=0.5\linewidth]{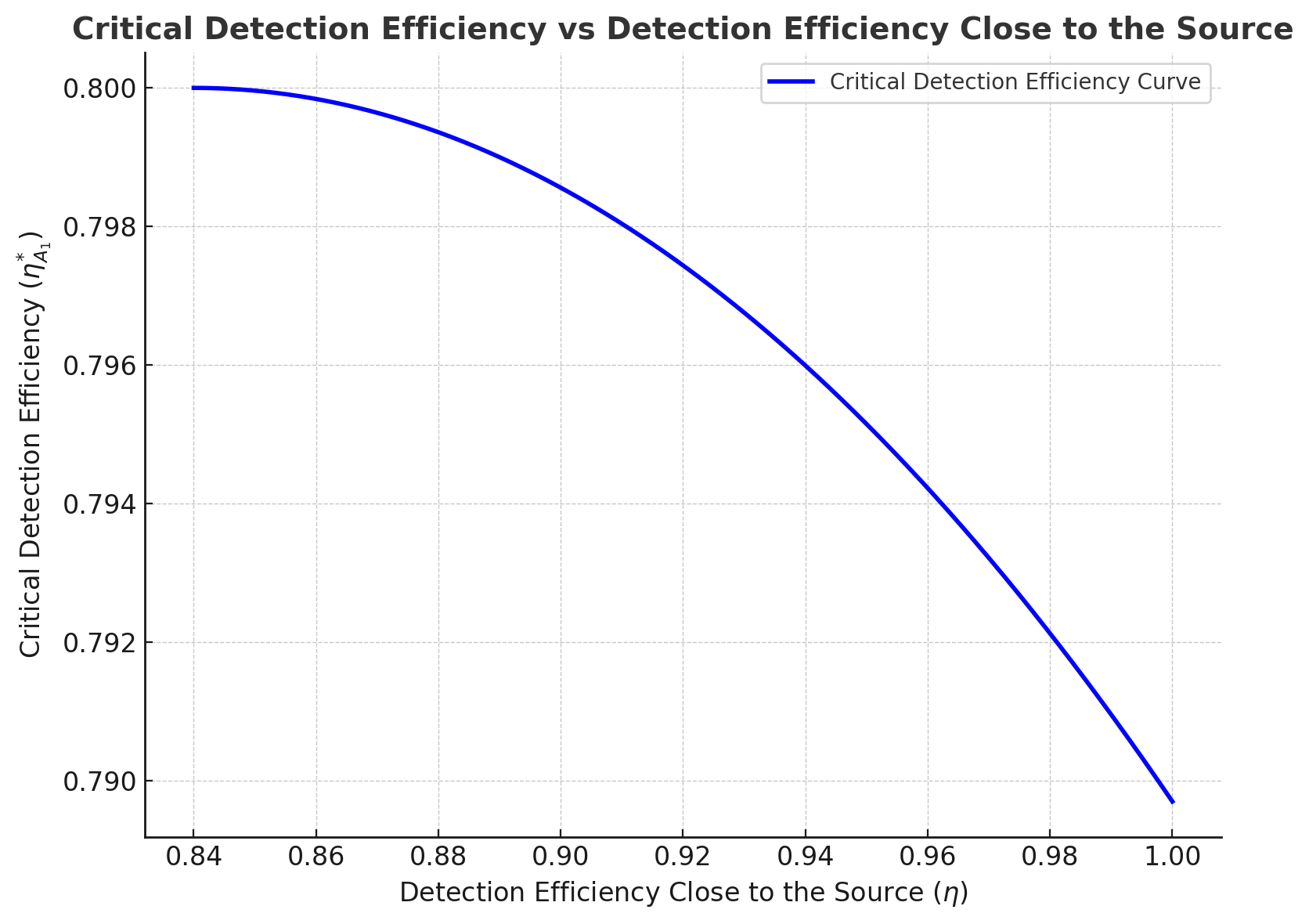}
    \caption{The detection efficiency $\eta_{A_1}^{*}$ versus detection efficiency $\eta$ obtained by equation \eqref{analyticaltradeoff} in symmetric case where $A_0$ and $B$ are equidistance from the source \cite{chaturvedi2024extending}. }
    \label{fig:CVP2024}
\end{figure*}
Consider now a strategy where the source, Bob's measurement device, the relay $R$, and the measurement device $A_0$ all behave as in CHSH expectations. Thus, any  $\beta_0\in[0,2\sqrt{2}]$ can be obtained by tuning $\eta_{A_0}$. For $i=0$, consider the case where at some point between $R$ and $A_1$ (before $A_1$) the second qubit is measured on the basis $z$, yielding
a binary result $\lambda$. This classical outcome is then transmitted to $A_1$
through some purely classical channel and upon receiving it, $A_1$ outputs it, irrespective of the input $x$. Therefore, we have $p(\lambda)=\frac{1}{2}$, $p(a|x,\lambda)=\delta_{a,\lambda}$, and $p(b|y,\lambda)=\textrm{Tr}(\rho_{\lambda}B_y)$ ($B_y$ is the corresponding observable for Bob's measurement) where $\rho_0=\ket{0}\bra{0}$ and $\rho_1=\ket{1}\bra{1}$ are the reduced states for Bob conditioned upon $\lambda$. Inserting in \eqref{CVPshortrange}, one gets the value $\beta_1=2$, which gives a violation of \eqref{analyticaltradeoff} while $A_1$ is fully classical. This shows that the form of non-classicality defined as \eqref{CVPshortrange} is weaker than $\mathcal{Q}_{SR}$ \ref{routedBIcorrelations}. \\
Lobo, Pauwels, and Pironio showed that there exist stronger versions of tradeoffs between long path (LP) and short path (SP) correlations in which, instead of CHSH, they considered the following Bell expression for LP \cite{Lobo2024} 
\begin{equation}
\label{JLP}
    J_{1}^{\theta}=\tan{\theta}\langle A_{01}B_{0}\rangle+\langle A_{11}B_0\rangle+ \langle A_{01}B_{1}\rangle-\tan{\theta}\langle A_{11}B_1\rangle,
\end{equation}
satisfying the following local and quantum bound 
\begin{equation}
    J_{1}^{\theta}\overset{L}{\leq}2\overset{Q}{\leq} 2/c_{\theta},
    \label{eq:JLP}
\end{equation}
where $c_{\theta}=\cos{\theta}$ ($\theta\in[0,\frac{\pi}{4}]$). For the case of observing the maximal CHSH violation for SP ($\beta_0=2\sqrt{2}$), SRQ correlations defined in \ref{routedBIcorrelations} satisfy the following bound 
\begin{equation}
 J_{1}^{\theta}\leq\frac{\sqrt{2}}{c_{\theta}}, 
 \label{SRQbound}
\end{equation}
The intuition of finding the bound \eqref{SRQbound} is that when Alice and Bob observe $\beta_0=2\sqrt{2}$, by self-testing the LP correlators are associated with the Pauli expectations i.e. $\mathrm{Tr}(PA_{x1})$ where $P\in\{\mathbb{I},\sigma_x,\sigma_z\}$, of the observables $A_{x1}$ and as a result, performing tomography of these observables is restricted to $ZX$ plane. In the case $\theta=0$ it was found that $\frac{1}{2}[\mathrm{Tr}(\sigma_x A_{11})+\mathrm{Tr}(\sigma_z A_{01})]\leq\sqrt{2}$ \cite{Pusey2015}. Therefore, the bound \eqref{SRQbound} can be obtained by rotation $R_{\theta}=\begin{bmatrix} \cos{\theta} & \sin{\theta} \\ -\sin{\theta} & \cos{\theta} \end{bmatrix}$ in the $ZX$ plane. \\
For any value $0\leq\theta<\frac{\pi}{4}$, the bound obtained \eqref{SRQbound} is strictly smaller than the local bound, which means that the SP CHSH test weakens the condition of witnessing a long-range quantum correlation based on $J_{1}^{\theta}$. This weakening is maximal in the case $\theta=0$, yielding
\begin{equation}
    J_{1}=\langle A_{1L}B_0\rangle+\langle A_{0L}B_1\rangle \leq 2,
\end{equation}
which coincides with the quantum bound in \eqref{eq:JLP}. Although the above inequality in the standard case (without any relay $R$) does not demonstrate a violation and may not appear to be a proper Bell inequality, in a routed Bell scenario equipped with a strategy that achieves $\beta_0 = 2\sqrt{2}$, the local bound is replaced by the SRQ bound \eqref{SRQbound}. 
\begin{equation}
     J_{1}=\langle A_{1L}B_0\rangle+\langle A_{0L}B_1\rangle \overset{L}{\leq} \sqrt{2} \overset{Q}{<} 2,
\end{equation} 
which is smaller than the quantum bound, making it a proper witness for long-range nonlocality ($J_1=J_{1}^{\theta=0}$). Therefore, since the SP CHSH is maximally violated, long-range quantum correlations can be witnessed whenever $J_{1}^{\theta}>\frac{\sqrt{2}}{c_{\theta}}$ instead of the more constraining criterion $J_{1}^{\theta}>2$.\\
The maximal SP CHSH violation $S_{0}$ is too strong and unrealistic in experimental settings, so deriving bounds based on non-maximal cases is important. For the case $\theta=0$ the following analytical tradeoff between $J_{1}$ and $\beta_{0}$ can be proved \cite{Lobo2024}
\begin{theorem}
    For any SRQ correlation defined in \ref{routedBIcorrelations}, the following inequality holds for $\beta_{0}\in[2,2\sqrt{2}]$
    \begin{equation}
        J_{1}\leq\frac{\beta_{0}+\sqrt{8-\beta_{0}^2}}{2}.
    \end{equation}
\end{theorem}
\noindent For other values of $\theta$, an SRQ bound of $J_{1}^{\theta}$ can be obtained numerically using the NPA hierarchy \ref{subsec:NPA}. 
\paragraph{Universal bounds on critical detection efficiency} Similarly to the case of the regular Bell experiment, lower bounds on critical detection efficiency can be obtained in routed Bell scenarios. The detection efficiencies in a routed Bell experiment can be denoted by the vector $\Vec{\eta}=(\eta_{A_0},\eta_{A_1},\eta_{B})$. In \cite{Lobo2024} it was proved that there exists an SRQ model if the following condition is satisfied.
\begin{equation}
    \eta_{A_1}\leq\frac{\eta_B(m_B-1)}{\eta_B(m_{A_1}m_B-1)-(m_{A_1}-1)},
    \label{SRQuniverslabound}
\end{equation}
which is independent of the number of measurement settings $m_{A_0}$. For the special case \(m_{A_0}=0\), this bound can be applied to the standard Bell experiment. Therefore, this bound places fundamental limits on the distance at which nonlocal correlations can be observed for both regular and routed Bell experiments. Since the right-hand side of \eqref{SRQuniverslabound} is always greater than \(1/m_{A_1}\), the detection efficiency of the remote device \(A_1\) cannot be lower than \(1/m_{A_1}\), even if the other detectors are perfect. \\
Although the bound in \eqref{SRQuniverslabound} applies to both standard and routed Bell experiments, this does not mean that the Routed Bell experiment cannot be more robust to photon losses than the regular ones. To see this, consider a protocol that the nearby detectors have the same efficiency $\eta_{A_0}=\eta_B=\eta_S$ and the remote detector has a lower one $\eta_{A_1}\leq \eta_{S}$.  Assuming the case that produces maximal CHSH violation ($\beta_0=2\sqrt{2}$ for $\eta_S=1$), and following anticommuting measurement settings for $A_1$ 
\begin{equation}
    A_{01}=s_{\theta}\sigma_x+c_{\theta}\sigma_z, \quad A_{11}=c_{\theta}\sigma_x-s_{\theta}\sigma_z.
\end{equation}
Considering the standard Bell experiment between $(A_1,B)$ by ignoring the relay $R$ in fig. \ref{fig:routedBI}. The violation of CHSH inequality is then implied 
\begin{equation}
    \beta_1=2\eta_{B}(c_{\theta}+s_{\theta})>2 \longrightarrow \eta_{B}>\frac{1}{c_{\theta}+s_\theta},
\end{equation}
for the routed Bell inequality using equation \eqref{SRQbound}, one can get the following bound for all values of $\theta$
\begin{equation}
    \eta_{B}\geq\frac{1}{\sqrt{2}}\approx 0.71,
\end{equation}
Considering that $\frac{1}{c_\theta+s_{\theta}}\geq0.71$ (the equality holds for $\theta=\frac{\pi}{4}$), for values $\theta\in(0,\frac{\pi}{4}]$ the routed Bell experiment can tolerate higher losses compared to standard Bell inequality.
The critical efficiency can be further reduced by the following LP inequality. 
\begin{align}
    J_{1}^{\theta_{+},\theta_{-}} &= (c_{\theta_{+}} + s_{\theta_{-}}s_{\theta_{+}})\langle A_{01}B_1\rangle 
    + (c_{\theta_{+}} - s_{\theta_{-}}s_{\theta_{+}})\langle A_{11}B_1\rangle \nonumber \\
    &\quad + (s_{\theta_{+}} - s_{\theta_{-}}c_{\theta_{+}})\langle A_{01}B_0\rangle 
    + (s_{\theta_{+}} + s_{\theta_{-}}c_{\theta_{+}})\langle A_{11}B_0\rangle \nonumber \\
    &\quad + c_{\theta_{-}}(\langle A_{0L}\rangle + \langle A_{1L}\rangle) \leq 2
\end{align}
where the SRQ bound $J_{1}^{\theta_{+},\theta_{-}}\leq 2$ is obtained assuming $\beta_0=2\sqrt{2}$ \cite{Lobo2024}. By considering the general projective measurements for $A_1$ of the form $A_{01}=s_{\theta_0}\sigma_x+c_{\theta_0}\sigma_z$ and $A_{11}=s_{\theta_1}\sigma_x+c_{\theta_1}\sigma_z$ the following lower bounds can be obtained for standard and routed Bell strategies 
 \begin{equation}
        \eta_{A_{1}}>
\begin{cases} 
\frac{1}{c_{\theta_{+}}(c_{\theta_{-}}+s_{\theta_{-}})} & \text{Standard Bell test}, \\
\frac{1}{1+c_{\theta_{-}}} & \text{Routed Bell test }.
\end{cases}
    \end{equation}
    As $\theta\rightarrow 0$, the critical efficiency in routed Bell scenario approaches $\frac{1}{2}$, which saturates the universal lowerbound \eqref{SRQuniverslabound}. There exist an explicit SRQ bound when $\eta_{A_L}=\frac{1}{1+c_{\theta_{-}}}$ implying that the above bound is tight \cite{Lobo2024}.  \\
 Sekatski et al. \cite{Sekatski2025} consider another test of nonlocality between $(A_1,B)$ where $A_1$ has a continuous number of settings $A_1\equiv\theta\in[0,2\pi)$ and they evaluate the following LP quantity 
 \begin{equation}
     \mathcal{C}=\int\frac{d\theta}{2\pi}\sum_{a_1,b=0,1}(-1)^{a+b}(c_{\theta}p(a_1,b|\theta,0,b)+s_{\theta}p(a_1,b|\theta,1,b))
 \end{equation}
    and satisfy the Bell inequality $\mathcal{C}\leq\frac{2\sqrt{2}}{\pi}\sin(\pi\frac{\mathcal{T}}{2})$ where $\mathcal{T}=\int\frac{d\theta}{2\pi}\sum_{a_1=0,1}p(a_1|\theta)$ is the average click probability of $A_1$ detector. Thanks to this Bell inequality, a strong routed Bell inequality can be proved as follows \cite{Sekatski2025} 
    \begin{theorem}
    \label{theorem:sekatski2025}
        All SRQ correlations satisfy the following tight routed Bell inequality
        \begin{equation}
            \mathcal{C}\leq\frac{2}{\pi}\sin\left(\frac{\pi}{2}\mathcal{T}\right) \begin{cases}
                \frac{\beta_0+\sqrt{8-\beta_0^2}}{2\sqrt{2}} \quad &\beta_0>2 \\
                \sqrt{2} \quad &\beta_0\leq 2
            \end{cases}
            \label{eq:sekatski2025}
        \end{equation}
    \end{theorem}
   \noindent The proof of this theorem is based on a steering scenario. Specifically, we first apply the SP-CHSH test to gather information about Bob's measurement settings, which is then used to determine whether the correlations observed in the LP are compatible with an SRQ model \eqref{SRQcorrelation}. For instance, when $ \beta_0 = 2\sqrt{2}$, one can certify that the shared state $\rho_{AB}$ is a two-qubit Bell state and  Bob's measurement settings correspond to the Pauli operators $\sigma_z$ and  $\sigma_x$. In this case, Eq.~\eqref{SRQcorrelation} becomes equivalent to a steering scenario in which Bob's states are remotely prepared by $A_1$ measurements, i.e., the assemblage $\rho_{A_1|\theta}$ admits a LHV model of the form $\rho_{A_1|\theta} = \sum_{\lambda} p(a_1|\theta,\lambda) \rho_{\lambda}$. If one can show that $\rho_{A_1|\theta}$ does not admit such a model, then it exhibits steering and therefore belongs to the set $\mathcal{M}_{qqq}$. Consequently, the violation of any steering inequality certifies that the full correlations are genuinely quantum. The complete proof of this theorem in \cite{Sekatski2025} uses such steering inequalities to derive Eq.~\eqref{eq:sekatski2025}.
\\
    If $\beta_0=2\sqrt{2}$, and $A_1$ perform all real projective qubit measurements of the form $c_{\theta}\sigma_z+s_{\theta}\sigma_x$, then  $\mathcal{C}=\mathcal{T}=\eta_{l}$ where $\eta_l$ is the transmission rate of the $LP$ link (see section \ref{sec:losses}). Putting these together, there is a violation of the inequality \eqref{eq:sekatski2025} that happens when 
    \begin{equation}
        \mathcal{C}=\eta_l>\frac{2}{\pi}\sin{\left(\frac{\pi\eta_l}{2}\right)},
\end{equation}
which can be satisfied for any $\eta_l>0$. As a result, \textit{the routed Bell inequality expressed in \ref{theorem:sekatski2025} can be violated for arbitrary transmission $\eta_l>0$, offering a dramatic advantage in terms of robustness to loss compared to standard Bell tests.} Even when the source and detectors are not ideal but sufficiently reliable, quantum nonlocal correlations can still be established for arbitrarily low transmission, provided that the number of measurement settings on $A_1$ ($n_{A_1}$) satisfies condition $n_{A_1} > \frac{1}{\eta_l}$. \\
\paragraph{Parallel repeated routed Bell experiments}  
The repeated routed Bell experiments introduced in \cite{Chaturvedi2025} work as follows: consider a Bell experiment where Alice and Bob both have $N$-bit strings $\boldsymbol{x}=(x_j)_{j=1}^N$ and $\boldsymbol{y}=(y_j)_{j=1}^N$ as inputs and produce $N$-bit strings $\boldsymbol{a}=(a_j)_{j=1}^N$ and $\boldsymbol{b}=(b_j)_{j=1}^N$. In the routed version, based on an additional input $i\in\{0,1\}$, Alice decides to direct the transmission from the source to her measurement device $A_i$. The source prepares a $2^N\times 2^N$ dimensional state $\ket{\Phi^{+}}^{\otimes N}$, and the devices close to the source, $(A_0,B)$ perform the following measurements 
\begin{equation}
A_{a|x} = \bigotimes_{j=1}^{N} A_{a_{j}|x_j}, \quad B_{b|y} = \bigotimes_{j=1}^{N} B_{b_{j}|y_j},
\end{equation}
where 
\begin{equation*}
A_{a_{j}|x_j} = \frac{1}{2} \left[ \bm{I}_{2} + (-1)^{a_j} \left( \delta_{x_j,0} \sigma_z + \delta_{x_j,1} \sigma_x \right) \right] \quad
B_{b_j|y_j} = \frac{1}{2} \left[ \bm{I}_{2} + \frac{(-1)^{b_j}}{\sqrt{2}} \left( \delta_{y_j,0} (\sigma_x + \sigma_z) + \delta_{y_j,1} (\sigma_x - \sigma_z) \right) \right]
\end{equation*}
This strategy achieves the maximal violation of the $N$-product CHSH inequality defined as 
\begin{equation}
\label{eq:CHSHNproduct}
C^{N}_{A_0B} = \frac{1}{2^{2N}} \sum_{\boldsymbol{x},\boldsymbol{y},\boldsymbol{a},\boldsymbol{b}} (\Pi_{j=1}^N\delta_{x_j  y_j = a_j \oplus b_j}) \, p(a, b \mid x, y, 0) = \alpha^{N},
\end{equation}
where $\alpha=\frac{1}{2}(1+\frac{1}{\sqrt{2}})$. A penalized version of this equation with the penalty parameter $q$ can be obtained as a generalization of the LP inequality in \eqref{JLP} (Also see equation (63) in \cite{Lobo2024}) 
\begin{equation}
\label{eq:q-penalized}
C^{N}_{A_1B}(q) = C^{N}_{A_1B} - q \, \frac{1}{2^{N}} \sum_{\boldsymbol{x},\,\boldsymbol{a} \in \{0,1\}^{N}} p_{A}(\boldsymbol{a}\mid \boldsymbol{x}, 1),
\end{equation}
where $q\in[0,\frac{1}{2^N}(1+\frac{1}{\sqrt{2}})^N]$. For a strategy with $A_1$ clicking with efficiency $\eta$ gives 
\begin{eqnarray}
    C^{N}_{A_1B}(q)=(\alpha^N-q)\eta.
\end{eqnarray}
It is shown in \cite{Chaturvedi2025} that when $A_1$ and $B$ witness the maximal quantum violation of \eqref{eq:CHSHNproduct} the maximal attainable value of the $q$-penalized $N$-product CHSH inequality \eqref{eq:q-penalized} with joint measurable measurements in $A_1$ satisfies the following upper bound
\begin{equation}
\label{CHSHNmax}
C^{N}_{A_1B(q)} \leq \frac{1}{2^{N}} \left( \alpha^{N} - q \right),
\end{equation}
for $q\in[\alpha^{N-1}\beta,\alpha^{N-1}]$ where $\beta=\frac{1}{8}(2+\sqrt{2}+\sqrt{4\sqrt{2}-2})$. Then the following result is immediately obtained 
\begin{theorem}
    The $q$-penalized $N$-product CHSH inequality \eqref{eq:q-penalized} can be violated by the strategy achieving \eqref{CHSHNmax} for $q\in[\alpha^{N-1}\beta,\alpha^{N-1}]$  whenever $\eta>\frac{1}{2^N}$. 
\end{theorem}
As a result, the threshold detection efficiency of the distant device, needed to certify non-jointly-measurable measurements, a prerequisite of secure DI-QKD, decreases exponentially with the number $N$ of parallel repetitions.

\paragraph{rDI-OKD: DI-QKD based on the routed Bell test}
\label{sec:rDIQKD}
\begin{figure}
\begin{tikzpicture}[scale=1, every node/.style={font=\small}]
\node[draw, circle, fill=gray!30, minimum size=30pt] (S) at (-3,0) {$\rho$};

\node[draw, rectangle, fill=blue!20, minimum width=1.5cm, minimum height=1cm, align=center] 
(B) at (6,0) {B};

\node[draw, rectangle, fill=red!20, minimum width=1.5cm, minimum height=1cm, align=center] 
(A0) at (-1,1.6) {T} ;

\node[draw, rectangle, fill=red!20, minimum width=1.5cm, minimum height=1cm, align=center] 
(A1) at (-5,0) { $A$} ;

\node[draw, circle, fill=green!20, minimum size=10pt] (R) at (-1.5,0) {R};
\draw[thick,->] (-5,1) -- (A1) node[midway, above, yshift=5pt] {$x$};
\draw[thick,->] (A1) -- (-5,-1) node[midway, above, yshift=-16pt] {$a$};
\draw[thick,->] (-1,2.6) -- (A0) node[midway, above, yshift=5pt] {$z$};
\draw[thick,->] (A0) -- (-1,0.6) node[midway, above, yshift=-16pt]{$c$};
\draw[thick,->] (-1.8,0.7) -- (R) node[midway, above, yshift=4pt] {$s$};
\draw[thick,->] (6,1) -- (B) node[midway, above, yshift=5pt] {$y$};
\draw[thick,->] (B) -- (6,-1) node[midway, above, yshift=-18pt] {$b$};
\draw[thick, ->] (S) -- (A1) node[midway, above] {};
\draw[thick, ->] (S) -- (R) node[midway, above] {};
\draw[thick, ->] (R) -- (B) node[midway, above] {};

\draw[thick, ->, dashed] (R) -- (A0) node[midway, above] {};
\end{tikzpicture}
\caption{rDI-QKD Protocol: Alice ($A$) and Bob ($B$) aim to establish a secret key over a long distance using the routed Bell setting. The relay ($R$) receives an input $s \in {0, 1}$ and, based on its value, sends the particle from the source to either the nearby device $T$ ($s = 0$) or the distant device $B$ ($s = 1$).  } 
\label{fig:rDIQKD}
\end{figure}
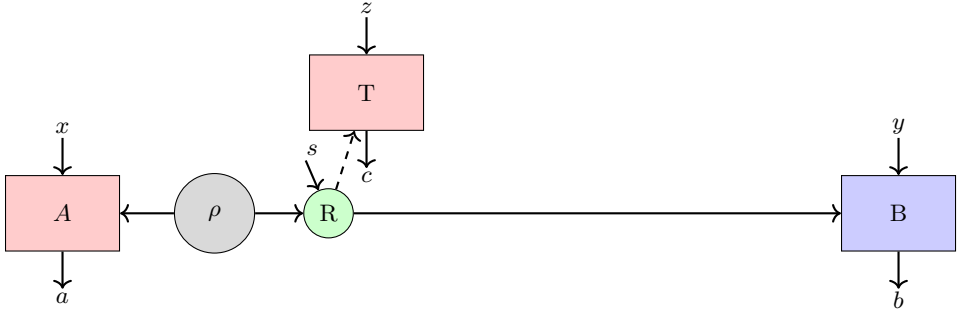
The routed Bell scenario introduced in \ref{sec:routedBI} exhibits features that enhance its applicability to DI-QKD and, in certain cases, provides advantages over standard Bell tests. For example, the BB84 correlations, which can be generated from a maximally entangled two-qubit state by performing measurements in the $\sigma_x$ and $\sigma_z$ bases, can be reproduced classically in a standard Bell setup which makes these correlations unsuitable for standard DI-QKD. However, their quantum nature can be demonstrated in a routed Bell scenario by performing random CHSH tests in the $(A,T)$ configuration.  
Consider the case in \ref{fig:rDIQKD} where $A$ and $T$ are two devices close to the source and the device $B$ is located far from the source and the purpose is to establish a key between $A$ and $B$. As in standard DI-QKD, all components including the relay $R$ are untrusted. The only assumption is that they obey certain no-signaling constraints preventing them from signaling arbitrarily to each other. \\
Unlike standard DI-QKD, some particles emitted from the source are routed to the testing device $T$ instead of $B$. The output of these cases always contributes to parameter estimation and cannot be used to generate a key.\\
The rDI-QKD protocol introduced in \cite{LeRoyDeloison2025} as follows: in each round of the protocol, Alice generates independent random variables $X_i\in\mathcal{X}$ and $s_i\in\{0,1\}$ and feed them to her device $A$ and the relay $R$, respectively. Based on the value of $s_i$, two cases can occur: 
\begin{enumerate}
    \item $s_i=0$:  Bob's quantum particle is routed to $T$, in this case, Alice generates a random variable $Z$ and feeds it to $T$. She records the output variable $A_i$ and publicly announces $S_i$.
    \item $s_i=1$ : Bob's quantum particle is routed to $B$ where Bob generates a random variable $Y_i\in\mathcal{Y}$, feeds it on his device, and records the output $B_i$.
\end{enumerate}
After all rounds, Alice and Bob check the date for which $s_i=0$ to see if they violate a routed Bell inequality. On the other hand, they agree on a subset of the rounds to generate the raw key. Finally, if the test of violating a routed Bell inequality is passed, they apply error correction and privacy amplification to extract the secret key. \\
An important condition must be satisfied for the security of rDI-QKD. This condition is captured in the following theorem \cite{LeRoyDeloison2025}.
\begin{theorem}
    Long-range quantum correlation as defined in \ref{routedBIcorrelations} is necessary for the security of rDI-QKD.
\end{theorem}
\begin{proof}
    The proof can be done by contradiction. Let us consider that Alice and Bob only generate the SRQ correlation of the form \eqref{SRQcorrelation}, then Eve can perform the parent POVM $N=\{N_{\lambda}\}$ on the public channel between $R$ and $B$. Since Eve can keep a copy of the classical outcome $\lambda$, the correlation between $A$ and $B$ was factorized when conditioned on Eve's information, i.e. $p(a,b|x,y,\lambda)=p(a|x,\lambda)p(b|y,\lambda)$, implying that no secure key can be extracted between $A$ and $B$ when only an SRQ correlation exists in the protocol.
\end{proof}
\paragraph{General behavior of Eavesdropper in rDI-QKD} Based on the location of the relay $R$, the device $T$, and the source, two situations can be considered: first, consider them as being outside Alice's and Bob's labs and in full control of the eavesdropper, and the second option is to assume that the relay $R$ and the device $T$ are all situated in Alice's lab by imposing the condition that they cannot arbitrarily communicate their private input to Alice's device. Clearly, security in the first case also implies security in the second case, since in the first case, the eavesdropper has more power. The protocol in \cite{LeRoyDeloison2025} considered the first case in which Eve has the measurement device $T$ that performs a measurement $T_z=\{T_{c|z}\}$ that acts jointly on the subsystems $B$ and $E$ when $s=0$, as this is the most general thing she can do to simulate the honest correlation between $A$ and $T$. Considering this behavior for Eve, the secret key rate can be calculated using the method explained in section \ref{subsec:rDIQKDsecurity}. 

To conclude the list of the CHSH-like protocol, in the following, we update the list of critical detection efficiency from Ref~\cite{Xu2022}.

\begin{table}[h]
\centering
\begin{tabular}{clc}
\toprule
\textbf{Protocols} & \textbf{Method} & \textbf{$\eta^*$} \\
\midrule
CHSH$_{p_{g}}$ & Standard analysis \cite{pabgs09andSangouard.4} & 92.4\% \\ 
CHSH$_{p_{g}}$& Efficient post-processing \cite{Ma2012a} & 90.9\% \\ 
CHSH$_{p_{g}}$ & Advantage distillation \cite{tan2020advantage} & 89.1\% \\ 
CHSH$_{\text{IMD}}$ & Iterated mean divergence \cite{Brown2021} & 84.5\% \\ 
CHSH$_T$ & Noisy preprocessing \cite{Ho2020} & 83.2\% \\ 
CH-SH & Asymmetric inequality \cite{Woodhead2021,Sekatski2021} & 82.6\% \\ 
CHSH$_T$ & Quasi-relative entropy \cite{Brown2024} & 80.5\% \\ 
CHSH$_{p\nu_p}$ & Random Postselection \cite{Xu2022} & 68.5\% \\ 
rDI-QKD & Routed BI \cite{chaturvedi2024extending} & 50\%\\
\bottomrule
\end{tabular}
\caption{A comparison of critical detection efficiency $\eta^*$ among different methods for DI-QKD.}
\label{tab:diqkd_threshold}
\end{table}
\paragraph{DI-QKD Protocols Beyond Binary Inputs and Outputs}
Up to now, we have discussed DI-QKD protocols with two inputs and two outputs. It is worth mentioning some results beyond this scenario: Ref~\cite{GonzalesUreta2021} for  BI $I_{4422}^4$ (also in \cite{Brunner2008}); Ref~\cite{Cabello2001} for BI $I_{234}$  with $T=\mathrm{id}$. Both of these protocols perform better than the CHSH$_\chi$ protocol \cite{pabgs09andSangouard.4}, but not as good as the CHSH$_{p_{\mathcal{V}_p}}$ protocol with $T=\mathrm{fp}$ \cite{Xu2022}. This motivates the research to characterize any Bell scenario~\cite{acin2002bell,collins2002bell,zukowski2002bell,Karczewski2022,ho2022entanglement}. Given a Bell scenario, find the largest gap $\beta_\QC-\beta_\LC$ is promising, but does not guarantee a good QKD performance as many aspects needs to be further addressed. Among them, \textit{memory loophole}.

\subsection{DI-QKD memory loopholes}\label{sec:memoryloopholeDIQKD}
We discussed the theoretical security proofs of (fully)-DI-QKD. However, to satisfy the definition of \textit{universal composable security} the theory demands many different devices at each run to close \textit{memory attack loophole}, i.e. a leakage of information due to correlated subsequent outputs \cite{barrett2013memory,Masanes2014,Jain2022}. This is experimentally critical.
Suppose that Alice and Bob have only one device each. Then a memory attack is possible:
the protocol is run on day $i$, generating $K^i_A=K^i_B$
, while informing Eve on day $i=1$ of the hash functions used by Alice for postprocessing. 

Eve can instruct the devices to proceed as follows. 
On day 2, Eve modulates $\rho_{AB}^c$ where $c$ is a classical ancilla to carry new instructions to the device in Alice’s lab as discussed in \textit{collective attacks in QT}. These instructions tell the device the hash functions
used on day 1, allowing Eve to compute $K^1_A=K^1_B$. 
The classical ancillae also instruct the device to bias the protocol for randomly selected inputs by producing specified bits from this secret key as outputs for example by programming the device to announce one key bit of day 1 as an output of some specified input. 
If any of these selected outputs are among those announced in the parameter estimation step, Eve learns the corresponding bits of day 1’s secret key. 
But, Alice or Bob might abort the protocol any day, thus, by waiting long enough, Eve can program them to communicate some or all information about their day 1 key to obtain \(N\) secret bits from day 1 and then program it to abort on the day $N+2$ since from this day Alice and Bob do not have any secret key available. This type of attack is called an \textit{abort attack} and cannot be detected until it is too late.
To defend against these attacks \textit{(i)} all quantum data and public communication come only from Alice, even if, Eve can still program Alice's device to leak $K_A^1$ or $K_A^i$, $i=1,2,\dots$ via the abort attack; \textit{(ii)} encrypt $c_\QC$ 
with some initial preshared seed randomness, even if fails if an abort attack 
involves communication with multiple users who may not be trustworthy; \textit{(iii)} additional independent measurement devices close memory attack loophole, but leave open the imposter attack \cite{barrett2013memory}. 

\section{Other device-independent protocols}
The device-independent framework can be extended to other cryptographic primitives that rely on similar physical assumptions and security techniques. In this section, we briefly discuss three such protocols: device-independent quantum random-number generation, quantum secret sharing, and quantum secure direct communication. In the same fashion as DI-QKD, these protocols exploit BI violations to achieve secure operation with untrusted devices, and their security analyses often rely on comparable tools such as entropy accumulation and randomness extraction.

\subsection{DI quantum random number generation (DIQRNG)} \label{subsec:DIQRNG}
The fact that the decay of the isotope that kills Schroedinger's cat is an intrinsically random event means that the cat is in the superposition state $\ket{\mathrm{cat}}=\alpha \ket{\mathrm{alive}}+\beta \ket{\mathrm{dead}}$. This is in contrast to   simple ignorance of the observer, who would describe the cat in the statistical mixture $\rho_\mathrm{cat}=\alpha \rho_\mathrm{alive}+ \beta \rho_\mathrm{dead}$.  Both examples give a probabilistic prediction concerning the status of the cat.  However, it is only in the former case that true unpredictability is present.   Associating random bits to the clicks of a Geiger counter, already in the 1940s, is based on this idea \cite{gaiger}. However, in practice, the ontological quantum randomness, appearing in the superposition state $\ket{\mathrm{cat}}$, is difficult to distinguish from the randomness due to the ``classical ignorance" present in the mixture $\rho_\mathrm{cat}$  \cite{PhysRevLett.48.291}. 

Realism in Bell’s hypothesis (Sec. \ref{def:realism}) rules out any intrinsic randomness. Therefore, a Bell‐inequality violation not only produces randomness but also— in a device‐independent way—certifies that this randomness cannot be mimicked by a classical mixture \cite{colbeck2009quantum,Colbeck_2011}.
More quantitatively, for iid trials the \emph{conditional smooth min--entropy}
\begin{equation}\label{eq:min-entropy}
H_{\min}^{\varepsilon}(AB|XYE)\ge n\,f(\beta_\QC)\equiv-n\log_{2}\left(\tfrac12+\tfrac12\sqrt{\tfrac{\beta_\QC^{2}}{4}-1},\right),
\end{equation}
where $E$ denotes the adversary's quantum side information, $n$ is the number of rounds is tight for CHSH\cite{Pironio2010,colbeckrenner2012free}.  Equation~\eqref{eq:min-entropy} lower--bounds the number of uniform bits that can be extracted with a strong randomness extractor.
The Bell tests use initial classical randomness for the choice of the settings, DIQRNGs {\it expand} randomness rather than {\it produce} it.
Since then, the terms \textit{randomness expansion} and \textit{randomness generation} have been used in the DIQRNG literature interchangeably. 
A DIQRNG consumes a short, uniformly--random seed $R_{\mathrm{in}}$ to choose the settings $(X,Y)$, and---upon observing a BI violation---outputs a longer string $R_{\mathrm{out}}$; the expansion factor is $|R_{\mathrm{out}}|/|R_{\mathrm{in}}|$.  Unbounded expansion is possible using two interleaved devices that alternately supply fresh settings for each other~\cite{Miller2016}.  Finite--size analyses must account for statistical fluctuations (via, e.g., Azuma--Hoeffding inequalities) and possible memory attacks that correlate different rounds.
The theoretical and experimental challenges follows:
\begin{itemize}
\item \textbf{Detection \& locality loopholes.}  A DIQRNG must simultaneously close the detection and locality loopholes; otherwise an adversary can fake a violation and inject deterministic bits.  This demands overall detection efficiencies $\eta\gtrsim 90\%$ and space--like separation of the measurement events.
\item \textbf{Finite--key effects.}  High Bell values are necessary but not sufficient; one must also accumulate a sufficiently large dataset before the entropy rate in~\eqref{eq:min-entropy} becomes positive.
\item \textbf{Rate versus security trade--off.}  Short gate times and high--brightness entangled sources improve rates but exacerbate multi--pair noise, which lowers $S$.
\item \textbf{Contrast with DI-QKD.}  In DI-QKD, the raw outcomes $A$ and $B$ are further processed into a shared secret key.  Besides randomness, one must bound the \emph{Holevo information} $\chi(E:K)$ leaked to an eavesdropper, perform information reconciliation, and authenticate classical communication, which reduces the asymptotic key rate per run to $R_{\mathrm{key}}\le H_{\min}(AB|E)-H_{\max}(A|B)$.  Moreover, DI-QKD typically requires distributing entanglement over \emph{kilometer} scales, whereas DIQRNG may keep both measurement stations on a single optical table.  Conversely, DI-QKD is tolerant to \emph{loss} (lost events are simply discarded) but highly sensitive to \emph{phase noise}, while DIQRNG demands near--unit detection efficiency but can operate with smaller spacings.
\end{itemize}
The first experimental demonstration of randomness expansion was presented in \cite{pironio2010random}. The experiment there was based on the violation of CHSH in a setup which used heralded entanglement generation through entanglement swapping. Later, experiments also based on CHSH but involving direct production of entanglement were performed \cite{Liu2018b,Bierhorst2018,Li2021,Liu2021b,Shalm2021}.
Table~\ref{tab:diqrng-milestones} compiles the most influential experimental demonstrations to date.

\begin{table}[ht]
\centering
\caption{Key milestones in DIQRNG and certified randomness.  Only experiments that reported a full finite--size analysis and extractor implementation are listed.}
\label{tab:diqrng-milestones}
\begin{tabular}{lcc}
\toprule
Platform & $\beta_\QC$ value & Net random bits \\
\midrule
Photonic (heralded) ``10 ~\cite{Pironio2010}& $2.41\pm0.10$ & $42$ bits \\
Photonic time--bin ``18 \cite{Liu2018b} & $2.63\pm0.04$ & $6\times10^{6}$ \\
Photonic (NV centers) ``18 \cite{Bierhorst2018}& $2.34\pm0.18$ & $1.02\times10^{7}$ \\
Photonic (SNSPD, spot--checking) ``21 \cite{Shalm2021} & $2.75\pm0.02$ & $1.18\times10^{9}$ \\
Superconducting qubits ``23 \cite{Storz2023} & $2.0747\pm0.0033$ & $\sim 10^{5}$ \\
Photonic two--photon interference ``23 \cite{mongia2024investigating}& $2.62\pm0.03$ & $\sim10^{6}$  \\
Trapped--ion processor (cloud) ``25 \cite{Liu2025} & RCS,* & $7.1\times10^{4}$ \\
\bottomrule
\end{tabular}
\flushleft Random circuit sampling certificate rather than CHSH.
\end{table}
Device--independent protocols continue to push both theoretical and technological boundaries~\cite{Amer2025}.  Near--term goals include (i) giga--bit--per--second DIQRNG via multiplexed sources and feed--forward switching, (ii) metropolitan--scale DI-QKD combining satellite links with ground--based quantum repeaters, and (iii) hybrid protocols that merge computation--based certification (as in random circuit sampling) with Bell--test randomness to relax detection efficiency requirements.

\blk
\subsection{ Device-independent quantum secret sharing} 
Quantum secret sharing (QSS), introduced first in \cite{hillery1999quantum}, is a multiparty communication protocol where a dealer tries to share a secret among a number of participants in such a way that a set of participants is needed to recover the secret. Each QSS includes three stages: Data gathering, Parameter estimation, and Key extraction.\\ 
The first device-independent quantum secret sharing (DI-QSS) protocol was introduced in \cite{Roy2019}, based on an $N$-partite $d$-dimensional XOR game; however, it was limited to cases where $d$ is an even number. Almost simultaneously, it was shown in \cite{Moreno2020} that in a three-party protocol, a maximal violation of Svetlichny's inequality \cite{Svetlichny1987} guarantees that an untrusted receiver cannot obtain any information about the secret unless they collaborate with the other parties. In cases of non-maximal violation, the adversary's access to the secret can still be limited. The advantage of using Svetlichny's inequality lies in its ability to certify a stronger form of nonlocality, one that even an untrusted party (as in three-party QSS) cannot simulate. The security analysis robust individual and collective attack is respectively in \cite{Moreno2020}  and \cite{Zhang2024}. Suppose that three separated parties, Alice, Bob, and Charlie,
measure their photons with their binary measurement settings $x_i$,$y_j$, and $z_k$, respectively, where $i,k\in\{0,1\}$, and $j\in\{0,1,2\}$ and obtain the outcomes $a_i$,$b_j$, and $c_k$, (all with binary outcomes $-1$ and $1$); therefore, the Svetlichny  polynomial $S_{ABC}$ can be written as the combination of the CHSH polynomials 
\begin{equation}
\label{Svetlichny}
    S_{ABC}^{\text{Svetlichny}}=\langle S_{AB}^{j=2,3}c_2\rangle + \langle S^{{\prime}}_{AB}c_1\rangle,
\end{equation}
where  $S_{AB}^{j=2,3}$ is the CHSH value between Alice and Bob with $y_j=\{2,3\}$ measurement settings, and $S^{\prime_{AB}}$  is obtained from  $S_{AB}^{j=2,3}$ by mapping $b_2\rightarrow b_3$ and $b_3\rightarrow -b_2$. As a result of \eqref{Svetlichny}, the violation of the Svetlichny inequality for three users is equal to the violation of the CHSH inequality for Alice and Bob. Consequently, the security proof of DI-QSS is similar to the one for DI-QKD given in section \ref{DI-QKD-collective}. By assuming that each user successfully detects the transmitted photon with a global detection efficiency $\eta$ and considering the white noise model with the fidelity $F$ for photon loss and decoherence ($\rho=F\ket{GHZ}\bra{GHZ}+\frac{1-F}{8}\bm{I}$), the following lower bound can be obtained \cite{Zhang2024}
\begin{equation}
    r\geq 1-h\left(\frac{\sqrt{2F^2\eta^6-1}}{2}+\frac{1}{2}\right)-h\left(1-\frac{1}{2}\eta^3-\frac{1}{2}\eta^3F\right).
\end{equation}
A positive value of $r$ requires high global detection efficiency and fidelity. For example, with no decoherence ($F = 1$), the detection efficiency threshold is $96.32\%$; at $F = 0.95$, it increases to $97.57\%$. The minimum tolerable fidelity is $85.1\%$. Below either threshold, no key can be generated.\\
The high detection efficiency thresholds in DI-QSS arise from the need to distribute multiphoton entanglement among multiple parties through noisy quantum channels. Photon loss in each channel can destroy the entanglement, making DI-QSS especially vulnerable to such losses. Consequently, it is crucial to develop methods that reduce the global detection efficiency requirement and improve noise tolerance for practical implementations. Similar techniques such as noise preprocessing and random postselection used for DI-QKD have been shown to enhance DI-QSS performance~\cite{Zhang2024}. Building on this, Ref.~\cite{Zhang2025} combined these techniques with a random key generation basis (see Section~\ref{subsubsec:randomkeybais}) and demonstrated that the global detection efficiency can be reduced to $93.41\%$.

\subsection{Device-independent quantum secure direct communication}
Unlike QKD, quantum secure direct communication (QSDC) enables the sender to transmit secret messages directly to the receiver without first establishing a shared key. The QSDC protocol was introduced in Ref.~\cite{Long2002}, leveraging quantum entanglement and block transmission techniques. In QSDC protocols, photons are transmitted in two rounds. Therefore, to achieve device-independent QSDC (DI-QSDC), a security check based on the BI violation is performed after each transmission round. A more robust security against collective attacks, instead, is in Ref.~\cite{Zhou2020}. The protocol works as follows: Alice has three measurement settings $x\in\{0,1,2\}$ and Bob has two measurement settings $y\in\{0,1\}$ same as CHSH$_\chi$ protocol. Alice prepares $n$ EPR pairs, divides them into two sequences and sends one to Bob. They perform the first security check (a CHSH test) by measuring a subset of EPR pairs. If the security check passes, Alice selects a subset of photons as checking photons and does not perform any operation on them. For the rest of them, she uses four operations, $\bm{I}$, $\sigma_x$, $\sigma_y$, and $i\sigma_z$ to encode the message and then sends all of her photons to Bob and reveals the position of each message photon and each checking photon through a classical channel. Bob performs the second security check on checking photon pairs on his own. If the second security check passes, Bob can read Alice's message by performing a Bell measurement on the message pairs. \\
Since in DI-QSDC there is no key generation between two parties, to study the security of the DI-QSDC, the communication efficiency $E_{c}$ is defined \cite{Zhou2020} as the amount of transmitted correct secure information qubits divided by the total amount of information qubits. Similar as the DI-QKD, for collective attacks it can be lower-bounded by the Devetak-Winter rate \eqref{DevetakWinter}. Since both security checking processes of the DI-QSDC are the same as that of DI-QKD, the following lower bound can be obtained 
\begin{equation}
    E_{c}\geq 1-h(Q)-h\left(\frac{1+\sqrt{(\frac{\beta_1}{2})^2-1}}{2}\right)-h\left(\frac{1+\sqrt{(\frac{\beta_2}{2})^2-1}}{2}\right),
\end{equation}
where $\beta_1$ and $\beta_2$ are the CHSH values for the first and second security checks, respectively. Since DI-QSDC requires two rounds of photon transmission, the influence of environmental noise and eavesdropping on DI-QSDC is more serious than DI-QKD.
To compensate for these negative effects, the methods entanglement-based noiseless linear amplification and entanglement purification protocols can be applied to DI-QSDC. In a modified protocol with both of these methods, Alice and Bob can have a near-perfect entanglement channel where the secure communication is effectively extended and Eve cannot steal any secure information without being detected. \\ 
Other DI-QSDC protocols have been proposed since then. A one-step DI-QSDC protocol was proposed in \cite{Zhou2022a}, wherein Bob prepares a large number of polarization–spatial mode hyperentangled photon pairs, transmit one photon from each pare to Alice in a single round, performs CHSH tests on both degrees of freedom, and Alice encodes her message in the polarization mode. Combining it with the hyperentanglement heralded amplification and the hyperentanglement purification can eliminate the message loss and the message error caused by the channel noise. A DI-QSDC protocol based on single-photon sources and the heralded Bell state measurement was introduced in \cite{Zhou2023a}, where Alice and Bob each have two photon sources send one photon to a third party who performs a  Bell state measurement, and if the measurement is successful, then they have an entanglement channel that can be used to send the message. The DI-QSDC protocol with hyper-encoding technique was studied in \cite{Zeng2024}, and non-Markovian dynamics have been shown to enhance protocol security compared to their Markovian counterparts significantly \cite{Roy2024}.
\section{Mathematical techniques for advanced security proofs}\label{sec:chap4}
\noindent In the previous section, we examined a secure protocol for generating a secret key $K$, along with its associated security proof, treating it as an independent \textit{module}. We analyzed lower bounds on Eve's uncertainty or randomness under various types of attacks—individual, collective, and coherent—to establish robust security guarantees. While Section~\ref{sec:DI-QKD} demonstrated the feasibility of such protocols, our current focus is on optimizing the key rate by considering correlations across protocol rounds. Broadly, there are two categories of protocols: those based on the (Generalized) Entropy Accumulation Theorem and Quantum Probability Estimation, as well as complementarity-based methods. These are typically used in \textit{sequential implementations}, where each round of the protocol depends on the outcomes of previous rounds. This contrasts with \textit{parallel approaches}, which require distinct analytical techniques. In this section, we explore advanced mathematical methods for constructing such security proofs in real-world settings.

\subsection{Entropy accumulation theorem (EAT)} \label{sec:EAT}
Let us begin with some key definitions required for a formal description of EAT.
\begin{figure}[h]
    \centering
    \begin{subfigure}{.49\textwidth}
        \centering
        \resizebox{\textwidth}{!}{
\begin{tikzpicture}[
  state/.style={rectangle, draw=purple, fill=purple!20, minimum size=1.2cm, text centered},
  arrow/.style={-latex, thick},
  curvedarrow/.style={-latex, thick},
  label/.style={font=\LARGE},
]

\node[label] (initial) at (-1, 0) {$\rho_{\text{in}}$};
\node[state] (M1) at (2,-0.7) {\LARGE{$\mathcal{M}_1$}};
\node[state] (M2) at (5, -0.7) {\LARGE{$\mathcal{M}_2$}};
\node[state] (Mn) at (9, -0.7) {\LARGE{$\mathcal{M}_n$}};
\node[label] (final) at (12, 0) {$\rho_{\text{out}}$};

\node[label] (R0) at (0.8, -0.3) {$R_0$};
\node[label] (R1) at (3.5, -0.3) {$R_1$};
\node[label] (R2) at (6., -0.3) {$R_2$};
\node[label] (p2) at (6.5, -0.7) {};
\node[label] (p3) at (7.3, -0.7) {};
\node[label] (Rn1) at (7.6, -0.3) {$R_{n-1}$};
\node[label] (Rn) at (10, -0.3) {$R_n$};

\node[label] (O1) at (1.55, -2) {$O_1$};
\node[label] (S1) at (2.45, -2) {$S_1$};
\node[label] (O2) at (4.65, -2) {$O_2$};
\node[label] (S2) at (5.55, -2) {$S_2$};

\node[label] (On) at (8.65, -2) {$O_n$};
\node[label] (Sn) at (9.55, -2) {$S_n$};

\draw[arrow] (M1) -- (M2);
\draw[arrow] (M2) -- (p2);
\draw[dashed,line] (p2) -- (p3);
\draw[arrow] (p3) -- (Mn);

\draw[curvedarrow, rounded corners=5pt] (initial) -- (0, 1) -- (11, 1) -- (final);
\draw[curvedarrow] (initial) -- (0,-0.7) -- (M1);
\draw[curvedarrow] (Mn) -- (11,-0.7) -- (final);

\draw[arrow] (M1) -- (O1);
\draw[arrow] (M1) -- (S1);

\draw[arrow] (M2) -- (O2);
\draw[arrow] (M2) -- (S2);

\draw[arrow] (Mn) -- (On);
\draw[arrow] (Mn) -- (Sn);

\node[label] at (5.5, 1.3) {$E$};

\end{tikzpicture}
}
        \caption{\textit{EAT}}
        \label{fig:EAT}
    \end{subfigure}
    \hfill
    \begin{subfigure}{.49\textwidth}
        \centering
        \resizebox{\textwidth}{!}{
\begin{tikzpicture}[
  state/.style={rectangle, draw=purple, fill=purple!10, minimum size=2cm, text centered},
  arrow/.style={-latex, thick},
  curvedarrow/.style={-latex, thick},
  label/.style={font=\LARGE},
]

\node[label] (initial) at (-1, -0.6) {$\rho_{\text{in}}$};
\node[state] (M1) at (2,-0.7) {\LARGE{$\mathcal{M}_1$}};
\node[state] (M2) at (5, -0.7) {\LARGE{$\mathcal{M}_2$}};
\node[state] (Mn) at (9, -0.7) {\LARGE{$\mathcal{M}_n$}};
\node[label] (final) at (12, -0.6) {$\rho_{\text{out}}$};

\node[label] (R0) at (0.4, -1.1) {$R_0$};
\node[label] (R1) at (3.5, -1.1) {$R_1$};
\node[label] (R2) at (6.3, -1.1) {$R_2$};
\node[label] (p2) at (6.5, -0.7) {};
\node[label] (p3) at (7.3, -0.7) {};
\node[label] (Rn1) at (7.3, -1.8) {$R_{n-1}$};
\node[label] (Rn) at (10.4, -1.) {$R_n$};
\node[label] (E0) at (0.4, 0.6) {$E_0$};
\node[label] (E1) at (3.5, 0.6) {$E_1$};
\node[label] (E2) at (6.3, 0.6) {$E_2$};
\node[label] (E2) at (6.5, 0.6) {};
\node[label] (E3) at (7.3, 0.6) {};
\node[label] (En1) at (7.6, 0.6) {$E_{n-1}$};
\node[label] (En) at (10.4, 0.6) {$E_n$};

\node[label] (O1) at (2, -3) {$O_1$};
\node[label] (O2) at (5, -3) {$O_2$};
\node[label] (On) at (9, -3) {$O_n$};
\draw[arrow] (3, -1.4) -- (4,-1.4);
\draw[arrow] (3, 0.07) -- (4,0.07);
\draw[arrow] (6, -1.4) -- (6.5,-1.4);
\draw[arrow] (6, 0.07) -- (6.5,0.07);
\draw[arrow] (7.5, -1.4) -- (8,-1.4);
\draw[arrow] (7.5, 0.07) -- (8,0.07);
\draw[dashed,line] (6.6, -1.4) -- (7.9,-1.4);
\draw[dashed,line] (6.6, 0.07) -- (7.9,0.07);

\draw[curvedarrow] (initial) -- (0, -1.4) -- (1,-1.4);
\draw[curvedarrow] (initial) -- (0, 0.07) -- (1,0.07);
\draw[curvedarrow] (10,0.07) -- (10.5, 0.07) -- (final);
\draw[curvedarrow] (10,-1.4) -- (10.5, -1.4) -- (final);

\draw[arrow] (M1) -- (O1);
\draw[arrow] (M2) -- (O2);
\draw[arrow] (Mn) -- (On);

\end{tikzpicture}
}
        \caption{\textit{GEAT
}}
        \label{fig:GEAT}
    \end{subfigure}
    \caption{\textit{(Generalized) Entropy accumulation theorem} -- \ref{fig:EAT} sequential processes $\bigcirc_{i=1}^{n}\mathcal{M}_{i}\otimes \mathrm{id}$ with $\mathcal{M}_i:R_{i-1}\mapsto R_iO_iS_iC_i$, and its generalization $\bigcirc_{i=1}^{n}\mathcal{M}_{i}$ with $\mathcal{M}_i:R_{i-1}E_{i-1}\mapsto R_iE_iO_iS_iC_i$ in
    \ref{fig:(G)EAT}.
    } 
    \label{fig:(G)EAT}
\end{figure}
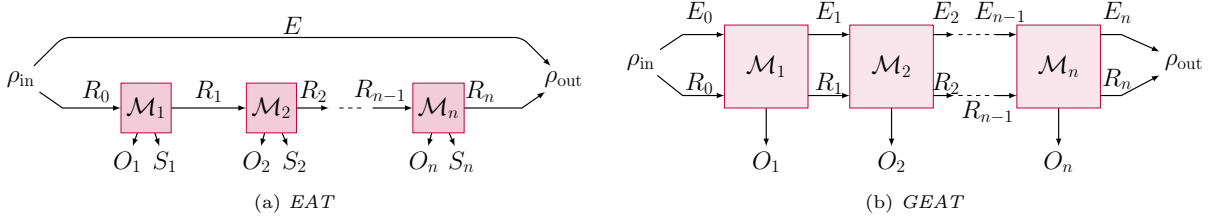
\begin{definition}(Sequential process)\label{def:seq}
We call \textit{sequential process} the composition map $\mathcal{M}=\bigcirc_{i=1}^{n}\mathcal{M}_{i}$, where $\mathcal{M}_i:R_{i-1}\mapsto R_iO_iS_iC_i$ are CPTP maps that transform the state on $R_{i-1}$ (quantum registers) into $R_i$, with output quantum system $O_i$ (readout observed outcome), $S_i$ (side information), $C_i$ (classical check), and where the composition $\bigcirc_{i=1}^{n} \mathcal{M}_{i}$ is defined as $\mathcal{M}_n \circ \cdots \circ \mathcal{M}_1$, acting sequentially on the quantum registers.
\end{definition}
As in Fig. \ref{fig:EAT}, in the $i$--th round, the internal state in the input memory $R_{i-1}$ is updated to the output memory $R_i$ ensuring that the state at the step $i$ depends on the previous one (non-iid). At each $i$, the quantum output system in the register $O_i$ \textit{accumulates} the entropy of Eve. The leaked information (about the measurements or outcomes) is in the partial state on the support of $S_i$ and the environment "controlled" by Eve in the Hilbert space $E$. The conditional entropy $H(O_1^n|S_1^nE)$  quantifies how much uncertainty remains about the update post-measurement state outputs 
$O_1^n$ after Eve learns the side information $S_1^n$ and external system $E$. The quantity $X^n_1$ refers to the whole process of $n$ rounds where each round $i$ is isomorphic the $i=1$.
The protocol is considered secure if the entropy in Eq. \eqref{eq:minHep} is higher than a lower bound from parameter estimation that is computed by other output $c_i^{(j)}$ or simply $c_i\in C_i$ stored in a classical register $C_i=\{c_i^{(j)}\}_j$ with probability distribution $p_i^{(j)}=p(c_i^{(j)})$ such that $ \sum_jp_i^{(j)}=1,\,p_i^{(j)}\ge0$. This is derived from the system $\rho_{O_iS_i}$ and used for BI violation.
\begin{definition}(Markovianity)
    Given $\mathcal{M}$ a sequential process from \ref{def:seq}. It is \textit{Markovian} iff 
    $O_{i-1}\leftrightarrow S_{i-1}E\leftrightarrow S_i$, i.e. the mutual information $I(O_{i-1}:S_i|S_{i-1}E)=0$
\end{definition}    
\begin{definition}(trade-off functions)
    The following quantum state set 
    \begin{equation}
        \Sigma_i(p_{i_j})=\{\rho_{R_iO_iS_iC_iE}=\mathcal{M}_i(\rho_{R_{i-1}}E)|\rho_{C_i}=\rho_{c_{i_j}}=p_{i_j}\in C_i\}
    \end{equation}
    with $\rho_{C_i}$ defines in the classical register $C_i$ the probability distribution with weight $p_{i_j}=\bra{c_{i_j}}\rho_{C_i}\ket{c_{i_j}}\ge0$ and $\sum_jp_{i_j}=1$  on the possible classical output $c_{i_j}$ in the $i$--th round. 
    Given $p_{i_j}$, then real functions $f_{\mathrm{min}}$ and $f_{\mathrm{max}}$ are called \textit{min(max)--tradeoff function} for $\mathcal{M}_i$ if respectively
    \begin{equation}
        f_{\mathrm{min}}(p)\le \inf_{\rho \in \Sigma_i(p)}H(O_i|S_i E)_\rho,\qquad
        f_{\mathrm{max}} (p)\ge \sup_{\rho \in \Sigma_i(p)}H(O_i|S_i E)_\rho
    \end{equation}
\end{definition}
The function $f$ is adequate to quantify the accumulated entropy in a single step of the process because it
balances between overly optimistic and pessimistic entropy estimates. A naive approach might use the conditional von Neumann entropy \( H(O_2 | O_1) \), which averages the entropy over all states and overestimates the extractable randomness. On the other hand, a worst-case min-entropy \( H_{\text{min}}^{\text{w.c.}} = \min_{o_1, o_2} \left[ -\log \Pr(o_2 | o_1) \right] \) is too pessimistic, as it fails to capture the realistic entropy when the systems are independent. The correct definition considers the worst-case state \( o_1 \) but averages the entropy contribution \( -\log \Pr(o_2 | o_1) \) over \( o_2 \), leading to \( \min_{o_1} \mathbb{E}_{o_2} \left[ -\log \Pr(o_2 | o_1) \right] = \min_{o_1} H(O_2 | O_1 = o_1) \). 
\begin{definition}(events on classical registers) The classical registers $C_i$ defines the following classical probability space $(\Omega,\mathcal{B}(\Omega),p)$ where the sample set
\begin{equation}
    \Omega=\{\omega=(c_1,\dots,c_n)|\forall i,\,c_i\in\{c_{i_j}\}_j\}\subseteq C_1\times\dots\times C_n\equiv C^n
\end{equation}
contains the results from each step extracted by $\rho_{O_iS_i}$ for $i=1,\dots,n$ so that the updated final state reduced to the classical registers $C^n$ is the probability distribution $\rho_{C^n}=p(\omega)$, with $\omega\in \Omega$. The updated final state at the output of the sequential process conditioned by the event $\omega$ is denoted as $\rho_{|\omega}\in \bigotimes_{i=1}^n\Sigma_i(p_{i})\subseteq \bigotimes_{i=1}^n R_iO_iS_iC_iE=R_1^nO_1^nS_1^nC_1^nE$ (denoting, e.g. $R^n\equiv R^{\otimes n}$). With this notation, each register is isomorphic to the corresponded register at round $i=1$. $\mathcal{B}(\Omega)$ is the Borel $\sigma$--algebra.
\end{definition}
Note that the trade-off functions applied for the probabilities $p(\omega)$ have consistent bound from all the other quantum output $O_1,\dots,O_n\equiv O_1^n$ and similarly $S_1^n$. For instance $ \inf_{\rho_{|\omega}}H(O_1^n|S_1^nE)_{\rho_{|\omega}}\ge f_{\mathrm{min}}(p(\omega))$. Using these definitions, EAT can be stated as the following theorem.
\begin{theorem}[Entropy accumulation theorem] 
    Given a Markovian sequential process $\mathcal{M}=\bigcirc_{i=1}^n\mathcal{M}_i\otimes\mathrm{id}$ (see fig. \ref{fig:EAT}) such that the output state is $\rho_{|\omega}=\mathcal{M}(\rho_\mathrm{in})$, a convex $f_{\mathrm{min}}(p(\omega))\ge t$ with $t\in\mathbb{R}$, and $\varepsilon \in (0,1)$, then 
    \begin{equation}
    H_{\mathrm{min}}^\varepsilon(O_1^n|S_1^nE)_{\rho_{|\omega}} \ge n t - \nu \sqrt{n},\qquad \nu=2\left(\log(1+2\dim O_i)+\lceil \parallel\nabla f_{\mathrm{min}}\parallel_\infty \rceil\right)\sqrt{1-2\log(\varepsilon p(\omega))}.
    \end{equation}
\end{theorem}
\begin{proof}
    (details in Refs. \cite{adfrv17,dfr16,Dupuis2020}).
    The smooth min-entropy \( H_{\text{min}}^{\epsilon} \) is related to the sandwiched Rényi entropy \( H_{\alpha} \) for some parameter \( \alpha > 1 \) used to decompose the entropy of the full sequence into a sum of conditional entropies for each round:
  $
  H_{\alpha}(O_1^n | S_1^n E) \approx \sum_{i=1}^n H_{\alpha}(O_i | S_i E)
  $.
  The Markovianity ensures that each term \( H_{\alpha}(O_i | S_i E) \) depends only on the previous rounds and not the entire sequence. Now, \( H_{\alpha} \) can be bounded by the von Neumann entropy \( H(O_i | S_i E) \) using properties of the Rényi entropy $  H_{\alpha}(A_i | B_i R) \approx H(A_i | B_i R) - \mathcal{O}(\alpha - 1)$. Combining this with the tradeoff function \( f_\mathrm{min} \), which lower bounds \( H(O_i | S_i E) \), we get:
  $H_{\alpha}(A_i | B_i R) \geq f(p(\omega)) - \mathcal{O}(\alpha - 1)$. The distribution $p(\omega)$ of the classical outcomes is used to bound the entropy of the sequence. The observed event \( \Omega \) ensures that $f(p(\omega)) \geq t$, i.e. the entropy rate averaged over all rounds is at least \( t \). The finite-size effects arise because \( n \) is finite and \( f_\mathrm{min} \) depends on the second order statistical fluctuations $\nu\sqrt{n}$ in the observed data. The sandwiched Rényi entropy \( H_{\alpha} \) is converted to the smooth min-entropy \( H_{\text{min}}^{\epsilon} \).
Thus, the total smooth min-entropy grows approximately \textit{linearly} with \( n \), up to finite-size corrections.
\end{proof}
\paragraph{Finite key analysis with EAT}
\begin{figure*} 
    \centering
    \includegraphics[width=1\linewidth]{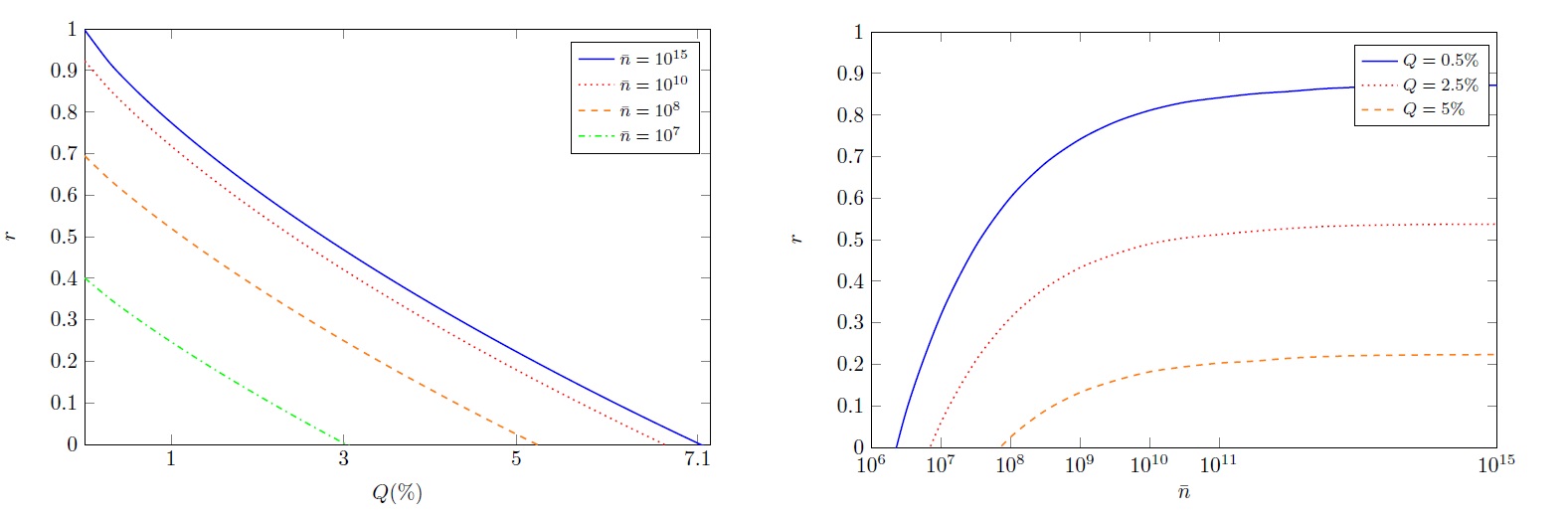}
    \caption{Figure from \cite{Tan2022}. (left) The expected key rate versus the QBER is lower for finite $n$ and with a security proof against coherent attacks than the one against individual and collective attacks. (Right) The expected key rate as a function of the number of rounds $\Bar{n}$.}
    \label{fig:Tan2022}
\end{figure*}
Thanks to EAT, the key rate \( r \) of CHSH protocol versus the QBER \( Q \) can be predicted at finite $n$ (see Fig. \ref{fig:Tan2022}). 
For \(n = 10^{15}\), the curve nearly overlaps the asymptotic iid case \ref{DIQKD-collective}, which was shown to be optimal allowing the protocol to tolerate up to $Q=0.071$.
Instead, for non-iid coherent attack analysis, the key rate obtained in \cite{vazirani2014fully} remains well below the lowest curve presented in Fig. \ref{fig:Tan2022}, even if the number of signals approaches infinity, with a maximum noise tolerance of only 1.6\%. 

Fig. \ref{fig:Tan2022} shows the key rate as a function of the number of rounds \( n \), for different values of \( Q \). Evidently, the rates achieved are significantly higher than those without EAT, 
and are comparable to the key rates of the practical device-dependent QKD with the difference that the device-independent one requires a larger $n$. 
Indeed, the finite-key analysis in Ref. \cite{Tan2022} shows that the experiments \cite{Hensen2015a} and \cite{Munich17} require, respectively, $n=10^8$ and $n=10^{10}$ rounds against coherent attacks for $r>0$. This analysis involves noisy preprocessing, random key measurements, and modified CHSH inequalities.
While this is a marked improvement wrt the basic protocol in \cite{pabgs09andSangouard.4} (Sec. \ref{DIQKD-collective}) (which yields zero asymptotic key rates for those experiments),
the values of $n$ are still impractical. Therefore, two modifications were applied: (\textit{i}) a pre-shared key, which results in a net key generation rate approximately double that of the original protocol; (\textit{ii}) relax the protocol by considering the collective-attacks assumption to alter its structure and enhance the key rate. However, despite the drop in $n\sim [10^6,10^7]$ the required number of rounds remains impractically large.
\paragraph{Entropy accumulation theorem with improved second order}
EAT theorem provides tight bounds only at the first order. The second-order term can be improved in many protocols of interest, where the entropy is estimated by testing positions with probability $O(n^{-1})$. Since $\nu_1\propto \parallel \nabla f_\mathrm{min}\parallel_\infty\propto O(n)$, this gives 
$\nu_1 \sqrt{n}\gg t n$.
Ref. \cite{Dupuis2019} show the correction,
$
H_{\min}^{\varepsilon}(K_{1}^n\mid S_{1}^n E)\ge nt - (\nu_1 \sqrt{n} + \nu_{2})
$, 
with $\nu_2$, a functions of $\varepsilon$, the maximum dimensions of the systems $K_i$ ($d_K$), and the variance of the function $f$. This further improvement contributes to reducing $n$.  In recent works on DI-QKD \cite{Nadlinger2022} and randomness expansion \cite{Liu2021b,Bhavsar2023}, the improved second-order term is explicitly employed to derive positive key rates and certify randomness with realistic experimental resources. Without this advancement, the required number of Bell test rounds would be orders of magnitude larger, making such protocols infeasible with current technology. 

\subsection{Generalized entropy accumulation theorem} \label{sec:GEAT}
EAT is incompatible with prepare-and-measure protocols because it assumes Markovianity, where side information $\rho_{s_i}$, once output, cannot be updated so that the total side information is in $\rho_{ES^n}$. But in prepare-and-measure protocols, Eve intercepts $\rho_i$ at the $i$-th round and updates her side information $\rho_{S_1,\dots,S_i}$ so that the total side information is higher than the one in $\rho_{ES_1^n}$. 
Although Markovianity allows estimating the smoothed min-entropy from a single round, it conflicts with the dynamic nature of side information in prepare-and-measure scenarios. For these protocols, to apply EAT one must first convert the protocol to an Entanglement-based one. 
To illustrate what could happen without
Markovianity, consider a case where $K_i$ is classical and no side information is output in the first $n-1$ rounds. Consider the side information of the last round in $\rho_{S_n}$ that contains a copy of
the systems $A_n$  
which can be passed along during the process in the systems $R_i$. Then, $H_{\min}^{\epsilon}=0$ while for the previous $n-1$ rounds, each single-round entropy bound that only considers the system $K_i$ and $S_i$ can be positive.
To address these issues, the \textit{Generalized Entropy Accumulation Theorem (GEAT)} replaces the Markov condition with a natural non-signaling condition between past outputs and future side information \cite{Metger2024}. 
\begin{definition}(non-signal process)
    Given $\mathcal{M}$ a sequential process from \ref{def:seq}. It is \textit{non-signal} if 
\begin{equation}
    \forall \mathcal{M}_i \qquad \exists \mathcal{R}_i:E_{i-1}\rightarrow E_{i} \mbox{ CPTP s.t } \qquad
    \mathrm{Tr}_{K_iR_i} \circ \mathcal{M}_i=\mathcal{R}_i\circ \mathrm{Tr}_{R_{i-1}}.
    \label{GEAT:nosignaling}
\end{equation}
\end{definition} 
Let us consider the systems $R_{i-1}$ and $R_i K_i$ as the inputs and outputs on ``Alice’s side'' of $\mathcal{M}_i$, and $E_{i-1}$ and $E_i$ as the inputs and outputs on Eve’s side, then Eq. \eqref{GEAT:nosignaling} states that the marginal of the output on Eve’s side cannot depend on the input on Alice’s side. This is exactly the non-signaling condition of Eq. \eqref{eq:no-signalling} in non-local quantum games. 

\begin{theorem}
   
    Given a non-signal sequential process $\mathcal{M}=\bigcirc_{i=1}^n\mathcal{M}_i$ with $\mathcal{M}_i:R_{i-1}E_{i-1}\rightarrow R_i K_i C_i E_i$ (see fig. \ref{fig:(G)EAT}) such that the output state is $\rho_{|\omega}=\mathcal{M}(\rho_\mathrm{in})$, an affine min-tradeoff 
    $f$ such that $t=\min f(p(\omega))$, $\varepsilon \in (0,1)$, $\alpha\in (1,\frac{3}{2})$, then 
    
    \begin{eqnarray}
    \label{eq:GEAT}
        H_{\min}^{\epsilon}(K^n|E_n)_{\rho_{|\omega}} \geq n\left(t -\frac{\alpha-1}{2-\alpha}\frac{\ln{2}}{2}V^2-\left(\frac{\alpha-1}{2-\alpha}\right)^2K^{\prime}(\alpha)\right)-\frac{g(\epsilon)-\alpha\log p(\omega)}{\alpha-1},
    \end{eqnarray}
    where $p(\omega)$ is the probability of observing event $\omega$, and 
    \begin{equation}
        g(\epsilon)=-\log(1-\sqrt{1-\epsilon^2})
        ,\quad
        V=\log(2d_{A}^2+1)+\sqrt{2+\Delta_f},
        \quad
        K^{\prime}(\alpha)=\frac{(2-\alpha)^3\ln^3(2^{\beta}+e^2)}{6(3-2\alpha)^
    3\ln{2}}2^{\frac{\alpha-1}{2-\alpha}(\beta+\log d_A)}
    \end{equation}
    with $d_A=\max_id_{A_i}$, $\Delta_f=\mathrm{Var}f$ and $\beta=\log d_A+\text{Max}(f)-\text{Min}_{\Sigma}(f)$
\end{theorem}
\noindent
The GEAT deals with a sequence of channels $\mathcal{M}_i$ that can update both the internal memory register $R_i$ and the side information register $E_i$ (subject
to the no-signaling condition of Eq.\eqref{GEAT:nosignaling}), while EAT sequential channels \textit{do} not update from each round the side information in the next rounds. As a result, GEAT is strictly more general than the EAT \cite{Metger2024}.
The B92 protocol and BB84 decoy-state protocol, lacking direct conversion to an entanglement-based form, cannot use EAT for security proof but it is based on GEAT \cite{Metger2023,Tan2024}.
 
Before (G)EAT the security proof bounds utilized de Finetti-type theorems combined with the QAEP, but with several drawbacks: (i) applicable only under specific assumptions regarding the symmetry of the protocols robust only against specific attacks; (ii) limited in the practically finite-size analysis; (iii) limited in a device-independent context. Entropy Accumulation Theorem (EAT) \cite{dfr16,Dupuis2019,Metger2022} applied for DI-QKD \cite{ArnonFriedman2019} solve these drawbacks. 
However, if condition (i) is satisfied, the security of DI-QKD against coherent attacks follows from security under the iid assumption.
Moreover, the dependence of the key rate on the number of rounds, $n$, is the same as that in the iid case, up to terms that scale like $\frac{1}{\sqrt{n}}$. 
As a consequence, one can extend tight results known for DI-QKD, under the iid assumption, to the most general setting. This yields the best rates known for any protocol for a DI cryptographic task as shown in fig. \ref{fig:Tan2022} for $n=10^{15}$.
\subsubsection{security of E91 with (G)EAT} 
\label{subsec:E91GEAT}
To study the security of E91 (introduced in Section~\ref{sec:E91}), we consider a simplified version where, in each round $i\in [n]$, Alice and Bob independently choose basis bits \(x_i, y_i \in \{0,1\}\) such that \(\Pr[x_i = y_i = 1] = \mu\) (diagonal basis) and \(\Pr[x_i = y_i = 0] = 1 - \mu\) (computational basis). Defining \(x^n := (x_1, \dots, x_n)\) and \(y^n := (y_1, \dots, y_n)\), they announce their bases and keep the subset $S = \{i : x_i = y_i\}$, with Alice's sifted string denoted $A_S$ and Eve's total side information represented by $E$.\\
Using EAT, one obtains the finite-size bound (Dupuis \textit{et al.} \cite{Dupuis2020} Theorem 5.1)   
\begin{equation}
H_{\min}^{\varepsilon}\!\bigl(A_S \mid x^{n}y^{n}E\bigr)_{\rho_{|\omega}}
   \ge
   n\!\bigl(1-2\mu-H(e)\bigr)-O(n),
\label{eq:EAT-min-entropy}
\end{equation}
where \(e\) is the observed phase-error rate on rounds with $x_i=y_i=1$. \\
Information reconciliation leaks at most \(n\vartheta_{\mathrm{EC}}\) syndrome bits,
so privacy amplification produces a \(\kappa\)-bit secret key whenever  
\begin{equation}
\kappa \;\le\;
   H_{\min}^{\varepsilon}\!\bigl(A_S \mid x^{n}y^{n}E\bigr)_{\rho_{|\omega}}
   \;-\; n\vartheta_{\mathrm{EC}} \;-\;O(1).
\label{eq:key-length}
\end{equation}
Dividing \eqref{eq:key-length} by \(n\) and inserting
\eqref{eq:EAT-min-entropy} yield the finite-size key-rate condition,  
\begin{equation}
r =\frac{\kappa}{n} < 1 - h(e) - \vartheta_{\mathrm{EC}} - 2\mu .
\label{eq:finite-rate}
\end{equation}

In the asymptotic limit \(n\!\to\!\infty\) with sampling fraction \(\mu\to 0\),
all finite-size terms vanish and \eqref{eq:finite-rate} reduces to  
\begin{equation}
r_\infty \;=\; 1 - h(e) - \vartheta_{\mathrm{EC}} .
\label{eq:asymp-rate}
\end{equation}
When applying GEAT to the same E91 protocol, it can be seen explicitly that the single-round trade-off function and testing step do not change, and the remainder of the security proof is exactly as EAT. As a result, GEAT does not change the key rate obtained by EAT; however, GEAT drops the Markov assumption, so one can treat protocols (e.g. prepare-and-measure or ones with device memory) that EAT could not handle at all.

\subsubsection{Security of rDI-QKD with GEAT}
\label{subsec:rDIQKDsecurity}
An rDI-QKD protocol introduced in \ref{sec:rDIQKD}  mainly differs from DI-QKD in the quantum measurement phases $\mathcal{M}_i$. To see it more clearly, let us focus on each step $i$. Conditioning on the input classical variables $x_i$, $s_i$, $z_i$, and $y_i$, each $\mathcal{M}_i$ can be described as a CPTP map  $\mathcal{M}_i:\mathcal{Q}_{A_{i-1}}\mathcal{Q}_{B_{i-1}}E_{i-1}\rightarrow A_i B_iC_i\mathcal{Q}_{A_i}\mathcal{Q}_{B_i}E_i$ that takes as input the quantum registers $\mathcal{Q}_{A_{i-1}}$ (Alice's private measurement device $A$), $\mathcal{Q}_{B_{i-1}}$ (B's private measurement device $B$), and $E_{i-1}$ (Eavesdropper Eve) and outputs the classical variables $A_i$, $B_i$, $C_i$ along with updated quantum registers $\mathcal{Q}_{A_i}$,$\mathcal{Q}_{B_i}$, and $E_i$. By including the additional data in rDI-QKD (compared to DI-QKD) i.e. random inputs $S_i$, $z_i$, and the outcome $c_i$ into the Eve's side information $E$, the non-signal condition in \ref{GEAT:nosignaling} remains unaffected and GEAT can be applied for the security proof of the protocol.\\ 
 An rDI-QKD protocol introduced in \ref{sec:rDIQKD} can be shown by a tuple $\mathcal{Q}_{r}=\{\rho_{AB},A_x,B_y,T_z\}$. which gives rise to the correlations $p(a,b|x,y)$ and $p(a,c|x,z)$. Based on the above discussion, as in the standard DI-QKD, the asymptotic key rate can be calculated by the iid Devetak-Winter rate $r=H(A|XE)-H(A|B)$. To lower-bound the term $H(A|XE)$, considering that the source initially produces a state $\rho_{ABE}$, without loss of generality, one can assume that this state is a pure state $\ket{\Psi_{ABE}}$ and all the measurement settings are projective. The possible quantum strategies that Eve can use are fully characterized by the pure state $\ket{\psi_{ABE}}$ and the projective measurements $\{A_{a|x}\}$,$B_{b|y}$, and $T_{c|z}$ conditioned to the fact that they return the honest correlations 
 \begin{align}
  \label{eq:rdihonestcorr}
     p(a,c|x,z)=\bra{\Psi_{ABE}}A_{a|x}\otimes T_{c|z}\ket{\Psi_{ABE}}, \\
     p(a,b|x,y)=\bra{\Psi_{ABE}}A_{a|x}\otimes B_{b|y}\otimes\bm{I}_E\ket{\Psi_{ABE}}, \nonumber
 \end{align}
where $T_{c|z}$ acts jointly on subsystems $B$ and $E$. To each strategy, one can associate the post-measurement state $\sigma_{AXE}=\sum_{a,x}p(x)\ket{ax}\bra{ax}\otimes\sigma_{E}^{a,x}$ where $\sigma_{E}^{a,x}=\mathrm{tr}_{AB}(\ket{\Psi_{ABE}}\bra{\Psi_{ABE}}(A_{a|x}\otimes \bm{I}_{B}\otimes\bm{I}_E))$ is the unnormalized state held by Eve conditioned to Alice's inputs and outputs. The conditional min-entropy can then be computed as 
\begin{equation}
H(A|XE)=\inf_{\mathcal{\hat{Q}}|p}H(A|XE)_{\sigma_{AXE}},
\end{equation}
where the optimization runs over all quantum strategies $\hat{\mathcal{Q}}$ compatible with the honest correlations \eqref{eq:rdihonestcorr}. Notice that this optimization is almost identical to the optimization problem in a standard DI-QKD protocol where Bob performs the measurements $T_z\otimes B_y$, with a difference that the measurements $T_z$ act on the joint systems $BE$, instead of just $B$. The method in section \ref{subsubsec:numeriacallowerbounds} (see Theorem \ref{theorem:BFF21}) then can be applied to lower bound the conditional entropy $H(A|XE)$ in rDI-QKD in almost the same way as in DI-QKD. For the rCHSH protocol which is the routed version of the DI-QKD CHSH protocol, if $A$ and $T$ have perfect detectors $\eta_A=\eta_B=1$, the key rats are very robust as $\eta_B$ decreases, remaining positive for $\eta_B\gtrsim 0.68$ \cite{Tan2024}. However, this value is not robust when $\eta_A$ and $\eta_B$ are decreased for example, in the case where all devices have the same detection efficiency $\eta$, the key rate is positive for $\eta\gtrsim 0.96$ which is worse than standard CHSH based DI-QKD protocols.
\subsubsection{GEAT and the security of monogamy-of-entanglement based DI-QKD}
While some protocols, such as the generalized CHSH game \ref{ch-sh} and the magic square game \ref{msg-di-qkd}, have also been considered, most of the protocols studied so far have been based on the CHSH game due to its simplicity of implementation. So, here an important question arise: \textit{Is it possible to explicitly prove secrecy of a DI-QKD protocol using an arbitrary monogamy-of-entanglement game?} This question was tackled in \cite{CerveroMartin2025}. 

Let us start by defining \textit{a non-local game}.  \\
\begin{definition}
A two-party nonlocal game is a tuple $G_2 = (\pi, \mathcal{X}, \mathcal{Y}, \mathcal{A}, \mathcal{B}, V),$
where $\pi$ is a probability distribution over input pairs  $(x, y) \in \mathcal{X} \times \mathcal{Y}$, and  $V : \mathcal{X} \times \mathcal{Y} \times \mathcal{A} \times \mathcal{B} \to \{0,1\}$ is the winning predicate. Alice and Bob receive input $(x,y)$, respond with output $a \in \mathcal{A}$, $b \in \mathcal{B}$, and win if $V(x, y, a, b) = 1$. Similarly, a three-party game $G_3$ is defined as one in which a third player (Eve) also contributes by receiving $(x, y)$ and outputs a bit $c \in \{0,1\}$.
\end{definition}
\noindent Using this definition, the main theorem in \cite{CerveroMartin2025} is expressed as the following theorem. 
\begin{theorem}
\label{theorem:monogamyofentanglementgames}
Consider a DI-QKD protocol based on a two-player non-local game $G_2$ between Alice and Bob, with quantum winning probability $\omega_2$. Suppose an adversary (Eve) holds quantum side information and may launch general coherent attacks, which can be modeled by extending the game to a three-party non-local game $G_3$ with quantum winning probability $\omega_3 < \omega_2$. Then, there exists an affine min-tradeoff function $f : [0,1]\to \mathbb{R}$, defined for any $\beta \in [\omega_3, \omega_2]$ by
\begin{equation}
f(p) = \frac{p-\beta}{\ln 2}(1 - \beta + \omega_3) - \log(1 - \beta + \omega_3),
\end{equation}
such that the smooth min-entropy of Alice's raw key conditioned on Eve's quantum side information and public communication, satisfies
\begin{equation}
H_{\min}^\varepsilon(A|E) \geq n f(p_{\mathrm{exp}}) - O(\sqrt{n}),
\end{equation}
where $p_{\mathrm{exp}} \in [0,1]$ is the observed winning probability in the testing rounds.
\end{theorem}
\noindent The monogamy-of-entanglement property in this setting is reflected in the fact that the optimal quantum winning probability $\omega_2$ of the two-player non-local game (played between Alice and Bob) exceeds the tripartite quantum winning probability $\omega_3$ of the extended game, in which a third party receives both inputs and attempts to guess the key bit produced by Alice and Bob. \par
Theorem~\ref{theorem:monogamyofentanglementgames} implies that it is indeed possible to construct DI-QKD — and prove their security — from any two-player non-local game that exhibits a sufficiently large gap $\omega_2 > \omega_3$ between the two-party and three-party instances of the game.

\subsubsection{(G)EAT vs. iid and non-iid techniques} (G)EAT is more general in the sense that it does not need to assume that the rounds of the experiment are \textit{independent and identically distributed} (iid). This, in particular, implies that (\textit{i}) the measurement devices are memoryless, i.e. they behave independently and in the same way in every round of the protocol;(\textit{ii}) the distributed state is the same for every round $\rho_{A_1^nB_1^nE}=\rho_{ABE}^{\otimes n}$. The iid simplification can be justified, for example, in experimental setups where Alice and Bob control, to some extent, the source and measurement devices, but do not have a full characterization of their working devices.
In this case, $H_{\min}^{\varepsilon}(K_{1}^n\mid E_{1}^n)$ can be directly related to the single-round conditional von Neumann entropy $H(K_i|E_i)$ and (G)EAT is equivalent to the quantum asymptotic equipartition property (AEP) \cite{TCR09} yielding
\begin{align}
	H_{\min}^{\varepsilon}(K_{1}^n\mid E_{1}^n)\ge n H(K_i\mid E_i)-c_\varepsilon \sqrt{n} ,\label{eq:AEP]}
\end{align}
where $c_\varepsilon$ is dependent only on $\varepsilon$ and $H(K_i\mid E_i)\le 1-\chi_0$ of Eq. \eqref{Holevo-upperbound} for CHSH protocol. 

%
%
%
%
%
%
(G)EAT improves the traditional DI-QKD security proofs under coherent attacks \cite{vazirani2014fully,vazirani2019fully}. This, in particular, assumes that Eve exploits all degree of freedoms of the quantum systems, applying global operations across all protocol rounds \(\rho_{ABE}=\rho_{A^n B^n E} \) and a global measurement \( \mathcal{M}_E \) on \( \rho_{E} \).
Let us consider CHSH protocol with abortion threshold $S\le 2\sqrt{2}(1-2Q)$, then
\begin{equation}\label{eq:Vaziranisecproof}
 H_{\min}^{\varepsilon}(A\mid E)_{\rho}> 
 -6(1-\tau')\log\left(\frac{11}{12}+\frac{3}{8}\sqrt{\frac{Q}{1-\tau}}\right)-O\left(\frac{\log(1/\varepsilon)}{2Q^2n}\right),\qquad \forall \tau+\tau'>1,
\end{equation}
with $n$ the rounds and $Q$ the QBER. 
After postprocessing 
$
    r\geq H_{\min}^{\epsilon}(A|E)- h(Q),
$
as in Eq. \eqref{devatekwinterformula}. 
Eq. \eqref{eq:Vaziranisecproof} relies on \textit{quantum reconstruction paradigm} (QRP) \cite{De2012}.

However, the key rate is lower compared to security proofs obtained via (G)EAT.
\subsubsection*{(G)EAT vs Rényi-divergence method}

Before the development of the EAT, one of the most complete finite-key security proofs for robust DI-QKD was introduced by Miller and Shi~\cite{Miller2016}. Their approach was based on bounding a Rényi-divergence that decreases linearly with the number of rounds, relying on a game-specific uncertainty relation for strong self-testing XOR games and an intricate induction technique using partially trusted measurements. This allowed them to achieve notable milestones such as a polylogarithmic seed length, constant quantum memory requirements for the honest parties, and explicit exponential soundness error bounds. However, their proof was highly specialized and tightly coupled to the particular structure of the underlying nonlocal game, making it difficult to generalize or adapt to new protocols, statistics, or noise models. In contrast, the (G)EAT framework generalizes this type of entropy accounting by providing a universal accumulation statement: once a suitable min-tradeoff function is established, the conditional smooth min-entropy grows predictably and additively, regardless of whether the device behavior is iid or exhibits arbitrary non-stationary correlations. This shift delivers sharper finite-size bounds, native composability, and a modular, game-independent toolkit that can be applied across a wide range of DI-QKD, DI-randomness expansion, and other sequential quantum protocols without rederiving core security proofs. As such, (G)EAT is widely regarded as the natural successor to the Rényi-divergence-based method for modern device-independent security analysis.

\subsection{Quantum Probability Estimation method}
\label{subsec:QPE}
Another useful method for proving security in device-independent protocols is to employ the quantum probability estimation method, which was introduced in \cite{Zhang2018d,Zhang2020a}. The analysis in \cite{Zhang2020a} enables a full security analysis for DIQRNG. Similarly to (G)EAT, in QPE, inputs and outputs are determined in a sequence of $n$ time-ordered trials, where the $i$ trial has input $Z_i$ and output $C_i$, so $\boldsymbol{Z}=(Z_i)_{i=1}^n$ and $\boldsymbol{C}=(C_i)_{i=1}^n$. QPE introduces a flexible tool called a Quantum Estimation Factor (QEF), which quantifies certifiable randomness in quantum experiments, even in the presence of quantum side information. QEF is defined as 
\begin{definition}[Quantum Estimation Factor]    
Let $C$ and $Z$ be classical random variables that represent, respectively, the output (e.g., measurement result) and input (e.g., measurement setting) in a quantum experiment, and let $E$ be a quantum system that may hold quantum side information by an adversary.
Then a function $F : \mathrm{Rng}(C) \times \mathrm{Rng}(Z) \to \mathbb{R}_{\geq 0}$
is called a Quantum Estimation Factor (QEF) with power $\beta > 0$ if, for every valid classical-quantum state $\rho_{CZE}$ , the following inequality holds:
\begin{equation}
\sum_{c \in \mathrm{Rng}(C)} \sum_{z \in \mathrm{Rng}(Z)} \mathrm{Tr}\left[\rho_E(c, z)\right] \cdot F(c, z) \cdot \hat{R}_\alpha\left( \rho_E(c, z) \,\middle\|\, \rho_E(z) \right) \leq 1,
\end{equation}
where $Rng(A)$ is the value space of a classical variable $A$ and
\begin{itemize}
  \item $ \alpha = 1 + \beta $,
  \item $\rho_E(c, z)$ is the (possibly subnormalized) quantum state of the adversary’s system $E$ conditioned on observing $C = c$ and input  $Z = z$,
  \item $\rho_E(z) = \sum_{c} \rho_E(c, z)$ is the marginal quantum state for input $z$,
  \item $\hat{R}_\alpha(\rho \| \sigma)$ is the normalized sandwiched R\'enyi power of order $\alpha$, defined as:
  \begin{equation}
  \hat{R}_\alpha(\rho \| \sigma) = \frac{1}{\mathrm{Tr}[\rho]} \cdot \mathrm{Tr}\left[ \left( \sigma^{-\frac{\beta}{2\alpha}} \rho \, \sigma^{-\frac{\beta}{2\alpha}} \right)^{\alpha} \right].
  \end{equation}
\end{itemize}
\end{definition}
Using this definition, the following lower bound on quantum smooth min-entropy can be proved \cite{Zhang2020a}.
\begin{theorem}[Lower Bound on Quantum Smooth Min-Entropy via QEFs]
Suppose that the experiment is repeated for $n$ rounds. For each round $i$, let  $F_i : \mathrm{Rng}(C) \times \mathrm{Rng}(Z) \to \mathbb{R}_{\geq 0}$ be a QEF with power $\beta>0$. By defining the total QEF value as $
f = \prod_{i=1}^{n} F_i(c_i, z_i)$ and  any desired soundness error \( \epsilon \in (0, 1] \), if
$f \geq \left( \frac{2^{-k} \cdot \epsilon^2}{2} \right)^{-\beta}$, then the smooth conditional min-entropy of the output string  $\boldsymbol{C}$ with given inputs $\boldsymbol{Z}$ satisfies:
\begin{equation}
H_{\min}^\epsilon(C \mid Z E) \geq k.
\end{equation}
\end{theorem}
\noindent The above theorem states that at least $k$ bits of the output string are private and unpredictable, which can be safely extracted using a quantum-proof randomness extractor.\\
Since certifying quantum smooth conditional min-entropies is a central task in QKD, QPE can be adapted to enhance the finite-data efficiency of DI-QKD. To achieve this, it is necessary to certify the quantum smooth conditional min-entropy evaluated on a classical-quantum state following the error-correction step, similar to the approach used in security proofs of DI-QKD via the EAT \cite{adfrv17} (see Section \ref{sec:EAT}).

\subsection{Complementarity-Based Security for DI-QKD}

So far, DI-QKD security proofs have primarily relied on entropic techniques such as the (G)EAT. While powerful, these methods are technically demanding and typically require large data sizes, limiting their experimental feasibility.\\
A different security proof for DI-QKD was proposed by Zhang et al. \cite{Zhang2023b} by employing a \textit{complementarity-based} approach. They recast the DI-QKD security problem as a form of \textit{quantum error correction} against \textbf{phase errors}, drawing on principles already familiar from standard device-dependent QKD \cite{Lo1999,Shor2000,Koashi2009}. \\
The central insight in \cite{Zhang2023b} is that a BI violation can be used to upper-bound the adversary's information via a \textit{phase error probability} derived from complementary observables and provide an operational connection between the Bell value and phase error rate using sample-entropy techniques together with a composable, finite-size security proof under coherent (non-iid) attacks. This leads to the following main result:

\begin{theorem}
In a CHSH-based DI-QKD protocol with memoryless devices executed over $n$ rounds, let $m$ rounds be used for key generation. If the observed average Bell value is $\bar{\beta} > 2$, then, except with small failure probability, the number of secure key bits that can be extracted is approximately:
\begin{equation}
k \approx m \cdot \left[1 - h(e_b) - h\left(e_p(\bar{\beta})\right)\right],
\end{equation}
where 
\[
e_p(\bar{\beta}) = \frac{1 - \sqrt{(\bar{\beta}/2)^2 - 1}}{2}
\]
is the phase error probability inferred from the Bell violation $\beta$.
\end{theorem}
 Simulation results demonstrate that the complementarity-based approach significantly improves the finite-size performance of DI-QKD. For example, in an ion-trap experiment with $\beta = 2.64$ and $Q=1.8\%$, approximately $3.16 \times 10^5$ measurement rounds are required to generate a secure key, which is less than one-third of the rounds required by EAT methods. This substantial reduction in data size makes the protocol more practical for experimental platforms with low entangled-pair rates, such as cold atoms and nitrogen-vacancy centers.
Moreover, it was shown that advantage distillation can be applied without assuming iid behavior \cite{Zhang2023b}, lowering the thresholds for required transmittance and fidelity and expanding the range of conditions under which secure DI-QKD is possible.

\subsection{Security Proof of Parallel DI-QKD}
The original security proof for parallel DI-QKD relates collision entropy of the raw key to the success probability in a derived game. By showing that even entangled strategies cannot win this game with high probability, they obtain the bound:
\begin{equation}
H^{\epsilon}_{\min}(A | E) - H^{\epsilon}_{\max}(A | B) \geq R(n)
\end{equation}
where \(R(N)\) is the extractable key rate, and \(R(N) \in \Omega(N)\) means it grows at least linearly with the number of rounds \(N\).\\
 Soon after, Vidick \cite{Vidick2017} provided a simplified security proof for such protocols relying on two powerful and general tools from quantum information theory, the immunization technique and anchored parallel repetition. The immunization method ensures that even entangled adversaries cannot win a modified "guessing" version of the Magic Square game with certainty, introducing unavoidable uncertainty that translates into key entropy. The anchored parallel repetition theorem then guarantees that this uncertainty accumulates exponentially with the number of rounds, even under coherent attacks. By combining these ideas, he shows that a positive key rate can be securely extracted with a far more direct and modular argument. 
Further improvement was done by Jain et al. \cite{Jain2022} by introducing the direct product theorem. They show that it is possible to do DI quantum cryptography without
the assumption that devices do not leak any information. By applying it to parallel DI-QKD, they showed that when the protocol is carried out with devices that are compatible with $n$ copies of the Magic Square game, it is possible to extract $\Omega(n)$ bits of key from it, even in the presence of $O(n)$ bits of leakage.

\subsection{Analytical bounds}\label{sec:analyticalBounds}

\subsubsection{2-input/2-output protocols}

The first analytical bound, as mentioned in Section \ref{DIQKD-collective}, was established in Ref.~\cite{acin2007device} against collective attacks.
\begin{theorem}
    Let $\ket{\psi_{ABE}}$ be a quantum state for a CHSH$_c$ protocol. Then, the following upper bound holds for the Holevo quantity:
    \begin{equation*}
      \chi(B_1:E) \leq h\left(\frac{1+\sqrt{(\beta/2)^2-1}}{2}\right),
    \end{equation*}
\end{theorem}

\noindent For the proof of this theorem, we use the following lemma, which we put here without proof.
\begin{lemma}
\label{lemma:belldiag}
For a Bell-diagonal state with eigenvalues $\lambda$ ordered as $\lambda_{{\Phi}^{+}}\geq\lambda_{\psi^{-}}$ and $\lambda_{\Phi^{-}}\geq\lambda_{\psi^{+}}$ and for measurements in the $xz$ plane, the following bound holds for the Holevo quantity $\chi_{\lambda}(B_1|E)$
\begin{equation*}
  \chi_{\lambda}(B_1|E) \leq F(\beta_{\lambda})\leq h\left(\frac{1+\sqrt{(\beta_{\lambda}/2)^2-1}}{2}\right),
\end{equation*}
where $\beta_{\lambda}$ is the largest violation of the CHSH inequality by the state $\rho_{\lambda}$.
\end{lemma} 
Using this lemma, the proof of the theorem can be established.
\begin{proof}
As stated at the beginning of sec \ref{sec:DI-QKD}, suppose that Eve sends to Alice and Bob a mixture $\rho_{AB}=\sum_c p_c \rho_{AB}^c$ of two-qubit states with a classical ancilla known to her which carries on the information about measurement settings on Alice and Bob side. Two measurements on Alice and Bob can be assumed as von Neumann measurements (if necessary by including ancillas in $\rho_{AB}$). Thus the measurements $A_{1,2}$ are Hermitian d-dimensional operators. Using the Jordan lemma \ref{Jordanlemma} one can show that $A_1$ and $A_2$ are block diagonal, with blocks of size $1\times 1$ or $2 \times 2$ i.e. $A_j=\sum_c P_c A_j P_c$ with $P_c$'s as projectors of rank 1 or 2. Therefore, from Alice's standpoint, $A_{1,2}$ amounts at projecting in one of the at most two-dimensional subspaces defined by the projectors $p_c$  followed by a measurement on the reduced state observable $P_cA_iP_c$. The same argument holds for Bob. As a result, one can conclude that in each round of the protocol Alice and Bob receive a two-qubit state. \\
Each state $\rho_{AB}^c$ can be taken to be a Bell diagonal state ($\sum_{\lambda}p_{\lambda}\rho_{\lambda}$), and the measurements of Alice and Bob to be measurements in the $xz$ plane which result in $\chi(B_1:E)=\sum_{\lambda}p_{\lambda}\beta_{\lambda}$ plane. Therefore, using the lemma \ref{lemma:belldiag} the concavity of function $F$ 
\begin{equation*}
    \chi(B_1:E)\leq \sum_{\lambda}p_{\lambda}F(\beta_{\lambda})\leq F\left(\sum_{\lambda}p_{\lambda}\beta_{\lambda}\right)\leq F(\beta),
\end{equation*}
the last inequality comes from the fact that $F$ is a monotonically decreasing function.
\end{proof}
Based on this bound, the following lower bound for the key rate can be derived:

\begin{equation}
		r \geq I(A_0:B_1)-h\left(\frac{1+\sqrt{(\beta/2)^2-1}}{2}\right),
  \label{Acin2007bound}
\end{equation}

The basic CHSH protocol based on the above lower bound is, however, not optimal in several respects. To address the drawbacks, Masini et al. \cite{Masini2022}, introduced a new and versatile
approach to bound the conditional entropy in the 2-input/2-output device-independent setting that is
conceptually and technically relatively simple. The starting point is to use Jordan's lemma to reduce the
analysis to convex combinations of qubit strategies. \\ 
The next step, as in a standard qubit QKD protocol like BB84,  is to bound the conditional entropy of Alice's key-generating
measurement, $A_1$, through an uncertainty relation involving the correlations $\langle \Bar{A}_1\otimes B \rangle$ where $\Bar{A}_1$ is an orthogonal measurement on Alice's subsystem and $B$ is a binary observable on Bob's system. Considering the situation where Alice's raw key bit  $A_1$ is obtained as the outcome of the measurement, then we have the following bounds which are qubit uncertainty relations of the standard entanglement-based BB84 protocols its variants:\\

\begin{tabular}{@{} p{0.4\textwidth} >{\centering\arraybackslash}p{0.55\textwidth} @{}}
\toprule

BB84 entropy bound \cite{Berta2010} &
\scalebox{1.0}{$H(A_1|E) \geq 1 - \phi\left(|\langle\Bar{A}_1 \otimes B\rangle|\right)$} \\
\midrule

BB84 bound with noisy preprocessing \cite{Woodhead2021,Woodhead2014} &
\scalebox{1.0}{$H(A_{1}^q|E) \geq f_q\left(|\langle\Bar{A}_1 \otimes B\rangle|\right)$} \\
\midrule

BB84 with noisy preprocessing and bias \cite{Masini2022} &
\scalebox{1.0}{$H(A_{1}^q|E) \geq g_q\left(|\langle A_1 \rangle|, |\langle\Bar{A}_1 \otimes B\rangle|\right)$} \\
\midrule

Two-basis bound \cite{Masini2022,Woodhead2021} &
\scalebox{1.0}{$H(A_{X}^q|E) \geq f_q\left(\sqrt{p\langle \Bar{A}_1 \otimes B\rangle^2 + (1-p)\langle \Bar{A}_2 \otimes B^{\prime} \rangle^2}\right)$} \\
\bottomrule
\end{tabular} 
\\ \\
where $\phi(x)=h(\frac{1}{2}+\frac{1}{2}x)$ and $h(x)$ is the binary entropy. Moreover, $f_q(x)=1+\phi(\sqrt{(1-2q)^2+4q(1-q)x^2})-\phi(x)$, and $g_q(z,x)=\phi(\frac{1}{2}(R_{+}+R_{-}))-\phi(\sqrt{z^2+x^2})$, with $R_{\pm}=\sqrt{(1-2q\pm z)^2+4q(1-q)x^2}$.  
 The second step approach consists in deriving a constraint on these correlators in terms of correlators involving only the observables $A_1$, $A_2$, $B_1$, $B_2$  measured by the devices \\
 
 \begin{tabular}{@{} p{0.4\textwidth} >{\centering\arraybackslash}p{0.55\textwidth} @{}}
\toprule

CHSH correlation bound \cite{Woodhead2021} &
\scalebox{1.0}{$ |\langle\Bar{A}_1 \otimes B\rangle|\geq \sqrt{\beta^2 /4 -1}$} \\
\midrule

asymmetric CHSH correlation bound \cite{Woodhead2021} &
\scalebox{1.0}{$|\langle\Bar{A}_1 \otimes B\rangle|\geq E_{\alpha}(\beta_{\alpha})$} \\
\midrule

Two-basis correlation bound \cite{Masini2022} &
\scalebox{1.0}{$p\langle\Bar{A}_1 \otimes B\rangle^2+(1-p)\langle\Bar{A}_2 \otimes B^{\prime}\rangle^2\geq E_p(\beta)^2$} \\
\bottomrule
\end{tabular} 
\\

where 
\begin{equation}
E_{\alpha}(\beta_{\alpha}) = 
\begin{cases} 
\sqrt{\frac{\beta_{\alpha}^{2}}{4} - \alpha^2}, & \text{if } |\alpha| \geq 1, \\
\sqrt{1 - \left(1 - \frac{1}{|\alpha|}\sqrt{(1 - \alpha_2)\left(\frac{\beta_{\alpha}^2}{4} - 1\right)}\right)^2}, & \text{if } |\alpha| < 1.
\end{cases}
\end{equation}
and $E_p(\beta)^2$ is the solution of a polynomial optimization problem of five real variables which for the case $p=\frac{1}{2}$ can be solved analytically \cite{Masini2022}. \\
By combining the aforementioned correlation bounds with the entropy bounds, one can derive device-independent bounds on conditional entropy. For instance, by integrating the BB84 bound with the CHSH correlation bound, the bound in \ref{Acin2007bound} can be obtained.\\
Similarly, using the asymmetric CHSH correlation bound within the BB84 noisy preprocessing framework, the bound from \cite{Woodhead2021} can be derived:
\begin{equation}
    H(A_{1}^q|E) \geq f_q(E_{\alpha}(\beta_{\alpha})),
\end{equation}

Moreover, by combining the CHSH correlation bound with BB84, incorporating noisy preprocessing and bias, the following bound is obtained \cite{Masini2022}:
\begin{equation}
    H(A_{1}^q|E) \geq g_q(|\langle A_1 \rangle|, \sqrt{\beta^2/4 - 1}),
\end{equation}

Finally, if we denote $\tilde{E}_{p}(\beta)^2$ as any lower bound on $E_p(\beta)^2$, another bound can be expressed as \cite{Masini2022}:
\begin{equation}
    H(A_{X}^q|XE) \geq f_q(\tilde{E}_p(\beta)),
\end{equation}
where $\tilde{E}_{p}(\beta)$ is defined as $\tilde{E}_{p}(\beta) = \sqrt{\tilde{E}_{p}(\beta)^2}$. \\
Since the obtained bounds are convex, they can be extended to give fully device-independent bounds in arbitrary dimensions. 

\subsection{Rényi--entropy bounds for CHSH\label{sec:RenyiCHSH}}
Building on the analytical CHSH bounds for the von Neumann entropy and the min-entropy/guessing-probability approach, Hahn \emph{et al.}~\cite{Hahn2025} derive closed-form trade-off functions that relate a device’s observed CHSH violation $\beta \in [2, 2\sqrt{2}]$ directly to Rényi conditional entropies. Their main result provides tight entropy rate functions for a wide range of Rényi entropies.

Before stating the theorem, we introduce the relevant definition:

\begin{definition}[H rate function for CHSH]
Let $H$ be a conditional entropy and $\beta \in [2, 2\sqrt{2}]$. A function 
\[
f_H : [2, 2\sqrt{2}] \rightarrow \mathbb{R}
\]
is a tight $H$ rate function for the CHSH Bell inequality if
\[
f_H(\beta) := \inf_{\mathcal{E}} H(A \mid X = 0, E) \quad \text{subject to} \quad \beta_{\text{CHSH}}(\mathcal{E}) = \beta,
\]
where the infimum is over all quantum strategies $\mathcal{E}$.
\end{definition}

The following analytical bounds hold for the sandwiched Rényi conditional entropies $\tilde H^{\uparrow}_{\alpha}(A\mid E)$ and $\tilde H^{\downarrow}_{\alpha}(A\mid E)$ for any order $\alpha > 1$:

\begin{theorem}
For any $\alpha > 1$, the rate functions take the form:
\begin{align}
f_{\tilde H^{\uparrow}_{\alpha}}(\beta) &= 1 + \frac{2\alpha - 1}{1 - \alpha} \log \varphi_{\frac{\alpha}{2\alpha - 1}}(\beta), \\
f_{\tilde H^{\downarrow}_{\alpha}}(\beta) &= 1 + \frac{\alpha}{1 - \alpha} \log \varphi_{1/\alpha}(\beta),
\end{align}
where
\[
\varphi_{\mu}(\beta) = \left(\frac{1 - \sqrt{\beta^2/4 - 1}}{2} \right)^{\mu} + \left(\frac{1 + \sqrt{\beta^2/4 - 1}}{2} \right)^{\mu}.
\]
\end{theorem}

These bounds recover known results in the limits $\alpha \to 1$ (von Neumann entropy) and $\alpha \to 2$ (min-entropy).

Notably, the minimum entropy is achieved by a single optimal eavesdropping strategy independent of $\alpha$. Alice measures $A_0 = \sigma_z$ and $A_1 = \sigma_x$, while Bob measures
\[
B_{0/1} = \frac{\sigma_z \pm g_\beta\, \sigma_x}{\sqrt{1 + g_\beta^2}}, \quad \text{where } g_\beta = \sqrt{\beta^2/4 - 1}.
\]
These observables act on a partially entangled Bell state dependent only on $\beta$.

The new rate function $f_{\tilde H^{\downarrow}_{\alpha}}$ can significantly enhance finite-key DI-QKD performance. For example, in the loophole-free demonstration of~\cite{Nadlinger2022}, it improves the extractable key length by up to a factor of three for $n \sim 10^9$ rounds, and raises the positive-rate noise threshold to approximately $Q \approx 8.4\%$ depolarizing noise.

\subsubsection{Entropy Bound for multiparty DI cryptography}
 Ribeiro et al. \cite{Ribeiro2018} and Grasselli et al. \cite{Grasselli2021a} extended DI protocols to multipartite scenarios by proposing a DI conference key agreement (DI-CKA) among $N$ parties. The security of their protocol relies on the violation of a Mermin-Ardehali-Belinskii-Klyshko (MABK) inequality \cite{Mermin1990a,Ardehali1992,Belinskii1993}, a generalization of the CHSH inequality. Specifically, they focused on the three-party case involving Alice, Bob, and Charlie. In this context, the MABK inequality is expressed as:
\begin{equation}
m=\langle M_3 \rangle = \mathrm{Tr}[M_3\rho] \overset{\text{Cl}}{\leq} 2 \overset{\text{GME}}{\leq} 2\sqrt{2} \overset{\text{Q}}{\leq} 4,
\end{equation}
where \(M_3 = A_0 \otimes B_0 \otimes C_1 + A_0 \otimes B_1 \otimes C_0 + A_1 \otimes B_0 \otimes C_0 - A_1 \otimes B_1 \otimes C_1\) is the MABK operator. Here, \(A_x\), \(B_y\), and \(C_z\) represent Alice's, Bob's, and Charlie's observables, respectively. A violation beyond the GME threshold implies that the parties share a genuine multipartite entangled (GME) state.
\par
They derived the following bound on the conditional entropy as a function of the observed MABK violation $m$ :
\[
H(A_{0}|E) \geq 1 - h\left(\frac{1}{2} + \frac{1}{2} \sqrt{\frac{m^2}{8} - 1}\right).
\]
By proving that this bound on the conditional entropy of a party's outcome is tight at the GME threshold, it can be concluded that genuine multipartite entanglement is essential to ensure the privacy of a party’s random outcome in any device-independent protocol based on the MABK inequality.
 \\  
Grasselli et al. \cite{Grasselli2023} advance the field by deriving tight analytical bounds on entropy as a function of the violation of the Holz inequality, a multipartite generalization of the CHSH inequality introduced in \cite{Holz2020}. The Holz inequality was specifically designed for DI-CKA protocols. In the tripartite case, the inequality takes the form:

\begin{equation}
   \beta_H= \langle A_1 B_{+} C_{+} \rangle - \langle A_0 B_{-} \rangle - \langle A_0 C_{-} \rangle - \langle B_{0} C_{-} \rangle \overset{L}{\leq} 1 \overset{Q}{\leq} \frac{3}{2},
\end{equation}

where \( B_{\pm} = \frac{1}{2}(B_0 \pm B_1) \) and \( C_{\pm} = \frac{1}{2}(C_0 \pm C_1) \).

If Alice, Bob, and Charlie test this inequality and obtain an expected Bell value \(\beta_H\), the following tight analytical bound on the conditional entropy of Alice's outcome \(A_0\) can be derived:

\begin{equation}
    H(A_0|E) \geq 1 - h\left[\frac{1}{4}\left(\beta_H + 1 + \sqrt{\beta_H^2 - 3}\right)\right].
\end{equation}

Moreover, the authors demonstrate that the entropy bounds for the Holz inequality remain non-zero below the GME threshold set by the MABK inequality. This implies that GME might not be a strict requirement for certifying the privacy of a single party’s outcome when testing multipartite Bell inequalities.

\subsection{Numerical techniques} \label{subsec:numericaltechniques}
The main theoretical problem in QKD is calculating how much of a secret key can be extracted by a given protocol. A crucial practical issue is that the QKD protocols that are easiest to implement with existing technology do not necessarily coincide with the protocols that are easiest to analyze theoretically. Furthermore, existing analytical methods for calculating the key rate are highly technical and often limited in scope
to particular protocols, and invoke inequalities that introduce looseness into the calculation. Therefore, putting efforts into numerical methods, which are inherently more robust to device imperfections and protocol structure changes, is necessary. \\
At the technical level, the key rate problem is an optimization problem, since one must minimize the well-known entropic formula $H(A|E)$ over all states $\rho_{AB}$ that satisfy Alice’s and Bob’s experimental data .
\begin{equation}
   r:=\min_{\rho_{AB}\in\mathcal{C}} [H(R_A\mid E)-H(R_A\mid R_B)],
\end{equation}
where $\mathcal{C}$ is the set of all states $\rho_{AB}$ consistent with the observed experimental data. 
Coles et al. \cite{Coles2016} showed that the key rate $r$ can be lower bounded with the use of the dual problem by the following maximization problem
\begin{equation}
    r \geq \frac{\Theta}{\ln{2}}-H(A\mid B),
\end{equation}
where $\Theta=\max_{\vec{\lambda}}\left(-\|\sum_{i} A_{a_i|\Bar{x}} R(\vec{\lambda})A_{a_i|\Bar{x}}\|-\vec{\lambda}\cdot\vec{\gamma}\right)$ ($\Bar{x}$ is the key generating measurement) and $R(\vec{\lambda})=\exp(-\bm{I}-\vec{\lambda}.\vec{\Gamma})$. $\vec{\Gamma}=\{\Gamma_i\}$ where $\Gamma_i$ are bounded Hermitian operators dependent on the observed experimental data and $\vec{\lambda}=\{\lambda_i\}$ ($\lambda_i=\text{Tr}(\rho_{AB}\Gamma_i)$). The key rate also can be lower bounded by applying direct optimization (primal problem) \cite{Winick2018}  
\begin{equation}
r \geq \alpha - p_{\text{pass}}\text{leak}^{EC},     
\end{equation}
such that $\alpha=\min_{\rho\in\mathcal{C}}f(\rho)$ where $f(\rho)$ is a convex function of $\rho$ and $\text{leak}^{EC}$ denotes the number of bits Alice publicly reveals during error correction.  \\
To apply the EAT, as discussed in the previous section, a trade-off function must be computed that lower-bounds the amount of randomness produced in a single round. Existing results for the CHSH game \cite{pabgs09andSangouard.4} are highly specific to this case, with limited generalizability to other games. A particularly promising approach involves using SDP relaxations\cite{Navascues2008, NV15} provide valuable techniques for studying classical and quantum advantages in DI and SDI protocols (see also \cite{openqkdsecurity}). Often, these methods can provide exact solutions to the problems at hand. However, the complexity of these techniques imposes limitations, especially when studying protocols involving higher-dimensional quantum systems.  \\
\paragraph{Lower bounds on the min-entropy}
A straightforward way to derive numerical lower bounds for von Neumann entropy minimization is through the use of min-entropy, as demonstrated in \cite{Bancal2014, NietoSilleras2014}. The corresponding optimization of min-entropy can be formulated as a noncommutative polynomial over measurement operators. This problem can be relaxed into a semidefinite program (SDP) using the NPA hierarchy \ref{subsec:NPA}, which can then be solved efficiently. While this method provides a simple and effective way to lower bound the rates of various device-independent (DI) tasks, the min-entropy is generally much smaller than the von Neumann entropy. As a result, this approach often yields suboptimal outcomes. Therefore, to achieve optimal bounds, obtaining upper bounds on von Neumann entropy is both more efficient and essential.
\paragraph{Lower bounds on the conditional von Neumann entropy} \label{subsubsec:numeriacallowerbounds}
Tan et al. \cite{Tan2021} approach the DI security problem with a universal computational toolbox that directly bounds the von Neumann entropy using the complete measurement statistics of a device-independent cryptographic protocol.  Suppose that the protocol estimates parameters of the form $l_j=\sum_{abxy}c_{abxy}^{(j)}p(ab|xy)$ for
some coefficients $c_{abxy}^{j}$. These parameters could be Bell inequalities in a DI scenario. Thanks to the Naimark theorem, one can assume all measurements as projective measurements,  $P_{a\mid x}$ for Alice's side and $P_{b\mid y}$ for Bob's side, in a higher but finite dimension space. The task is to find lower bounds on $\inf H(A_0\mid E)$ such that $\langle L_j \rangle=l_j$ where $L_j=\sum_{abxy} c_{abxy} P_{a|x}\otimes P_{b|y}$ and the infimum takes place over $\psi_{ABE}$ and any uncharacterized measurements. The central result of Tan et al. \cite{Tan2021} is expressed as the following theorem
\begin{theorem} \label{Tan2021theorem}
For a DI scenario, the minimum value of $H(A_0|E)$ subject to constraints $\langle L_j \rangle=l_j$  is lower-bounded by
\begin{equation}
    \sup_{\vec{\lambda}}\left(\sum_{j}\lambda_j l_j-\ln{\left(\sup_{\substack{\rho_{AB},P_{a|x},P_{b|y}\\ \text{s.t. }\langle L_j\rangle_{\rho_{AB}} = l_j}}\langle K\rangle_{\rho_{AB}}\right)}\right),
    \label{Tan2021bound}
\end{equation}
where 
\begin{equation}
    K=T\left[\int_{\mathbb{R}}dt \beta(t) \left|\prod_{xy}\sum_{ab}e^{\kappa_{abxy}}P_{a|x}\otimes P_{b|y}\right|^2 \right]
\end{equation}
with $T[\sigma_{AB}]=\sum_{a}(P_{a|0}\otimes \mathbb{I}_B)\sigma_{AB}(P_{a|0}\otimes \mathbb{I}_B)$ , $\beta(t)=(\pi/2)(\cos(\pi t)+1)^{-1}$, and $\kappa_{abxy}=(1+it)\sum_j \lambda_j c_{abxy}^{(j)}/2$.
\end{theorem}
The previous best bound on \( H(A_0 \mid E) \) was established in \cite{pabgs09andSangouard.4} (see section\ref{sec:analyticalBounds}), where only the CHSH value was used instead of the full probability distribution. In contrast, the method proposed in Theorem \ref{Tan2021theorem} directly bounds \( H(A_0 \mid E) \) using the complete input-output distribution. This approach yields results that are comparable to or slightly better than the bound in \cite{pabgs09andSangouard.4}. It also demonstrates that in scenarios with limited detection efficiency, better bounds on \( H(A_0 \mid E) \) can be achieved by considering the full distribution rather than relying solely on the CHSH value. This suggests that optimizing experimental parameters to maximize the CHSH value may not be the most effective strategy; instead, optimizing a different Bell value could lead to further improvements.
 \\
The numerical results presented in \cite{Tan2021} are very promising, providing significant improvements in the rates when compared to the min-entropy approach and also improving over the analytical results \cite{pabgs09andSangouard.4}.
However, the approach is relatively computationally intensive requiring the optimization of a degree 6 polynomial in the simplest setting. To reduce the complexity, Brown et al. \cite{Brown2021} take a different approach, defining a new
family of quantum Rényi divergences, the iterated mean (IM) divergences which for the sequence $\alpha_k=1+\frac{1}{2^k -1}$ 
for $k\in \mathbb{N}$ is defined as 
\begin{equation}
    D_{(\alpha_k)}(\rho||\sigma):=\frac{1}{\alpha_k -1}\log{Q_{(\alpha_k)}(\rho||\sigma)},
\end{equation}
where 
\begin{equation}
    Q_{(\alpha_k)}(\rho||\sigma)=\max_{V_1,\cdots,V_k,Z} \alpha_k \text{Tr}\left[\rho\frac{(V_1+V_{1}^{*})}{2}\right]-(\alpha_k -1)\text{Tr}[\sigma Z]
\end{equation}
such that 
\begin{equation}
    V_1+V_{1}^*\geq 0, \enspace\enspace \frac{V_2+V_{2}^{*}}{2}\geq V_{1}^{*}V_1, \enspace\enspace \cdots, \enspace\enspace Z\geq  V_{k}^{*}V_k .
\end{equation}
The crucial property that makes these divergences well-adapted for device-independent optimization is the fact that
$Q_{(\alpha_k)}(\rho||\sigma)$ has a free variational formula as a supremum of linear functions in $\rho$ and $\sigma$. Given a bipartite quantum state $\rho_{AB}$  and a divergence $D_{(\alpha_k)}(\rho||\sigma)$ the corresponding conditional entropy can be defined as $H_{\alpha_k}^{\downarrow}=-D_{(\alpha_k)}(\rho_{\rho_{AB}}||I_A\otimes\rho_B)$ together with its optimized version $H_{\alpha_k}^{\uparrow}=\sup_{\sigma_B} -D_{(\alpha_k)}(\rho_{\rho_{AB}}||I_A\otimes\sigma_B)$,then
the following theorem gives an explicit characterization of $H^{\uparrow}$ for the iterated mean divergences 
\begin{theorem} (BFF method)
\label{theorem:BFF21}
    For a bipartite state $\rho_{AB}$ 
    \begin{equation}
        H_{\alpha_k}^{\uparrow}(A|B)_{\rho_{AB}}=\frac{1}{1-\alpha_k}\log{Q^{\uparrow}_{(\alpha_k)}},
    \end{equation}
    where 
    \begin{equation}
    Q^{\uparrow}_{(\alpha_k)}=\max_{V_1,\cdots,V_k}\left(\text{Tr}[\rho_{AB}\frac{V_1+V_{1}^{*}}{2}]\right)^{\alpha_k},
    \end{equation}
    such that
    \begin{equation*}
    \text{Tr}[V_{k}^{*}V_k]\leq I_B,\enspace\enspace  V_1+V_{1}^{*} \geq 0, \enspace\enspace \text{and} \enspace\enspace \begin{pmatrix}
        I & V_i \\
        V_{i}^{*} & \frac{V_{i+1}+V_{i+1}^{*}}{2}
    \end{pmatrix}\geq 0,
     \end{equation*}
     where in the last constraint $1\leq i\leq k-1$.
\end{theorem}
Brown et al. \cite{Brown2021} showed that for each $\alpha_k$ and any pair of $\rho$ and $\sigma$, $D_{(\alpha_k)}(\rho||\sigma)\geq \tilde{D}_{(\alpha_k)}(\rho||\sigma)$
Where $\tilde{D}_{(\alpha)}$ denotes the sandwiched Renyi divergence \cite{MuellerLennert2013} which result in $H_{\alpha_k}^{\uparrow}(A|B)\leq \tilde{H}_{\alpha_k}^{\uparrow}(A|B)\leq H(A|B)$  for all $\alpha >1$
that can be used to compute lower bounds on the rates of various device-independent protocols.
 \\
Both above methods improve upon the min-entropy method, but neither has been shown to give tight bounds on the actual rate of a protocol, and in general, there appears to be significant room for improvement. As such, the question remains as to whether
one can give a computationally tractable method to compute tight lower bounds on the rates of protocols. To address this question, Brown et al. \cite{Brown2024} derived a converging sequence of upper bounds on the relative entropy between
two positive linear functionals on a von Neumann algebra and demonstrated how to use this sequence of upper bounds to derive a sequence of lower bounds on the conditional von Neumann entropy. The main technical result of their work is the following theorem 
\begin{theorem} \label{Brown2024theorem}
Assume that $\rho$ and $\sigma$ are two positive operators on a finite-dimensional Hilbert space and  $\lambda >0$ is such that $\rho\leq\lambda\sigma$. Then for any $m\in \mathbb{N}$ there exists a choice of $t_1,\cdots t_m \in(0,1]$ and $\omega_1,\cdots,\omega_m >0$ such that 
\begin{eqnarray}
    D(\rho||\sigma)\leq -c_m &-&\sum_{i=1}^{m-1}\frac{\omega_i}{t_i \ln{2}}\inf_{Z}\text{Tr}[\rho (Z+Z^{*}+(1-t_i)Z^{*}Z]+t_i\text{Tr}[\sigma ZZ^{*}], \\
    \text{s.t.} \enspace\enspace ||Z||&\leq& \frac{3}{2}\max \{\frac{1}{t_i},\frac{\lambda}{1-t_i}\},
\end{eqnarray}
where $c_m=\text{Tr}[\rho](\sum_{i=1}^m \frac{\omega_i}{t_i\ln{2}}-\frac{\lambda}{m^2\ln{2}})$. As $m\rightarrow \infty$, the right-hand side of the above equality converges to $D(\rho||\sigma)$.  
\end{theorem}
The theorem \ref{Brown2024theorem} provides a convergent sequence of upper bounds on the relative entropy in the form of an optimization problem and can turn into SDP lower bounds on the rate of DI protocols. For the case of DI-QKD and for the devices are constrained by quantum theory the following noncommutative polynomial optimization problem gives a lower bound $H(A|E)$ 
\begin{eqnarray}
    c_m+\inf \sum_{i=1}^{m-1} \frac{\omega_i}{t_i\ln{2}}\sum_{a}\bra{\psi}(M_{a|x=x^{X^{*}}}(Z_{a,i}+Z^{*}_{a,i}+(1-t_i)Z_{a,i}Z^{*}_{a,i})+t_i Z_{a,i}Z^{*}_{a,i}\ket{\psi}
\end{eqnarray}
such that
\begin{eqnarray}
    \sum_{abxy}c_{abxy}^{j}\bra{\psi}M_{a|x}M_{b|y}\ket{\psi}\geq v_j,  
\end{eqnarray}
and $[M_{a|x},N_{b|y}]=[M_{a|x},Z^{(*)}_{b,i}]=[N_{b|y},Z^{*}_{a|i}]=0$ where $M_{a|x}$ and $N_{b|y}$ are POVM elements of Alice and Bob measurements respectively which are bounded operators together with $Z_{a,i}$. \\
By applying the NPA hierarchy \ref{subsec:NPA}, this optimization can be relaxed into a sequence of SDPs that yield a converging series of lower bounds on the optimal value. This, in turn, provides a lower bound on the protocol's rate. When calculating key rates for DI-QKD, a significant improvement (below $0.8$) in the minimum detection efficiency required to generate a secret key can be achieved, bringing it well within the capabilities of current device-independent experiments. Araujo et al. \cite{Araujo2023} adapt the same SDP hierarchy to the case of QKD with characterized devices.

\subsection{Upper bounds} \label{subsec:upperbounds}
Up to this point, only the lower bounds on key rates have been explored for all the protocols mentioned. In this section, we address a different question:\\
\textit{What is a non-trivial upper bound on the secret key rate that can be extracted from a DI-QKD protocol?} \\
Understanding upper bounds on key rates is crucial from a practical perspective, as it reveals the inherent limitations of an entire class of protocols rather than focusing on individual protocols and analyzing them in isolation.\\
This question was first posed by Kaur et al. \cite{Kaur2020}, who introduced information-theoretic measures of nonlocality, termed intrinsic nonlocality and quantum intrinsic nonlocality. They demonstrated that these measures serve as upper bounds for DI-QKD protocols, specifically against no-signaling and quantum adversaries, respectively. Instead of using intrinsic nonlocality, Arnon-Friedman et al. \cite{ArnonFriedman2021} examined a closely related information-theoretic quantity known as intrinsic information, which they employed to derive an upper bound on the key rates of DI-QKD protocols.\\
Winczewski et al. \cite{Winczewski2022} initiated a systematic study of upper bounds on secret key rates within the no-signaling DI scenario. They introduced a computable function, termed squashed nonlocality, as one such bound. Their numerical analysis suggests that quantum devices with two binary inputs and two binary outputs can extract only a limited amount of key. Moreover, they found that isotropic devices with less than $80\%$ of the Popescu-Rohrlich box weight are generally key-undistillable.\\
Since DI-QKD has a higher security demand than QKD, one has the trivial bounds $r^{\text{DI}}\leq r^{\text{DD}}$ ($r^{\text{DD}}$ is the key rate of a standard device-dependent QKD protocol). Christandl et al. \cite{Christandl2021} use this fact to find an upper bound on a DI-QKD as follows: assume that the POVMs $\{A_{a|x}\}$ and $\{B_{b|y}\}$ are chosen such that the key-rate $r$ is optimal, there might be different measurement $A_{a|x}^{\prime}$  and $B_{b|y}^{\prime}$  and state $\rho^{\prime}$ leading to the same distribution   
\begin{equation}
    p(a,b\mid x,y):=\mathrm{Tr}[(A_{a|x}\otimes B_{b|y})\rho]=\mathrm{Tr}[(A_{a|x}^{\prime}\otimes B_{b|y}^{\prime})\rho^{\prime}],
\end{equation}
the above equality is shown as $(\mathcal{M},\rho)\equiv(\mathcal{M}^{\prime},\rho^{\prime})$. Since the maximal achievable key rate for $\rho$ is also achievable for $\rho^{\prime}$ ($r^{\text{DI}}(\rho)\leq r^{\text{DI}}(\rho^{\prime})$ ) then combining it with $r^{\text{DI}}(\rho^{\prime}) \leq r^{\text{DD}}(\rho^{\prime}) $ the following bound can be obtained \cite{Christandl2021} 
\begin{equation}
    r^{\text{DI}}(\rho)\leq \sup_{\mathcal{M}}\inf_{\substack{(\mathcal{M}^{\prime},\rho^{\prime})\\(\mathcal{M},\rho)\equiv(\mathcal{M}^{\prime},\rho^{\prime})}}r^{\text{DD}}(\rho^{\prime}), 
\end{equation}
Consider that $\rho$ is a PPT state $\rho^{\Gamma}\geq 0$ ($\Gamma$ denotes partial transpose) because the transpose of a POVM element is a POVM element then one can find 
\begin{equation}
    r^{\text{DI}}(\rho)\leq \min \{r^{\text{DD}}(\rho),r^{\text{DD}}(\rho^{\Gamma})\},
\end{equation}
The significance of the above result,  can be seen by an example. Consider the $2d\times 2d$ state $\sigma_d$ as
\[
\sigma_d = \frac{1}{2} \begin{bmatrix}
(1-p)\sqrt{XX^{\dagger}} & 0 & 0 & (1-p)X \\
0 & pY & 0 & 0 \\
0 & 0 & pY & 0 \\
(1-p)X & 0 & 0 & (1-p)\sqrt{XX^{\dagger}}
\end{bmatrix}
\]
where $Y=\frac{1}{d}\sum_{i=0}^{d-1}\ket{ii}\bra{ii}$ and $X=\frac{1}{d\sqrt{d}}\sum_{i,j=0}^{d-1}u_{ij}\ket{ij}\bra{ij}$, where $u_{ij}$'s are the elements of a unitary matrix such that $|u_{ij}|=\frac{1}{d}$.   
For this state, it has been found in \cite{Christandl2021} that for the case of $d=2^{20}$, $r^{\text{DD}}(\sigma_{2^{20}})\geq 0.98$ and $r^{\text{DD}}(\sigma_{2^{20}}^{\Gamma})\leq \frac{1}{2^{10}+1}$. Therefore, we see that whereas in QKD, the obtained bit in this setting is secure, the upper bound
tells us that this bit is not secure in a device-independent setting. Therefore the state, and any of its parts, cannot be tested independently of the device. This example can be also regarded as supported evidence for the revised Peres conjecture for DI-QKD in \cite{ArnonFriedman2021} which states that bound entangled states cannot be used as a resource for DI-QKD. \\
Kaur et al. \cite{Kaur2022} develop the above bounds by going beyond PPT states and arrive at the following upper bound for general DI-QKD protocols based on the relative entropy of entanglement \cite{Vedral1997} 
\begin{equation}
     r^{\text{DI}}(\rho)\leq (1-p) \inf_{(\sigma^{NL},\mathcal{N})=(\rho^{NL},\mathcal{M})} E_R(\sigma^{NL})+p \inf_{(\sigma_L,\mathcal{N})=(\rho^{L},\mathcal{M})} E_R(\sigma^L),
     \label{KHDbound}
\end{equation}
where $\rho=(1-p)\rho^{NL}+p\rho^{L}$ such that $(\sigma^{L},\mathcal{N}),(\rho^{L},\mathcal{M})\in \text{LHV}$ where LHV denoted the set of devices with locally realistic hidden variable models. For the CHSH-based protocols, with $\omega$ denoting the CHSH violation, the following analogous upper bound can be obtained 
\begin{equation}
    r^{\text{DI}}(\rho)\leq (1-p) \inf_{\omega(\sigma^{NL},\mathcal{N})=\omega(\rho^{NL},\mathcal{M})} E_R(\sigma^{NL}),
\end{equation}
One implication of this bound is that all PPT states satisfy the CHSH inequality, resulting in a zero device-independent key rate for CHSH-based protocols. This, in turn, proves the revised Peres conjecture for such protocols. It is important to note that while there exist bound entangled states from which a private key can be distilled in device-dependent protocols, these states are useless for DI-QKD in CHSH-based protocols. Furthermore, for CHSH-based protocols, they show that the convex hull of above-mentioned bounds is a tighter upper bound on the device-independent key rates.
\paragraph{Upper bounds based on convex-combination attacks} 
Farkas et al. \cite{farkas2021bell} studied the problem of upper bounding the key rate of a DI-QKD problem by applying a convex combination attack. In a convex combination attack, Eve distributes local deterministic correlations with certain probabilities that give rise to a local correlation $p_{AB}^{L}(a,b|x,y)$ with overall probability
$q_L$, and a nonlocal quantum correlation $p_{AB}^{NL}(a,b|x,y)$ with probability $1-q_L$. Eventually, the observed correlation of Alice and
Bob takes the form 
\begin{equation}
p_{AB}(x,y)=q_L p_{AB}^{L}(a,b|x,y)+(1-q_L)p_{AB}^{NL}(a,b|x,y),  
\label{ccattack}
\end{equation}
Since Alice and Bob announce their input for each round, Eve knows their results in all rounds in which she distributes a local correlation, so to gain more information about the key, she should maximize $q_L$ for the given observed correlation \eqref{ccattack}. By denoting $e$ as the classical variable representing Eve's knowledge which $e=(a,b)$ when a local correlation was distributed and $e=?$ for the cases of nonlocal correlations were distributed.
The corresponding joint probability of \eqref{ccattack} distribution among Alice, Bob, and Eve is then written as
\begin{equation}
    p_{ABE}(x,y,e)=q_L p_{AB}^{L}(a,b|x,y)\delta_{e,(a,b)}+(1-q_L)p_{AB}^{NL}(a,b|x,y)\delta_{e,?},
    \label{pABE}
\end{equation}
where $\delta$ is the Kronecker delta. Then, the following bound on the key rate can be obtained 
\begin{equation}
    r\leq \sum_{x,y}p_{xy}I_{xy}(A:B\downarrow E),
    \label{Farkasbound}
\end{equation}
where $p_{xy}$ is the probability of Alice and Bob choosing the settings $x$ and $y$ and $I_{xy}(A:B\downarrow E)$ is the intrinsic information of the distribution \eqref{pABE} which is defined as $I(A,B\downarrow E)=\min_{E\rightarrow F}I(A:B|F)$ where $I(A:B|F)$ is the conditional mutual information and the minimization is taken over all stochastic maps $E\rightarrow F$ that map the variable $E$  to a new variable $F$ such that the alphabet size of $F$ is at most that of $E$.\\
The upper bound \eqref{Farkasbound} resulted from a well-established result in classical cryptography that the asymptotic rate extractable from a distribution $p_{ABE}(a,b,f)$ is bounded by intrinsic information \cite{maurer1997isit}. \\
 Additionally, Farkas et al. \cite{farkas2021bell} investigate the problem by applying the upper bound on a standard protocol implemented on a two-qubit Werner state with visibility $v$ using an arbitrary number of projective measurements. They showed that for a range of visibilities for which the Werner state is nonlocal the upper bound on the key rate is zero. This means that nonlocal quantum states exist that cannot be used for standard DI-QKD with projective measurements. Therefore, \textit{the bell nonlocality is generally insufficient for the security of standard device-independent quantum key distribution protocols}. \\
 \begin{figure*}
    \centering
    \includegraphics[width=0.5\linewidth]{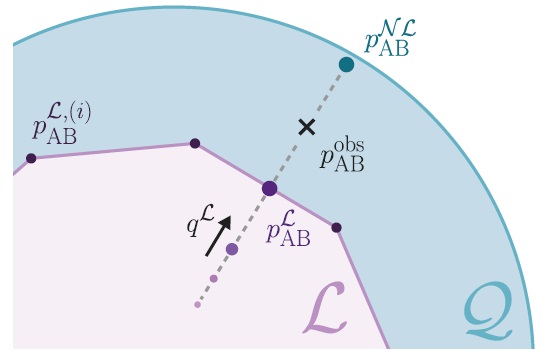}
    \caption{Figure from \cite{Lukanowski2023}. Geometry formulation of the CC attack.}
    \label{fig:Lukanowski2023}
\end{figure*}
 To find the maximum value of \( q_L \) in \eqref{ccattack}, Lukanowski et al. \cite{Lukanowski2023} provided a geometric interpretation of the CC attack, which describes its optimization in terms of a linear program for determining the tightest upper bound on the key rate. A given correlation \( p(a,b|x,y) \) represents a point within the probability space (Figure \ref{fig:Lukanowski2023}). Maximizing \( q_L \) involves identifying two additional points that are collinear with this point: \( p_{AB}^{L}(a,b|x,y) \), located within the local set \( \mathcal{L} \), and \( p_{NL} \), which lies outside the local set but within the quantum set \( \mathcal{Q} \). The goal is to minimize the ratio of the distances from \( p_{AB}^{L}(a,b|x,y) \) to \( p_{AB}(a,b|x,y) \) and from \( p_{AB}^{NL}(a,b|x,y) \) to \( p_{AB}(a,b|x,y) \).\\
The assertion that Eve has perfect knowledge of all outcomes when distributing a local correlation to Alice and Bob is based on the fact that the local set \( \mathcal{L} \) forms a convex polytope within the probability space. As such, any local correlation \( p_{AB}^{L}(a,b|x,y) \) can be expressed as a combination of the extremal points of the polytope, which can be represented by the vector \( \boldsymbol{p}^L = (p_i^L(a,b|x,y))_i \). Furthermore, in this geometric construction, maximizing the local weight \( q_L \) results in the optimal local correlation \( p_{AB}^{L}(a,b|x,y) \) lying on the boundary of the local polytope \( \mathcal{L} \), meaning it must always reside on one of the polytope's facets. \\
Analogous to \( \boldsymbol{p}^L \), one can define the vector \( \boldsymbol{q}^L = (q_i^L(a,b|x,y))_i \), which represents the probabilities assigned by Eve to each local correlation in the CC attack. Similarly, the average non-local correlation that Eve distributes can be modeled as a mixture of preselected non-local quantum correlations, forming the vector \( \boldsymbol{p}^{NL} = (p_i^{NL}(a,b|x,y))_i \), with the corresponding probability vector \( \boldsymbol{q}^{NL} = (q_i^{NL}(a,b|x,y))_i \) indicating the probabilities assigned by Eve to each non-local correlation.\\
To optimize the CC attack, Eve seeks a probability vector \( \boldsymbol{q} = \boldsymbol{q}^L \oplus \boldsymbol{q}^{NL} \), ensuring that local correlations are distributed as frequently as possible. This requires solving the following linear program, which maximizes the overall probability of sending any local boxes:

\[
\boldsymbol{q}_{CC}[\boldsymbol{p}^{NL}, p(a,b|x,y)] = \arg\max \sum_i q_i^L,
\label{pl-linearprogram}
\]
subject to the following constraints:
\[
\boldsymbol{q}^L \cdot \boldsymbol{p}^L + \boldsymbol{q}^{NL} \cdot \boldsymbol{p}^{NL} = p(a,b|x,y), \nonumber
\]
\[
\sum_i q_i^L + \sum_j q_j^{NL} = 1, \nonumber
\]
\[
\forall i,j : 0 \leq q_i^L, q_j^{NL} \leq 1. \nonumber
\]
The first condition is nothing but eq. \eqref{ccattack} and the other constraints ensure $\boldsymbol{q}$ is a valid probability vector.
The set of extremal local correlation, $\boldsymbol{p}^L$, is a predetermined fixed collection fixed by the scenario.  \\
To apply the CC attack and upper bound the key rate in a DI-QKD protocol, one must first specify the ideal correlation that would be shared by the parties in the absence of noise, denoted by \( Q_{AB}(a,b|x,y) \). However, due to practical imperfections—such as finite detection efficiency $\eta$ and visibility $v$—the actual noisy correlation observed is \( p_{AB}(a,b|x,y) \). This correlation is then decomposed within the attack into local and nonlocal parts as equation \eqref{ccattack}. \\
The method in \cite{Lukanowski2023} works such that the eavesdropper must specify in advance the set of nonlocal correlations \( p_{NL} \) to be used in the convex decomposition and then apply the linear program \eqref{pl-linearprogram} to determine the local contribution. For a CHSH protocol involving maximally entangled states with finite detection efficiency \(\eta\), the maximum local weight can be analytically determined as
\[
q^L = (1 - \eta)(1 + (3 + 2\sqrt{2})\eta) \quad \text{for} \quad \eta \geq \eta_{\text{loc}},
\]
where \(\eta_{\text{loc}} = 2(\sqrt{2} - 1) \approx 82.8\%\) is the detection efficiency threshold, below which (\(\eta < \eta_{\text{loc}}\)) the correlation \( p(a, b | x, y) \) becomes local. \\
As a result, the following bounds can be derived for the one-way and two-way protocols:
\begin{align}
    r_{\text{1-way}} &\leq (3 + 2\sqrt{2})\eta^2 - 2(1 + \sqrt{2})\eta - \frac{\eta}{2}h(\eta) - (1 - \eta)h\left(\frac{\eta}{2}\right), \\
    r_{\text{2-way}} &\leq \eta \left( 2(1 + \sqrt{2})\eta - 2\sqrt{2} - 1 \right) \left( 1 - h\left( \frac{1 - \eta}{1 - 2(1 + \sqrt{2})(1 - \eta)} \right) \right).
\end{align}
There are critical values for one-way (\(\eta_{\text{1-crit}} \approx 89.18\%\)) and two-way (\(\eta_{\text{2-crit}} \approx 85.36\%\)) protocols, below which the key rates become negative. This demonstrates that, for detection efficiencies in the range \(\eta_{\text{loc}} \leq \eta \leq \eta_{\text{i-crit}}\), no DI-QKD protocol is feasible, even though the shared correlation remains nonlocal. The same result can be extended to the case of finite visibility (\(v < 1\)); i.e., there are critical visibilities that introduce nonlocal intervals in which no DI-QKD protocol is possible. \\
Zhang et al. \cite{Zhang2022} applied the CC attack to a two-way protocol by optimizing the non-local points in the CC decomposition using the NPA-hierarchy \ref{subsec:NPA}. They demonstrated that noise reduction can be achieved by employing the B-step procedure \cite{gottesman2003}.

\section{Semi-Device-independent Quantum Key Distribution}\label{chap5}
The implementation of fully DI-QKD schemes is hampered by the stringent hardware requirements that currently prohibit a reasonable key rate over practical distances.  One way to make DI methods more viable is to slightly relax the notion of \textit{device independence}, and establish a minimal set of reasonable assumptions. This approach, named semi-device-independent (SDI) QKD reduces hardware demands, so that a more reasonable key (or randomness) rate can be achieved with current technology \cite{PB11}.
Examples of these additional assumptions include (i) an upper bound on the system's dimension \cite{PB11}, (ii) shared randomness \cite{LBLLBMZB15}, or (iii) honest construction of part of the device \cite{CZM15}. Still other approaches within the general semi-DI philosophy include measurement-device independent (MDI) QKD \cite{Lo2012} and one-sided DI-QKD (1SDI-QKD) \cite{branciard2012one}. In the following sections, we present these methods.

\begin{table*}[ht]
\small
\centering
\caption{Overview of semi-device-independent QKD (SDI-QKD) protocols.}
\begin{tabular}{
@{} l >
{\raggedright\arraybackslash}p{2.3cm} >{\raggedright\arraybackslash}p{5.cm} >{\raggedright\arraybackslash}p{5.cm} @{}}
\toprule
\textbf{Protocol} & \textbf{Main Assumption(s)} & \textbf{Key Features} & \textbf{Security Notes} \\
\midrule
P\&M SDI-QKD & Known dimension of quantum states & One-way communication; black-box devices for Alice and Bob & Uses dimension witnesses; robust against individual attacks \\
\midrule
RDI-QKD & Trusted sender; bounded overlaps $\gamma_{ij} = \langle \psi_i | \psi_j \rangle$ & No assumptions on Bob's device & Resilient to side-channel attacks on Bob \\
\midrule
MDI-QKD & Untrusted measurement devices; trusted sources & Bell-state measurement by untrusted third party & Removes detector-side channel attacks \\
\midrule
DDI-QKD & Trusted preparation; black-box detection system & Uses single-photon Bell-state measurement & Simpler than MDI-QKD, but vulnerable to preparation attacks \\
\midrule
1SDI-QKD & One trusted party; entanglement-based & Demonstrates quantum steering; asymmetric trust & Security linked to steering inequalities \\
\bottomrule
\end{tabular}
\label{tab:sdi-qkd-summary}
\end{table*}

\subsection{Prepare-Measure semi-device-independence}
In DI-QKD protocols, the security is based on testing non-locality between two parties. One question that can arise is whether such a strong form of security could be established for prepare and measure scenarios. This question is especially important since many commercially available QKD systems operate in one-way configurations, in which a transmitter (Alice) prepares a quantum state and sends it to a receiver (Bob). This question was first addressed by Paw\l{}owski and Brunner in 2011 \cite{PB11}, in a scenario which they called semi-DI (SDI). In their approach, the Hilbert space dimension of the quantum system is known, but the quantum preparation and measurement devices are uncharacterized,  such that the devices of Alice and Bob can be seen as black boxes. 
The assumptions of the SDI protocols are the following:

\begin{procedurelist}{SDI-QKD scenario \cite{PB11}}
	\item Alice’s black box is a “state preparator” which
	has the freedom to choose among a certain set of preparations
	$\rho_a\in \mathcal{B}(\mathbb{C}^\mathrm{d})$ with $a\in \{1,\dots,N\}$ unentangled from
	Eve, but knows nothing about
	these quantum states apart from their dimensionality. She sends the prepared state to Bob.
	\item  Bob’s  measurement device is a black box. He can choose to perform an uncharacterized measurement $M_y$ with
	$y\in \{1,\dots,m\}$ and gets the outcome $b\in \{1,\dots,k\}$. 
	\item The boxes may feature shared classical variables $\lambda$, known to Eve, but uncorrelated from the choice of preparation (measurement) made by Alice (Bob).
	\item After repeating this procedure many times, Alice and Bob
	can estimate the probability distributions $P(b|a,y)=\mathrm{Tr}(\rho_a M_y^b)$
	which denotes the probability of Bob finding outcome $b$ when
	he performed measurement $M_y$ and Alice prepared $\rho_a$. 
	\item The protocol is restricted to individual attacks.
\end{procedurelist}
Notice that if Alice’s preparations were entangled with Eve’s system, then the communication capacity would be effectively
doubled using dense coding \cite{densecoding92}.\\
To prove the security of SDI-QKD from the table of probabilities $P(b|a,y)$,  a \textit{dimension witness} can be used to estimate the minimal dimension of the state sent from Alice to Bob.  Introduced by Gallego et al. \cite{GBHA10}, a dimension witness is defined as $W=\sum_{a,y,b}w_{aby}P(b|ay)$, where the real coefficients $w$ are chosen such that one can  derive lower bounds on the dimension of classical or quantum systems that is necessary to reproduce the measurement data. For example, in the simplest case of a two-dimensional system with three preparations and two binary measurements, for a classical system (i.e., a bit), they derived the witness $I_3$:
	\begin{equation}
		I_3=|E_{11}+E_{12}+E_{21}-E_{22}-E_{31}| \leq 3,
		\label{I3}
	\end{equation}
	where $E_{ij}=P(b=+1|ij)-P(b=-1|ij)$. To surpass the upper bound with a classical system, dimension of at least three (trits) is required, giving an algebraic maximum of $I_3=5$. Looking now to the quantum case, the same witness can be employed. After solving the maximization problem, one finds that $\mathrm{max}_{\ket{\psi}\in \mathcal{H}_2}I_3=1+2\sqrt{2}$, and the witness takes the form $I_3\leq{1+2\sqrt{2}}$.\\
 The first four terms in \eqref{I3} can be seen as a CHSH inequality (whose maximum is $2\sqrt{2}$), and the maximization does not involve the fifth term ($E_{31}$), which can be set to $-1$. Thus, the violation of \eqref{I3} corresponds to the violation of the CHSH inequality and can be seen as a device-independent protocol by estimating the CHSH inequality. Moreover, the witness $I_3$ can be used to distinguish between bits and qubits. If the dimension witness $W$ satisfy the following condition
\begin{equation}
	C_d < W \le Q_d,
\end{equation}
where $C_d$ and $Q_d$ are the classical and quantum bounds respectively, for $\dim\mathcal{H}_A=d$. Specifically, suppose Alice's device creates $d$--dimensional quantum systems, a value $W>C_d$  means that it becomes infeasible to replicate the quantum data table using $d$--dimensional classical systems, or equivalently quantum states emitted by Alice's device that are \textit{orthogonal} to Bob's measurements. The inability to replicate the data table using $d$--dimensional classical systems, witnessed by $W>C_d$, confirms that Eve cannot access the full information about the system. Relaxing the assumption on the dimension would enable Eve to use a classical system with sufficient higher dimensions to reproduce such a table.

To prove the security of SDI-QKD in ~\cite{PB11}, following the geometrical method in \cite{GBHA10}, the authors utilize a dimensional witness as the main tool to assess the security of SDI-QKD. Consider that Alice's device prepares qubits and is limited to four specific preparations ($N = 4$),  denoted by two bits, $a_0$ and $a_1$, while Bob's device can perform two binary measurements, and they can evaluate the correlators $E_{a_0a_1,y}$. They can evaluate a dimension witness of the form
	\begin{equation}
S=E_{00,0}+E_{00,1}+E_{01,0}-E_{01,1}-E_{10,0}+E_{10,1}-E_{11,0}-E_{11,1} \leq 2.
		\label{PBwitness}
	\end{equation}

Applying this dimension witness to the states and measurements of the BB84 protocol yields \( S = 2 \), demonstrating the insecurity of the BB84 scheme. However, this conclusion is not limited to BB84, and in fact any protocol that utilizes the same states and measurements as BB84, such as the SARG protocol \cite{scarani2004quantum}, is also insecure when viewed from this perspective. Rather, to obtain a positive key rate in the SDI scenario, Bob must perform measurements in a basis that is rotated with respect to the BB84 bases, as will be discussed below. 
\paragraph{Security of SDI-QKD} 
A secret key can be extracted if a positive value for the key rate \( r = I(A : B) - I(A : E) \) is obtained, where \( I(A:X) = \sum_j{1 - h(P_{X}(a_{y_j}))} \) represents the mutual information. Consequently, the sufficient condition for security is expressed as follows:
\[
I(A:B) > I(A:E) \Longrightarrow P_B > P_E,
\]
where \( P_X = \frac{1}{2}(p_X(a_0) + p_X(a_1)) \) denotes the average probability of party \( X \) correctly guessing  when Alice sends the state \( \rho_{a_0 a_1} \), based on the two random bits \( a_0 \) and \( a_1 \) that she generates. By evaluating Bob's success probability, it can be established that \( P_B \) is a function of \( S \) as given by:
\[
P_B = \frac{S + 4}{8}.
\]
Now, consider the following scenario: Alice receives an \( n \)-bit string as input, and Bob is tasked with guessing the value of a function from the set \( \{F_n\}_n \) (where \( \{F_n\}_n \) represents all Boolean functions on \( n \)-bit strings) after receiving \( s \) qubits from Alice. The average probability of Bob's success is bounded above by:
\[
P_n \leq \frac{1}{2} \left( 1 + \sqrt{\frac{2^s - 1}{2^n - 1}} \right).
\]
For \( n = 2 \), the optimal probability of guessing a function \( F_n \) or its negation is equivalent. Thus, when Alice sends a single qubit to Bob (\( s = 1 \)), we have:
\[
P_B(a_0) + P_B(a_1) + P_B(a_0 \oplus a_1) \leq \frac{3}{2} \left( 1 + \frac{1}{\sqrt{3}} \right),
\]
which also holds when Bob collaborates with Eve. By utilizing the relationships \( P_{BE}(a_i) \geq P_B(a_i) \) and \( P_{BE}(a_i) \geq P_E(a_i) \), along with the inequality:
\[
P_{BE}(a_0 \oplus a_1) \geq P_{BE}(a_0, a_1) \geq P_{BE}(a_0) + P_{BE}(a_1) - 1,
\]
one can derive the following equation:
\[
P_{BE}(a_0) + P_{BE}(a_1) + P_{BE}(a_0 \oplus a_1) \geq 2P_B(a_0) + 2P_E(a_1) - 1.
\]
Using the earlier inequality, this leads to:
\[
P_B(a_0) + P_E(a_1) \leq \frac{5 + \sqrt{3}}{4}.
\]
A similar inequality can be derived by interchanging \( a_0 \) and \( a_1 \). This illustrates that when Eve attempts to guess a different bit than Bob, she will inevitably disturb Bob's statistical outcomes. From inequality \( P_B(a_0) + P_E(a_1) \leq \frac{5 + \sqrt{3}}{4} \) and its symmetry with respect to \( a_0 \) and \( a_1 \), we conclude:
\[
P_B + P_E \leq \frac{5 + \sqrt{3}}{4}.
\]
This implies that \( P_B > P_E \) if:
\[
P_B > \frac{5 + \sqrt{3}}{8} \approx 0.8415.
\label{Pbideal}
\]
When Bob uses measurement operators $(\sigma_x \pm \sigma_z)/\sqrt{2}$, $P_B\sim0.8536$ and the  key rate is 
\[
r = I(A : B) - I(A : E) \approx 0.0581.
\]
An additional conclusion can be drawn from the discussion above. In DI-QKD, nonlocality is necessary but not sufficient \cite{farkas2021bell} (see \ref{subsec:upperbounds}). Similarly, we can deduce the analogous result: \textit{mere violation of a dimension witness ($S > 2 \Longrightarrow P_B>3/4$) is not sufficient to guarantee the security of SDI-QKD}. \par
In the ideal scenario where perfect detectors are assumed, meaning all systems leaving Alice's laboratory are detected by Bob, $P_B$ is the sole security parameter. However, in the presence of losses, the average detection efficiency of Bob's detectors, denoted as $\eta_B$, becomes an additional security parameter. It is crucial to define how the parties handle rounds where no particle is detected. 
Chaturvedi et al. \cite{Chaturvedi2018} chose the simplest case, where no-detection rounds are discarded from the statistics. This choice allows the parties to estimate the average success probability close to the optimal one in \cite{PB11}. If we split Bob's detection efficiency as $\eta_B = \eta + \eta^{\prime}$, where $\eta = P(\text{Click}|e \neq b)$ represents the detection efficiency when Eve's and Bob's inputs are different, and $\eta^{\prime}$ corresponds to the case where their inputs are the same, then Eve maximizes $\eta^{\prime}$ because she wants Bob's device to return outcomes as often as possible. Since Eve has no control over Bob's setting, this leads to $\eta_B = \frac{1 + \eta}{2}$. Therefore, the condition for establishing a secret key is translated to
\begin{equation}
    P_B(\eta)>P_E(\eta).
\end{equation}
Chaturvedi et al. \cite{Chaturvedi2018} studied security against two types of quantum eavesdroppers, those with and without access to quantum memory. They showed that, in the general case where Eve could control both Alice's and Bob's devices, the security condition for both cases (with and without memory) is 
\begin{equation}
P_B > \frac{1}{2}\left(1+\frac{1}{1+\eta}\right).
\end{equation}

Moreover, they considered a minimal characterization of the preparation device, with the restriction that, while Eve can choose the states that leave Alice’s laboratory, she cannot alter them during the protocol. This is a reasonable assumption, as manipulations inside Alice’s laboratory are significantly more difficult for Eve. They found that the optimal states are mutually unbiased bases, and the security condition in this case is 
\begin{equation}
P_B > \frac{1}{4}\left(2 + \cos{\alpha_{\eta}} + \frac{1-\eta}{1+\eta}\sin{\alpha_{\eta}}\right),
\end{equation}

where $\alpha_{\eta} = \arctan\left(\frac{1-\eta}{1+\eta}\right)$. The fact that these conditions are the same for both cases, with and without memory, proves that access to a small quantum memory (a qubit) does not help the eavesdropper in attacking the SDI-QKD protocol.
\par
Additionally, as a straightforward generalization of the original protocol \cite{PB11}, Chaturvedi et al. \cite{Chaturvedi2018} presented a modified SDI-QKD protocol based on the ($3 \rightarrow 1$) scenario. In this protocol, Alice is given three bits, and depending on them, she prepares a state and sends it to Bob, who has as input a classical trit $b \in \{0,1,2\}$ to select his measurement. Although the generalized protocol has lower key rates, the security requirements are significantly reduced.

\paragraph{Relation between SDI and DI} Just as the violation of a Bell inequality in the DI case tells us that
the measured system cannot have a classical description, the violation of a dimension witness in the SDI case tells us
that the communicated system cannot be a classical bit. In both cases, violation of the classical bound is a necessary (though
not always sufficient) condition, with the difference residing in the form of the inequalities. 
Finding the correspondence between these two objects is equivalent to
finding the correspondence between the scenarios. 
A typical bell inequality $I$ can be written as $I=\sum_{a,b,x,y}\alpha_{a,b,x,y}P(a,b|x,y)$. Using the relation $p(a,b|x,y)=p(a|x,y)p(b|a,x,y)$, then $I$ can be rewritten as $\sum_{a,b,x,y}\alpha_{a,b,x,y}p(a|x,y)p(b|a,x,y)$. By considering $a$ as an input of Alice  ($x^{\prime}=(x,a)$), $p(a|x,y)$ can be seen as the part of Alice's input is $a$. Since in
the parameter estimation phase of the protocol the inputs are chosen according to a uniform distribution, we set $P(a|x,y) =\frac{1}{A}$, where $A$ is the size of the alphabet of $a$.  Then, the Bell inequality $I$ takes the form of a dimension witness. Using this method, the method for going from SDI to DI and vice versa was introduced in \cite{Li2013}. For going from SDI to DI, Alice’s input $x$ must be divided into a pair comprising a setting and an outcome. Let us consider the SDI randomness generation (based on $n \rightarrow 1$ quantum random access code $q$) where Alice input $x^{\prime}$ is a collection of $n$ independent bits $a_0,\cdots,a_{n-1}$. For this case, Alice's input can be divided into pairs of outcome $a=a_0$ and setting 
$x=(a_1,\cdots,a_{n-1})$ and a family of Bell inequalities can be obtained which they were found useful to implement entanglement-assisted random access codes works in DI randomness and DI-QKD protocols. The other side is also possible, and one can show a DI protocol can be converted to an SDI protocol. The example of this conversion was also shown in \cite{Li2013}. \par
Using a similar approach, Woodhead et al. \cite{Woodhead2015} demonstrated that the fundamental bound on Alice’s min-entropy in the DI setting also applies to the semi-DI setting. They achieved this by utilizing the PM version of the CHSH correlator, defined as $S=\frac{1}{2}\sum_{abxy}(-1)^{a+b+xy}P(b|axy)$, instead of Eq. (\ref{eq:CHSH1}). This result helps bring the semi-DI setting, where security proofs are still lacking, more in line with the established security results for DI-QKD.

\paragraph{SDI protocols based on other assumptions}  
In addition to the previously mentioned protocols, it is possible to introduce other SDI protocols based on alternative constraints. A prerequisite for developing any DI or SDI protocol is to examine the set of available correlations under the given assumptions. In light of this, Himbeeck et al. \cite{HWCGP17} introduced a general framework for SDI prepare-and-measure scenarios and modified it to account for a physical constraint, namely, the mean value of an observable. This results in a restriction on the quantum messages $\rho_x$, which can be expressed as a constraint on the corresponding mean values $H_x = \textrm{tr}(H\rho_x)$ of the observable. \\
More specifically, Himbeeck et al. \cite{HWCGP17} considered two types of constraints on the mean values of $H$. The first, called the max-average assumption, assumes upper bounds on the mean values: 
\begin{equation}
   H_x = \textrm{tr}(H\rho_x) = \sum_{\lambda}p_{\lambda} \textrm{tr}(H\rho_{x}^{\lambda}) \leq \omega_x, \quad \forall x .
\end{equation}  

For example, if $H$ is the photon-number operator, one can trust that, for all states $\rho_x$ emitted by the source, the mean photon numbers $H_x$ are below a certain threshold. 
If the states emitted by the source (Alice) vary from run to run according to some random parameter $\lambda$, the max-average assumption only bounds the mean value averaged over all possible values of $H_{x|\lambda} = \textrm{tr}(H\rho_{x|\lambda})$. However, it does not constrain the maximum values of $H_{x|\lambda}$, which could, in principle, be arbitrarily high. To address this, a stronger assumption, known as the max-peak assumption, was introduced:  

\begin{equation}
   \max_{\lambda} H_{x|\lambda} = \max_{\lambda} \textrm{tr}(H\rho_{x}^{\lambda}) \leq \omega_x, \quad \forall x . 
\end{equation}  

Again, if $H$ is the photon-number operator, this second condition still allows fluctuations in photon numbers within each state. It does not imply truncation of the Fock space, as the constraint only imposes a bound on the mean values $\textrm{tr}(H\rho_{x}^{\lambda})$ of $H$ for each $\rho_{x}^{\lambda}$. In particular, the states may still have non-zero amplitudes in any number-basis states.
\par
The max-average assumption has the advantage of being verifiable externally by testing the average emitted states without requiring knowledge of the internal workings of the source. On the other hand, verifying the max-peak assumption typically depends on modeling the source. Its primary advantage is that it is more restrictive and can certify useful properties that would not be certified under the max-average assumption.

\subsection{Measurement-device-independent QKD}
The purpose of SDI-QKD is to eliminate the most serious possible hardware vulnerabilities at the expense of a minimum of security assumptions.  The measurement devices and detectors used in QKD present considerable opportunities for side-channel attacks by Eve, in particular, due to the fact that Eve can both probe and/or manipulate the measurement system using external light, allowing her to determine measurement settings, blind a detector, or force a detector to ``click". To combat this type of attack, Lo et al. proposed the idea of measurement-device-independent QKD (MDI-QKD) in 2012 \cite{Lo2012}. Since its introduction, many variants and improvements of MDI-QKD have been proposed, a summary of which is shown in Fig. \ref{fig:MDIQKDhier}. In the following section, we present a brief overview of MDI-QKD.

\begin{figure}[h!]
\centering    
\begin{tikzpicture}[
    box/.style={draw, rounded corners, minimum width=2.8cm, minimum height=.8cm, thick, fill=white},
    box1/.style={draw, rounded corners, text width=3cm, align=center, minimum height=.8cm, thick, fill=white},
    box2/.style={draw, rounded corners, text width=2.5cm, align=center, minimum height=.8cm, thick, fill=white},
    font=\small
  ]

\node[box] (s) at (0, 9.5) {\textbf{MDI-QKD}};

\node[box] (b1) at (-6, 8) {Encoding Variants};
\node[box] (b2) at (-2., 8) {Protocol Extensions};
\node[box2] (b3) at (1.8, 8) {Practical Implementations};
\node[box] (b4) at (5.5, 8) {Advanced Variants};

\node[box] (b11) at (-6, 6.4) {Phase Encoding};
\node[box] (b12) at (-6, 5.4) {Time Bin Encoding};

\node[box1] (b21) at (-2.1, 6.8) {Decoy-state MDI-QKD};
\node[box1] (b22) at (-2.2, 5.8) {Coherent-state MDI-QKD};
\node[box1] (b23) at (-2.3, 4.8) {Higher dimensional variants};

\node[box] (b31) at (1.8, 6.3) {Fiber-based};
\node[box] (b32) at (1.8, 5) {Free-space based};

\node[box] (b41) at (5.5, 7) {Twin-Field QKD};
\node[box1] (b42) at (5.6, 6) {Asymmetric MDI-QKD};
\node[box1] (b43) at (5.7, 5) {Mode-pairing MDI-QKD};
\node[box] (b44) at (5.9, 4) {Network MDI-QKD};

\draw[->, thick] (s) to[out=200, in=90] (b1.north);
\draw[->, thick] (s) to[out=220, in=90] (b2.north);
\draw[->, thick] (s) to[out=-40, in=90] (b3.north);
\draw[->, thick] (s) to[out=-20, in=90] (b4.north);

\draw[->, thick] (b1.west) to[out=180, in=90] (b11.north west);
\draw[->, thick] (b1.west) to[out=180, in=90] (b12.north west);

\draw[->, thick] (b2.west) to[out=180, in=90] (b21.north west);
\draw[->, thick] (b2.west) to[out=180, in=90] (b22.north west);
\draw[->, thick] (b2.west) to[out=180, in=90] (b23.north west);

\draw[->, thick] (b3.east) to[out=0, in=90] (b31.north east);
\draw[->, thick] (b3.east) to[out=0, in=90] (b32.north east);

\draw[->, thick] (b4.east) to[out=0, in=90] (b41.north east);
\draw[->, thick] (b4.east) to[out=0, in=90] (b42.north east);
\draw[->, thick] (b4.east) to[out=0, in=90] (b43.north east);
\draw[->, thick] (b4.east) to[out=0, in=90] (b44.north east);
\end{tikzpicture}
    \caption{\label{fig:MDIQKDhier}Hierarchical overview of MDI-QKD advancements.}
    \label{fig:enter-label}
\end{figure}
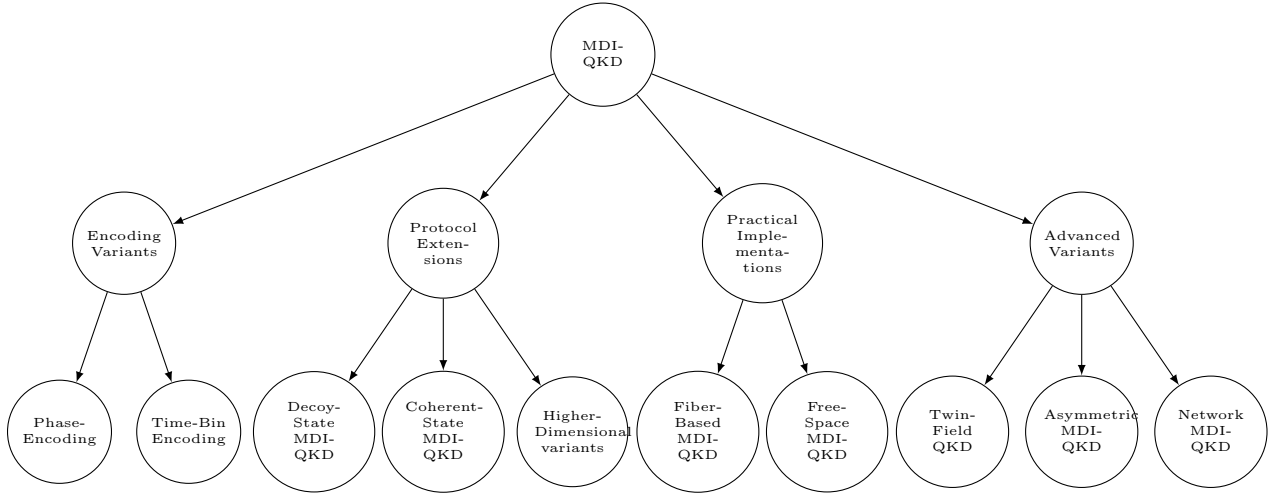

\subsubsection{Original MDI-QKD Protocol}
\label{sec:OrigMDIQKD}
The main idea of the original MDI-QKD scenario \cite{Lo2012} is sketched in Fig. \ref{fig:mdi}. For simplicity, consider first that Alice and Bob each send a single photon to Charlie, who implements a linear-optical Bell-state measurement (BSM), which can be described by measurement operators $\Pi_+,\Pi_-,\Pi_0$ (see also Fig. \ref{fig:eventready}). Here $\Pi_\pm$ is a rank-one projector onto one of the Bell states $\ket{\psi^\pm}=(\ket{H}_A\ket{V}_B\pm \ket{V}_A\ket{H}_B)/\sqrt{2}$, where $H$ and $V$ are the rectilinear polarization states. $\Pi_0$, on the other hand, is a rank-two projector onto the remaining two Bell states.   When Alice and Bob prepare their pulses in either the $H$ or $V$ polarization state, a successful BSM (``$\Pi_\pm$" results) indicates that their pulses were prepared in orthogonal states. In the ideal scenario, they have an error rate of zero. These events can be used to generate a shared key--say--by defining $H_A\equiv$ ``0" and $V_A\equiv$ ``1", where Bob will flip all of his bits so that their bit strings are correlated. 
When the photons are prepared with diagonal polarization (states A, D in Fig. \ref{fig:mdi}), a detection of the state $\psi^-$ ($\psi^+$) occurs only when orthogonally (parallel) polarized pulses were prepared. An error in this case corresponds to detection events $\Pi_+$ ($\Pi_+$) when orthogonal (the same) diagonal polarization states are sent.  With the error rates for both bases, post-processing steps similar to BB84 can be realized. 
\par
Remarkably, the same logic can be applied to the case when Alice and Bob use weak coherent pulses (WCP) \cite{Lo2012}. To fully realize MDI-QKD, the pulses are prepared in one of the four polarization states, determined by random bit strings held by Alice and Bob, and then sent to Charlie, who publicly announces the results when the BSM is successful.  Alice and Bob subsequently post-select their data, retaining only the outcomes with a successful BSM, and the same preparation bases were used. The data corresponding to the diagonal basis is used for assessing bit and phase error rates. If the error rates fall below a predetermined threshold, Alice and Bob proceed with classical error correction and privacy amplification processes to obtain a secure key; otherwise, they terminate the protocol.  Most notably, MDI-QKD does not require a trusted measurement apparatus.  To see this, notice that a successful BSM event $\psi^\pm$ does not provide any information about whether Alice and Bob prepared $H_A,V_B$ (= secret bit value 0) or $V_A,H_B$ (= secret bit value 1), and moreover there is no basis selection at the measurement device, so no information can be acquired if Eve somehow probes the BSM station.  In addition, forcing BSM detectors to click, or announcing false results will inevitably lead to errors when diagonal states are sent.       

Thus, MDI-QKD effectively mitigates all potential detector side channels, and in fact, the BSM station can be implemented by an untrusted third party and completely under Eve's control. However, one still needs to consider the imperfection of signal resources, such as the basis-dependent flaw that results in a decrease in achievable distance. Therefore, a generalization of \cite{Lo2012} is essential for practical purposes.
\par
Published back to back with Ref. \cite{Lo2012}, Braunstein et al. (2012) adopted a similar approach by devising a side-channel-free protocol to account for all potential side channels that might arise during a QKD implementation \cite{braunstein2012side}. In its simplest form, this protocol corresponds to an entanglement-swapping experiment, where the dual teleportation channels serve as ideal Hilbert space filters to eliminate the possibility of side-channel attacks.
 To prevent any attempts at probing side channels within Alice's and Bob's laboratories, they implemented state-generation via partial measurement of a bipartite entangled state. This strategic move effectively isolates any external probes from the state-generation device.
This approach ensures the complete protection of both Alice and Bob's private spaces against any potential side-channel attacks. In this protocol, the secret key rate can be lower bounded through the use of quantum memory and by calculating the entanglement distillation rate over the distributed state as follows:
	\begin{equation}
		r \geq I(A \rangle B | L') + \Delta. 
	\end{equation}
	Here, $I(A\rangle B | L^{\prime})$ represents the coherent information conditioned on Eve's fake variable $L^{\prime}$ which Eve sends to both parties to mislead them, instead of the true variable $L$ ($I(A\rangle B)=S(\rho_B)-S(\rho_{AB})=-S(A|B)$) and $\Delta$ denotes the amount of classical cheating.  \\ 
\begin{figure}
\centering
\begin{tikzpicture}[
    node distance=2.5cm, 
    every node/.style={font=\sffamily},
    align=center]

\node[draw, fill=blue!20, text width=3cm, text centered, rounded corners] (alice) {Alice\\Quantum State\\Source (A)};
\node[draw, fill=green!30, text width=3cm, text centered, rounded corners, right=of alice] (charlie) {Charlie/Eve\\Quantum\\Measurement\\Device\\ (in Bell basis)};
\node[draw, fill=blue!20, text width=3cm, text centered, rounded corners, right=of charlie] (bob) {Bob\\Quantum State\\Source (B)};

\draw[->, thick] (alice) -- (charlie);
\draw[<-, thick] (charlie) -- (bob);
\draw[->, dashed, thick] ([yshift=8pt]charlie.west) -- ([yshift=8pt]alice.east);
\draw[->, dashed, thick] ([yshift=8pt]charlie.east) -- ([yshift=8pt]bob.west);

\end{tikzpicture}

\vspace{1cm} 

\begin{tabular}{c|c|c|c|c}
  \textbf{A/ B}  & H & V & D & A \\
  \hline
   H  & $\Pi_0$ & $\Pi_\pm$ & X & X \\
    V  & $\Pi_\pm$ & $\Pi_0$ & X & X \\
     D  & X & X & $\Pi_+ ,\Pi_0$  & $\Pi_- ,\Pi_0$ \\
     A  & X & X & $\Pi_- ,\Pi_0$ & $\Pi_+ ,\Pi_0$ \\
\end{tabular}
\caption{Basic set-up of an MDI-QKD protocol. Alice and Bob each send a single photon to Charlie, who implements a linear-optical Bell-state measurement, which can be described by measurement operators $\Pi_+,\Pi_-,\Pi_0$.  The tables includes expected detection results at BSM station in MDI-QKD.  ``X" refers to discarded results in which Alice and Bob chose different bases. $H$, $V$, $D$, $A$ are horizontal, vertical, diagonal and anti-diagonal linear polarization states, respectively.}
\label{fig:mdi}
\end{figure}
By modeling MDI-QKD as a special case of LOCC-assisted communication over a multiplex quantum channel, Das et al. \cite{Das2021} derived a tight strong-converse upper bound on the achievable key rate. The bound is expressed in terms of entanglement measures, such as the max-relative entropy of entanglement and, for teleportation-simulable channels, the relative entropy of entanglement of the effective channel's Choi state. For practical photon-based MDI-QKD setups (e.g., with erasure or depolarizing channels), their results account for both channel losses and imperfections at the measurement node, offering a more precise benchmark than earlier network-cut-based bounds.
 
 \paragraph{Practical loopholes in MDI-QKD}
Practical loopholes in MDI-QKD shift the focus from measurement device imperfections to state preparation or source flaws, as these can be exploited to compromise security. Liu et al. \cite{Lu2023a} demonstrated a hacking strategy leveraging modulation errors to obtain all key bits. One significant imperfection is the basis-dependent flaw arising from discrepancies in density matrices in BB84 states, while the birefringence effect in optical fibers highlights the practicality of phase encoding over polarization encoding. Tamaki et al. \cite{Tamaki2012} addressed these issues with two MDI-QKD schemes: one using phase locking of separate lasers and a double BB84 protocol to control basis-dependent flaws, and another employing phase encoding for longer distances. Primaatmaja et al. \cite{Primaatmaja2019} introduced a numerical technique using semidefinite programming to analyze phase-error rates, showing phase-encoding MDI-QKD's potential to outperform decoy-state MDI-QKD at short distances. Zhu et al. \cite{Zhu2021a} improved analysis by modulating different intensities in key and test bases. Bourassa et al. \cite{Bourassa2022} identified a time-dependent side channel in sources employing Faraday mirrors, showing divergences between three-state and BB84 protocols. Alternative schemes \cite{ma2012alternative} simplified encoding and decoding with minimal performance compromise. Xu et al. \cite{xu2013practical} analyzed error sources like polarization misalignment and mode mismatch, showing that MDI-QKD tolerates up to $6.7\%$ polarization misalignment at $0$ km and $5\%$ at $120$ km, while mode mismatch tolerance decreases from $80\%$ to $50\%$ over the same range. Wang et al. \cite{wang2014} estimated gains, error rates, and key rates under arbitrary photon mixtures, while Li et al. \cite{Li2023d} proposed a polarization-alignment method using fewer devices to reduce photon loss. Phase-randomized weakly coherent pulses (PR-WCPs), commonly used due to a lack of mature single-photon sources, introduce errors in the $X$ basis. Li et al. \cite{Li2022} incorporated these errors into security analysis, showing equivalence between PR-WCPs and Poisson-distributed photon states, with a tighter key rate than \cite{Ma2012}. Lu et al. \cite{Lu2022unbalanced} introduced an MDI-QKD protocol addressing modulation errors, achieving high performance despite $X$ basis imbalances and asymmetric channel transmittances. Yin et al. \cite{Yin2013, Yin2014} removed encoding state characterization assumptions, demonstrating practicality and tolerance for high loss and errors over $160$ km. Hwang et al. \cite{Hwang2017} improved phase error estimation, and Zhou et al. \cite{Zhou2019} extended uncharacterized qubit protocols to weak coherent sources using decoy-state methods. Kang et al. \cite{Kang2019, Kang2020} developed protocols with uncharacterized coherent states under collective attacks. Li et al. \cite{Li2014a} proposed the CHSH-MDI-QKD protocol to mitigate state preparation assumptions, using the CHSH inequality and decoy states \cite{Zhang2014a} to enhance accuracy in single-photon yield estimation, though increased parameter estimation complexity limited its effectiveness.

\paragraph{Finite-key analysis}
Finite-key analysis of MDI-QKD was first conducted in \cite{Song2012} and \cite{Ma2012}, where they derived secure bounds under the influence of statistical fluctuations in relative frequency. This analysis applies to practical detectors with low efficiency and highly lossy channels. Their study demonstrates the possibility of achieving secure transmission over distances exceeding 10 kilometers with a success rate of $10^{10}$ outcomes, making it directly applicable in practical implementations. However, when the number of successful outcomes falls below $10^{8}$, achieving a nonzero key rate becomes impossible. \par
Both studies mentioned above focused on security assessments against specific types of attacks. The first study to explore security proofs within the finite-key regime against general attacks and to satisfy the composability definition was conducted by Curty et al. \cite{curty2014finite}. They utilized the principles of large deviation theory, specifically employing a multiplicative form of the Chernoff bound, for critical parameter estimation. This step was crucial in demonstrating the feasibility of implementing MDI-QKD over long distances and within a reasonable timeframe.
Their findings demonstrated that even with the technology available at the time, an MDI-QKD protocol could be realized without the need for high-efficiency detectors. Importantly, they showcased the potential for long-distance MDI-QKD protocols, extending up to approximately 150 kilometers, for finite-sized data sets ranging from $10^{12}$ to $10^{14}$ signals. This achievement was made possible using practical signal sources, such as WCPs.\\

\subsubsection{Decoy-state measurement-device-independent QKD} 

The sources used in MDI-QKD must be trusted, necessitating a complete characterization of the source. Commonly, weak coherent sources replace perfect single-photon sources, though they remain susceptible to photon number splitting (PNS) attacks due to multiphoton fractions \cite{Brassard2000}. To counteract this, the decoy-state method was originally proposed by Hwang \cite{Hwang2003}, with the first security proofs independently established by Wang \cite{Wang2005} and by Lo, Ma, and Chen \cite{Lo2005}. This method has since been adapted for MDI-QKD to estimate single-photon contributions efficiently \cite{ma2012alternative}. Wang et al. \cite{wang2013three} further optimized this by employing three-intensity decoy states, which addressed basis-dependent coding errors. Subsequent enhancements, such as vacuum and weak decoy states by Sun et al. \cite{Sun2013}, showed improved performance but highlighted the limitations of certain methods. Advances continued with modified coherent states introduced by Li et al. \cite{Li2014b}, reducing multiphoton distributions and enhancing key rates. Techniques to refine single-photon yield and phase error estimation, as demonstrated by Zhu et al. \cite{Zhu2016} and Ding et al. \cite{Ding2018}, increased the accuracy of MDI-QKD parameters and extended secure transmission distances. Other studies, such as those by Mao et al. \cite{Mao2019}, explored new decoy-state frameworks that surpassed prior methods, further enhancing both distance and key rates. \par

Incorporating heralded single-photon sources (HSPS) offers notable benefits, including reduced dark count rates and lower QBER \cite{Yu2013}. Wang et al. \cite{Wang2013} showed that combining triggered and non-triggered events in HSPS-based protocols enhances both key rates and transmission distances. Subsequent works \cite{Zhang2018} applied biased decoy-state schemes with HSPS, yielding superior results for small datasets. Similarly,  Zhou et al. \cite{Zhou2013}, introduced passive decoy methods to spontaneous parametric down-conversion (SPDC) sources, to minimize side-channel leaks and improve performance compared to weak coherent states. While SPDC sources offer advantages, challenges such as spectral entanglement were addressed by Zhan et al. \cite{Zhan2023}, underscoring the need for high-purity sources. \par

Optimization of decoy-state parameters has played a pivotal role in enhancing MDI-QKD protocols. Techniques such as local search algorithms \cite{Xu2014}, statistical fluctuation considerations \cite{Yu2015}, and advanced joint constraints \cite{Jiang2021a} have significantly improved key rates and extended transmission distances. Furthermore, protocols integrating memory-assisted techniques \cite{Abruzzo2014} and MDI-QKD mode-pairing \cite{Zeng2022,Xie2022,Liu2023a,Wang2023,Lu2024,Li2024a,Zhou2025a} further push the boundaries of MDI-QKD capabilities. Asymmetric protocols (where Alice and Bob have different distances from the source) \cite{Wang2019b} and reference-frame-independent methods \cite{Yin2014a,Zhou2025} address practical challenges like channel asymmetry and misalignment, making MDI-QKD more adaptable for real-world applications.

\subsubsection{High-dimensional measurement-device-independent QKD}
So far, all mentioned protocols were for two-dimensional encoding systems using $Z$ and $X$ bases. In this section, we review protocols developed for higher dimensions. Chau et al. \cite{Chau2016} introduced the protocol which they called the mother-of-all QKD protocol and its MDI variants for qudits, including the round-robin differential phase protocol \cite{sasaki2014practical} and the Chau15 protocol \cite{chau2015quantum}. However, these were experimentally infeasible due to challenges in realizing high-dimensional Bell states. Hwang et al. \cite{hwang2016n} proposed a d-dimensional MDI-QKD protocol, proven secure under the condition of zero QBER. Jo et al. \cite{Jo2016} proposed a three-dimensional MDI-QKD (3d-MDI-QKD) protocol with mutually unbiased bases (MUBs) comprising time and energy bases \cite{bechmann2000quantum}. Bell state measurements in 3d-MDI-QKD use nine maximally entangled states in a three-dimensional bipartite system, enabling a secret key rate of $\Tilde{r} \geq \log_2 3 - 2Q - 2h(Q)$, where $Q$ represents state error rate. This protocol achieves higher secret key rates than the original MDI-QKD for low transmission losses, suitable for short-distance communication, but faces feasibility challenges in realizing high-dimensional Bell state measurements with linear optics \cite{jo2019enhanced}. Sekga et al. \cite{Sekga2023} introduced a qutrit-based MDI-QKD protocol employing biphotons and Mach-Zehnder interferometers, achieving significant secret key rates for moderate distances. Dellantonio et al. \cite{dellantonio2018high} extended QKD to generalized $Z$ and $X$ bases in $d$ dimensions, demonstrating unconditional security with improved performance in low dark-count scenarios. Cui et al. \cite{cui2019measurement} proposed a high-dimensional MDI-QKD protocol utilizing hyper-encoded qudits with polarization and spatial-mode degrees of freedom, yielding a fivefold improvement in secret key rates. This was further extended by Yan et al. \cite{Yan2020} and Li et al. \cite{Li2023c} to multi-degree-of-freedom encoding. The limitations of long-distance QKD due to decoherence prompted solutions like quantum repeaters, as discussed by Erkilic et al. \cite{Erkilic2023}. Their MDI-QKD protocol surpasses the PLOB bound \cite{Pirandola2017} using high-dimensional states optimized for increased key rates at shorter distances, though these advantages diminish with greater transmission distances due to photon loss.

\subsubsection{Continuous-variable measurement-device-independent QKD} 
While two-dimensional discrete protocols can achieve long-distance communication, they often suffer from low key rates, making them unsuited for metropolitan network requirements. A solution to this challenge can be found in adopting continuous-variable (CV) systems. One significant advantage of a CV-QKD protocol is its compatibility with standard telecommunication technology, particularly because it does not rely on single-photon sources, which are the most vulnerable to attacks in discrete-variable QKD (DV-QKD) protocols. Another significant advantage is that, in a typical QKD protocol, users often need to allocate a portion of their raw data to estimate communication channel parameters, such as the error rate. This results in a trade-off between the secret key rate and the accuracy of parameter estimation in the finite-size regime. However, it has been demonstrated that this constraint does not apply to continuous variable QKD. In continuous variable QKD, the entire set of raw keys can be utilized for both parameter estimation and secret key generation without compromising security \cite{Lupo2018}.  In addition, CV-QKD systems might be more suitable to coexist with classical data transmission in optical fibers, since the local oscillator required for homodyne detection can act as a mode filter, reducing classical Raman noise from the quantum signal.
\par
As such, there is considerable interest in continuous-variable MDI-QKD (CV-MDI-QKD). The first CV-MDI-QKD protocols were originally introduced by Pirandola et al. \cite{pirandola2015high}, Li et al. \cite{Li2014} and Ma et al. \cite{ma2014gaussian}. The protocol operates as follows: Alice and Bob randomly prepare coherent states, denoted as $\ket{\alpha}$ and $\ket{\beta}$, respectively, where the amplitudes $\alpha$ and $\beta$ are modulated by Gaussian distributions with zero mean and sufficiently large variances. These prepared states are then sent to an intermediary party (Charlie) for measurement. To establish secret correlations, Charlie performs a CV Bell measurement and communicates the outcomes to Alice and Bob. This Bell measurement is executed by mixing the incoming modes using a balanced beamsplitter. The measurement corresponds to the quadrature operators $\hat{q_{-}}=(\hat{q}_A-\hat{q}_B)/{\sqrt{2}}$ and $\hat{p_{+}}=(\hat{p}_A+\hat{p}_B)/{\sqrt{2}}$, and the classical outcomes are combined into a complex variable denoted as $\gamma=(q_{-}+ip_{+})/{\sqrt{2}}$. 
The most general eavesdropping strategy  involves a joint attack encompassing both Charlie's measurement device and the two communication links, namely Alice-Charlie and Charlie-Bob. Since the protocol is based on Gaussian modulation and the detection of Gaussian states, the optimal eavesdropping technique employs a Gaussian unitary approach \cite{garcia2006unconditional}.
 By introducing a reconciliation efficiency factor denoted as $\epsilon\leq 1$, the secret key formula can be modified as follows:
\begin{equation}
r:=\epsilon I(A:B)-I_E.
\end{equation}
where $I_E$ is the upper bound on Eve's information.
\par
An investigation into the performance of the protocol under ideal reconciliation conditions ($\epsilon=1$) reveals the potential for achieving remarkably high secret key rates, approaching one bit per use. Notable, symmetric configurations, where the transmissivities are the same between Alice-Charlie and Bob-Charlie, are not the most secure option, particularly for longer distances. The optimal configuration is asymmetric, corresponding to minimal loss in Alice's link, which allows Bob's link to have a low transmissivity. Specifically, if Charlie's position can be situated close to Bob, the total transmission distance, i.e., the distance between Alice and Bob, can theoretically extend up to $ 80$ km. Taking into account realistic reconciliation performance, the experimental rates closely approach the maximum theoretical predictions. In particular, with $\epsilon\approx 0.97$, the experimental rates can achieve remarkably high values over typical connection lengths within a metropolitan network. 
\par
Zhang et al. \cite{Zhang2014} introduced a CV-MDI QKD protocol using squeezed states and demonstrated that its secret key rate consistently surpasses the coherent-state-based protocol, particularly under collective attacks, with a total maximum transmission distance increase of $6.1$ km under both perfect and imperfect detectors. The transmission distance further increases in asymmetric scenarios. In the extreme case where Charlie is on Bob’s side, such that the coherent-state-based protocol achieves zero transmission distance,  the squeezed-state protocol significantly outperforms it, especially with the introduction of optimal Gaussian noise levels on Bob’s side, as determined for maximizing key rate and transmission distance under reverse reconciliation. Chen et al. \cite{chen2018composable} extended this protocol against general attacks using entropic uncertainty relations, yielding a composable security analysis and demonstrating the system’s resilience to a maximum channel loss of $0.64$ dB. One key limitation of Gaussian-modulated protocols is their low reconciliation efficiency in long-distance transmissions, which has driven interest in discrete modulation. Ma et al. \cite{ma2019long} proposed a four-state discrete-modulated CV-MDI-QKD protocol, leveraging nonorthogonal coherent states for encoding bits, achieving longer transmission distances and simplified implementation compared to Gaussian modulation, with the eight-state protocol \cite{Zhao2020a} further improving key rates and modulation variances. Wilkinson et al. \cite{Wilkinson2020} introduced postselection in long-distance CV-MDI-QKD, extending the communication range to $14$ km over standard optical fiber, while protocols employing quantum catalysis \cite{ye2020continuous, ye2021enhancing}, quantum scissors \cite{kong2022improvement}, and multimode signals \cite{ding2021multi} further improved performance by improving transmission distance and reducing noise. Practical implementation challenges, such as independent light sources, phase reference calibration, and external disturbances, require mitigation to prevent overestimation of key rates. Ma et al. \cite{ma2019security} studied phase calibration imperfections and their thermal noise equivalence, proposing models for realistic security analysis, while Zhao et al. \cite{Zhao2020a} introduced Bayesian phase-noise estimation to eliminate local oscillator transmission. Simplified implementations, such as the plug-and-play scheme \cite{liao2018dual, zhou2023plug}, address synchronization issues by deriving local oscillators from a shared laser, reducing complexity and enhancing stability. However, imperfections in state preparation also introduce Gaussian noise, as modeled by Ma et al. \cite{ma2019security}, who explored intensity error impacts under various distributions and emphasized placing stable sources on the encoder’s side for optimization. Countermeasures like Huang et al.’s one-time calibration method \cite{Huang2023} and noise characterization approaches are critical for enhancing practical security. Addressing transmittance fluctuations, Zheng et al. \cite{zheng2022security} highlighted performance degradation under varying channel conditions, proposing Gaussian post-selection to mitigate risks of denial-of-service attacks, while Li et al. \cite{Li2023} studied the effects of non-ideal Bell detection due to angle errors, showing significant transmission distance reductions even with minor errors. Efforts to reduce CV-MDI-QKD complexity include self-referenced schemes \cite{wang2018self}, shared optical path methods \cite{yin2019phase}, and unidimensional modulation \cite{bai2020unidimensional}, achieving comparable performance with reduced system demands. Semi-Quantum Key Distribution (SQKD), introduced by Boyer et al. \cite{Boyer2007}, evolved into a continuous variable version \cite{Zhou2022}, enabling secure communication between classical and quantum users, leveraging Charlie’s full quantum capabilities to balance cost-effectiveness and security under various attack scenarios.

 \paragraph{Finite-size effects} The impact of finite-size effects on the key rate of CV-MDI-QKD was initially investigated by Papanastasiou et al. \cite{Papanastasiou2017}, considering two-mode Gaussian attacks, and by Zhang et al. \cite{Zhang2017}, examining collective attacks. To study the security of the protocol, a potent approach is to employ the entanglement-based representation, where the description of the dynamics occurs within an extended Hilbert space, allowing the use of pure states.
The protocol is outlined as follows: Alice and Bob employ sources of coherent states, which are purified, assuming they start from two-mode squeezed vacuum states $\rho_{aA}$ and $\rho_{bB}$, where modes $A$ and $B$ are transmitted over the communication links, while local modes $a$ and $b$ are heterodyned. The measurements project the traveling modes into pure coherent states. The attenuation of the channel on modes $A$ and $B$ is modeled using two beam splitters with transmissivities $\tau_A$ and $\tau_B$, where $0 \leq \tau_A, \tau_B \leq 1$. These processes manipulate Alice and Bob's signals with a pair of Eve's ancillary systems $E1$ and $E2$, which generally belong to a broader reservoir of modes controlled by the eavesdropper. 
\par
The key rate, accounting for finite-size effects, is expressed as:
\begin{equation}
r^{\text{finite}} = \frac{n}{N}\left(r - \Delta(n)\right),
\end{equation}
where $n$ represents the number of signals used for key preparation, $N$ is the total number of exchanged signals, and $r$ denotes the asymptotic key rate. The correction function $\Delta(n)$ is employed to compensate for the utilization of the Holevo function in the context of a finite number of signals. It is a function that relies on the number of signals used for key preparation ($n$) and the probability of error associated with the privacy amplification procedure $\epsilon_{PA}$ ($\Delta(n) \sim \sqrt{\frac{1}{n} \log_2(2/\epsilon_{PA})}$).
\par
Numerical results indicate that under realistic conditions and considering finite-size effects, CV-MDI-QKD is viable for metropolitan distances within experimental constraints. In particular, a total number of signals exchanged in the range of $N = 10^6$ to $10^9$ is sufficient to achieve a high key rate of $10^{-2}$ bits per use over metropolitan distances, even in the presence of excess noise of approximately $0.01$.
\par
For the protocol considering collective attacks, the CV-MDI-QKD protocol with an asymmetric structure and finite-size effects can securely transmit over approximately $86$ km under ideal reconciliation efficiency and optimal modulation variance conditions for $n = 10^{10}$ block size. When the reconciliation efficiency is $96.9\%$, the maximum transmission distance achievable is approximately $75$ km. 
\par
Lupo et al. \cite{Lupo2018} studied the security proof for coherent attacks. The advantage of their study compared to the previous ones is that the correlations between Alice and Bob are generated through the variable $Z$ announced by the relay which allows Alice and Bob to do parameter estimation with a negligible amount of public communication. Therefore, the whole raw key can be exploited for both parameter estimation and secret-key extraction. They first investigated the security against collective attacks by presenting an improved estimation of the conditional smooth min-entropy obtained by applying a new entropic inequality and found the following lower bound on the secret-key rate: 
\begin{eqnarray}
    r^{\text{finite}}\geq r- \frac{1}{\sqrt{n}}\Delta_{\text{AEP}}(\frac{2}{3}p\epsilon_s,d)
    +\frac{1}{n}\log(p-\frac{2}{3}p\epsilon_s)+\frac{1}{n}2\log(2\epsilon), 
\end{eqnarray}
where $p$ is the probability of successful error correction, $\epsilon_s$ is the smoothing parameter entering the smooth conditional min-entropy and $\Delta_{\text{AEP}}(\delta,d)$ is a function of dimensionality $d$ ($\Delta_{\text{AEP}}(\delta,d)\leq 4(d+1)\sqrt{\log(2/\delta^2)}$). 
\par
The secret key rate for coherent attacks can be modified by applying the results of \cite{leverrier2017security} as 
\begin{eqnarray}
    r^{\text{finite}}\geq \frac{n-k}{k} r^{\infty}- \frac{\sqrt{n-k}}{{n}}\Delta_{\text{AEP}}(\frac{2}{3}p\epsilon_s,d)
    +\frac{1}{n}\log(p-\frac{2}{3}p\epsilon_s)+\frac{2}{n}\log(2\epsilon)-\frac{2}{n}\log\begin{pmatrix}
    K+4 \\
    4
\end{pmatrix}, 
\end{eqnarray}
where $k$ is the number of signals used for the energy test and $K\sim n$. 
Based on numerical examples,  it is in principle possible to generate a secret key against the most general class of coherent attacks 
for block sizes of the order of $10^7$–$10^9$, depending on loss and noise. In particular, this composable security analysis confirms that CV-MDI protocols allow for high QKD rates on the metropolitan scale, supporting the results of the asymptotic analysis of Pirandola et al. \cite{pirandola2015high}. The viability of utilizing the entire raw key for both parameter estimation and key extraction was later demonstrated by Lupo et al. \cite{Lupo2018a}. Their work CV-MDI-QKD revealed that parameter estimation in this scheme can be achieved with minimal public communication, as correlations are postselected by the central relay. Consequently, the public variable announced by the relay encompasses all the information regarding the correlations between Alice and Bob, making it sufficient, along with the local data, to estimate the covariance matrix. This crucial discovery eliminates the trade-off between the secret key rate and the accuracy of parameter estimation in the finite-size regime of CV-QKD.  Similar results are presented in \cite{wu2019security,wu2020simultaneous}.  
\par
Non-Gaussian postselection, such as virtual photon subtraction from a coherent state source, improves CV-QKD protocols by enhancing secret key rates and tolerating excess noise over longer distances \cite{huang2013performance, guo2017performance}. Zhao et al. \cite{Zhao2018} and Ma et al. \cite{Ma2018} demonstrated its application in coherent-state CV-MDI-QKD, optimizing performance through Alice’s photon subtraction with carefully chosen parameters while maintaining protocol security. Kumar et al. \cite{kumar2019coherence} showed that photon subtraction on two-mode squeezed coherent states extends transmission distances up to $68$ km but reduces key rates compared to vacuum states. Practical applications, such as photon subtraction over fiber-to-water channels \cite{yu2022photon}, further validate this approach. Recently, Papanastasiou et al. \cite{Papanastasiou2023a} and Ghalaii et al. \cite{ghalaii2023continuous} analyzed composable finite key generation, demonstrating secure CV-MDI QKD over free-space optical links under realistic conditions.
 
\subsubsection{Measurement-device-independent Multiparty Quantum Communication}
 Multiparty quantum communication protocols strive to ensure information-theoretic security in the realm of highly sensitive and confidential multiuser communication. Using the principles of quantum mechanics, these protocols exhibit superior physical performance compared to their classical counterparts. Their versatile applications encompass a spectrum of scenarios such as secret multiparty conferences, remote voting, online auctions, management of payment system master keys, collaborative scrutiny of accounts containing quantum money, and the facilitation of secure distributed quantum computation.
 \par
Specifically, Quantum Cryptographic Conferencing (QCC) is a protocol designed for multiparty Quantum Key Distribution. QCC ensures the secure sharing of a key among legitimate users, even in the presence of potential eavesdroppers. Another notable protocol, Quantum Secret Sharing (QSS), involves the fragmentation of a message into multiple parts distributed among a group of participants. Each participant is allocated a share of the secret, and consequently, the complete set of shares is required to comprehensively decipher the message. For instance, QSS can be employed to guarantee that no single individual possesses the capability to launch a nuclear missile or access a bank vault independently. Instead, the collective participation of all legitimate users is essential for these critical actions. 

\paragraph{Discrete variable protocols}
The Greenberger-Horne-Zeilinger (GHZ) entanglement is an important resource for multiparty quantum communication tasks, especially for the measurement-device-independent versions of QCC (MDI-QCC) and QSS (MDI-QSS). However, the practical applications of GHZ states are quite limited due to the lack of two important factors—(i) high-intensity sources and (ii) reliable distribution of the GHZ states. To tackle these limitations, Fu et al. \cite{Fu2015} take advantage of postselected GHZ states among three legitimate users (typically called Alice, Bob, and Charlie) to perform secure multiparty quantum communication. As a typical MDI-QKD protocol, the postselecting measurement device here can be regarded as a black box that can be manipulated by anyone including the eavesdropper. Therefore, the scheme is naturally immune to all detection-side attacks and can be regarded as the combination of time-reversed GHZ state distribution and measurement. Moreover, by employing the decoy-state method, the scheme can defeat photon-number-splitting attacks. The protocol in \cite{Fu2015} is as follows: Alice, Bob, and Charlie independently and randomly prepare quantum states with phase-randomized weak coherent pulses in two complementary bases ($Z$ basis and $X$ basis). They send the pulses to the untrusted fourth party located in the middle node, David, to perform a GHZ-state measurement which projects the incoming signals onto a GHZ state. After performing the measurement, David announces the events through public channels whether he has obtained a GHZ state and which GHZ state he has received. Alice, Bob, and Charlie only keep the raw data of successful GHZ-state measurements and discard the rest. They post-select the events where they use the same basis in their transmission through an authenticated public channel. Notice that Alice performs a bit flip when Alice, Bob, and Charlie all choose $X$ basis and David obtains a GHZ state $\frac{1}{\sqrt{2}}(\ket{000}-\ket{111})$. The data of $Z$ basis is used to generate the key, while the data of $X$ basis are used to estimate errors. After classical error correction and privacy amplification, Alice, Bob, and Charlie extract secure cryptographic conferencing keys. In the asymptotic limit, the MDI-QCC key generation rate and the MDI-QSS key rate are given by
\begin{eqnarray}
    R_{QCC}=Q_{\nu}^Z+Q_{111}^Z[1-H(e_{111}^{BX})]-H(E_{\mu\nu\omega}^{Z*})f Q_{{\mu\nu\omega}}^Z, \\
    R_{QSS}=Q_{\nu}^X+Q_{111}^X[1-H(e_{111}^{BZ})]-H(E_{\mu\nu\omega}^X)f Q_{{\mu\nu\omega}}^X
\end{eqnarray}
where $Q_{{\mu\nu\omega}}^Z$ and $E_{\mu\nu\omega}^{Z*}$ ($Q_{{\mu\nu\omega}}^X$ and $E_{\mu\nu\omega}^{X})$) are
 the gain and quantum bit error rate of $Z$ ($X$) basis respectively. The subscripts $\mu$,$\nu$, and $\omega$ are the pulse intensities of Alice, Bob, and Charlie respectively. $Q_{111}^Z$ ($Q_{111}^X$) is the gain of of $Z$ ($X$) basis and $e_{111}^{BX}$ ($e_{111}^{BZ}$) is the bit error rate of $X$ ($Z$) basis. The parameter $f$ is the error correction efficiency. 
 \par
Simulation results for QCC show that the estimation using two decoy states gives a secure key rate nearly the same as the corresponding one using infinite decoy states. In the case of asymptotic data with two decoy states, the secure transmission distance between Alice and the middle node of MDI-QCC is about $190$ km for the detection efficiency of $40\%$ ($210$ km for the detection efficiency of $93\%$). For MDI-QSS, the secure transmission distance is about $130$ km for the detection efficiency of $40\%$ ($150$ km for the detection efficiency of $93\%$) between the middle node and any user. Hua et al. proposed a similar scheme based on a GHZ entangled state, which is different from the above protocol and uses a GHZ entangled
state and the polarization state prepared by users to execute BSM and realize multi-user sharing of a common secret key \cite{Hua2022}. They derived the secure key rate when users employ an ideal single photon source and a weak coherent source and showed that the secure distance between each user and the measurement device can reach more than $280$ km while reducing the complexity of the quantum network. 
Despite the efficiency of the protocol, its scalability diminishes exponentially with the number of users, and security issues in QSS protocols, such as in \cite{Fu2015}, remain underanalyzed, particularly for participant attacks \cite{Gao2008}. To address these challenges, Li et al. \cite{Li2023b} proposed an MDI-QSS protocol based on spatial multiplexing, achieving a transmission distance over $300$ km and a secret key rate two orders of magnitude higher than \cite{Gao2008}, while addressing security concerns like participant attacks. Ju et al. \cite{Ju2022} introduced a hyper-encoding MDI-QSS protocol using polarization and spatial-mode degrees of freedom for enhanced error resilience and achieved a key rate improvement of three orders of magnitude over the original MDI-QSS protocol under a $100$ km transmission distance. Zhang et al. \cite{Zhang2023} developed a secure protocol against Trojan horse attacks. Memory-assisted MDI-QKD, combined with HSPS and the decoy-state method, was introduced in \cite{Zhang2025a}.
\par
 Chen et al. \cite{Chen2016} demonstrated finite-key performance using a biased decoy-state approach, further extended by an asymmetric decoy-state method achieving secure communication over $43.6$ km \cite{Chen2017}. Protocols using W-states \cite{Zhu2015} and cluster states \cite{Liu2017} showed feasibility for distances over $150$ km and $280$ km, respectively.  MDI key agreement protocols are discussed in \cite{Cai2022,Yang2023,Liu2023}.
\paragraph{Continuous variable protocols}
A continuous-variable Measurement-Device-Independent multiparty quantum communication protocol was initially explored by Wu et al. \cite{Wu2016}, utilizing squeezed states of light and homodyne measurements to optimize the secret key rate. To execute QCC and QSS communication protocols, they employ a Continuous-Variable GHZ state, a multipartite entangled state with squeezed uncertainties in relative position and total momentum \cite{Loock2003}. In the case of the tripartite CV GHZ state, their positions and momenta satisfy the relations $\hat{X}_1-\hat{X}_2\rightarrow 0$, $\hat{X}_2-\hat{X}_3\rightarrow 0$, and $\hat{P}_1+\hat{P}_2+\hat{P}_3\rightarrow 0$.\\
The security analysis in \cite{Wu2016} addresses two types of attacks: entangling cloner and coherent attacks. Under the entangling cloner attack, the maximal transmission distances can be extended in scenarios of unbalanced distribution. In contrast, the coherent attack notably diminishes the maximum transmission distances. A coherent state-based MDI multiparty protocol was investigated in \cite{Zhou2017}, demonstrating superior performance compared to the squeezed state-based MDI protocol in terms of experimental realizations. \\
The three-party CV GHZ state in \cite{Wu2016} is not prepared and then distributed; instead, it is obtained through postprocessing using the concept of entanglement swapping. Conversely, Guo et al. \cite{Guo2017} employ a four-party GHZ state to execute the CV-MDI QSS protocol. Specifically, the four participants prepare and transmit modulated states to a relay for the generation of a four-party GHZ state. In this protocol, three participants collaborate to acquire the fourth person’s secret key by leveraging the GHZ state. Furthermore, given that the detection apparatus inherently possesses imperfections, which do not compromise security but can diminish the generation rate of the final secret key, optical amplifiers are deployed to enhance the signal and compensate for these inherent imperfections. This deployment results in an increased transmission distance. The same conclusion was found for the CV-MDI QCC \cite{Li2017}. The continuous variable measurement-device-independent quantum secret sharing and quantum conference based on a four-mode cluster state with different structures were conducted by Wang et al. \cite{Wang2019}. \\
Ottaviani et al. introduced an MDI-modular network in their work \cite{Ottaviani2019}, presenting a modular design for continuous-variable networks. In this architecture, each module functions as a MDI star network. Within each module, users transmit modulated coherent states to an untrusted relay, thereby establishing multipartite secret correlations through a generalized Bell detection mechanism. Their investigation revealed that under ideal conditions, up to 50 users can achieve private communication exceeding 0.1 bit per use within a radius of 40 m, comparable to the size of a large building. Fletcher et al. \cite{Fletcher2022} utilized the same generalized Bell detection technique to establish multipartite correlations between user variables. Their study demonstrates that postselection procedures based on performing reconciliation on the signs of prepared quadratures of coherent states can be effectively used to broaden the protocol's operational range.

\subsubsection{Experiments in Measurement Device Independent QKD}

MDI-QKD was an important advance in that it reduces vulnerability to detector attacks, while being feasible with current technology. Not long after the concept was introduced, several proof-of-principle realizations were achieved. Ref. \cite{rubenok13} reported a demonstration of MDI-QKD over more than 80 km of spooled fiber as well as in inter-city fiber links. A demonstration using polarization qubits over two optical fiber links of 8.5 km each employed a full-polarization control system to stabilize and control the polarization drift in the fibers \cite{FerreiradaSilva2013}. Moreover, the feasibility of  MDI-QKD with polarization encoding was demonstrated in 10km of telecom fiber using standard off-of-shelf devices \cite{tang14}.  Other sophisticated implementations using decoy-state MDI-QKD have been realized over tens and even hundreds of km of optical fibers~\cite{liu13,tangpan14,yin16}. MDI-QKD
Progress has advanced quite rapidly, a summary MDI-QKD in terms of bit rate/distance is shown in Fig. 
\ref{fig:mdiqkd}. 
\par
Regarding the actual establishment of metropolitan communication networks based on the security of the MDI-QKD protocol, many advances have already been made, including the construction of a star-type quantum network in a metropolitan area of 200-square-km, which in addition to providing a high transmission bit rate, also proven to be safe against detection attacks~\cite{tangpan16}. MDI-QKD has been implemented in quantum channels that coexist in the same fiber with classical data channels \cite{Valivarthi2019,Berrevoets2022}.
\subsection{Detector-device-independent quantum key distribution}
As discussed in the previous section, implementing the MDI-QKD protocol requires the interference of photons from two separate lasers, making its implementation more challenging than conventional QKD schemes. Another issue lies in the finite-key analysis, which demands a relatively large post-processing data block to achieve optimal performance.
\par
To address these challenges, an alternative approach called detector-device-independent QKD (DDI-QKD) has been proposed \cite{Lim2014, Gonzalez2015, Liang2015}. While DDI-QKD shares a similar conceptual framework with MDI-QKD, it differs in its use of a `black box' model. In this method, Alice and Bob ensure that their measurement systems do not leak any unwanted information to external sources. This is accomplished, in principle, by replacing the measurement apparatus in Bob's laboratory with a device built by them but not necessarily characterized. Additionally, DDI-QKD replaces the two-photon Bell state measurement (BSM) with a 2-qubit single-photon BSM, eliminating the need for two-photon interference from independent light sources.
\par
An example of such a protocol works as follows: Alice encodes BB84 polarization states in single photons, which she sends to Bob. Bob encodes his information into the spatial degree of freedom of the incoming photons (two modes). This is achieved using a 50:50 beam splitter along with a phase modulator that applies a random phase to each incoming signal. Bob then performs a BSM on the two qubits (polarization, spatial modes) to project each input photon into a Bell state. The remaining steps of the protocol are identical to MDI-QKD.
\par
Despite the anticipated strong performance and partial security proofs, Sajeed et al. \cite{Sajeed2016} demonstrated that the security of DDI-QKD cannot rely on the same principles as MDI-QKD. They demonstrated two key security vulnerabilities. First, DDI-QKD's security is not based on postselected entanglement, and a blinding attack renders its security. Second, Sajeed et al. \cite{Sajeed2016} revealed that DDI-QKD is vulnerable to detector side-channel attacks, as well as other side-channel attacks that exploit imperfections in Bob's receiver. The source of these vulnerabilities seems to stem from Bob’s preparation process, which, unlike MDI-QKD, can be influenced by Eve through the signals she sends to Bob.

\subsection{Receiver-device-independent QKD}
Ioannou et al. introduced another prepare-and-measure SDI-QKD protocol \cite{Ioannou2022a}, where the sender’s device is partially trusted, while the receiver’s device is treated as a black box. They called these protocols ``receiver-device-independent quantum key distribution (RDI-QKD).'' The main assumption in RDI protocols is to bound the pairwise (possibly complex) overlaps between the various states prepared by Alice, denoted as $\gamma_{ij} = \bra{\psi_i}\psi_j\rangle$. The states $\ket{\psi_x}$ represent the quantum systems prepared by Alice’s devices or, more generally, the states of all systems outside Alice’s lab, conditioned on her applying the preparation sequence labeled by $x$. If Alice’s states are mixed, their purifications must also satisfy the overlap bounds. These bounds prevent any side-channel from leaking additional information about $x$ to Eve. No characterization of the receiver’s (Bob's) device is required, and no fair-sampling assumption is made. As a result, these protocols are resilient to attacks where Eve controls Bob’s device. 
\\
\paragraph{Simplest protocol} 
In the simplest protocol, given a key bit $k$, Alice prepares one of two possible states by setting $x=k$. using a coherent state $\ket{\alpha}$ with two possible polarization states $\ket{\phi_x}=\cos{\frac{\theta}{2}}\ket{H}+e^{i\pi x} \sin{\frac{\theta}{2}}\ket{V}$, she prepares one of the following states ($x=0,1$)
\begin{equation}
    \ket{\psi_x}=\ket{\alpha\cos{\frac{\theta}{2}}}_H \ket{\alpha\sin{\left(\frac{\theta}{2}\right)e^{i\pi x}}}_V.
\end{equation}
The overlap between two preparation is given by $\bra{\psi_1}{\psi_0}\rangle=e^{-2|\alpha|^2}\sin^2{\theta}$ and the main assumption is then written as
\begin{equation}
    \gamma=\bra{\psi_1}\psi_0\rangle\geq C,
\end{equation}
where $C$ is a parameter chosen by the user. \\
Bob, then, performs a measurement of the polarization states. For $y=0$ ($y=1$), he projects the incoming signal to $\ket{\phi_0^{\perp}}$ ($\ket{\phi_1^{\perp}}$). If he gets a click, then the round is conclusive and he outputs $b=0$, otherwise, the round is inconclusive and he outputs $b=1$ and the round will be discarded during sifting.  In the case of an ideal channel without noise and loss, the following statistics will be observed by Alice and Bob
\begin{equation}
    p(b=0|x,y)=1-e^{-|\alpha|^2\sin{(\theta)}^2\sin{(\frac{\pi(x-y)}{2})}^2},
\end{equation}
which is nonzero only when $x\neq y$. Therefore, the raw key can be constructed after removing the inconclusive rounds, by Bob flipping all his bits.    
\paragraph{General case of $n>2$ different preparations \cite{Ioannou2022}} Consider a given ensemble of states $\ket{\psi_x}_{x=0}^{n-1}$ that Alice can prepare, and that Bob can perform binary measurements
$\{ B_{0|y}, B_{1|y} \}_{y=0}^{n-1}$ corresponds to projections onto the polarization states orthogonal to the states that Alice prepares.  Alice randomly chooses a pair of integers $\bold{r}=(r_0,r_1)$ where $0\leq r_0\leq r_1\leq n-1$ and a bit $k$ and sends the state $\ket{\psi}_{x=r_k}$ to Bob, who randomly chooses an integer $y$ ($0\leq y \leq n-1$) and performs the binary measurement $\{B_{0|y}, B_{1|y}\}$. If the outcome is $b=1$, the round will be discarded, otherwise $b=0$. Bob then asks Alice to reveal $\bm{r}$. If $y=r_0$ or $y=r_1$ Bob informs Alice that the round is conclusive otherwise it is aborted.  The main assumption concerns the complex pairwise overlaps between preparation states, which is encompassed in the Gram matrix $G$, whose entries are given by $G_{ij}=\bra{\psi_i}\psi_j\rangle$.
\subsubsection{Security analysis} 
Eve’s information about the secret bit $k$ is bounded by assuming that the Gram matrix $G$ of the set of encoding states is fully characterized and that the probabilities $p(b|x, y)$ are perfectly estimated by Alice and Bob. There is no other restriction on the protocol, no bound on the dimension, and neither any characterization on the prepared states, transmission channel, or measurement device. To attack the protocol, Eve can correlate herself to the state Alice sent and design Bob's measurement. Moreover, Eve can benefit from having a quantum memory. By denoting $p_{\text{succ}}$ as the probability that a round is not discarded, the asymptotic key rate is lower bounded by using the Devetak-Winter key rate formula and gives 
\begin{equation*}
    r^{\text{RDI}}\geq [H(k|\text{Eve},\text{succ})-H(k|\text{Bob},\text{succ})]p(\text{succ})\geq [-\log_2(p_g(e=k|\text{succ}))-h(Q)]p(\text{succ}).
\end{equation*}
Here the second inequality comes from the fact that Bob's entropy can be upper-bounded as $H(k|\mathrm{Bob},\text{succ})\leq h(Q)$ and  Eve's conditional entropy can be lower-bounded by conditional min-entropy $H_{\min}(k|\text{Eve},\text{succ})=-\log_2(p_g(e=k|\text{succ}))$.  $p_g(e=k|\text{succ})$ is the maximal probability that Eve guess the bit $k$ correctly, and is the only quantity that needs to be upper bounded to give a lower bound on $r^{\text{RDI}}$, since the QBER $Q$ and $p(\text{succ})$ can be extracted from the observed statistics $p(b|x,y)$.  
The guessing probability $p_g(e=k|\text{succ})=\frac{p(e=k,\text{succ})}{p(\text{succ})}$ then can be upper bound  by an upper bound on $p(e=k,\text{succ})$. In \cite{Ioannou2022}, it is shown that by using SDP an upper bound on $p(e=x,\text{succ})$ can be obtained. 
\paragraph{SDP method for upper bounding $p(e=x,\text{succ})$}
Let us define the set $\{S_i\}_{i=0}^{s-1}$ where its elements are monomials of the operators $B_{b|y}$ and $E_{e|\mu}$ (Eve's measurements). The $ns\times ns$ moment matrix $\Gamma$ then can be defined as
\begin{equation}
    \Gamma=\sum_{i,j=0}^{n-1}\Gamma_{x x^{\prime}}\otimes \ket{e_x}\bra{e_x}
\end{equation}
with the sub-blocks $\Gamma_{xx^{\prime}}$ defined as $\Gamma_{xx^{\prime}}=\sum_{i,j=0}^{s-1}\otimes  \ket{\hat{e_j}}\bra{\hat{{e_j}}}$ where $\{\ket{e_x}\}_{x=0}^{n-1}$ ($\{\ket{\hat{e_i}}\}_{i=0}^{s-1}$) is an orthonormal basis on $\mathbb{R}^n$ ($\mathbb{R}^s$). If we define $\Gamma_{xx^{\prime}}^{ST}:=\bra{\psi_x}S^{\dagger}T\ket{\psi_x^{\prime}}$ ($S,T\in{\mathbb{S}}$), then the SDP upper bounding $p(e=x,\text{succ})$ is given by  
\begin{eqnarray*}
    \max_{\Gamma} \frac{1}{(n-1)n^2}\sum_{r=0}^{\scriptsize{\begin{pmatrix} n \\ 2 \end{pmatrix}}}\sum_{k=0}^{1}\sum_{y=0}^{n-1}\Gamma_{r_k r_k}^{B_{0|y}{E_{r_k|r}}}(\delta_{y,r_0}+\delta_{y,r_1}),
\end{eqnarray*}
such that 
\begin{eqnarray}
    \Gamma_{xx^{\prime}}^{\mathbb{I} \mathbb{I}}&=&\bra{\psi_x}\psi_{x^{\prime}}\rangle=\gamma_{xx^{\prime}} \forall x,x^{\prime}, \\
    \Gamma_{xx}^{\mathbb{I} B_{b|y}}&=&p(b|x,y), \forall b,x,y \nonumber \\
    \mathrm{tr} (\Gamma_{xx^{\prime}}F_k)&=&f_k,  k=0,\cdots, \nonumber \\
    \Gamma &\geq& 0 \nonumber
\end{eqnarray} 
The first condition is the overlap constraint between the sets of states. The second equation ensures that the moment matrix $\Gamma$ is compatible with the observed correlation $p(b|x,y)$. The matrices $F_k$ and the coefficients $f_k$ are Hermitian and complex, respectively, and are defined to satisfy the constraints on the measurement operators for Bob and Eve. These constraints include positivity, completeness, commutativity $[M_{b|y}, E_{e|y}] = 0$, and the requirement that both $M_{b|y}$ and $E_{e|y}$ are projectors. 
\paragraph{Case study: Ideal qubit protocol} As an example \cite{Ioannou2022}, consider the case where Alice prepares states from a set of \( n \) single-qubit states \( \{\ket{\psi_x}\}_{x=0}^{n-1} \), where
\[
    \ket{\psi_x} = \cos\left(\frac{\theta}{2}\right)\ket{0} + e^{i\frac{\pi x}{n}} \sin\left(\frac{\theta}{2}\right)\ket{1},
\]
for a given \(\theta\). In the presence of loss and noise, the Gram matrix and the probability distribution for this set of states are given by:
\[
    G_{ij} = \cos^2\left(\frac{\theta}{2}\right) + e^{i\frac{2\pi (i-j)}{n}} \sin^2\left(\frac{\theta}{2}\right),
\]
and
\[
    p(b=0|x,y) = \zeta \left(\frac{\lambda}{2} + (1 - \lambda) \sin^2(\theta)\sin^2\left(\frac{\pi (x - y)}{n}\right)\right),
\]
where \(\lambda \in [0,1]\) is the noise parameter, modeled as a depolarizing channel, and \(\zeta \in (0,1]\) represents the loss, modeled by a binary erasure channel with erasure probability \(1 - \zeta\).
\\
By optimizing over $\theta$, for different QBERs ($Q$) and values of $n$, the raw key rate as a function of the transmission $\zeta$ can be derived. In their calculations, Ioannou et al. \cite{Ioannou2022} demonstrated that the lower bound of the key rate asymptotically approaches zero as $\zeta \rightarrow \frac{1}{n}$. This is considered optimal because at $\zeta \rightarrow \frac{1}{n}$, Eve can compromise security by intercepting the states sent by Alice and manipulating Bob’s detector based on her outcome and Bob’s input. Therefore, for any prepare-and-measure protocol, the key rate becomes zero for $\zeta \leq \frac{1}{n}$. Furthermore, the proposed protocol surpasses the B92 protocol (a specific case of the proposed protocol with $n=2$ and fixed $\theta=\frac{\pi}{4}$) in terms of both transmission efficiency and noise tolerance. Similarly, BB84 is also outperformed by a qubit-based RDI protocol using three states.

\subsection{One-sided device-independent quantum key distribution}
\label{sec:1SDIQKD}
\subsubsection{Standard 1SDI-QKD}

Another approach to relaxing technical requirements of DI-QKD systems is to consider an asymmetric scenario, known as one-sided DI-QKD (1SDI-QKD), in which one party trusts their device and the other does not.  Note that this situation might describe QKD between a user who is technologically sophisticated enough characterize and trust their equipment (such as a bank or government agency), with a client with untrusted devices. 
\par
The protocol first introduced by Branciard et al.  \cite{branciard2012one} is as follows: Alice and Bob each receive part of an entangled photon pair. Alice has two binary measurement settings $A_1$ and $A_2$. Since she does not trust her measurement apparatus, it is treated as as a block box with a single bit input to choose between settings. On the other side, Bob has two fully trustful projective measurements $B_1$ and $B_2$ in some qubit subspaces. Alice and Bob might not always detect their photons due to channel losses or inefficient detectors. Alice and Bob will to try to extract a secret key from the outcomes of $A_1$ and $B_1$ and can estimate the information of a possible Adversary (Eve) using the results of $A_2$ and $B_2$. 
Since Bob fully trusts his measurement device, he can safely discard the events where he does not detect a photon, since Eve cannot gain any information from these.  Alice, on the other hand, cannot, and must include no-click events in her analysis. 
For the security proof and key rate in this protocol, let us denote by $\boldsymbol{A_i}=(\boldsymbol{A_i}^{\textrm{ps}},\boldsymbol{A_i}^{\textrm{dis}})$ and $\boldsymbol{B_i}=(\boldsymbol{B_i}^{\textrm{ps}},\boldsymbol{B_i}^{\textrm{dis}})$ the strings of classical bits of Alice and Bob get from the recording of their measurements results. Here $\boldsymbol{A_i}^{\textrm{ps}}$ and $\boldsymbol{B_i}^{\textrm{ps}}$ applied for actual detection (ps for post selection) and $\boldsymbol{A_i}^{\textrm{dis}}$ and $\boldsymbol{B_i}^{\textrm{dis}}$ are for no detection and they will be discarded for the key extraction. Then from $n$-bit strings of $\boldsymbol{A_1}^{\textrm{ps}}$ and $\boldsymbol{B_1}^{\textrm{ps}}$ on which Eve can have some information, Alice and Bob can extract a secret key of length $l$ \cite{renes2012one} where  
\begin{equation}
    l\approx H_{\min}^{\epsilon}(\boldsymbol{B_1}|E)-nh(Q_{1}^{\textrm{ps}}),
\end{equation}
here  $Q_{1}^{\textrm{ps}}$ is the bit error rate between $\boldsymbol{A_1}^{\textrm{ps}}$ and $\boldsymbol{B_1}^{\textrm{ps}}$. By bounding $H_{\min}^{\epsilon}(\boldsymbol{B_1}|E)$ using the chain rule and the data-processing inequality for smooth min-entropy, the following bound on the secret key rate can be obtained 
\begin{equation}
    r\geq \eta_A[1-h(Q_{1}^{\textrm{ps}})]-h(Q_2)-(1-q),
\end{equation}
where $q$ is a measure of how distinct Bob’s two measurements are. $Q_2$ is the bit error rate between $A_2$ and $B_2$. 
\par
In analogy to the connection between DI-QKD and the violation of Bell inequalities, here the security of this one-sided DI-QKD is related to the demonstration of quantum steering. That is, $\eta_A[1-h(Q_{1}^{\textrm{ps}})]-h(Q_2)-(1-q)\leq 0$ can be understood as an EPR-steering inequality \cite{Wiseman07,Jones07,Schneeloch13}. Because closing the detection loophole in a steering experiment is  easier than in a Bell test, 1SDI-QKD is more feasible to realize experimentally. For example, consider a typical experimental setup, where a
source sends maximally entangled two-qubit states to Alice and Bob through a depolarizing channel with visibility $v$, with measurement settings $A_1=B_1=\sigma_z$ and $A_2=B_2=\sigma_x$.  Then, for perfect visibility $v=1$,
 a positive secret key can be obtained when Alice's detection efficiency  $\eta_A>65.9\%$, which is much lower compared to those in DI-QKD. As a comparison, for the cases where Alice and Bob have the same detection efficiency, to close the detection loophole in DI-QKD requires $\eta>94.6\%$ are needed ( $\eta>91.1\%$ for post-selected data).
 \par
As is the case for most conventional proofs, the security for the above 1SDI-QKD protocol provided for the asymptotic limit of infinitely long keys. In practical implementations, the number of signals used for establishing a secure key is finite. For the case of 1SDI-QKD, finite key analysis was addressed by Wang et.al \cite{wang2013finite} based on the asymptotic of 1SDI-QKD presented above. They present the secure key rate of 1SDI-QKD with finite resources by employing the smooth min-entropy and smooth max-entropy\cite{renner2008security,scarani2008quantum}:
\begin{equation}
    l\approx H_{\min}^{\epsilon^{\prime}}(Y_1^{\textrm{ps}}|E)- H_{\max}^{\epsilon^{\prime}}(Y_1^{\textrm{ps}}|X_1^{\textrm{ps}}).
\end{equation}
Using the uncertainty relation for smooth entropies \cite{tomamichel2011uncertainty} and the upper bound for smooth-max entropy \cite{tomamichel2012tight}, the following bound for the key rate will be obtained  
\begin{equation}
    r\geq \eta_A P_{Z}^2 [1-h(Q_{1}^{\textrm{ps}})]- P_{Z}^2[1-q+h(Q_2+\mu)]-\frac{1}{N}\log_2\frac{2}{\epsilon_{\textrm{cor}}},
\end{equation}
where $P_Z$ is the probability that Alice (and also Bob) chooses the measurement in $Z$ basis and $\mu=\sqrt{\frac{n+k}{nk}.\frac{k+1}{k}.\ln\frac{2}{\epsilon_{\textrm{sec}}}}$, with $n$ and $k$ being the length of the raw key and the length of the bit string used for parameter estimation respectively. $\epsilon_{\textrm{cor}}$ is the security parameter bounding the possibility that Alice and Bob have different outputs. 
\par
For comparison purposes the simulation results were done in \cite{wang2013finite} and show that the sifted key rate is consistently lower than that predicted by the asymptotic case, particularly when considering finite-key analysis. Furthermore, the outcomes reveal that the relative difference between the asymptotic and non-asymptotic cases ($\delta=\frac{r_{\infty}-r_{N}}{r_{\infty}}$) gradually diminishes as the detection efficiency $\eta_A$ increases.
Notably, the investigation also pinpoints the minimum number of exchanged quantum signals required for achieving efficient detection efficiencies. The results illustrate the potential for a non-zero final secret key rate, approaching $9\times10^6$, specifically when $\eta_A$ reaches 0.67. This underscores the viability of attaining substantial secret key rates even in scenarios involving moderate detection efficiency.  \par
The protocols mentioned above are the QKD schemes that encode a discrete variable (DV) key in a two-dimensional space, typically encoded into a pair of entangled photons. Considerable attention has also been devoted to schemes that instead utilize the quadratures of the optical field, in which one has access to deterministic, high-efficiency broadband source and detectors. In this case, the secret key is now a continuous variable (CV) that is encoded in states living in an infinite-dimensional Hilbert space. This kind of protocol has some advantages over the discrete variable counterpart. The very important ones are that in the CV case, detection-loophole-free tests have been experimentally feasible for over 30 years \cite{ou1992realization} and very strong violations of steering inequalities have been demonstrated. These benefits provide enough motivation for studying the possible one-sided device-independent CV QKD (1SDI-CV-QKD). This was done by Walk et.al \cite{walk2016experimental} where they studied Gaussian CV-QKD protocols from the perspective of 1SDI-QKD against collective attacks, and showed that 1SDI-CV-QKD is possible even with coherent states. The existence of non-zero key rates was connected to the steering parameters for Gaussian states. An experimental implementation achieved positive secret keys under a lossy channel for both entanglement based and coherent state protocols. A version of a 1SDI-CV-QKD protocol that generates a finite and composable key and is secured against coherent attacks was reported by Gehring et.al \cite{gehring2015implementation}. The experiment used two continuous wave optical light fields whose amplitude and phase quadrature amplitude modulations were mutually entangled, and CV equivalent of the BBM92 protocol for discrete variables was implemented. This scheme  is  secure against memory-free attacks performed on Bob’s untrusted detector, that is,
attacks that are independent of Bob’s previous measurement, and secure against Trojan-horse attacks on the source that usually threaten electro-optical modulators commonly used in Gaussian-modulation QKD protocols. A hybrid scheme where Alice uses a Gaussian-modulated coherent state while Bob uses a two-mode squeezed state was studied in \cite{xin2020one}.
\par
A 1SDI-QKD protocol using high-dimensional time-energy entanglement was proposed in Ref. \cite{bao2017time}. The security of this scheme was established by applying the entropic uncertainty relation introduced in \cite{niu2016finite} against coherent attacks. Their numerical results demonstrate that the protocol achieves higher bit rates per two-photon coincidence count while requiring lower detection efficiencies compared to the original 1SDI-QKD protocol (achieving a key rate with $\eta_A=50\%$). This improvement stems from the limitation imposed by photon information efficiency in the original 1SDI-QKD protocol, which restricted the key generation rate to no more than 1 bit per coincidence. Encoding information in high-dimensional photonic degrees of freedom proves to be an efficient approach for overcoming this limitation.
\subsubsection{Generalized 1SDI-QKD and Quantum Secret Sharing}
In 2017, Kogias et.al. \cite{kogias2017unconditional} tackled the problem by considering the protocol as a generalized 1SDI-QKD problem for a continuous-variable version of QSS \cite{hillery1999quantum}. They started with the simplest case involving three parties, Alice, Bob, and Charlie. Alice is fully trusted and shares the secret using a three-mode continuous-variable entangled-state. She keeps one mode and sends the other modes to the untrusted players, Bob and Charlie, through individual unknown quantum channels. In this way, the protocol can be seen as a generalized 1SDI-QKD protocol from Alice (trusted part) to Bob and Charlie as untrusted players. Alice is assumed to perform two homodyne measurements of two canonically conjugate quadratures $\hat{x}$ and $\hat{p}$, and her goal is to establish a unique secret
key, not with Bob’s or Charlie’s individual measurements (as in standard two-party QKD), but with a collective (nonlocal) degree of freedom for Bob and Charlie that strongly correlates with one of Alice’s quadratures ($X_A$ for example). Alice sends a sufficient number of states to the players and in each run, they randomly choose measurements and measure their parts, and collect the outcomes $X_i$ and $P_i$. In the next step, all parties announce their measurements and keep the data originating from correlated measurements, using it for extracting a secret key. 
\par
The final bound on the asymptotic key rate to provide unconditional security against general attacks of an eavesdropper, and against arbitrary (individual) cheating methods of both Bob and Charlie, which include the announcement of faked measurements and general attacks of Bob on Charlie’s system and of Charlie on Bob’s system, can be written as 
\begin{equation}
\label{qsskeyrate}
    r\geq -\log(e\sqrt{V_{X_A|\bar{X}}\max\{V_{P_A|P_C},V_{P_A|P_B}\}}), 
\end{equation}
where 
\begin{equation}
    V_{X_A|\bar{X}}=\int d{\bar{X}}p(\bar{X})(\langle X_{A}^2\rangle_{\bar{X}}-\langle X_{A}\rangle_{\bar{X}}^2),
\end{equation}
and $\bar{X}$ is Bob and Charlie's collective degree of freedom that strongly correlated with Alice's quadrature $X_A$. \par
While the resource behind the standard 1SDI-QKD is known to be (bipartite) steering, one could suspect a similar connection with the case of multi-player QSS, which is indeed the case for the case of Gaussian measurements \cite{xiang2017multipartite}. For a generic Gaussian $(n+m)$-mode state $\rho_{AB}$ of a bipartite system, composed of a subsystem $A$ (for Alice) of $n$ modes and a subsystem $B$ (for Bob) of $m$ modes, one can define a steering measure as \cite{kogias2015quantification} 
\begin{equation}
    \mathcal{G}^{A\rightarrow B}(\sigma_{AB})=\max\left\{0,\frac{1}{2}\ln{\frac{\det A}{\det \sigma_{AB}}}\right\}=\max \{0,S(A)-S(\sigma_{AB})\},
\end{equation}
where $\sigma_{AB}=\begin{bmatrix}
    A & C \\
    C^{\textrm{T}} & B \\
\end{bmatrix}$ is the covariant matrix of the state $\rho_{AB}$ \footnote{Any Gaussian state $\rho_{AB}$ is fully specified, up to local displacements, by its covariance matrix $\sigma_{AB}$ with the elements $\sigma_{ij}=\mathrm{Tr}[\{\hat{R_i},\hat{R_j}\}\rho_{AB}]$ and $\hat{R}=(\hat{x}_{1}^A,\hat{p}_{1}^A,\cdots,\hat{x}_{n}^A,\hat{p}_{n}^A,\hat{x}_{1}^B,\hat{p}_{1}^B,\cdots,\hat{x}_{m}^B,\hat{p}_{m}^B)^{\textrm{T}}$}. 
This measure has an operational meaning in 1SDI-QKD. For a two-mode entangled Gaussian state with covariance matrix $\sigma_{AB}$, the key rate  can be readily expressed in terms of the $B\rightarrow A$ Gaussian steerability of $\sigma_{AB}$ \cite{walk2016experimental}, yielding 
\begin{equation}
    r\geq \max \{0,\mathcal{G}^{B\rightarrow A}+\ln{2}-1\}.
\end{equation}
\par
The Gaussian steering measure $\mathcal{G}$ is monogamous and then satisfies a Coffman-Kundu-Wootters type monogamy inequality \cite{xiang2017multipartite},  in direct analogy with entanglement \cite{coffman2000distributed}. For an $m$-mode Gaussian state with covariance matrix $\sigma_{A_1,\cdots,A_m}$, the following inequalities hold for each party $A_j$ composed of a single mode ($n_j=1, 1\leq j\leq m$): 
\begin{equation}
\begin{array}{cc}
     \mathcal{G}^{(A_1,\cdots,A_{k-1},A_{k+1},...,A_m)\rightarrow A_k}(\sigma_{A_1,\cdots,A_m})&-\sum_{j\neq k}\mathcal{G}^{A_j\rightarrow A_k}(\sigma_{A_1,\cdots,A_m})\geq 0,  \\ 
     \mathcal{G}^{A_k\rightarrow (A_1,\cdots,A_{k-1},A_{k+1},...,A_m)}(\sigma_{A_1,\cdots,A_m})&-\sum_{j\neq k}\mathcal{G}^{A_k\rightarrow A_j}(\sigma_{A_1,\cdots,A_m})\geq 0.
\end{array}
\end{equation}
 For the tripartite case, this becomes 
\begin{eqnarray}
\label{monogamy}
    \mathcal{G}^{(AB)\rightarrow C}(\sigma_{ABC})- \mathcal{G}^{A\rightarrow C}(\sigma_{ABC})-\mathcal{G}^{B\rightarrow C}(\sigma_{ABC})\geq 0, \\
    \mathcal{G}^{C\rightarrow (AB)}(\sigma_{ABC})- \mathcal{G}^{C\rightarrow A}(\sigma_{ABC})-\mathcal{G}^{C\rightarrow B}(\sigma_{ABC})\geq 0. \nonumber
\end{eqnarray}
In analogy with the case of entanglement, residual Gaussian steering (RGS) can be defined by calculating the residuals from \eqref{monogamy} and minimization over all mode permutations. Therefore, in the case of a pure three-mode Gaussian state, the RGS can be defined as 
\begin{equation}
    \mathcal{G}^{A:B:C}(\sigma_{ABC}^{\textrm{pure}})=\min_{\langle i,j,k \rangle} \{\mathcal{G}^{(jk)\rightarrow i}-\mathcal{G}^{j\rightarrow i}-\mathcal{G}^{k\rightarrow i}\}.
    \label{eq:RGS}
\end{equation}
This quantity is a monotone under Gaussian local operations and classical communication, such that a nonzero value of the RGS certifies genuine tripartite steering \cite{he2013genuine}. Therefore, it can be regarded as a meaningful quantitative indicator of genuine tripartite steering for pure three-mode Gaussian states under Gaussian measurements. Returning to the key rate of the QSS protocol \eqref{qsskeyrate}, the mode-invariant QSS key rate bound $K_{\textrm{full}}^{A:B:C}$ 
 that takes into account eavesdropping and potential dishonesty of the players can be obtained by minimizing the right-hand side
of \eqref{qsskeyrate}  over permutations of $A$,$B$,and $C$. It was found that it admits the exact linear upper and lower bounds as a function of the RGS \eqref{eq:RGS}:
\begin{eqnarray}
\frac{\mathcal{G}^{A:B:C}(\sigma_{ABC}^{\textrm{pure}})}{2}-\ln{\frac{e}{2}}\leq K_{\textrm{full}}^{A:B:C}(\sigma_{ABC}^{\textrm{pure}}) \leq  \mathcal{G}^{A:B:C}(\sigma_{ABC}^{\textrm{pure}}) -\ln{\frac{e}{2}}.
\end{eqnarray}
Thus, partial DI QSS yields a direct operational interpretation for the RGS in terms of the guaranteed key rate of the protocol. \par
 EPR steering is a necessary requirement for non-zero key rates in all of the protocols mentioned above. Therefore, it is essential to have a procedure for generating EPR steering between two or more distant parties. Xiang et al. \cite{xiang2019multipartite} designed a protocol that allows the distribution of one-way Gaussian steering. This can be subsequently employed for 1SDI-QKD and also for three-user scenarios to distribute richer steerability properties, including one-to-multimode steering and collective steering, which can be utilized for 1SDI quantum secret sharing. Since all of their protocols can be implemented with squeezed states, beam splitters, and displacements, they can be readily realized experimentally. A related experiment was done by Wang et.al. \cite{wang2020deterministic} which experimentally demonstrate the deterministic distribution of Gaussian entanglement and steering with separable ancillary states both in two-user and multi-user scenarios by preparing independent squeezed states and applying classical displacements on them, which makes initial states fully separable. In a later development in 2023, Lv et al. \cite{lv2023sharing} demonstrated that a 2-qubit entangled state can consistently produce steering through sequential and independent pairs of observers, given that the initial pair shares either a pure entangled state or a specific category of mixed entangled states.

\subsubsection{Experiments: One-sided Device Independent QKD}
1SDI-QKD is rigorously based upon the loophole-free observation of EPR-steering (also known as quantum steering) \cite{Wiseman07,Jones07}.  As EPR-steering is below Bell nonlocality in the hierarchy of correlations, 1SDI-QKD provides a security paradigm that is less robust than that of full-DI QKD. However, it is much easier to close the detection loophole for EPR-steering than in Bell non-locality.  As such 1SDI-QKD can be realized with much lower detection efficiencies. In fact, EPR-steering can be observed for arbitrarily large losses in the DV context, provided that a sufficiently large number of measurements can be realized on a bipartite state with sufficient entanglement  \cite{Bennet12,Evans13}.  1SDI-QKD has also been implemented in continuous variable systems, with the advantage that Gaussian states and measurements can be used \cite{Gehring2015,Walk2016}.      
\blk 

\section{Towards future a DI-QKD network:  Requirements, Challenges and Solutions}\label{sec:experiments}
DI-QKD has the appeal that it can help resolve security risks associated to implementation issues, as it aims to provide information theoretic security with minimum physical assumptions and uncharacterized hardware, thus reducing or eliminating many of the side-channels and security concerns in real-world deployment. However, DI-QKD requires satisfactory demonstration of Bell non-locality over long distances, and as such introduces demanding technical requirements, in particular related to the distance limitations (signal loss requires advanced quantum technology such as quantum repeaters) and high efficiency (high-quality detectors, sources, devices) required to achieve reasonable key rates.        

In this section we discuss the current outlook towards real-world implementation of DI-QKD, presenting the first set of DI-QKD experiments and additional promising experimental platforms, the current technical challenges, and possible solutions. In Section \ref{sec:integration}, we discuss the efforts towards real-world deployment of QKD in general, including efforts towards standardization, interoperability, and integration into cybersecurity and network architecture, since future implementation of DI-QKD will most likely benefit from most of this groundwork.  When possible, we mention MDI-QKD benchmarks, and  highlight specific or unique challenges that DI-QKD will likely encounter on the road to real-world deployment. 
\subsection{The first DI-QKD Experiments}
In July 2022, three independent research groups (one based in the UK~\cite{Nadlinger2022}, one in Germany~\cite{Zhang2022a}, and one in China~\cite{Liu2022})
successfully implemented DI-QKD in three different platforms.
The first successful photonic-based experiment~\cite{Liu2022} demonstrated DI-QKD (against collective attack) over fiber lengths of $d=20,100$ and $220$ m and verified that the measured correlations between the entangled photons were strong enough to guarantee a positive secret key rate. Optimized high-efficiency entangled photon source and single-photon detectors resulted in a single photon heralding efficiency slightly greater than $87\%$. To address the locality loophole, the experiment employs a shielding assumption that prohibits unnecessary communications between untrusted devices and a potential adversary. Essentially, it ensures that the information about the input choices and output results of one party remains unknown to the other party and Eve.
\par
Proof-of-principle memory-based experiments reported in Oxford~\cite{Nadlinger2022} and Munich~\cite{Zhang2022a} used entangled strontium ions and entangled rubidium atoms, respectively. Each system has its advantages and disadvantages. However, when compared to all-photonic implementations, these experiments have the advantage that they can be ``event-ready", so that link efficiencies can be partially mitigated. 
The Oxford-based experiment was the only one of the three to complete the entire DI-QKD protocol, with a generated secret key of 95 kbits between Alice and Bob achieved over about 8 hours and with a distance of $d=2$ m between the two quantum memories. The Munich-based experiment created heralded entanglement between atoms separated by a line of sight distance of $d=400$m ($700$m of optical fiber). An asymptotic security analysis was performed, and a secret key rate of 0.07 bits per entanglement generation event was obtained from results accumulated over 75 hours.
\par
These three landmark experiments demonstrate that experimental DI-QKD is feasible.  In addition, they highlight the  technical challenges to be face in achieving real-world utility.  In the next two sections we discuss some of the fundamental requirements and technical challenges to be faced for the next generation of DI-QKD experiments, and present some of the techniques and proposals meant to address them.

\subsection{Bell loopholes in the DI-QKD scenario}
Chapter 2  discussed several of the loopholes in conclusive Bell tests.  Most of these loopholes are also present in DI-QKD, with several caveats. 
The locality loophole, that is, having space-like separation between measurement processes, as discussed in Section~\ref{sec:locloop}, is necessary to rule out LHV theories. However, to guarantee  security in DI-QKD we must ensure that the users stations do not leak any information to an adversary Eve, even at sub-luminal velocity. From Eve's point of view it is most likely much easier to install a backdoor that sends information from a user's devices to her station, rather than make them communicate with each other to fake a Bell inequality violation.

Thus, DI-QKD calls for complete isolation of the measurement stations, involving shielding, electromagnetic or otherwise, to avoid broadcasting of any type of signal related to measurement basis choices and outcomes.  In addition, whether the quantum systems are photonic or stationary, the users stations are connected by a photonic channel, which in principle opens a backdoor for side-channel attacks using external light sources, as has been exploited for fake Bell violations~\cite{Gerhardt2011}, in QKD~\cite{Qi2005,Lamas-Linares:07,zhao08,lydersen10,wiechers11,qin18} and QRNGs~\cite{garcia20,smith21}. To isolate the users stations, a switch or shutter mechanism should be used to block the optical channel after the relevant optical signal has passed, and before the measurements are performed. In a recent DI-QKD implementation with trapped ions, this was achieved by shifting the ions out of the focal point of the collection lens, thus decoupling them from the optical link, and also scrambling the quantum state after measurement, so that the state after measurement (and thus the measurement result) could not be determined by a third party probing the ions~\cite{Nadlinger2022}. 

Another difference between Bell and DI-QKD scenario lies in the memory loophole (see Section~\ref{sec:memoryloopholeDIQKD}.). The ability of the devices to remember the inputs and outcomes of the previous rounds to be used in the future has been proved to be of very little consequence for Bell inequalities already in Ref.~\cite{gill03}. For DI-QKD, on the other hand, memory attacks pose a very serious threat~\cite{barrett2013memory,Masanes2014}. While some countermeasures against them are possible, there is no known method of full protection. 

The experimental loophole which is of crucial importance for both Bell and DI-QKD scenarios is the detection efficiency loophole, which currently is the main problem in experimental realizations and implementations, as we address in the next section.

\subsection{Detection efficiency and channel losses}
\label{sec:losses}
The issues of detection efficiency and channel losses are intimately related in determining the performance characteristics of a DI-QKD link. In both all-photonic setups and those with stationary qubits, the efficiency in which photons can be detected at a distant measurement station is a critical metric in determining the overall performance characteristics of the system.   

The overall detection efficiency of a photon can be expressed as $\eta=\eta_c \eta_\ell \eta_m \eta_d$, where $\eta_c$ is the coupling efficiency  from the source to the optical link (ex: optical fiber), $\eta_\ell$ is the transmission efficiency of the optical link, $\eta_m$ is the efficiency of the measurement device, and  $\eta_d$ is the quantum efficiency of the detector.  

The efficiencies $\eta_c$, $\eta_m$ and $\eta_d$ depend upon specific characteristics of the source, the optical components of the measurement device and detectors. For example, coupling or ``collection" losses can range from near zero  to several tens of dB in the case of coupling from an optical source into an optical fiber  \footnote{Here we give efficiency in terms of probability, and losses in terms of dB.  For loss $L$, one has $\eta=10^{-L/10}$.}. While bulk optical components such as (polarizing or non-polarizing) beam splitters can present very low losses $\leq 0.5\%$ ($\sim$ 0.02dB), fiber-based components can have losses up to a few dB.
State-of-the-art commercial superconducting single-photon detectors typically have $\eta_d \leq 0.85$, but efficiencies reaching over 0.95 have been reported \cite{Reddy19,Chang2021}. These efficiencies do not typically depend upon the propagation distance within the optical channel.  

The link efficiency $\eta_\ell$, on the other hand, does depend upon the propagation distance, decaying exponentially with the length $\ell$ of the channel \cite{OFLAA,Pirandola2017,Takeoka2014}. In particular, $\eta_\ell=10^{-\gamma \ell/10}$, where $\gamma$ is the attenuation coefficient in dB/distance. Losses in an optical fiber link include contributions from attenuation (typically value $\sim$ 0.2dB/km in the telecom band) that accumulate over distance, as well as from fiber splices connecting different sections of fiber. Mechanical splices using barrel connectors typically have losses greater than 0.5dB/connection, while fusion splicing can give losses less than 0.01dB/splice in standard single-mode fiber, showing the necessity of dedicated high-quality optical fiber links for DI-QKD. Moreover, networking hardware, such as optical switches, can also present losses of several dB.  

Achieving the critical detection efficiency required for DI-QKD (typically $\eta > 80\%$, see Section \ref{sec:DI-QKD}) in an all-photonic setup (such as a single SPDC source) presents significant technical challenges.
Even considering  $\eta_c=\eta_m=\eta_d=1$ and  that all link loss is due to attenuation, $\eta_\ell=0.8$ (or about 1dB loss) corresponds to  $\sim 4.85$km of propagation in an optical fiber.  Fortunately, ``event ready" setups can be used to overcome the probabilistic nature of most sources of entangled particles, as well as low collection efficiency and losses between source and the detection stations. These are schemes in which the presence of the entangled state at the respective detection sites is heralded by a separate detection event~\cite{bell80,zukowski93,simon03}. While there have been proposals and experiments involving all-photonic event-ready setups, a considerable advantage arises when employing stationary quantum systems such as ions, atoms, quantum dots and NV centers, since these can be measured with efficiency close to unity, making event-ready setups involving entangled stationary qubits one of the most promising path towards useful implementation of DI-QKD. We discuss event-ready sections in further detail in Section \ref{sec:eventready}.  
In addition, we note here that there has been theoretical progress in reducing the CDE for DI-QKD, by including pre- and post-processing, as discussed in Section \ref{sec:DI-QKD}. 
\subsection{High-quality entanglement sources}
\begin{table}[]
     \centering
     \begin{tabular}{|c|c|c|c|c|}
     \hline
        Setup type & CHSH value & QBER & Raw bit rate & Distance  \\
        \hline
        point-to-point SPDC & $\eta$, $\mathcal{E}$ & $\mathcal{E}$ &  $\eta$, $B$ & $\eta$  \\
        \hline
       Event ready & $T_2$,$\mathcal{E}$ &$T_2$,$\mathcal{E}$ & $\eta$, $B$ & $T_2$ \\
       \hline
     \end{tabular}
     \caption{Summary of relevant DI-QKD parameters and technical characteristics affecting them for both point-to-point photonic setups with SPDC sources, and event-ready setups using stationary qubits. Here $\eta$ is the overall photonic detection efficiency, $\mathcal{E}$ is the decoherence of photonic quantum systems in the optical link, $T_2$ is the decoherence time of stationary node qubits, and $B$ is the overall ``brightness" (entangled pairs created/sec) of the source. 
\label{tab:setupsum}}
\end{table}
High-quality sources of entangled quantum systems are a necessary resource for DI-QKD. Quality refers not only to robust violation of detection loophole-free Bell-inequalities, but also a high brightness $B$ (or repetition rate $R$), as these two characteristics have a direct effect on the key rates obtainable. In addition to the overall detection efficiency discussed in Section~\ref{sec:losses}, decoherence in the channel (such as depolarization, dephasing, etc) will also degrade the quality of entanglement.  
In the next two subsections, we describe the two principle entanglement sources used for DI-QKD. A summary of the merits of these sources for DI-QKD in terms of the relevant experimental parameters is given in Table~\ref{tab:setupsum}.
\subsubsection{Spontaneous Parametric Down-Conversion sources}
\begin{figure}
    \centering
    \includegraphics[width=0.6\linewidth,trim=0cm 2.5cm 0cm 0cm,clip]{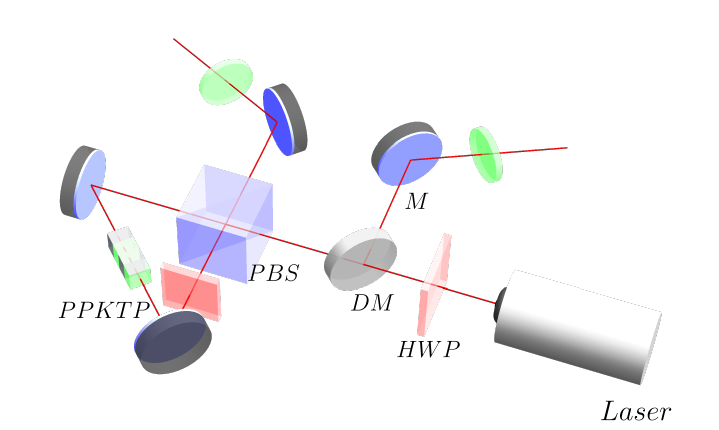}
    \caption{Sagnac source of entangled photons: A pump laser is directed through a half-wave plate (HWP) to control polarization, followed by a dichroic mirror (DM), which separates the pump beam by reflecting one wavelength while transmitting the other. The beam enters a polarizing beam splitter (PBS), splitting it into orthogonal polarization propagating paths. These paths pass through a periodically poled nonlinear crystal (PPKTP) inside the interferometer loop, generating photon pairs via spontaneous parametric down-conversion (SPDC). Mirrors (M) guide the beams, which recombine at the PBS. A HWP is placed in the reflecting path to adjust de polarization. The entangled photon pairs are separated using dichroic mirrors and sent to high pass filters (green) to be detected at single-photon detectors.}
    \label{fig:spdcsource}
\end{figure}
A major step in experimental Bell tests was the development of spontaneous parametric down-conversion sources (SPDC) as a source of entangled photon pairs in the 1990's~\cite{kwiat95,kwiat99}, which offered much higher count rates than the first generation of experiments based on atomic cascade~\cite{freedman72,aspect1982experimental}. The most efficient SPDC sources today are based on periodically poled nonlinear crystals in Sagnac interferometers, as shown in Fig.~\ref{fig:spdcsource}. An adequate choice of crystal length and optics produces highly pure entangled polarization states, reaching state fidelities over 99.5\%, where the transverse spatial mode of the photons is optimized for coupling into single-mode fibers~\cite{Gomez2019,Meraner2021,Liu2022}. Coupling efficiencies over 95\% have been achieved~\cite{Liu2022}. 
SPDC is a probabilistic source of photon pairs, and the state fidelities refer to the post-selected state obtained when two photons are registered. The full SPDC output is described by the two-mode squeezed vacuum state 
\begin{equation}
    |TMSV\rangle = \sum_{n=0}^\infty \frac{\bar{n}^{n/2}}{(1+\bar{n})^{(n+1)/2}}  |n\rangle_1 \otimes |n\rangle_2,
    \label{eq:TMSV}
\end{equation}
where $\bar{n}$ is mean photon number in each mode $1,2$. The ratio of multiple pair events to single pair events is $\bar{n}$, and the probability of a multi-pair event is $P_{n>1}=\bar{n}^2/(\bar{n}+1)^2$. Thus, there is a trade-off between the brightness achievable and the fidelity, since multiple-pair events become non-negligible at high pump beam intensity. Since photons from different pairs are uncorrelated, multiple pair events limit the quality of the two-photon state, especially for pulsed sources~\cite{cosme08PRA}.  Indeed, it has been shown that the absence of post-selection results in a limited violation of Bell inequalities~\cite{Tsujimoto2018}.  The full photon statistics of state \eqref{eq:TMSV}  have been considered in key rate analysis for DI-QKD \cite{Ho2020}, where it was shown that unity key rate is impossible, even with perfect detection efficiency.  In addition, the authors showed that, by adding artificial noise,   the critical detection efficiency can be lowered to $83.2\%$, even with the multi-photon events of SPDC sources.  
\par SPDC  sources have been used for point-to-point DI-QKD~\cite{Liu2022}, and can also be incorporated into event-ready setups using absorptive quantum memories (see next section). The main parameters affecting performance of SPDC sources for point-to-point DI-QKD are the detection efficiency $\eta$, the overall brightness $B$ (pairs emitted/time) and the quality of the entangled states reaching the measurement devices, which for simplicity we will describe in terms of a general decoherence channel $\mathcal{E}$, which may incorporate imperfections in the source as well as noise in the optical channel. The detection efficiency $\eta$ affects the obtainable bit rate of these sources, since not only do both photons need to be detected to establish a key bit, but also through the obtainable loophole-free CHSH violation (see Section \ref{sec:DI-QKD}). Decoherence $\mathcal{E}$ can affect both the QBER as well as the CHSH violation. In some cases, decoherence is due to random unitaries operations (such as phase fluctuations or polarization rotations), which can be monitored and corrected. The source brightness $B$ affects only the raw bit rate in principle, aside from multi-pair events that will result in uncorrelated events, increasing the QBER and lowering the CHSH value.        

\subsubsection{Heralded or Event Ready Setups}
\label{sec:eventready}

\begin{figure}
\centering
    \begin{subfigure}{0.3\textwidth}
    \resizebox{\textwidth}{!}{
    \includegraphics{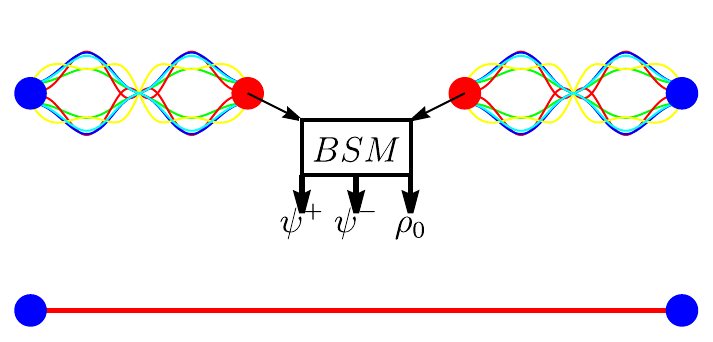}
    }
    \caption{\textit{Basic concept}}
    \label{fig:bsm1}
    \end{subfigure}
    \hfill
    \begin{subfigure}{0.33\textwidth}
    \resizebox{\textwidth}{!}{
    \includegraphics{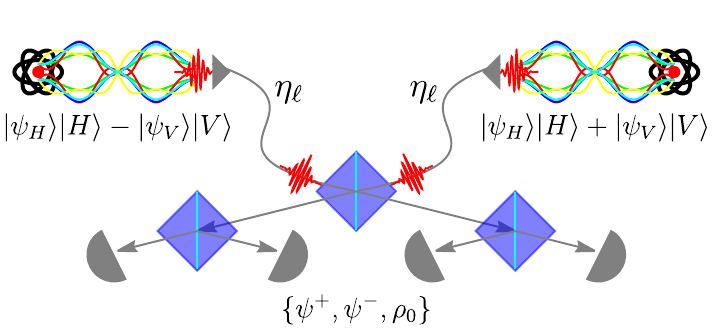}
    }
    \caption{\textit{Polarization mode}}
    \label{fig:bsm2}
    \end{subfigure}
    \hfill
    \begin{subfigure}{0.33\textwidth}
    \resizebox{\textwidth}{!}{
    \includegraphics{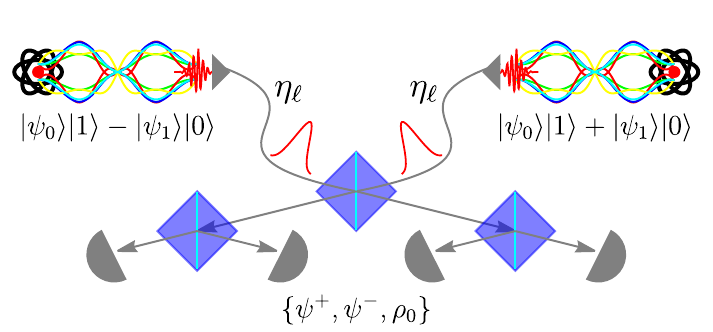}
    }
    \caption{\textit{Counting mode}}
    \label{fig:bsm3}
    \end{subfigure}
    \caption{\textit{Event--ready setups} -- a) Remote subsystems $A$ and $D$ of two entangled pairs are entangled via a Bell state measurement on their entangled partners ($B$ and $C$). \ref{fig:bsm2}) Event ready scheme with stationary qubits and two photon heralds. Stationary qubits are each entangled with the polarization state of a photon, which are are sent to a central station for the two-photon Bell-state measurement (BSM). Joint detection events at pairs of detectors signal preparation of an entangled state $\psi^\pm=\ket{\psi_H}_A\ket{\psi_H}_B\pm \ket{\psi_V}_A\ket{\psi_V}_B$. \ref{fig:bsm3}) Stationary qubits each emit a photon with probability $p$. The optical modes are coupled, so that detection of one and only one photon results in an entangled state of the stationary qubits. Here, $\psi^\pm=\ket{\psi_0}_A\ket{\psi_1}_B\pm \ket{\psi_1}_A\ket{\psi_0}_B$}
    \label{fig:eventready}
\end{figure}
Two recent DI-QKD experiments used entanglement sources that are ``event-ready" ~\cite{Nadlinger2022,Zhang2022a}.
These sources use entanglement swapping to herald the creation of a remote entangled pair \cite{bell80,zukowski93,simon03}, as illustrated in Fig.~\ref{fig:eventready} a).
 While entanglement swapping can be used to create all-photonic event-ready entanglement sources \cite{pan1998experimental,Takesue2009swapping, Jin2015swapping,Sun17, Samara2021swapping}, the near-unity detection efficiencies achievable with stationary quantum systems such as ions, atoms, quantum dots or NV centers, and the possibility to use them as quantum memory, make these systems most attractive for DI-QKD and quantum networks as a whole. For generality, let us refer to these stationary systems as ``nodes".  
 Through the application of external fields, the node qubit can be entangled with a flying qubit, in the form of emission of an optical pulse that can be coupled into an optical channel (fibers). When the pulses emitted from two nodes are combined at a beam splitter, a Bell state measurement (BSM) can be realized  resulting in an entangled state of the two nodes $A$ and $B$. A BSM with 50\% efficiency can be realized with linear optics (see section \ref{sec:OrigMDIQKD} and Fig. \ref{fig:mdi}), where classical communication from the BSM station to $A,B$ is required for heralding. Optical decoherence $\mathcal{E}$ in the link, such as phase or polarization fluctuations, can prohibit the creation of high-quality entanglement.  If these two nodes are already the end points of the quantum channel, they can be measured after successful node-photon entanglement, and the final results post-selected upon a successful BSM \cite{stolk2024b}. However, the critical advantage to the node systems is in storing entanglement in quantum memories to be used in repeater protocols \cite{azuma2023}.  Once created, the entangled state begins to deteriorate due to a number of possible decoherence processes. The quality of the memory can be quantified by the coherence time $T_2$, which is the time during which phase coherence of the quantum state can be maintained. When the goal is heralding an entangled pair of nodes, $T_2$ also determines the maximum separation distance $L$, since the coherence must be maintained long enough for the BSM station to communicate to the nodes, and subsequent measurements at $A$ and $B$ to occur. That is, $L << v T_2$, where $v$ is the velocity of light in the optical link.  
A coherence time of 10ms has been recently observed for Rubidium atoms, allowing for a link distance over 100km~\cite{Zhou2024Rb}.   

Thus, developing good quantum memories, resulting in increased $T_2$, is crucial for increasing separation distance between nodes, while maintaining the high CHSH violation and QBER required for DI-QKD. 
Entangled quantum memories will also play a crucial role in quantum repeaters, required for quantum networking and establishing long-distance entanglement (see Ref. \cite{azuma2023} for a comparison of decoherence times of candidate platforms and Ref.~\cite{Covey2023} for a review of quantum networks with neutral atoms).
 
Event-ready setups can be divided into two main categories: those that use two photons as heralds and those with a single-photon herald. A double-herald setup requires one photon from each node $A$ and $B$ to arrive at the BSM station (see Fig.~\ref{fig:bsm2}). The BSM relies on two-photon interference, which is inherently robust to phase instability~\cite{mattle96} (low optical decoherence), but is less efficient, as it  requires the emission, arrival and detection of two photons. Thus, if the overall efficiency to detect a single photon at the BSM is $\eta$, the double-herald setup has efficiency $\sim \eta^2$.   
Event-ready entangled states have been generated over tens of km of optical fibers with atomic ensembles~\cite{yu2020} and single atoms~\cite{vanleent2022}. 

To improve efficiency, single-photon event-ready setups can be employed. These again require entanglement between the node system and a photonic system, however in this case the state is of the form $\ket{\psi_0}\ket{1}\pm \ket{\psi_1}\ket{0}$, where $\ket{\psi_j}$ are states of the node system, and $\ket{n}$ are $n$-photon Fock states. This scheme has the advantage in that the relevant events are those where only one of the nodes emits a photon. When the two optical channels are combined at a beam splitter, so that one cannot determine which node emitted the photon, a detection in either output results in an entangled state at A and B (see Fig.~\ref{fig:bsm3}).  Using single-photon events increases the event rate (efficiency $\sim \eta$), but requires optical phase stability to attain high-quality entanglement~\cite{stockill17, humphreys18,pompili21,hermans2023,stolk2024}.

Event-ready schemes are probabilistic, and for heralding and storing entanglement, the generation attempts can only be retried after a time interval that permits two-way communication between the devices and the BSM setup, creating a balance between distance, node decoherence and generation rate. If the distance is too large, then decoherence sets in at the node qubits before successful heralding can be confirmed. Several important figures of merit can be considered. In addition to the fidelity of the state produced, there is the entanglement generation rate, given by $r_{ent}=p_{ent}\times R$, where $p_{ent}$ is the probability that entanglement attempt is successful, and $R$ is the attempt (repetition) rate of the protocol.   The probability $p_{ent}$ is determined by the emission probability of the emitter(s),  optical transmission and collection losses, and losses due to the probabilistic nature of the BSM.  The ratio between $r_{ent}$ and the decoherence rate $r_{dec}$ is known as the ``link efficiency" $\eta_{link}=r_{ent}/r_{dec}$ \cite{humphreys18}.  When $\eta_{link}$ is on the order of unity or better, then entanglement can be created faster than it is destroyed, which can be used to create entangled nodes deterministically. This was achieved in Ref.~\cite{humphreys18} by delivering a separable state in the cases when  entanglement generation fails. More specifically, entanglement creation was attempted repeatedly, and when entanglement was heralded successfully, it was stored for use at the end of the time window. If all attempts failed, then a separable state was delivered at the end of the time window. This results in a state of the form $\rho_{det}=p_{success} \rho_{success}+p_{fail} \rho_{fail}$. If $\rho_{det}$ has a fidelity greater than 1/2 with a maximally entangled state (for $2 \times 2$ systems), then the system delivers entanglement deterministically at the given time intervals.    

Event-ready schemes can be used to eliminate the importance of the optical link inefficiency in DI-QKD, but the entanglement generation rate diminishes exponentially with distance due to the attenuation during light propagation to the BSM. Moreover, most quantum memories emit photons in the visible or near-IR regime, where attenuation can be 1-2 orders of magnitude larger than the telecom bands. A solution to this problem is through  quantum frequency  conversion  to the telecom window using difference-frequency conversion with an intense pump pulse~\cite{mann2023,geus2024}, which has been used to establish entanglement across distances of tens of kilometers for spin systems~\cite{stolk2024b} and 
atomic ensembles~\cite{yu2020,vanleent2022,Luo2022}.

Nodes consisting of absorptive quantum memories~\cite{Lago-Rivera2021, Liu2021afc}, can be entangled in a similar scheme using entangled photons from SPDC. Here each node consists of a memory that is coupled to an SPDC source, so that a single photon can be absorbed, while the other photon is sent to the BSM station. Similar to the single-photon event-ready schemes, if one cannot determine which node produced the photon detected at the BSM, the result is a pair of entangled nodes. In comparison to single-qubit nodes discussed above, here the non-degenerate frequencies of the SPDC photons can be tuned such that one photon is produced at the required frequency for level transition in the memory, and the other at telecom wavelengths for optimal transmission in the optical link to the BSM. 

Finally, we note that several direct entanglement generation schemes of stationary nodes have been demonstrated. In Ref.~\cite{Luo2022}, entanglement between two Rb atomic memories separated by 12.5km was generated by sending a single photon emitted from one memory to be absorbed by the other.  Here, atom-photon entanglement was generated in the first memory node between an atom and a photon at 795nm, which was frequency converted to 1342nm for transmission in the optical fibers. It was sent to the second node, where it was  converted back to 795nm using sum-frequency generation, and stored in the second memory. A second experiment~\cite{kucera2024} produced entangled photons from SPDC, stored the state of one of them in a single-ion quantum memory, while the other was sent to a remote detection station via a 14km deployed urban fiber link.      

An all-optical approach for an event-ready setup is through qubit amplifiers, several of which have been proposed~\cite{Gisin2010,Curty2011,pitkanen2011,meyer2013}, and are based on a previous proposal for probabilistic noiseless amplification for quantum optical signals~\cite{ralph2009}. Based on quantum teleportation, these qubit amplifiers cannot only act as a herald of an incoming signal but also introduce an optical gain on the desired optical mode.
This can reduce or eliminate the effects of transmission losses, but also increases technical demands due to the need for ancilla photons or photon pairs on demand, which must be coupled and detected efficiently with the linear optical measurement device.  A recent finite-key analysis analysis shows that detection efficiencies greater than 96.5\% are required~\cite{zapatero2019} to achieve a positive key rate with 39dB of overall transmission loss (about 195km distance if only fiber attenuation is considered).

Many types of encoding and protocols can be used to produce photon-mediated entanglement for quantum networking (for details see a recent tutorial~\cite{beukers2024}).  To realize long-distance DI-QKD, it will be crucial to realize quantum protocols through quantum networks of different physical types.  In this direction, a quantum network stack has been defined~\cite{dahlberg2019} and realized~\cite{pompili2022}, in analogy with classical networking models such the Open Systems Interconnection (OSI). 
\par
Despite impressive progress in ``event-ready" entanglement generation using solid-state and atomic qubits, and the realization of proof-of-principle DI-QKD with these setups \cite{Nadlinger2022,Zhang2022a}, there still remains a substantial gap between the current state of the art of these systems and the realization of a practical (even short-haul) DI-QKD system. To date, the entanglement generation rates prohibit realization of DI-QKD over larger distances, where the losses in optical fiber become critical.  Moreover, there is a tradeoff between state fidelity and entanglement generation rate \cite{beukers2024}. Multi-photon events can be produced during the production of the node qubit/photonic qubit entanglement.  This requires lower excitation strength or filtering mechanisms \cite{Zhang2022a}, which reduces $p_{ent}$. The photonic interference in the BSM requires spatial, spectral and temporal mode matching, which also requires filtering and leads to losses, lowering $p_{ent}$.    Recent experiments have reported  $p_{ent}\sim 10^{-4}-10^{-6}$ \cite{Nadlinger2022,Zhang2022a,stolk2024b,beukers2024}. Thus, a repetition rate $R$ of a few MHz would lead to a few hundred events per second in the best case scenario.  Thus, considerable technical improvements are still necessary for DI-QKD to reach reasonable key rates over useful distances. One approach to improve these numbers might be through multiplexing of event-ready sources.  

\subsubsection{Link Relays} 
The transmission losses in optical fibers is a limiting factor for all quantum networking protocols, and limits point-to-point links to a few hundred kilometers in length.  
As is well known, the no-cloning theorem prevents quantum information from being copied deterministically, so classical optical amplification techniques cannot be used in the quantum regime. Several solutions exist to overcome this limitation. 

Current fiber-based QKD systems over several hundred kilometers use trusted classical relays to extend transmission length~\cite{Chen2021sat}.   
Conceptually, these trusted relays consist of hardware security modules where the keys from the two neighboring links are stored confidentially. These keys (at all connecting relay points) are then post-processed\footnote{A bitwise XOR, publication of result and correction at one side, is a simple example.}, resulting in a shared key between the two endpoints.  However, DI-QKD is not possible in a trusted relay infrastructure, since entanglement cannot be shared between the two end points.  

To overcome the need for trusted classical hardware devices, quantum relays consisting of quantum repeaters \cite{briegel98,azuma2023} are required.  Many of the same event ready setups described above can be used to build quantum repeaters, which employ multiple stages of entanglement swapping between intermediate nodes to construct a long-distance entangled state between edge nodes. Since efficient entanglement preparation and swapping are typically probabilistic processes, quantum memory devices are needed to store quantum information from one link while swapping is performed on others. The development of robust quantum memory is one of the principle challenges in creating large-scale quantum networks for DI-QKD and other applications. For a review of recent progress on quantum memories and repeaters, see Refs.~\cite{azuma2023,Covey2023}.   

\subsection{Integration of QKD into Cybersecurity Infrastructure} 
\label{sec:integration}
Significant advancements have been made worldwide in the proof-of-concept implementation of QKD networks in real-world scenarios, and their integration into cybersecurity infrastructure.  Of critical importance was to demonstrate how QKD, which establishes a shared key between users in a point-to-point configuration, can be employed within the network architectures used in modern communications. As early as 2002-2006 the DARPA network demonstrated a multi-node QKD network with optical switching, connecting fiber and free-space links using weak coherent pulses and also entangled photons~\cite{Elliot2007darpa}. The issues of routing, trusted relays, key management and integration into communication protocols such as IPSec were also addressed for the first time.  The SECOQC network operated from 2004-2008 in and around Vienna, Austria~\cite{Peev2009}, and included weak coherent pulse, entanglement-based, and continuous variable QKD systems developed by several different groups and institutions.  The SECOQC network demonstrated compatibility and interoperability of these different systems, and employed a ``hop-by-hop" relay scheme, in which a cypher key (to be used for classical encryption) is sent along the chain of trusted nodes using one-time pad encryption between each connected pair of nodes.  Routing and key consumption were also addressed.  From 2009 to 2011, the three-node SwissQuantum network was deployed in Geneva, Switzerland~\cite{Stucki2011}. The keys generated were tested for various applications, including high-speed commercial OS layer 2 encryptors (10 Gbit/s Ethernet), research platforms for encryption and authentication, and IPSec encryptors. In 2010, high-speed QKD systems running at GHz clock rates were developed and deployed in Tokyo, enabling encrypted video conferencing over 45km using a one-time pad (OTP)~\cite{Sasaki2011}.  In addition, a key management layer was included to control and coordinate key consumption.  To date, China has constructed the largest QKD network, spanning over 2,000 km and linking cities from Beijing to Shanghai using trusted relays~\cite{Chen2021sat}. Furthermore, in 2016 the Micius quantum science satellite was launched. Micius has facilitated quantum key distribution between various locations in China~\cite{Chen2021sat} and Europe and enabled real-time encrypted video calls between Beijing and Vienna~\cite{Lu2022}. The Cambridge quantum network achieved $\sim 2$Mbps key rates coexistent with 100Gbps data traffic over metropolitan distances, and used link redundancy to mitigate denial of service risk.  Another key advancement is the realization of a quantum network consisting of stationary nodes of solid-state spin qubits coupled by photons, and subsequent demonstration of post-selection free protocols in a scaleable architecture \cite{pompili21}, also leading to definition of a quantum network stack ~\cite{dahlberg2019,pompili2022}. A thorough overview of QKD networks implemented to date can be found in Ref.~\cite{Cao2022Review}.

In terms of DI approaches, MDI-QKD has been field tested in a  metropolitan network, where three users in a star configuration could communicate with each other through a central untrusted relay~\cite{tangpan16}.  MDI-QKD systems have been realized  in fiber optical links together with classical IP network signals~\cite{Valivarthi2019,Berrevoets2022}.  Recently, free-space satellite links and fiber optic channels have been integrated for MDI-QKD \cite{Li2023MDIQKD},  showing improvement to background noise when compared with BB84, and which could greatly increase transmission distances.  There has also been considerable progress in CV MDI-QKD \cite{Hajomer2023}.

Though these field-tested QKD networks and other important advances have propelled QKD towards real-world use, there are still many challenges to be faced before QKD (of any type) can be fully integrated into the existing cybersecurity infrastructure.   
Despite the progress mentioned above, many national security organizations and regulatory agencies worldwide still do not classify QKD as a viable replacement for key distribution based on public-key cryptosystems. To protect against the threat of quantum computation,  post-quantum cryptography is currently seen as a more cost effective and robust solution to key agreement~\cite{nsawebsite,anssiwebsite,ncscwebsite}.  In one form or another, the critical issues most often cited are:  
\begin{enumerate}
\item Implementation security - specialty hardware and implementation particularities (laser pulses are not single photons, detectors can be vulnerable to side-channel attacks) can introduce additional vulnerabilities that may not be considered into theoretical security proofs.
\item Authentication -  Unlike public key cryptosystems, QKD does not provide a method for authentication protocols, which are widely used for handshaking, signatures, etc.  Moreover, the security of QKD relies on an authenticated classical channel for post-processing.
\item Trusted Relays - without quantum repeaters, extending QKD to large distances requires intermediate trusted relays, where security depends upon a classical hardware device. 
\item Denial of service risk - If a private key cannot be established, the QKD protocol aborts, opening the door to denial-of-service attacks in which the communication channel is shut down.  
\item Cost - special-purpose equipment is required for QKD.  These devices, such as single-photon detectors, are typically expensive, raising costs of installation, operation and maintenance. 
\item Compatibility -  QKD needs to coexist and integrate with  classical encryption systems and networks, which is complex due to the different operational frameworks of quantum and classical cryptography.
\end{enumerate}
Here we provide a brief description of how these issues are currently being tackled by QKD community.
\paragraph{Implementation Security} 
Quantum cryptographic protocols can be shown to be information-theoretic secure in principle.  However, practical implementations can open the door to a wide range of vulnerabilities that might not be considered in security proofs~\cite{SCARANI2014,etsiwp27}. Thus, implementation security is of paramount importance in taking QKD into the real world.  Of course QKD is not special in this regard, the same is true for all encryption techniques, which are based on security claims or assumptions that might not be valid upon implementation.  It is essential that all components of any cybersecurity system be extensively vetted and routinely tested.    

The effort to achieve implementation security in QKD has been two-fold.  On the one hand is the effort to remedy the practical issues with specific solutions, either by including additional techniques to QKD protocols (as in the case of decoy-state QKD, for example) and/or characterization of the devices, or by adapting security proofs to include these practical details.  In addition, there has been effort to develop certification procedures for QKD equipment that can be carried out by third parties~\cite{Makarov2024}, as is done for conventional IT and security equipment.

On the other hand, the development of device-independent protocols can provide a more broad solution with its goal to achieve information-theoretic security with as few assumptions as possible. DI-QKD, MDI-QKD and SDI-QKD can solve many of the most important implementation issues.  While DI techniques allow for the main concepts behind security to be rooted to fundamental laws of physics, practical implementation will inevitably introduce new issues that may not have been considered.   These  need to be identified and scrutinized in order for QKD (DI or otherwise) to have widespread use.  This is one of the objectives of the ongoing standardization process of QKD systems (see next section).   

\paragraph{Authentication}
To prevent man-in-the-middle attacks, QKD requires two-way authentication of the classical channel between users for the classical post-processing stage of the QKD protocol (basis sifting, error correction, privacy amplification).  In small networks, pre-shared keys can be used. However, this severely limits the network, as not only do the keys need to be stored securely, but new users should be able to join without having previously established a key.  In classical communications,  the public-key infrastructure  (PKI) provides methods for authentication, which will soon include post-quantum cryptography (PQC).  Though it is not information-theoretic secure, PQC and crypto-agility is the current next step to be adopted in protecting public-key cryptosystems from quantum computing~\cite{NIST-Crypto-Agility}.  PQC has already been used to authenticate classical communications in QKD sessions in several network topologies~\cite{Wang2021PQCQKD,Yang2021PQCQKD}.  Importantly, since PQC is used only for authentication (key exchange for data encoding is realized with QKD), only short-term security is required.  That is, if the PQC method used is broken in the future, the encoded data is still safe. Thus, PQC+QKD can offer a more robust security paradigm.  
\paragraph{Trusted Relays}
As discussed above, currently trusted relays are required to construct QKD links over several hundred kilometers.  As the development of quantum repeaters evolves~\cite{azuma2023}, these classical relays can be exchanged for quantum relays, which will solve this issue.   In a quantum network architecture, distributing several keys over multiple paths incorporating different sets of nodes will improve security, should one or more trusted nodes become compromised.  Post-processing of the keys by the end users can reduce any leaked information.  In addition, MDI-QKD can be used to transform some of the trusted relay stations into untrusted ones.    
For DI-QKD at large-scale distances, however, quantum repeaters are indeed a requirement.

\paragraph{Denial of service} Since QKD involves sending a single quantum state over a channel, any interruption in transmission, such as simply cutting the optical fiber, or introducing high amounts of noise, will prevent key exchange.  This risk, known as denial of service, has been reduced in several proof-of-concept implementations by using quantum link redundancy, which takes advantage of the quantum network architecture to distribute keys over multiple paths~\cite{Dynes2019}. In addition,  hybrid approaches using QKD and PQC can also mitigate denial of service~\cite{Chaqra2024,zeng2024pract}.  
\paragraph{Cost} 
When evaluating the cost of cybersecurity, it should be compared to the cost of cybercrime, which worldwide is the equivalent third largest economy in the world ($\sim$ 9.5 trillion USD in 2023 and growing)~\cite{Bloomberg_cybercrime}.  In this regard, massive investment in cybersecurity is warranted, as exemplified by the US governments migration to post-quantum encryption, which is estimated to be 7.1 billion USD over ten years~\cite{usgovpqc}. 
Second, the last century has shown that the evolution of technology typically leads to cheaper and better devices, as is the case of the microchip, for example. In this regard, integrated photonics will inevitably bring not only miniaturization and improvement but also the cost-reduction of quantum photonic devices, as it has done for classical  equipment  (see Ref.~\cite{Pelucchi2022} for a recent review). It should also be noted that research and investment in quantum technology in general will accelerate development in quantum communications, since quantum photonic devices are widely used in most applications. In regards to QKD, manufacture of on-chip transmitters and receivers should facilitate the standardization and deployment of QKD systems.  While chip integration of sources and optical circuits is quite advanced, the current technological roadblock is the integration of on-chip single-photon detectors, which are currently at a proof-of-principle or development stage~\cite{Pelucchi2022}. 

Much progress has been made in integrating QKD into existing telecommunications infrastructure (see below), which will help decrease costs and requirement of special purpose equipment. In addition, QKD network architecture can be designed for cost reduction. For example, Hub and spoke~\cite{Tang2016MDIhub}  or multi-user~\cite{Zhong2022}  architectures with MDI-QKD or standard QKD~\cite{Frohlich2013}, incorporating a central detection station, reduce the need for multiple detectors, which should minimize infrastructure costs. 
Finally, we note that that quantum communication systems might also find use as dual-purpose devices. For example, as large-scale quantum sensors capable of vibrational sensing~\cite{Chen2022-658km}, which could also motivate investment, development and deployment. 

\paragraph{Compatibility} 
   For over thirty years, the data capacity of fiber-optical communications has increased by a factor of ten every four years.  The demand for increased capacity has not subsided, leading to even more optical intensity within the fibers as channel density increases.  Though many of the QKD links demonstrated in field-deployed systems employed dark fiber (no copropagating classical data channels) \cite{Elliot2005,Sasaki2011,Wang2014Imp}, QKD will most likely need to coexist with classical data transmission in the same telecommunications network infrastructure. Moreover, to avoid loss of capacity, a QKD channel should not occupy much more bandwidth than a classical one.  Multiplexing quantum and classical communication channels together is a considerable challenge, since noise from Raman scattering of light from the classical data channel, in which photons in the optical fiber are scattered inelastically, can easily contaminate a quantum signal.  We note that Raman noise is not such an issue for CV-QKD, since the local oscillator used in homodyne detection acts as a mode filter, thus eliminating a large part of the Raman background~\cite{Kumar2015,Eriksson2019}.  Several methods have been proposed to mitigate this problem for DV-QKD. One method to minimize Raman noise, demonstrated as early as 1997~\cite{Townsend1997}, is to employ a quantum signal with shorter wavelength, such as the telecommunications O-band ($\sim 1260-1360$nm),  with the classical channels in the C-band ($\sim 1530-1565$nm) or L-band ($\sim 1565-1625$nm).  In this way, the Raman noise is less prevalent.
This approach has allowed QKD with keyrates of 4.5 kbps and 5.1 kbps for O-band quantum signals co-propagating and counter-propagating with 3.6 Tbps C-band classical signal  ($\sim$21 dBm), over a 66km   commercial backbone network~\cite{Mao2018Imp}. An MDI-QKD session was realized in a deployed link of about 25km, resulting in a positive key rate with up to about 45dB of link loss, when the quantum signal (@$1310$nm) was multiplexed with classical telecommunications signals at 10 Gb/s (@ $1550$ nm) and 10Mb/s ($\sim 1510$ nm)~\cite{Berrevoets2022}. Recently, a quantum link sending one O-band photon of an entangled pair through 47km of fiber with 18dBm of classical signal (C-band) was demonstrated~\cite{Thomas2023}, also showing improved performance for wavelengths less than 1300nm. A similar setup was recently used for quantum teleportation coexistent with classical communications~\cite{Thomas2024}.
Despite the difficulties due to Raman noise,  wavelength-division multiplexing (WDM) has been used to implement  QKD in C-band channels coexisting with 100Gb/s the encrypted classical channel (C-band) in a metropolitan network~\cite{Dynes2019}.   
A possible way to minimize Raman noise is to use hollow-core fiber, which also reduces noise arising from nonlinear effects.  Noise reduction of roughly 35dBm compared to standard SMF28 fiber has been observed in QKD trials~\cite{Honz2023}. 

Another route for coexistence of classical and quantum signals is space division multiplexing (SDM), where multiple spatial modes are used as communication channels. SDM, involving new types of optical fiber, is currently seen as a necessary step to solve the current capacity crunch in optical fiber communications~\cite{richardson2013}. Multi-core optical fibers contain several single-mode cores within the same cladding material, and can be used to transmit independent classical and quantum signals~\cite{Dynes2016,Cai2019,Xavier2020}. Other types of specialty fibers, such as few-mode fibers and ring-core fibers, can support multiple transverse modes, which can each function as an independent channel. Propagation of quantum and classical signals in these fibers is currently being investigated for future communications infrastructure~\cite{Wang2020,Cao2022Review}.    
 
Concerning DI-QKD specifically, Raman noise presents a considerable obstacle for deployment in commercial fibers along with classical data channels. Quantum process tomography of the effect of Raman noise on DV QKD protocols has shown that it can be accurately described by a depolarizing channel for both co-propagating and counter-propagating signals~\cite{Chapman2023}, where the degree of depolarization is a function of fiber link length.  Since depolarization reduces and can destroy entanglement, it is most likely that DI-QKD will require a dedicated standard fiber, or one of the more advanced noise-reduction solutions involving specialty fiber mentioned above.
\par

The second observation that is typically made about the compatibility concerns integration and interoperability with existing cybersecurity hardware and protocols.  In this regard, QKD has already been integrated with various cybersecurity protocols, including IPSec, TLS, Kerberos, AES, etc, as briefly discussed above.  We will further discuss interoperability in the next section.

\subsubsection{Standardization and Interoperability}
In addition to the scientific and technical challenges of realizing QKD in real-world conditions, there is also a need for coordinated effort towards standardization and interoperability to enable the integration of QKD into practical security services~\cite{Deventer2022}.  As QKD is an evolving technology, there are additional challenges ranging from immediate concerns, such as ensuring the security and interoperability of trusted relay-based QKD networks to medium- and long-term considerations like the large-scale integration of quantum and classical telecommunications networks, expanding the applications of QKD, and scaling up the network using quantum repeaters.  Moreover, the global deployment of QKD may employ multiple types of links (fiber, free space) depending on the type of network and application~\cite{wang2022deploy}. 

Cryptographic hardware and software in use today has been developed under a set of industry standards that help maintain a consistent and high level of security across different systems and networks.  This involves defining industry-wide guidelines, best practices, compliance and regulatory criteria, as well as interoperability parameters.  
In the US, standards for IT equipment are produced by National Institute of Standards and Technology (NIST) as Federal Information Processing Standards (FIPS) and approved by the Secretary of Commerce.  In Europe, the International Organization for Standardization (ISO) developed the  Common Criteria standard (ISO/EN 15408) (http://www.commoncriteriaportal.org/index.cfm). These standards provide a mechanism for certification of IT and security devices. 

QKD equipment, protocols and methods need to be standardized, so that they can be certified for use by government agencies and/or third parties.    Standardization should be realized with QKD protocols and security proofs that closely match the real-world conditions of the QKD implementation.  While almost all QKD systems in operation today implement some form of device-dependent prepare and measure QKD, these efforts are equally important to the future deployment of DI-QKD in that they will accelerate its adoption as the relevant technology comes to maturity, since many of the technical issues related to integration, interoperability and standardization will have already been solved at least partially.   Government security agencies have noted the need for standardization of QKD before the technology can considered for adoption on a broad scale~\cite{nsawebsite,anssiwebsite,ncscwebsite}.   This includes developing protocols for connectivity and interoperability, so that QKD systems can be linked with cryptographic key management systems and the application layer.  These standards not only ensure quality and security, but ensure that equipment from different future vendors can interoperate together, and are important to establish an industry supply chain by defining interfaces and technical specifications for components and modules in various systems or distributed networks.  

The successful deployment of QKD testbeds and proof-of-concept integration with cybersecurity hardware  demonstrated that QKD technology and networking was sufficiently advanced  for the standardization process to begin. 
In 2008, the European Telecommunications Standards Institute (ETSI) created the industry specification group for QKD (ISG-QKD)~\cite{Langer2009}, which has produced recommendations regarding QKD architecture, use cases, certification, security proofs and assurances, integration into standard optical networks, interoperability and interfacing, among other topics~\cite{ETSI-QKD}.  Notably,  in 2023 a Common Criteria Protection Profile for QKD was recently published (GS QKD 016),  which ``will help manufacturers to submit pairs of ‘prepare and measure’ QKD modules for evaluation under a security certification process. Such modules can be used by telecom operators and enterprises in securing their networks with the knowledge that certified products have been subjected to the scrutiny of a formal security evaluation process"~\cite{ETSI-QKD-ActivityReport-2023}. The International Organization for Standardization (ISO) and the International Electrotechnical Commission (IEC) have developed the ISO/IEC 23837 series,  which specifies security requirements, testing procedures, and evaluation methods for quantum key distribution (QKD) modules~\cite{ISO-Standard-23837, ISO-Standard-23838}. This series is structured under the framework of the ISO/IEC 15408 series, commonly referred to as the Common Criteria for Information Technology Security Evaluation. By establishing rigorous standards and assessment methods, the objectives of the ISO/IEC 23837 series are the security and reliability of QKD technologies for practical and secure implementations. These standards were developed under subcommittee 27 of joint technical committee (JTC) 1  (Information security, cybersecurity and privacy protection)~\cite{ISO-JTC1}. In 2024, a ISO/IEC JTC on quantum technologies was established~\cite{ISO-JT3}.
The 
International Telecommunications Union (ITU) 
have also published documents containing definitions and recommendations in the ITU-T Y.3800 series (quantum communication) and ITU-T X.1700 series (QKD networks). An overview of standardization processes and documents can be found in Ref.~\cite{Cao2022Review}, and on the organization websites~\cite{ETSI-QKD,ISO-JTC1,ISO-JT3,ITU-Y-website}.  The certification of MDI-QKD devices has been studied in Ref.~\cite{Garban2023}, in which it is noted that similarities between these findings and ETSI GS QKD 016 suggest that a generic framework could be created to permit certification of various implementations and protocols, including MDI-QKD.  

\subsubsection{Quantum key distribution network architecture}
\begin{figure}
    \centering
=\includegraphics[width=0.8\linewidth]{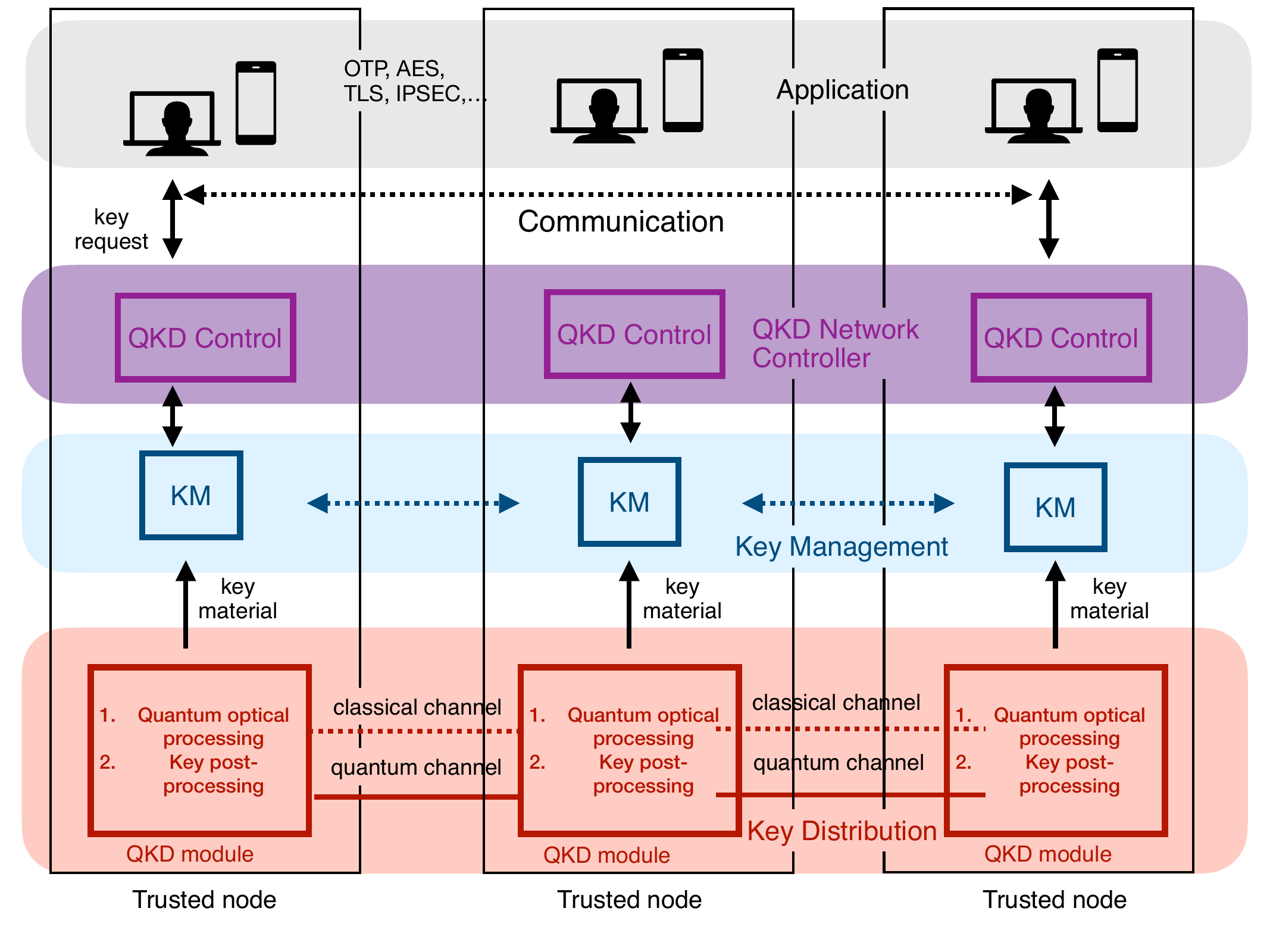}
\caption{\label{fig:qkdplat} Illustration of integrated QKD network}
\end{figure}
Part of the challenge of implementing real-world QKD is determining how this new quantum layer will integrate into the existing cybersecurity infrastructure.  The field tests realized over the last two decades have been important in accelerating this integration.   
 A number of authors have discussed network layer architectures for QKD systems~\cite{Sasaki2011,Sasaki2017,Tysowski2018,Broberg2022,Mehic2020Review,Cao2022Review}, and similar models have appeared in technical recommendations by international agencies, such as the ITU (documents Y.3800-Y.3805)~\cite{ITU-Y-website}. 
 
Fig.~\ref{fig:qkdplat} shows a simplified illustration of a QKD network architecture containing only three users, similar to the model put forth in Recommendation ITU-T Y.3800~\cite{ITU-Y3800}.  The network consists of a QKD layer, a key management layer, a QKD control layer and the application layer.  The users reside in the application layer, which contains all hardware devices and software that will consume cryptographic keys, for use in protocols such as TLS for secure connection to web servers, IPSec for VPN connections, AES for data encryption, etc.   

In this simplified example, each node is responsible for generating and managing keys between users in their local network and users in the local networks of other nodes.  Depending on the network architecture, each node could serve as an end point as well as a trusted node used for linking other end points. The raw key material is generated by the QKD modules residing in two connected nodes, which are linked by quantum and classical channels used for point-to-point QKD sessions.  In the case of DI-QKD, the quantum channel would consist of entangled states shared across the link, and the trusted node would  employ a repeater station to connect the two neighboring links.  In the near term, the trusted nodes are the classical trusted relay nodes described above.  We note that this architecture permits the key distribution layer to be constructed from different types of QKD systems or protocols, or to employ parallel QKD links between nodes to increase the  key rate and reduce denial of service risk.  In addition, in a network architecture, two users might be linked through different sets of intermediary nodes to the same effect. 

 Through the QKD protocol(s) used, cryptographic key material in the form of shared random bit strings is produced between linked QKD modules and uploaded to the key managers, which store it for future use.  
 When end users need to be connected, the key managers at the intermediate trusted nodes perform the necessary post processing to produce shared keys between the users. Upon request, the key managers at the endpoints can format the keys and deliver them to the security application that will use them.  Key managers at different trusted nodes must communicate to synchronize the key request and delivery, to assure that two end users can communicate with the same key.  
 
 The QKD control layer manages the end-to-end connectivity from one user to the another through the required trusted nodes, so that the middle nodes perform the appropriate processing to enable the link between end users. 
The QKD controllers are responsible for routing control for key relays, management of QKD and KM links, session control for QKD services, authentication and authorization, as well as ensuring quality of service.  The QKD control layer might also employ a centralized architecture. In addition to the layers shown in Fig.~\ref{fig:qkdplat}, management layers (not shown) monitor the entire stack and ensure quality of service and that the required functionality is met.  

As technology progresses, the QKD layer can be upgraded from device-dependent QKD to semi-DI and eventually full-DI.  A roadmap for development of QKD architecture and rollout in the EU is shown in Fig. \ref{fig:euroadmap}, where we note that device independence is included as a key benchmark.   As quantum repeaters come on line, quantum connectivity between end users would be managed by the QKD control layer, which would allow for the realization of full DI-QKD in principle.  Conceivably, the network could employ several types of QKD protocols, depending on the security profile of each user group, and the types of available hardware in each section of the network. 
\begin{figure}
    \centering
    \includegraphics[width=\linewidth]{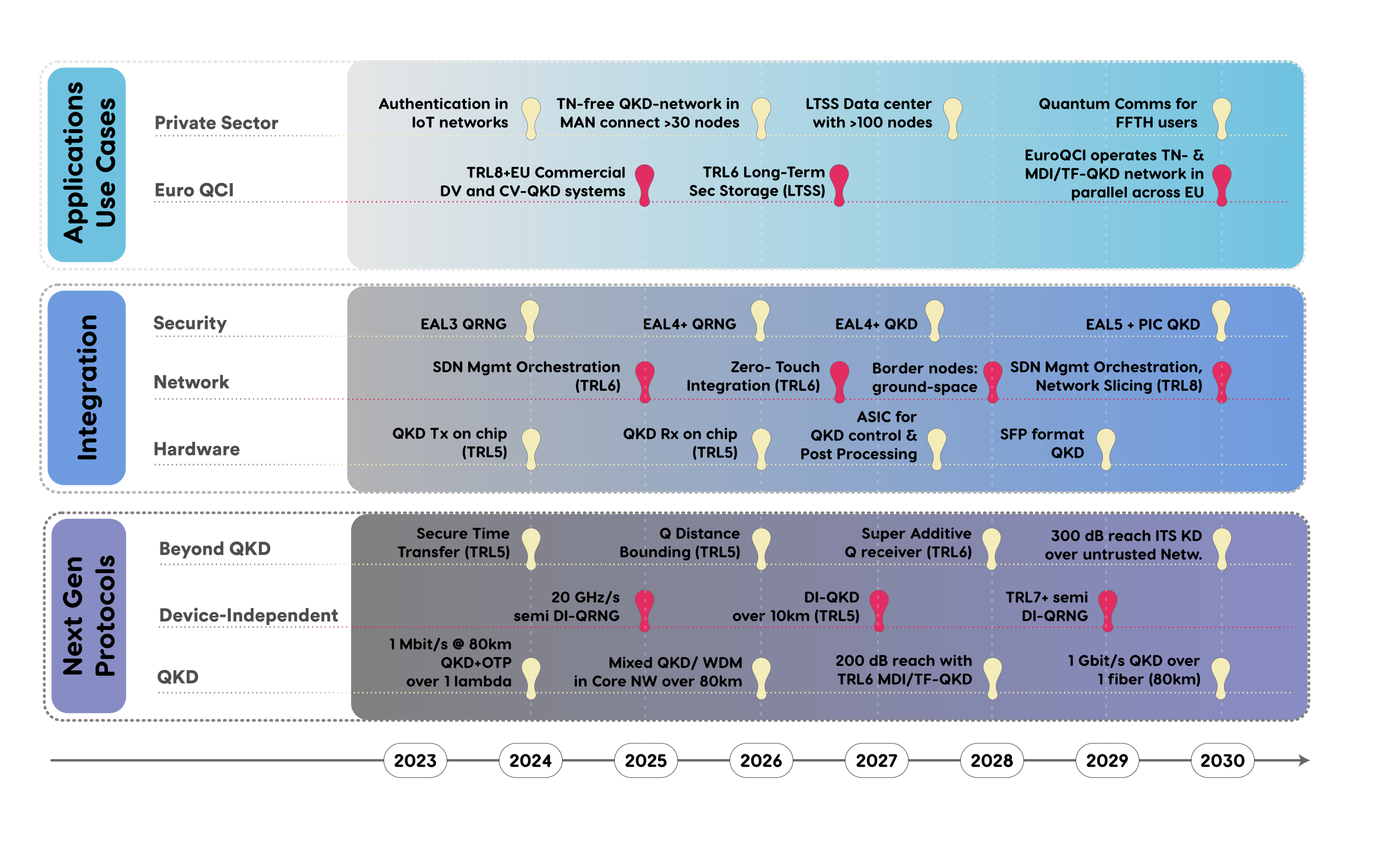}
    \caption{European roadmap for the QKD architecture deployment (from \href{qnsp}{https://qsnp.eu/}).}
    \label{fig:euroadmap}
\end{figure}
\section*{Conclusion}
DI-QKD represents a transformative advancement in QKD, addressing fundamental security challenges by using nonlocal correlations rather than relying on the trustworthiness of quantum devices. In this review, we have highlighted both the fundamental theoretical aspects and the progress in implementing experimental setups. Moreover, the growing exploration of semi-device-independent protocols such as  MDI-QKD,  RDI-QKD, and 1SDI-QKD have been presented. 

DI-QKD achieves its robust security through the violation of Bell inequalities, ensuring that any eavesdropping attempts disturb the nonlocal correlations, thereby making such third parties detectable. While the first successful implementations of DI-QKD marked a milestone by addressing all Bell test loopholes, practical challenges related to scalability and technology readiness persist. Current experimental realizations have achieved limited distances of a few hundred meters with low key rates, far short of the scales required for widespread commercial deployment.

Despite the remaining challenges in practical deployment, DI-QKD is poised to redefine the future of cryptographic security. The ongoing researches on DI-QKD protocols together with relaxed versions of semi-device-independent frameworks are paving the way for this groundbreaking technology to transition from the laboratory to real-world applications, ensuring unconditional security for the next generation of quantum communication networks.  A view as to how the rollout of DI technique might unfold is provided in the European roadmap shown in Fig. \ref{fig:euroadmap}.

\section*{Acknowledgments}
We thank C. Lupo, H. Weinfurter, M. \.Zukowski, G. Lima, E. G\'omez, N. Gigena, R. Wolf, M. Ziman, T. Vértesi, Yu-Bo Sheng, Q. Zhang, and A. Ac\' in for the fruitful discussions.
This project has received funding from QuantERA/2/2020, an ERA-Net co-fund in Quantum Technologies, under the eDICT project, the European Union’s Horizon Europe research and innovation program under the project "Quantum Secure Networks Partnership" (QSNP, grant agreement No 101114043), and INFN through the project "QUANTUM". SAG was supported by the Stefan Schwarz Support Fund and project QENTAPP (09103-03-V04-00777). SPW and LLT were supported by Fondo Nacional de Desarrollo Científico y Tecnológico (ANID) (Grants 1200266, 1240746), ANID – Millennium Science Initiative Program – ICN17\_012 and project
 UCO 1866.



%

\end{document}